\documentclass[11pt,letterpaper]{article}
\usepackage[hmargin=0.8in,vmargin=1in]{geometry}
\usepackage[T1]{fontenc}
\usepackage[utf8]{inputenc}
\usepackage{newtxtext}
\usepackage[protrusion=true,expansion=true]{microtype}
\usepackage{amsthm}
\usepackage{rotating}
\usepackage{etoolbox}
\usepackage{float}
\usepackage{placeins}
\usepackage{caption}
\usepackage{mathtools}
\usepackage{amssymb}
\usepackage{newtxmath}
\let\leq\leqslant
\let\le\leqslant
\let\geq\geqslant
\let\ge\geqslant
\usepackage{bm}
\usepackage{graphicx}
\usepackage{booktabs}
\usepackage{array}
\usepackage{multirow}
\usepackage{tabularx}
\usepackage{longtable}
\usepackage{threeparttable}
\usepackage{siunitx}
\usepackage{pdflscape}

\graphicspath{{images/}}
\newcounter{CloudContinuedFloat}
\AtBeginDocument{%
  \let\CloudOriginalContinuedFloat\ContinuedFloat
  \renewcommand{\ContinuedFloat}{%
    \stepcounter{CloudContinuedFloat}\CloudOriginalContinuedFloat}%
  \renewcommand{\theHtable}{\arabic{table}.\arabic{CloudContinuedFloat}}%
}

\newcommand{\CloudStatTableStyle}{%
  \small%
  \setlength{\tabcolsep}{4pt}%
  \renewcommand{\arraystretch}{1.10}%
}

\usepackage{xcolor}
\usepackage[authoryear,round]{natbib}
\usepackage{hyperref}

\newcommand{\CloudTitlePrefix}{%
  The Continuous Latent Ornstein--Uhlenbeck Dynamics Framework:%
}
\newcommand{\CloudTitleSuffix}{%
  A Lyapunov-based Multivariate Latent Process Model
  for Longitudinal Categorical Data%
}
\newcommand{\CloudTitle}{\CloudTitlePrefix\space\CloudTitleSuffix}
\newcommand{\CloudTitlePlain}{%
  The Continuous Latent Ornstein-Uhlenbeck Dynamics Framework:
  A Lyapunov-based Multivariate Latent Process Model for Longitudinal Categorical Data%
}

\newcommand{\CloudAuthorOne}{Zhennan Wu}
\newcommand{\CloudAuthorTwo}{Yijie Wang}
\newcommand{\CloudAuthorThree}{Xiaoqing Huang}
\newcommand{\CloudEmailOne}{zwu1@iu.edu}
\newcommand{\CloudEmailTwo}{yijwang@iu.edu}
\newcommand{\CloudEmailThree}{huanxi@iu.edu}
\newcommand{\CloudAuthorsPlain}{%
  \CloudAuthorOne, \CloudAuthorTwo, and \CloudAuthorThree%
}
\newcommand{\CloudAuthorsWithMarks}{%
  \CloudAuthorOne\textsuperscript{1},
  \CloudAuthorTwo\textsuperscript{1}, and
  \CloudAuthorThree\textsuperscript{2}%
}

\newcommand{\CloudDeptOne}{Department of Computer Science}
\newcommand{\CloudSchoolOne}{%
  Luddy School of Informatics, Computing, and Engineering%
}
\newcommand{\CloudInstitutionOne}{Indiana University Bloomington}
\newcommand{\CloudLocationOne}{Bloomington, Indiana, USA}
\newcommand{\CloudDeptTwo}{Department of Biostatistics and Health Data Science}
\newcommand{\CloudInstitutionTwo}{Indiana University School of Medicine}
\newcommand{\CloudLocationTwo}{Indianapolis, Indiana, USA}

\newcommand{\CloudKeywords}{%
  Bayesian inference; Factor analysis; Item response theory;
  Latent structural models; Longitudinal data analysis;
  Multivariate Ornstein--Uhlenbeck process%
}
\newcommand{\CloudKeywordsPlain}{%
  Bayesian inference; Factor analysis; Item response theory;
  Latent structural models; Longitudinal data analysis;
  Multivariate Ornstein--Uhlenbeck process%
}

\hypersetup{hidelinks,hypertexnames=false,
  pdfauthor={\CloudAuthorsPlain},pdftitle={\CloudTitlePlain},
  pdfkeywords={\CloudKeywordsPlain}}
\newtheorem{thrm}{Theorem}
\newtheorem{lem}{Lemma}
\newtheorem{prop}{Proposition}

\newtheorem{lemma}{Lemma}
\newtheorem{corollary}{Corollary}
\theoremstyle{definition}
\newtheorem{assumption}{Assumption}
\theoremstyle{remark}
\newtheorem{remark}{Remark}

\newcommand{\CloudSupplementAppendixName}{Appendix}
\title{\CloudTitle}
\author{\CloudAuthorsWithMarks\\[0.8em]
  \small\textsuperscript{1}\CloudDeptOne\\
  \small\CloudSchoolOne\\
  \small\CloudInstitutionOne, \CloudLocationOne\\[0.5em]
  \small\textsuperscript{2}\CloudDeptTwo\\
  \small\CloudInstitutionTwo, \CloudLocationTwo\\[0.5em]
  \small\texttt{\CloudEmailOne}; \texttt{\CloudEmailTwo}; \texttt{\CloudEmailThree}}
\date{}

\begin{document}
\maketitle
\begin{abstract}
  Longitudinal biomedical studies increasingly collect irregularly sampled, multivariate categorical measurements that are noisy manifestations of underlying continuous disease processes. 
These data pose several challenges: subjects exhibit heterogeneous progression patterns; multiple observed items reflect unobserved biological traits that may evolve interdependently; and observation schedules often have unequal intervals both within and across subjects. 
To address these challenges, we present the Continuous Latent Ornstein--Uhlenbeck Dynamics (CLOUD) framework for modeling complex disease trajectories from multivariate longitudinal categorical data. CLOUD links multivariate categorical observations to underlying latent functional domains and characterizes their coupled temporal evolution by integrating item response theory, factor analysis, and multivariate Ornstein--Uhlenbeck (OU) processes. Methodologically, we introduce a time-inhomogeneous OU process that incorporates covariate-dependent progression slopes into the moving mean function of the latent dynamics, allowing baseline biomarkers and clinical characteristics to modulate individual disease trajectories while preserving the analytical tractability of the OU process. We further propose a Lyapunov-based structural parameterization of the joint drift and stationary covariance that enforces valid dynamics by construction and enables efficient Markov chain Monte Carlo inference. We establish well-posedness and conditional identifiability. Under additional regularity conditions, we establish consistency and asymptotic properties of the model. Through simulation studies and an application to Parkinson's Progression Markers Initiative data, we demonstrate that CLOUD can characterize subject-specific disease evolution across multiple interacting functional domains.

\end{abstract}
\noindent\textbf{Keywords:} \CloudKeywords.

\section{Introduction}
\label{sec:intro}
Longitudinal biomedical studies provide insights into disease progression and changes in health outcomes through repeated measurements of biological and clinical markers.
However, extracting meaningful insights from these data requires addressing several inherent complexities, including unobserved and potentially interacting biological traits underlying noisy measurements; heterogeneous trajectories associated with different baseline characteristics; and sparse, unbalanced observations arising from irregular sampling and missing data.
Given these challenges, joint modeling frameworks have emerged as a powerful approach. Such frameworks typically comprise a measurement model and a dynamic model.
The measurement model maps high-dimensional observations onto a lower-dimensional continuous latent space to filter measurement noise and consolidate correlated clinical or molecular markers into interpretable disease domains. This projection is often achieved using factor analysis (FA) or nonnegative matrix factorization (NMF) in applications such as recurrent event modeling \citep{Chen2024DynamicFA}, gene expression trajectory analysis \citep{Cai2023DynamicFA}, and disease trajectory analysis \citep{hou2025longitudinal}. When the observations are categorical, item response theory (IRT) provides a probabilistic framework to link discrete observations to continuous latent traits \citep{de2013theory}.
The dynamic model characterizes how latent variables evolve while accounting for irregular observation intervals. Linear mixed-effects models (LMMs) and Gaussian processes (GPs) specify trajectories through mean and covariance functions; GP formulations can also represent cross-domain dependence and lead--lag relationships \citep{Zhong2025BayesianAO,Cai2023DynamicFA,Mu2024GaussianPF}. Complementary formulations include continuous-time hidden Markov models (CT-HMMs), which govern stochastic jumps among discrete latent regimes \citep{YanLeeLou2025}, and ordinary and stochastic differential equations (ODEs/SDEs), which describe continuously valued state evolution \citep{Zhao2025ManifoldconstrainedGP,zhou2025dynamic,Jiang2024APN,Alatkar2025ARTEMISIA}.
Beyond accommodating irregular observation times, many biomedical applications require the dynamic model to represent multiple latent domains jointly, allowing distinct domains to follow different temporal trajectories while accounting for dependence among them. Such multivariate structure can be scientifically important when related disease processes progress at different rates yet remain coupled. For example, multidimensional longitudinal IRT models have revealed distinct but correlated cognitive and functional progression in Alzheimer's disease \citep{Zhou2024AssessingTE}, while dependent GP models characterize correlations among interacting latent biological pathway trajectories \citep{Cai2023DynamicFA}. Recent latent factor models likewise explicitly account for dependence both between variables and across time in high-dimensional longitudinal data \citep{Lee2026ALV}.
Among these approaches, continuous-state SDEs are particularly attractive for latent disease processes because they naturally accommodate arbitrary observation intervals while representing stochastic temporal variation and interactions among continuously valued latent traits. A particularly tractable and interpretable member of this class is the multivariate Ornstein--Uhlenbeck (OU) process, which can be viewed simultaneously as a linear SDE and a continuous-time Gaussian Markov process. This structure provides analytically tractable transition distributions while allowing temporal persistence and cross-process dependence to vary continuously with the elapsed time between observations. These properties make the multivariate OU process an attractive tool for modeling nonmonotonic and potentially oscillatory coupled trajectories and interindividual variability \citep{Oravecz2009AHO,Oravecz2016BayesianDA}. For multivariate categorical longitudinal outcomes, \citet{Tran2020LatentOM} introduced the latent Ornstein--Uhlenbeck (LOU) model by combining an IRT measurement model with multivariate OU dynamics. This work highlights the importance of asymmetric dependence between latent variables and provides the closest statistical point of departure for our work.

Although OU-based models have shown considerable promise for modeling irregular longitudinal data, several important challenges remain. First, joint inference on the drift and stationary covariance while maintaining the OU validity conditions is an important computational consideration.
The explicit low-dimensional constraints in \citet{Tran2020LatentOM} and the related constrained formulation in \citet{Abbott2024ABJ} motivate an alternative that is straightforward to implement at larger fixed latent dimensions. General stable-matrix constructions already exist through Lyapunov geometry \citep{OharaKitamori1993,OharaAmari1994,GillisSharma2017}, and recent work develops structured parameterizations specifically for OU computation \citep{Monti2026}. Our aim is to incorporate a dimension-generic construction into an identified ordinal latent-variable model with unknown stationary covariance.
Second, connecting baseline biomarker profiles to latent progression slopes is important for multivariate categorical longitudinal models. Clinical studies associate disease progression rates with baseline biomarkers and clinical characteristics \citep{Huang2020LongitudinalBI,Benatar2024PrognosticCA}. Hierarchical OU models already accommodate individual differences and covariate associations in process parameters \citep{Oravecz2016BayesianDA}. Building on the categorical LOU and joint longitudinal-survival formulations \citep{Tran2020LatentOM,Abbott2024ABJ}, CLOUD specifically places baseline covariates in a linear moving mean for each latent domain, allowing covariate-dependent average progression while retaining coupled OU residual dynamics.
Furthermore, integrating latent OU dynamics with a measurement model introduces additional identifiability challenges. Because the latent dynamic process is not directly observed but instead inferred from noisy measurements, the scale, orientation, and temporal dependence of the latent process may become confounded with the parameters of the measurement model. Careful treatment of these identifiability issues is therefore essential for recovering interpretable latent trajectories and obtaining reliable estimates of dynamic interactions.

To address these challenges, we propose Continuous Latent Ornstein--Uhlenbeck Dynamics (CLOUD), a unified probabilistic framework that combines a measurement model based on IRT and FA with a dynamic model based on multivariate OU processes. CLOUD makes three primary methodological contributions. First, we adapt a Lyapunov-based structural parameterization for joint inference on the OU drift and stationary covariance. This replaces dimension-specific drift constraints that must be derived manually for each dimension with a construction that is unchanged across fixed finite latent dimensions, while retaining asymmetric and oscillatory dynamics and enabling efficient Markov chain Monte Carlo (MCMC) inference. Second, extending the latent OU measurement--dynamics formulation of \citet{Tran2020LatentOM}, we specify a covariate-dependent moving mean function that allows baseline characteristics and biomarkers to modulate individual disease trajectories, thereby capturing subject-specific progression dynamics within a continuous-time latent process. Third, we establish the theoretical properties of the proposed framework by proving the well-posedness of the time-inhomogeneous OU process, the exactness of the parameterization for the identified admissible drift--covariance pair space, and conditional identifiability of its parameters under explicit ordinal-measurement and longitudinal-design assumptions. Under additional regularity conditions, we further establish consistency and asymptotic normality for the identified finite-dimensional population parameter and smooth dynamic functionals.

The remainder of this paper is organized as follows. Section~\ref{sec:model} introduces the CLOUD model together with the corresponding theoretical guarantees. Section~\ref{sec:simu} evaluates the proposed method through simulation studies. Section~\ref{sec:ppmi} demonstrates its practical utility using data from the Parkinson's Progression Markers Initiative (PPMI). Finally, Section~\ref{sec:dis} concludes with a discussion of the main findings, limitations, and directions for future research.

\section{Methods}
\label{sec:model}
\subsection{Model specification}\label{sec:model_spec}
For a longitudinal dataset with subjects indexed by $i=1,\ldots,N$, let $Y_{ijk}$ denote the $k$th of $K$ categorical responses observed at time $t_{ij}$ ($j=1,\ldots,n_i$, where $n_i$ is the number of visits). Let $\boldsymbol\xi_i(t)\in\mathbb R^R$ denote the $R$-dimensional state vector of the $R$ underlying latent variables. We distinguish visit-level measurement covariates $\mathbf x_{ij}^{(1)}\in\mathbb R^p$ from baseline covariates $\mathbf x_i^{(2)}\in\mathbb R^q$ governing progression. Following cumulative-link ordinal regression \citep{McCullagh1980Ordinal} and its latent OU formulation in \citet{Tran2020LatentOM}, we model the relationship between the observed data and the latent state through the measurement model
\begin{equation}
 h\{\Pr(Y_{ijk}\leq m)\}=\theta_{km}-\boldsymbol\Lambda_k^\top\boldsymbol\xi_i(t_{ij})-\boldsymbol\beta_k^\top\mathbf x_{ij}^{(1)}-b_{ik},
 \label{eq:irt}
\end{equation}
where $Y_{ijk}\in\{0,\ldots,c_k-1\}$ and $m=0,\ldots,c_k-2$ indexes the ordered cutpoints; $h$ is a fixed link (logit in our analyses); $\boldsymbol\Lambda_k$ is the transpose of row $k$ of the $K\times R$ loading matrix; and $b_{ik}\sim\mathcal N(0,\sigma_{bk}^2)$ is an item-specific subject random intercept. The random intercepts are independent of one another and of the latent process. The measurement covariates enter directly through the fixed effects $\boldsymbol\beta_k$.

We model the temporal dynamics of the observed data through a multivariate OU process for the latent states. This process comprises a covariate-dependent moving mean $\boldsymbol\mu_i(t)$ and centered deviations $\mathbf r_i(t)=\boldsymbol\xi_i(t)-\boldsymbol\mu_i(t)$. Over an arbitrary observation interval $\Delta t>0$, the latent states have the following conditional Gaussian transition distribution:
% \begin{align}
% \boldsymbol\xi_i(t+\Delta t)\mid\boldsymbol\xi_i(t),\mathbf x_i^{(2)}
% &\sim\mathcal N\bigl(\mathbf m_i(t,\Delta t),\mathbf Q(\Delta t)\bigr),\label{eq:ou}\\
% \mathbf m_i(t,\Delta t)
% &=\boldsymbol\mu_i(t+\Delta t)+e^{-\boldsymbol\Gamma\Delta t}r_i(t),\notag\\
% \mathbf Q(\Delta t)
% &=\boldsymbol\Omega-e^{-\boldsymbol\Gamma\Delta t}\boldsymbol\Omega e^{-\boldsymbol\Gamma^\top\Delta t},\notag\\
% \boldsymbol\mu_i(t)&=(\boldsymbol\alpha+\boldsymbol\Phi\mathbf x_i^{(2)})t.\label{eq:shift}
% \end{align}
\begin{align}
\boldsymbol\xi_i(t+\Delta t)\mid\boldsymbol\xi_i(t),\mathbf x_i^{(2)}
&\sim\mathcal N\bigl(\mathbf m_i(t,\Delta t),\mathbf Q(\Delta t)\bigr),\label{eq:ou}\\
\mathbf m_i(t,\Delta t)
=\boldsymbol\mu_i(t+\Delta t)+e^{-\boldsymbol\Gamma\Delta t}\mathbf r_i(t), &\quad 
\mathbf Q(\Delta t)
=\boldsymbol\Omega-e^{-\boldsymbol\Gamma\Delta t}\boldsymbol\Omega e^{-\boldsymbol\Gamma^\top\Delta t}\notag
\end{align}
Here $\boldsymbol\Gamma$ governs mean reversion and conditional cross-domain transitions of $\mathbf r_i$, while $\boldsymbol\Omega$ is its stationary covariance. Admissible dynamics require symmetric positive-definite (SPD) stationary and diffusion covariance matrices and a positive-stable drift matrix, meaning that every drift eigenvalue has positive real part:
\begin{align}
&\boldsymbol\Omega\succ0,\label{req:sta_cov}\\
&\operatorname{Re}\{\lambda(\boldsymbol\Gamma)\}>0\ \text{for every eigenvalue},\label{req:ps}\\
&\boldsymbol\Sigma:=\boldsymbol\Gamma\boldsymbol\Omega+\boldsymbol\Omega\boldsymbol\Gamma^\top\succ0.\label{req:inf_cov}
\end{align}
The matrix $\boldsymbol\Sigma$ is the infinitesimal diffusion covariance, whereas $\mathbf Q(\Delta t)$ is the finite-interval innovation covariance. With $\mathbf r_i(0)\sim\mathcal N(\mathbf0,\boldsymbol\Omega)$, the centered deviation $\mathbf r_i$ is stationary though the latent state $\boldsymbol\xi_i$ has a time-varying mean (Appendix~\ref{appd:tiou_proof}). We choose a moving mean that is linear in time, $\boldsymbol\mu_i(t) = \mathbf{v}_it$, where $\mathbf{v}_i$ is an affine function of the baseline covariates: $\mathbf v_i=\boldsymbol\alpha+\boldsymbol\Phi\mathbf x_i^{(2)}$. Here, $\boldsymbol\alpha$ is the population slope, and $\boldsymbol\Phi$ represents the baseline-covariate contrasts in latent progression slopes. The equivalent SDE is $d\boldsymbol\xi_i(t)=\{\mathbf v_i-\boldsymbol\Gamma[\boldsymbol\xi_i(t)-\mathbf v_it]\}dt+\mathbf B\,d\mathbf W_i(t)$, where $\mathbf W_i$ is an $R$-dimensional standard Brownian motion and $\mathbf B\mathbf B^\top=\boldsymbol\Sigma$. Thus, baseline covariates determine the average progression rate, while the drift matrix $\boldsymbol\Gamma$ determines interactions among the latent variables. The transition matrix $\mathbf T(\tau)=e^{-\boldsymbol\Gamma\tau}$ propagates effects among latent variables over an arbitrary time span $\tau$.

Although the moving-mean CLOUD process is generally nonstationary, its centered deviation $\mathbf r_i(t)$ remains a standard OU process initialized at stationarity and satisfying the classical Lyapunov condition. Lemma~\ref{lemma:conditional_kernel_stability} in Appendix~\ref{appd:tiou_proof} formalizes the stability property under conditions~\ref{req:sta_cov}--\ref{req:inf_cov}.
In particular, we assume that $\boldsymbol\Omega$ is positive definite, rather than merely positive semidefinite, to rule out perfect linear dependence among latent variables and prevent the latent process from collapsing into a lower-dimensional subspace.

% This formulation extends the model of~\citet{Tran2020LatentOM}. By modeling the mean trajectory $\boldsymbol{\mu}_i(t)$ as a function of time $t$ and dynamic-effect covariates $\mathbf{x}_{i}^{(2)}$, it captures the heterogeneity in progressive changes influenced by baseline subject characteristics, providing a representation of observed covariate-associated differences in disease progression.
% % 
\subsection{Parameter estimation}\label{sec:param_es}

Let $\boldsymbol\Psi
=
\left\{
    \{\{\theta_{km}\}_{m=0}^{c_k-2},
    \boldsymbol\Lambda_k,
    \boldsymbol \beta_k,
    \boldsymbol\sigma_{bk}\}_{k=1}^K,
    \boldsymbol\Gamma,
    \boldsymbol\Omega,
    \boldsymbol\Phi,
    \boldsymbol\alpha
\right\}$
denote the collection of model parameters described in Equations~\ref{eq:irt}--\ref{req:inf_cov}. Because the likelihood (provided in Appendix~\ref{appd:likelihood}) does not admit a closed-form expression, we perform Bayesian inference using the No-U-Turn Sampler (NUTS) implemented in Stan \citep{hoffman2014no,carpenter2017stan} to estimate $\boldsymbol\Psi$.

Joint inference on latent dependence and stochastic variability poses a computational and statistical challenge. Specifically, the matrices $\boldsymbol{\Omega}$ and $\boldsymbol{\Gamma}$ in $\boldsymbol{\Psi}$ must be parameterized to satisfy conditions~\ref{req:sta_cov}--\ref{req:inf_cov} without imposing additional restrictions that could introduce artifacts into model fitting, while preserving computational feasibility in general multidimensional latent spaces. We address this challenge by proposing the following reparameterization.

Using the Lyapunov-based stable-matrix construction \citep{OharaKitamori1993,OharaAmari1994,GillisSharma2017}, we parameterize the drift matrix $\boldsymbol{\Gamma}$ as
\begin{align}
\boldsymbol{\Gamma}
=
(\mathbf{S}+\mathbf{A})\boldsymbol{\Omega}^{-1},
\qquad
\mathbf{S}=\mathbf{S}^{\top}\succ0,
\qquad
\mathbf{A}^{\top}=-\mathbf{A},
\label{eq:gamma_param}
\end{align}
where $\mathbf{S}$ is an SPD matrix, and $\mathbf{A}$ is a skew-symmetric (SS) matrix. We further express $\mathbf{S}$ through its Cholesky decomposition,
$\mathbf{S}=\mathbf{L}_S\mathbf{L}_S^\top$,
where $\mathbf{L}_S$ is a lower-triangular matrix with strictly positive diagonal entries, and represent $\mathbf{A}$ by its strictly lower-triangular elements $\mathbf{L}_A$ through $\mathbf{A}=\mathbf{L}_A^{\top}-\mathbf{L}_A$. We also construct the OU covariance matrix $\boldsymbol{\Omega}$ via its Cholesky decomposition $\boldsymbol{\Omega} = \mathbf{L}_{\Omega}\mathbf{L}_{\Omega}^{\top}$. Consequently, both $\boldsymbol{\Gamma}$ and $\boldsymbol{\Omega}$ are fully parameterized by the triangular matrices
$\{\mathbf{L}_{\Omega},\mathbf{L}_{S},\mathbf{L}_{A}\}$.

The following theorem establishes that the representation in Equation~\ref{eq:gamma_param} is bijective.
\begin{thrm}[Generality of the CLOUD dynamic parameterization]\label{thm:main_generality}
    The parameterization in Equation~\ref{eq:gamma_param} represents exactly the class of valid drift--covariance pairs satisfying the OU conditions in Equations~\ref{req:sta_cov}--\ref{req:inf_cov}: every valid pair $(\boldsymbol\Gamma,\boldsymbol\Omega)$ has a unique representation through $(\mathbf{S},\mathbf{A},\boldsymbol\Omega)$, and every admissible $(\mathbf{S},\mathbf{A},\boldsymbol\Omega)$ produces a valid OU process. Thus, the proposed parameterization imposes no additional restriction on the latent interaction structure and retains all $R^2$ degrees of freedom of a general $R\times R$ drift matrix. Moreover, unlike the directly constrained pair $(\boldsymbol\Gamma,\boldsymbol\Omega)$, whose admissible values are coupled through Equation~\ref{req:inf_cov}, the structural components $\mathbf{S}$, $\mathbf{A}$, and $\boldsymbol\Omega$ have variation-independent parameter spaces.
\end{thrm}
The proof and a characterization of the identified admissible parameter space are in Appendix~\ref{appd:helm_proof}.

To illustrate the improvement in sampling efficiency that the proposed reparameterization could provide, we state the following proposition on direct constraints on the $\boldsymbol\Gamma$ matrix.
\begin{prop}[Dimension-dependent feasibility under direct parameterization]\label{prop:main_direct_feasibility}
    For independent, centrally symmetric row distributions of $\boldsymbol\Gamma$, the probability of satisfying $\boldsymbol\Gamma\boldsymbol\Omega+\boldsymbol\Omega\boldsymbol\Gamma^{\top}\succ0$, conditional on any fixed $\boldsymbol\Omega\succ0$, is at most $2^{-R}$.
\end{prop}
By contrast, Equation~\ref{eq:gamma_param} ensures $\boldsymbol\Gamma\boldsymbol\Omega+\boldsymbol\Omega\boldsymbol\Gamma^\top=2\mathbf{S}\succ0$ automatically. Proposition~\ref{prop:main_direct_feasibility} concerns naive centered direct parameterizations, not the feasibility of more sophisticated constrained methods; Appendix~\ref{appd:helm_proof} provides the proof and scope. This distinction explains a practical advantage over dimension-specific drift constraints such as those used by \citet{Tran2020LatentOM}, while building on the broader use of unconstrained parameterizations in statistical inference \citep{PinheiroBates1996}. The positive-definite/skew-symmetric construction itself is rooted in Lyapunov parameterizations of stable systems. For posterior computation, we also apply the noncentered parameterization (NCP) \citep{papaspiliopoulos2007general} to Equation~\ref{eq:ou}; details are provided in Appendix~\ref{appd:ncp}.

\subsection{Identifiability and asymptotic properties}\label{sec:iden}
The parameterization in Section~\ref{sec:param_es} ensures admissible latent dynamics. To interpret covariate-dependent progression and cross-domain transitions, the latent coordinate system must also be identified. Without further constraints, simultaneous nonsingular transformations of the latent states, loadings, and dynamic parameters leave the categorical-response distribution unchanged, producing $R^2$ degrees of indeterminacy. Thus, $R^2$ additional constraints are required. We introduce $R$ constraints by normalizing the stationary covariance to a correlation matrix to fix the latent scale:
\begin{equation}
 \boldsymbol\Omega=\mathbf D_{\mathrm{scale}}^{-1}\boldsymbol\Omega'\mathbf D_{\mathrm{scale}}^{-1},\quad
 \mathbf D_{\mathrm{scale}}=\operatorname{diag}(\sqrt{\Omega'_{11}},\ldots,\sqrt{\Omega'_{RR}}).\label{eq:scal_norm}
\end{equation}
We impose $R^2-R$ additional restrictions on the loading matrix to fix orientation and signs:
\begin{align}
 \boldsymbol\Lambda_{\mathcal A}=\operatorname{diag}(\lambda_{a_1 1},\ldots,\lambda_{a_R R}),\qquad
 \lambda_{a_r r}>0,\quad r=1,\ldots,R.\label{eq:struct}
\end{align}
Here, $\mathcal A=\{a_1,\ldots,a_R\}$ is a prespecified set of anchor items. Together, these restrictions define the scale, orientation, and signs of the latent dimensions while leaving non-anchor loadings unrestricted. Related indeterminacies occur in factor analysis and dynamic latent-variable models \citep{Chandra2023InferringCS,Chen2024DynamicFA,Cai2023DynamicFA,Lee2026ALV}. The formal population-identifiability result, including the required measurement and recording conditions, is given in Appendix~\ref{appd:id_proof}.

Under this identified formulation, the progression slopes $\mathbf v_i$ and finite-lag transition matrices $\mathbf T(\tau)=e^{-\boldsymbol\Gamma\tau}$ are well-defined functionals of the population parameters. We also establish their asymptotic properties for a fixed-dimensional model as the number of independent subjects increases, even when visits per subject remain bounded. Under the regularity conditions in Appendix~\ref{appd:inf_proof}, Theorem~\ref{thm:J2} establishes consistency and asymptotic normality of the marginal-likelihood estimator and an asymptotically Gaussian posterior for the population parameters. These limits extend to regular smooth progression and finite-lag transition functionals, providing a theoretical basis for large-sample uncertainty quantification in the CLOUD model.

\section{Simulations}
\label{sec:simu}
We conducted a simulation study to evaluate the performance of the proposed CLOUD framework systematically. Specifically, we assessed (1) the ability of the CLOUD parameterization to recover interactions among latent variables under varying latent dimensionality and dependence structures, including both low-dimensional settings where existing approaches are applicable \citep{Tran2020LatentOM} and higher-dimensional scenarios used to evaluate the dimension-generic construction; (2) the sampling-efficiency difference between the CLOUD parameterization and directly constrained estimation; (3) the recovery of the covariate-dependent latent moving mean function by generating subject-specific trajectories under baseline covariate effects; and (4) parameter recovery under realistic longitudinal sampling, missingness, and measurement error while satisfying the identifiability conditions established in Section~\ref{sec:iden}.

\subsection{Simulation setup}
We generated data under six scenarios (S1--S6), each with 100 independent datasets of $N=600$ subjects. S1--S2 used $R=2$ factors and $K=7$ items (3 binary, 4 ordinal); S3--S4 used $R=4$ and $K=12$ (5 binary, 7 ordinal); and S5--S6 used $R=8$ and $K=24$ (10 binary, 14 ordinal). The drift matrix $\boldsymbol\Gamma$ had real, positive eigenvalues in odd-numbered scenarios and complex eigenvalues with positive real parts in even-numbered scenarios. These settings allowed us to assess coupled nonoscillatory and oscillatory dynamics across dimensions. All scenarios used the same $q=2$ dynamic-effect covariates $\mathbf{x}^{(2)}$; dimension-specific $\boldsymbol\Phi$ and $\boldsymbol\alpha$ encoded positive and negative progression effects, yielding subject-specific latent trajectories $\boldsymbol\xi(t)$ with nonzero mean shifts $\boldsymbol\mu(t)$. For identifiability (Section~\ref{sec:iden}), each item loaded on a single factor, with all other loadings fixed to zero (simple-structure $\boldsymbol\Lambda$). Complete data-generating matrices and settings appear in Appendix~\ref{appd:simu_supp}.

Each dataset was fitted with CLOUD and three restricted versions: TrendyOU retains coupled dynamics but sets $\boldsymbol\Phi=\mathbf0$, giving a common slope $\boldsymbol\alpha$; StationaryOU additionally sets $\boldsymbol\alpha=\mathbf0$; and DiagOU retains covariate-dependent slopes but restricts $\boldsymbol\Gamma$ to be diagonal and sets $\boldsymbol\Omega=\mathbf I_R$, removing cross-factor dependence. Priors for drift, correlation, and measurement parameters appear in Supplementary Table~\ref{tab:priors}. Visits were irregularly scheduled, and post-baseline item missingness depended on measurement covariates and preceding responses, following \citet{Tran2020LatentOM}; see Appendix~\ref{appd:simu_supp} for the full mechanism.

We assessed relative bias (RB), mean squared error (MSE), 95\% posterior-interval coverage probability (CP), and effective sample size (ESS). Recovery fits used three chains with 2,000 warm-up and 2,000 sampling iterations each and a target acceptance probability of 0.95. Summaries include fits with maximum $\widehat R<1.10$ and no divergent transitions; retained counts differ by model.

\subsection{Simulation results}

% Complete parameter-level results are provided in Supplementary Tables~\ref{tab:s1_comparison_1}--\ref{tab:s4_comparison_1} and \ref{tab:r8_s5_full_recovery}--\ref{tab:r8_s6_full_recovery}.

To evaluate CLOUD's ability to recover interactions among latent variables, we compared it with DiagOU. CLOUD achieved mean drift-interval coverage of 92.2\%--100.0\%, compared with 0.0\%--15.6\% for DiagOU (Table~\ref{tab:simulation_summary}). Excluding interactions also biased within-domain drift estimates. For $\Gamma_{1,1}$, DiagOU yielded an RB of 13.645 and an MSE of 96.576 in the two-dimensional nonoscillatory scenario S1, versus $-0.059$ and 0.003 for CLOUD (Supplementary Table~\ref{tab:s1_comparison_1}); in the four-dimensional oscillatory scenario S4, the corresponding values were 4.635 and 8.780 versus 0.076 and 0.074 (Supplementary Table~\ref{tab:s4_comparison_1}). In the eight-dimensional oscillatory scenario S6, mean drift MSE was 6.098 for DiagOU versus 0.137 for CLOUD. Although DiagOU achieved slightly lower drift MSE (0.122 versus 0.134) and median absolute loading RB (0.033 versus 0.058) in the nonoscillatory scenario S5, its drift coverage was only 11.4\%, compared with 100.0\% for CLOUD (Table~\ref{tab:simulation_summary}; Supplementary Tables~\ref{tab:r8_s5_summary}--\ref{tab:r8_s6_summary}). These results indicate that ignoring latent interactions can distort both cross-domain and within-domain dynamics, whereas CLOUD maintains accurate drift recovery across latent dimensions.

\newcommand{\rEightNewRow}[7]{#1 & #2 & #3 & #4 & #5 & #6 & #7 \\}

\begin{table}[H]

\centering
\caption{Convergence-screened statistical recovery across six simulation scenarios.}
\label{tab:simulation_summary}
\small
\setlength{\tabcolsep}{4.4pt}
\renewcommand{\arraystretch}{1.10}
\begin{threeparttable}
\begin{tabular*}{\linewidth}{@{\extracolsep{\fill}}llccccc@{}}
\toprule
 & & & \multicolumn{2}{c}{Drift $\boldsymbol\Gamma$} & \multicolumn{2}{c}{Loadings $\boldsymbol\Lambda$} \\
\cmidrule(lr){4-5}\cmidrule(l){6-7}
Scenario & Model & Retained & Med. $|\mathrm{RB}|$ & Mean CP (\%) & Med. $|\mathrm{RB}|$ & Mean CP (\%) \\
\midrule
S1 ($R=2$) & \textbf{CLOUD} & 90/100 & 0.031 & 92.2 & 0.015 & 93.9 \\
 & TrendyOU & 100/100 & 0.796 & 25.0 & 0.586 & 20.0 \\
 & StationaryOU & 92/100 & 0.846 & 0.0 & 0.779 & 0.6 \\
 & DiagOU & 100/100 & 6.897 & 0.0 & 0.119 & 77.0 \\
\addlinespace[2pt]
S2 ($R=2$) & \textbf{CLOUD} & 92/100 & 0.010 & 95.7 & 0.021 & 93.9 \\
 & TrendyOU & 96/100 & 0.377 & 49.0 & 0.586 & 28.6 \\
 & StationaryOU & 82/100 & 0.852 & 0.0 & 0.868 & 0.2 \\
 & DiagOU & 100/100 & 1.000 & 0.0 & 0.081 & 82.1 \\
\addlinespace[2pt]
S3 ($R=4$) & \textbf{CLOUD} & 100/100 & 0.027 & 94.7 & 0.026 & 91.7 \\
 & TrendyOU & 100/100 & 0.248 & 60.8 & 0.121 & 61.0 \\
 & StationaryOU & 96/100 & 0.866 & 17.7 & 0.266 & 33.3 \\
 & DiagOU & 100/100 & 1.000 & 15.6 & 0.139 & 63.2 \\
\addlinespace[2pt]

S4 ($R=4$) & \textbf{CLOUD} & 84/100 & 0.021 & 99.3 & 0.042 & 93.5 \\
 & TrendyOU & 98/100 & 1.596 & 95.2 & 0.179 & 73.7 \\
 & StationaryOU & 88/100 & 0.685 & 42.5 & 0.637 & 8.3 \\
 & DiagOU & 99/100 & 1.000 & 13.4 & 0.130 & 70.9 \\
\addlinespace[2pt]
\rEightNewRow{S5 ($R=8$)}{\textbf{CLOUD}}{100/100}{0.059}{100.0}{0.058}{88.5}
\rEightNewRow{}{TrendyOU}{68/100}{1.498}{99.3}{0.178}{61.1}
\rEightNewRow{}{StationaryOU}{37/100}{1.522}{92.1}{0.612}{15.7}
\rEightNewRow{}{DiagOU}{100/100}{1.000}{11.4}{0.033}{94.8}
\rEightNewRow{S6 ($R=8$)}{\textbf{CLOUD}}{98/100}{0.053}{97.8}{0.034}{92.2}
\rEightNewRow{}{TrendyOU}{73/100}{0.561}{73.1}{0.169}{56.7}
\rEightNewRow{}{StationaryOU}{47/100}{0.725}{57.4}{0.623}{12.5}
\rEightNewRow{}{DiagOU}{97/100}{1.000}{0.8}{0.052}{92.6}
\addlinespace[2pt]
\bottomrule

\end{tabular*}
\begin{tablenotes}
\footnotesize
\item Each scenario contains 100 generated datasets. ``Retained'' is the number meeting the prespecified convergence criterion; all recovery summaries are conditional on these retained fits. S1 and S2 use $R=2$ latent factors; S3 and S4 use $R=4$. S5 and S6 use $R=8$. S1, S3, and S5 have real, positive drift eigenvalues; S2, S4, and S6 have complex drift eigenvalues with positive real parts.
\item RB is the per-parameter mean relative bias across retained replications. The reported median is taken over the absolute values of these per-parameter RBs. CP is the arithmetic mean of the per-parameter empirical 95\% credible-interval coverages. Relative bias excludes entries whose data-generating value is zero (where RB is undefined); mean CP includes these entries.
\item Fits completed for all 100 datasets under every method in S5--S6; the retained counts reflect the same prespecified convergence screen.
\item Drift summaries evaluate the same data-generating matrix entries in each scenario. For DiagOU, nonzero cross-domain drift entries constrained to zero are retained in the comparison (and hence have zero interval coverage). Full parameter-level results are in Supplementary Tables~\ref{tab:s1_comparison_1}--\ref{tab:s4_comparison_1}. Complementary grouped diagnostics for R=8 are in Supplementary Tables~\ref{tab:r8_s5_summary}--\ref{tab:r8_s6_summary}; complete S5--S6 parameter-level results are in Supplementary Tables~\ref{tab:r8_s5_full_recovery}--\ref{tab:r8_s6_full_recovery}.
\end{tablenotes}
\end{threeparttable}
\end{table}

To assess computational performance relative to directly constrained LOU inference, we reproduced the two-dimensional setting of \citet{Tran2020LatentOM}, comparing CLOUD with directly constrained LOU under matched stationary means, likelihoods, and priors. The convergence target required two consecutive checkpoints satisfying prespecified mixing, effective sample size, Monte Carlo precision, divergence, and maximum-tree-depth criteria (Appendix~\ref{appd:efficiency_comparison}). CLOUD reached the target in all 200 datasets, compared with 36/100 and 55/100 datasets for LOU in the oscillatory and nonoscillatory settings, respectively, within the computational allowance; two LOU fits in the nonoscillatory setting were censored by system limits (Table~\ref{tab:sim_computational_summary}). Among the 91 datasets in which both methods reached the target, LOU was slower in every paired comparison, with median LOU/CLOUD time ratios of 4.11 and 3.93 for the oscillatory and nonoscillatory settings, respectively (Supplementary Table~\ref{tab:outcomes}; Supplementary Figure~\ref{fig:common}A). Under fixed budgets of eight and 16 hours, CLOUD reached the target in 199/200 and 200/200 datasets, versus 1/200 and 45/200 for LOU (Figure~\ref{fig:budget}).

\begin{table}[H]

\centering
\caption{Target-reaching performance and effective sampling efficiency in the matched $R=2$ computational benchmark.}
\label{tab:sim_computational_summary}
\small
\setlength{\tabcolsep}{6pt}
\renewcommand{\arraystretch}{1.10}
\begin{threeparttable}
\begin{tabular*}{\linewidth}{@{\extracolsep{\fill}}llccc@{}}
\toprule
Drift regime & Model & Target met & Median target time (h) & Median min. bulk ESS/s \\
\midrule
Oscillatory & \textbf{CLOUD} & 100/100 & 4.22 & 0.0252 \\
 & LOU & 36/100 & 18.30 & 0.00297 \\
\addlinespace[2pt]
Nonoscillatory & \textbf{CLOUD} & 100/100 & 3.49 & 0.0338 \\
 & LOU & 55/100 & 12.71 & 0.00379 \\
\bottomrule

\end{tabular*}
\begin{tablenotes}
\footnotesize
\item Target met is the number of datasets for which the prespecified target was reached within the 90-hour allowance (out of 100 per drift regime and method). Median target time is computed only among fits that reached the target. Minimum bulk ESS/s is the median across all 100 checkpoint fits of the minimum bulk effective sample size divided by elapsed wall-clock seconds, at a fixed checkpoint after 1,000 post-warm-up draws; it is not conditioned on reaching the final target. The benchmark is separate from the S1--S6 parameter-recovery simulations. See Appendix~\ref{appd:efficiency_comparison}, Supplementary Tables~\ref{tab:outcomes}--\ref{tab:geometry}, and Figure~\ref{fig:budget} for definitions and diagnostics.
\end{tablenotes}
\end{threeparttable}
\end{table}

\begin{figure}[H]
\centering
\includegraphics[width=0.96\textwidth]{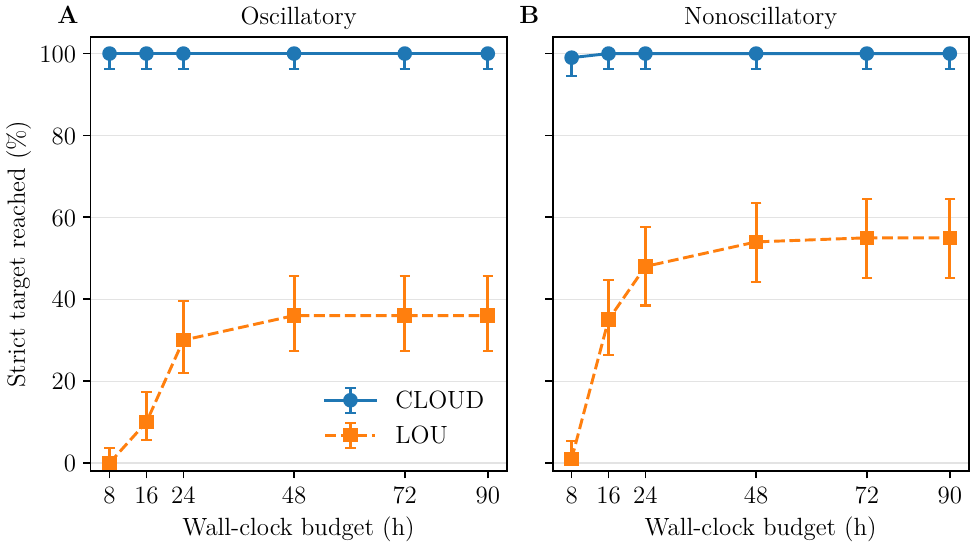}
\caption{\textbf{Computational reliability under matched low-dimensional benchmarks.} For each of the oscillatory and nonoscillatory $R=2$ scenarios, the figure displays the empirical proportion of 100 independently simulated datasets meeting the prespecified strict computational target by each wall-clock budget. Error bars denote Wilson 95\% confidence intervals. CLOUD reached the target in 199 of 200 datasets by eight hours and in all 200 by 16 hours. By 90 hours, LOU had met the target in 36 of 100 datasets in the oscillatory setting and 55 of 100 datasets in the nonoscillatory setting. Two LOU fits were censored by infrastructure limits rather than algorithmic nonattainment. The target definition and complete computational settings are given in Appendix~\ref{appd:efficiency_comparison}.}
\label{fig:budget}
\end{figure}

At the common checkpoint after 1,000 post-warm-up draws, CLOUD achieved higher median minimum bulk ESS (263 versus 83 in the oscillatory setting; 301 versus 113 in the nonoscillatory setting), with median paired ESS-per-second ratios of 9.11 and 8.59 and a lower median number of leapfrog steps (127 versus 511) (Supplementary Table~\ref{tab:geometry}; Supplementary Figure~\ref{fig:geometry}A--C). CLOUD exhibited no divergences, maximum-tree-depth events, or minimum chain-level energy Bayesian fraction of missing information (E-BFMI) values below 0.3, whereas LOU exhibited divergences in 30 and 15 datasets, maximum-tree-depth events in 40 and 38 datasets, and E-BFMI below 0.3 in all fits. Despite these computational differences, median drift root mean squared error (RMSE) remained comparable: 0.0420 versus 0.0422 in the oscillatory setting and 0.0431 versus 0.0427 in the nonoscillatory setting (Supplementary Table~\ref{tab:geometry}; Supplementary Figure~\ref{fig:geometry}D). Among datasets in which both methods reached the target, pooled drift-interval coverage was 95.6\% for CLOUD and 95.3\% for LOU, with similarly close correlation and measurement-parameter recovery (Supplementary Table~\ref{tab:recovery}; Supplementary Figure~\ref{fig:common}B). Thus, CLOUD substantially improved computational reliability and sampling efficiency without appreciable loss of estimation accuracy.

To assess the covariate-dependent moving mean, we compared CLOUD with StationaryOU and TrendyOU. Mean misspecification affected both latent dynamics and measurement parameters. In S2, StationaryOU estimated $\beta_{1,1}$ with an RB of $-2.021$ and 18.3\% coverage, versus 0.010 and 95.7\% for CLOUD (Supplementary Table~\ref{tab:s2_comparison_2}). TrendyOU reduced median absolute loading RB relative to StationaryOU in every scenario, but residual bias remained: in S2, the values were 0.868, 0.586, and 0.021 for StationaryOU, TrendyOU, and CLOUD, respectively (Table~\ref{tab:simulation_summary}). The corresponding values were 0.612, 0.178, and 0.058 in S5, and 0.623, 0.169, and 0.034 in S6. In S6, mean off-diagonal correlation MSE similarly decreased from 0.203 to 0.103 and 0.007 (Supplementary Table~\ref{tab:r8_s6_summary}). CLOUD also recovered covariate-dependent slopes $\boldsymbol\Phi$ with median absolute RB of 0.081 and 0.030 and mean coverage of 93.0\% and 95.3\% in S5 and S6, respectively. Coverage for the common slope $\boldsymbol\alpha$ was likewise higher under CLOUD than under TrendyOU (95.2\% versus 39.7\% in S5; 96.0\% versus 32.2\% in S6; Supplementary Tables~\ref{tab:r8_s5_summary}--\ref{tab:r8_s6_summary}). These findings indicate that ignoring covariate-dependent heterogeneity could result in biased measurement and dependence estimates.

Together, the results support accurate parameter recovery across multidimensional interacting latent processes, improved computational efficiency through structured drift parameterization, and reduced misspecification bias through covariate-dependent latent progression.

\section[Application]{Application to Parkinson's disease data}
\label{sec:ppmi}
Parkinson's disease encompasses heterogeneous combinations of bradykinesia (slowness of movement), rigidity, tremor, and axial impairment, with substantial variation in the subsequent clinical course \citep{Fereshtehnejad2015PDSubtypes}. Longitudinal analyses of the De Novo Parkinson (DeNoPa) and PPMI cohorts identified progression in rigidity, gait, and limb-movement items, whereas tremor items showed less consistent worsening \citep{Bartl2022MotorProgression}. In an independent longitudinal cohort, bradykinesia and rigidity worsened while rest tremor remained stable and postural and kinetic tremor decreased \citep{vanDenBerg2025ActionTremor}. These findings motivate domain-specific trajectories: an aggregate motor score can combine worsening in one component with stability or improvement in another.

The clinical domains are distinguishable but related. The original validation of the Movement Disorder Society--Unified Parkinson's Disease Rating Scale (MDS-UPDRS) identified seven factors within its Part III motor examination: axial/midline function, rigidity, lower-limb bradykinesia, bradykinesia of the left and right upper limbs, rest tremor, and postural/kinetic tremor \citep{Goetz2008MDSUPDRS}. This structure supports representing ordinal item responses as manifestations of several latent motor domains. Associations among motor deficits also have longitudinal relevance: baseline freezing and rigidity were associated with subsequent deterioration \citep{Bartl2022MotorProgression}. Together, these findings motivate estimating temporal dependence jointly across domains.

Baseline characteristics further distinguish progression patterns. In PPMI, older baseline age predicted earlier attainment of clinically relevant motor and nonmotor milestones \citep{Brumm2023PPMIMilestones}. Probable rapid eye movement (REM) sleep behavior disorder was associated with faster motor progression over five years, with the association depending on underlying biomarker profiles \citep{Pagano2018RBDProgression}. Symptoms of this disorder are often assessed using the REM Sleep Behavior Disorder Screening Questionnaire (RBDSQ) on a 0--13 scale \citep{StiasnyKolster2007RBDSQ}. These findings therefore provide a clinical rationale for using age and baseline sleep symptoms as candidate predictors of progression.

We selected longitudinal motor examinations from PPMI \citep{Marek2011PPMI} to illustrate the practical utility of the proposed framework. These data exhibit many of the methodological challenges that motivate our work, including heterogeneous disease progression, irregular follow-up schedules, and complex interactions among multiple functional domains. Although we focus on PPMI here, the proposed framework is broadly applicable to other longitudinal biomedical studies with similar data structures.

\subsection{Analysis setup}
\label{sec:ppmi_setup}
We analyzed $K=33$ MDS-UPDRS Part III ordinal item scores, ranging from 0 (normal) to 4 (severe impairment), assigned to $R=7$ prespecified motor domains \citep{Goetz2008MDSUPDRS} (Supplementary Table~\ref{tab:ppmi_item_domain_mapping}). The dataset comprised 1,679 participants and 19,101 assessments within five years of each participant's analysis time origin ($t=0$); time was measured in years. A fixed participant-level split allocated 1,343 participants to training and 336 to testing.
Training-set screening of 12 candidate predictors selected age at the analysis origin and RBDSQ score as latent progression covariates $\mathbf x_i^{(2)}$, and sex, examination state, PDMEDYN, and DBSYN as measurement covariates $\mathbf x_{ij}^{(1)}$. Continuous covariates were standardized using training-set statistics. Item specification, selection, and estimation procedures are provided in Appendix~\ref{appd:data_prep}.

\subsection{Measurement structure and observed motor progression}
\label{sec:overall_fitting}
We first assessed the loading structure used in the dynamic analysis. Each item was linked to its prespecified domain by an anchored, simple-structure loading matrix. Item-loading ranks were similar across CLOUD and the three restricted OU specifications (Spearman correlations 0.970--0.997; Supplementary Table~\ref{tab:ppmi_measurement_alignment}), providing a common measurement basis for comparing dynamic estimates.

We then examined observed motor-domain changes (Figure~\ref{fig:ppmi_progression}A). At 33 months, mean changes from each participant's first complete-domain assessment were $+1.222$ points (95\% participant-bootstrap interval [0.986, 1.458]) for axial/midline, $+0.775$ [0.592, 0.965] for lower limb, and $+0.591$ [0.393, 0.781] for rigidity. Upper-limb scores increased, rest tremor showed little net change ($-0.024$ [$-0.208$, 0.183]), and postural/kinetic tremor decreased ($-0.189$ [$-0.319$, $-0.070$]). After adjustment for seven comparisons, participant fixed-effects analyses supported increases in axial/midline, rigidity, lower-limb, and left-upper-limb scores, but not the right-upper-limb increase (Supplementary Table~\ref{tab:ppmi_annual_change}). Approximately 630--690 participants contributed to each 33-month domain summary (Supplementary Figures~\ref{fig:ppmi_web_slope_dist1}--\ref{fig:ppmi_web_followup_composition}).

\begin{figure}[H]
\centering
\includegraphics[width=\linewidth,height=0.76\textheight,keepaspectratio]{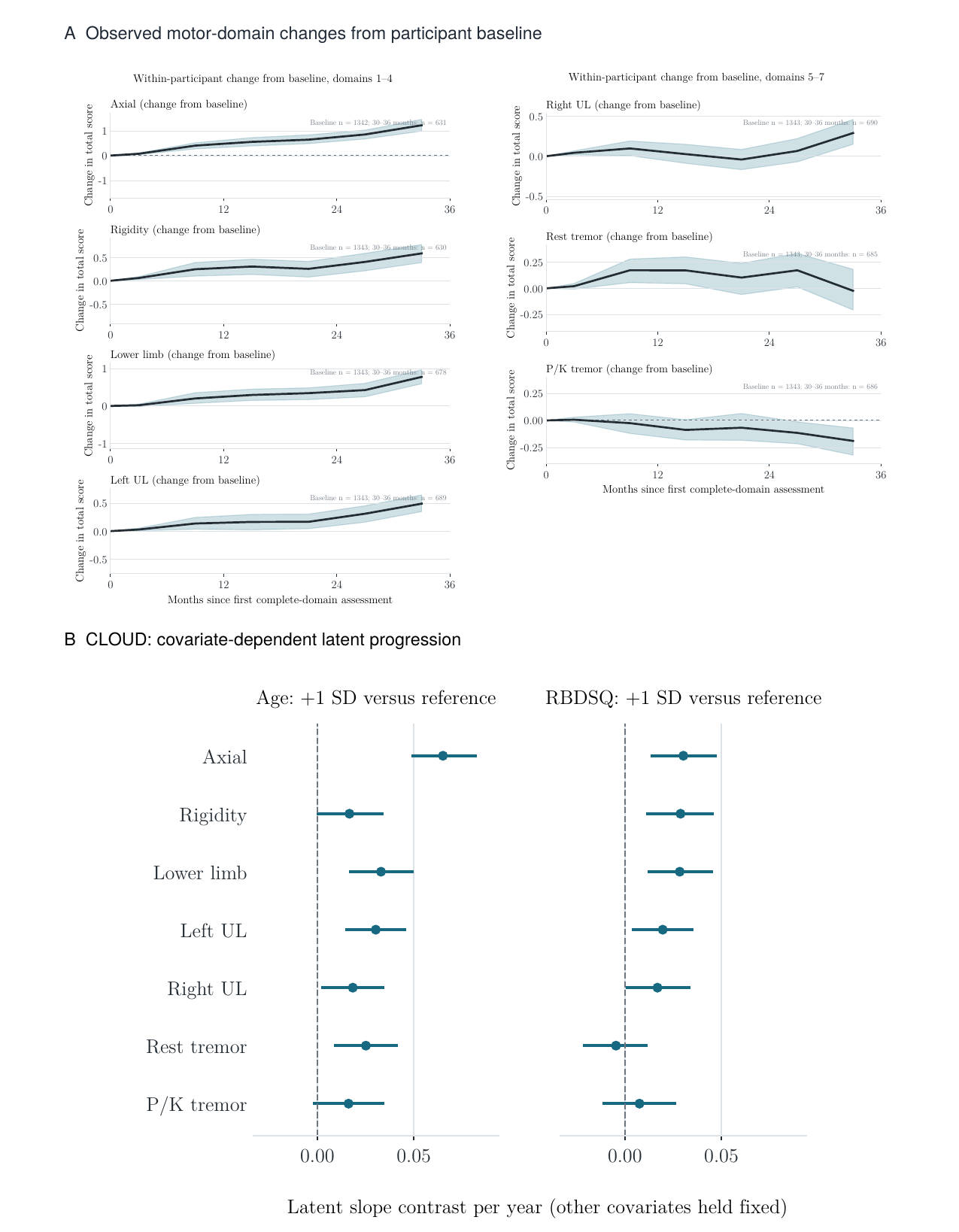}
\caption{Clinical-score progression and covariate-associated latent progression in PPMI. (A) Mean change in seven MDS-UPDRS Part III domain scores from each participant's first complete-domain assessment; the last plotted follow-up is 33 months. Shaded bands are participant-bootstrap 95\% intervals. Curves reflect the observed follow-up composition and treatment context. (B) CLOUD posterior medians and marginal 95\% credible intervals for one-standard-deviation (SD) age and RBDSQ-score contrasts in annual latent mean slopes. RBDSQ denotes the Rapid Eye Movement Sleep Behavior Disorder Screening Questionnaire; panel A is on the clinical domain-score scale and panel B is on the latent-slope scale.}
\label{fig:ppmi_progression}
\end{figure}

The differences in progression rates observed across motor domains are consistent with the literature. \citet{Bartl2022MotorProgression} similarly reported progression in rigidity, gait, and limb movement relative to tremor, while \citet{vanDenBerg2025ActionTremor} reported worsening bradykinesia and rigidity, stable rest tremor, and decreasing postural/kinetic tremor.

\subsection{Baseline covariate effects on latent progression}
\label{sec:covariate_effects}
To assess between-participant progression heterogeneity, we examined the covariate effects in CLOUD's latent moving mean (Figure~\ref{fig:ppmi_progression}B). Each coefficient represents an annual latent-slope contrast per one-standard-deviation change in a baseline predictor, holding the other predictor fixed. For age (9.95 years per standard deviation), the axial/midline contrast was $+0.065$ standardized latent units per year (marginal 95\% credible interval [0.049, 0.082]); six of seven age coefficients had intervals above zero. For RBDSQ (3.01 points per standard deviation), the corresponding contrasts were $+0.030$ [0.014, 0.047] for axial/midline, $+0.029$ [0.012, 0.045] for rigidity, and $+0.028$ [0.012, 0.045] for lower limb; five of seven RBDSQ coefficients had intervals above zero.
The positive age contrasts correspond in direction to PPMI findings linking older age with earlier disability milestones \citep{Brumm2023PPMIMilestones}. The positive RBDSQ contrasts accord with associations of probable REM sleep behavior disorder with faster motor progression \citep{Pagano2018RBDProgression} and subsequent freezing of gait and falls \citep{Liu2021RBD}.
These estimated latent progression rates were not evident from simple sums of the domain scores. In regressions of summed observed-domain scores, none of the 14 age/RBDSQ-by-time effects met the Benjamini--Hochberg-adjusted 0.05 criterion (Supplementary Table~\ref{tab:ppmi_empirical_covariates}); history-adjusted observed-category calibration contrasts also had intervals containing zero (Appendix~\ref{appd:data_prep}). These analyses concern different estimands and use different measurement models, so the contrast between their significance results alone does not establish unequal item contributions or a predictive advantage for joint modeling.

\subsection{Conditional temporal relationships between motor domains}
\label{sec:interdependency}
We next evaluated temporal dependence after accounting for covariate-dependent mean progression. Figure~\ref{fig:ppmi_interactions} compares regressions of annualized observed-score changes on preceding domain scores (panel A) with CLOUD's one-year latent transition matrix $\mathbf T(1)=\exp(-\boldsymbol\Gamma)$ (panel B). Rows and columns denote target and source domains. Off-diagonal entries relate current source-domain deviations from the moving mean to target-domain deviations one year later.

\begin{figure}[!htbp]
\centering
\includegraphics[width=\linewidth,height=0.77\textheight,keepaspectratio]{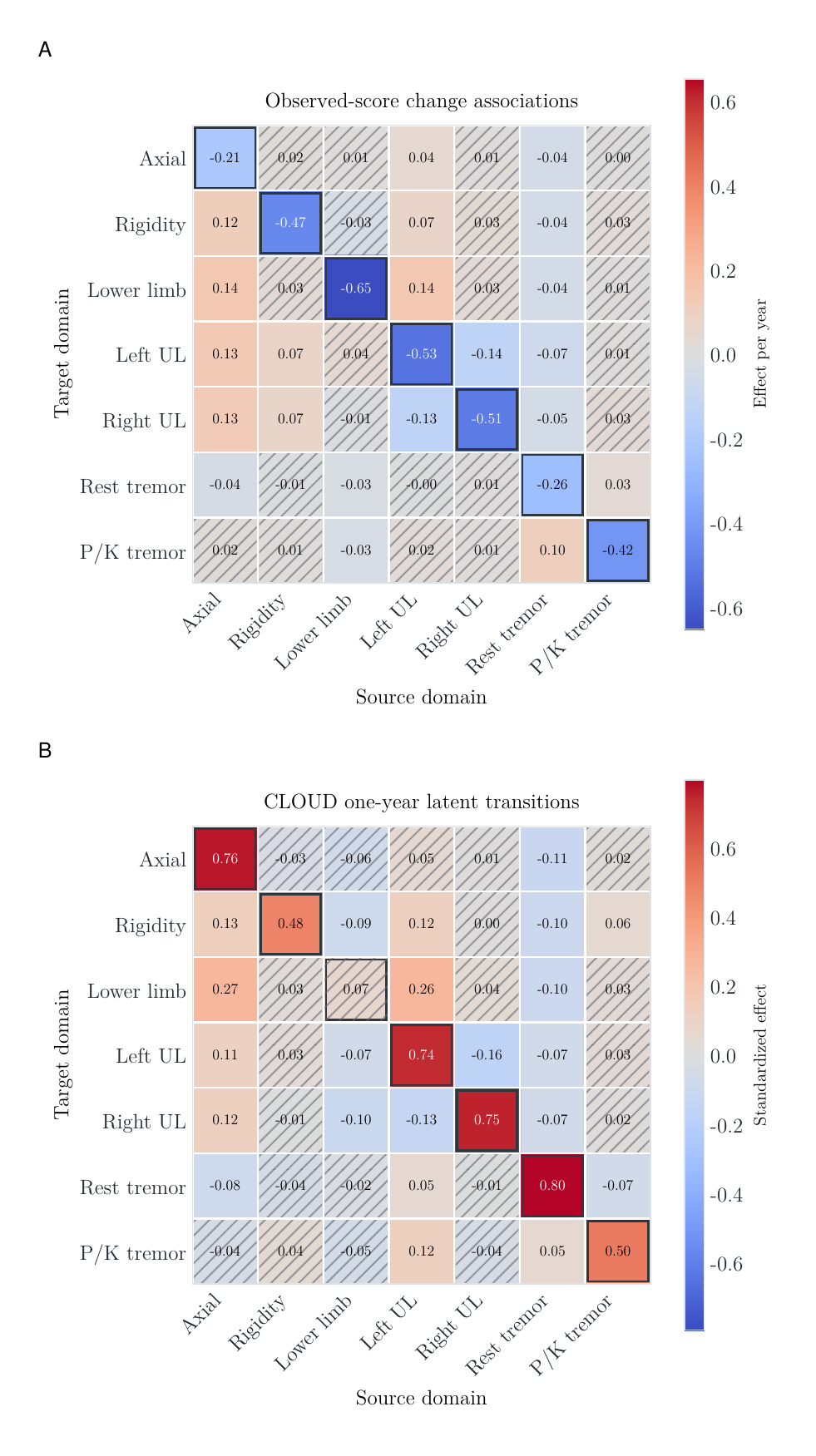}
\caption{Conditional temporal relationships among seven PPMI motor domains. (A) Regressions of annualized native-score change on prior domain scores in the training set. (B) Standardized one-year CLOUD latent transition matrix. Rows denote target domains, columns denote source domains, and all diagonal and off-diagonal entries are shown. Hatching indicates nonsignificant empirical associations after Benjamini--Hochberg adjustment in panel A, and 99\% posterior credible intervals that include zero in panel B; outlines identify diagonal cells. Panel A reports observed-score change associations, and panel B reports model-conditional latent transitions.}
\label{fig:ppmi_interactions}
\end{figure}

The axial/midline-to-lower-limb coefficient was $+0.27$ (99\% credible interval [0.185, 0.346]); the axial/midline-to-rigidity coefficient was $+0.13$ [0.058, 0.207] (Figure~\ref{fig:ppmi_interactions}B; Supplementary Table~\ref{tab:ppmi_selected_associations}). These estimates associate an above-mean axial latent state with subsequent above-mean lower-limb and rigidity states, conditional on other domains. Corresponding observed-score change coefficients were positive (0.14 and 0.12 per year) and passed multiplicity adjustment (Figure~\ref{fig:ppmi_interactions}A). Among 22 off-diagonal transitions with absolute standardized point estimates of at least 0.05, 19 agreed in sign across the two analyses and 15 met the empirical false-discovery-rate criterion (Supplementary Table~\ref{tab:ppmi_selected_associations}).

The clinical correspondence is specific: axial examinations assess gait, freezing, and postural stability; lower-limb items assess toe tapping and leg agility; and rigidity items assess resistance to passive movement \citep{Goetz2008MDSUPDRS}. These features also showed deterioration or prognostic relevance in longitudinal Parkinson's disease cohorts \citep{Bartl2022MotorProgression}. CLOUD therefore quantifies time-ordered dependencies among recognized motor deficits. The complete transition matrix appears in Figure~\ref{fig:ppmi_interactions}B and Supplementary Table~\ref{tab:ppmi_selected_associations}.

\FloatBarrier

\subsection{Predictive performance}
\label{sec:ppmi_benchmark}
Finally, we compared CLOUD's forecasting performance with that of three restricted OU models (StationaryOU, TrendyOU, and DiagOU), a study-specific item-response model with Mat\'ern latent Gaussian-process trajectories (GP-IRT) \citep{RasmussenWilliams2006}, a multidimensional latent-trait linear mixed-model comparator (MLTLMM) motivated by \citet{wang2017multidimensional}, and ordinal generalized linear and additive mixed models (GLMM and GAMM) \citep{McCullagh1980Ordinal,BreslowClayton1993GLMM,LinZhang1999GAMM,Wood2011REML}. The evaluation covered 237 test participants and 7,787 observed items at approximately one year (Table~\ref{tab:ppmi_forecast}). We used negative log predictive density (NLPD) and normalized ranked probability score (RPS) to assess predictive distributions \citep{GneitingRaftery2007}; mean-score error, exact-category accuracy, and quadratic weighted kappa provided complementary summaries. Participant-balanced scoring and paired bootstrap intervals accounted for clustering (Supplementary Table~\ref{tab:ppmi_v174_models}).

\begin{table}[H]
\centering
\caption{Comparison of eight longitudinal ordinal models for forecasting from the third eligible motor examination occasion to an assessment 9--15 months later (approximately one-year-ahead forecasts).}
\label{tab:ppmi_forecast}
\CloudStatTableStyle
\begin{threeparttable}
\begin{tabular*}{\linewidth}{@{\extracolsep{\fill}}lrrrrrl@{}}
\toprule
Model & NLPD $\downarrow$ & RPS $\downarrow$ & \shortstack{Mean-score\\MAE $\downarrow$} & Acc. (\%) $\uparrow$ & QWK $\uparrow$ & $\Delta$NLPD [95\% CI] \\
\midrule
CLOUD & 0.7901 & 0.06914 & 0.4623 & 65.89 & 0.620 & Reference \\
TrendyOU & 0.7899 & 0.06910 & 0.4620 & 65.98 & 0.620 & $+0.0002\quad[-0.0019, +0.0023]$ \\
StationaryOU & 0.7879 & 0.06891 & 0.4585 & 65.69 & 0.608 & $+0.0022\quad[-0.0014, +0.0060]$ \\
DiagOU & 0.7916 & 0.06931 & 0.4630 & 65.44 & 0.608 & $-0.0015\quad[-0.0063, +0.0032]$ \\
MLTLMM & 0.8607 & 0.07795 & 0.5137 & 61.92 & 0.531 & $-0.0706\quad[-0.0833, -0.0582]$ \\
GP-IRT & 0.8222 & 0.07314 & 0.4928 & 64.56 & 0.559 & $-0.0322\quad[-0.0417, -0.0225]$ \\
GLMM & 0.7919 & 0.06917 & 0.4494 & 66.22 & 0.637 & $-0.0019\quad[-0.0101, +0.0062]$ \\
GAMM & 0.7919 & 0.06914 & 0.4507 & 66.13 & 0.640 & $-0.0018\quad[-0.0106, +0.0067]$ \\
\bottomrule
\end{tabular*}
\begin{tablenotes}
\footnotesize\item \emph{Notes.} $N_{\mathrm{test}}=237$ held-out participants; 7,787 scored ordinal motor-item outcomes. NLPD, negative log predictive density; RPS, ranked probability score; mean-score MAE, mean absolute error of the probability-weighted predicted ordinal score $\bigl|\sum_{c=0}^{4}c\,p_c-y\bigr|$, averaged within each participant and then equally across participants; Acc., exact-category accuracy; QWK, pooled item-level quadratic weighted kappa. The mean-score MAE differs from the expected absolute ordinal loss $\sum_{c=0}^{4}p_c|c-y|$. $\Delta\mathrm{NLPD}=\mathrm{NLPD}_{\mathrm{CLOUD}}-\mathrm{NLPD}_{\mathrm{comparator}}$; negative differences favor CLOUD. Brackets give participant-paired 95\% bootstrap intervals for the differences, not for the raw scores. Intervals against TrendyOU, StationaryOU, and DiagOU include zero and do not establish equivalence. The target assessment was selected using timestamps and identifiers alone; predictions used only observed history available at the forecast landmark and integration over future examination context estimated from the development data. This comparison of approximately one-year-ahead forecasts from the third occasion remains exploratory because test outcomes were inspected during earlier model development. Reported values are rounded; supplementary tables provide domain-specific analyses and complementary score comparisons.
\end{tablenotes}
\end{threeparttable}
\end{table}

CLOUD had lower NLPD than GP-IRT (paired difference $-0.0322$, 95\% interval $[-0.0417,-0.0225]$) and MLTLMM ($-0.0706$ [$-0.0833,-0.0582$]), with corresponding RPS improvements (Supplementary Table~\ref{tab:ppmi_rps_differences}). GLMM and GAMM had comparable NLPD and slightly lower mean-score error (Supplementary Figure~\ref{fig:ppmi_web_domain_prediction1}). Among OU models, CLOUD's NLPD was 0.7901 versus 0.7879 for StationaryOU, a paired difference of $+0.0022$ (95\% interval $[-0.0014,0.0060]$), and intervals for comparisons with TrendyOU and DiagOU also included zero (Appendix~\ref{appd:data_prep}). CLOUD's quadratic weighted kappa was slightly higher than that of StationaryOU and DiagOU (Supplementary Table~\ref{tab:ppmi_kappa_differences}). In this evaluation, CLOUD retained competitive predictive performance alongside interpretable latent associations. Predictive scores alone did not clearly distinguish the OU variants.

\section{Discussion}
\label{sec:dis}
In this paper, we proposed CLOUD, a continuous-time latent variable framework for multivariate longitudinal categorical data observed at irregular time points.
CLOUD builds on existing OU-based latent variable formulations \citep{Tran2020LatentOM,Abbott2024ABJ} by combining a covariate-dependent latent moving mean function for heterogeneous disease progression with a dimension-generic drift parameterization for interactions among latent variables in general finite-dimensional latent spaces. These methodological developments substantially broaden the applicability of continuous-time latent variable models while preserving their interpretability and theoretical tractability.

This work makes three methodological contributions. First, a Lyapunov-based parameterization of the latent drift matrix enables unrestricted modeling of latent interactions in the identified admissible drift--covariance class at any fixed finite dimension. The positive-definite/skew-symmetric representation itself is rooted in classical Lyapunov parameterizations of stable linear systems \citep{OharaKitamori1993,OharaAmari1994,GillisSharma2017} and related nonreversible Gaussian-diffusion constructions \citep{Lelievre2013,Duncan2016}. CLOUD's contribution is to use that geometry for joint inference on the unknown drift and stationary covariance within the identified latent OU process model. This removes the dimension-specific stability inequalities used in earlier low-dimensional latent OU formulations and provides a common parameterization for studying cross-domain temporal dependence at larger fixed latent dimensions. Second, a covariate-dependent, time-inhomogeneous latent OU process enables modeling of heterogeneous disease progression. Whereas earlier hierarchical OU models allow person-specific parameters and covariate associations \citep{Oravecz2016BayesianDA}, CLOUD targets baseline-covariate effects on linear progression in an ordinal multivariate latent model. By allowing the latent mean to evolve as a function of baseline covariates, CLOUD directly models subject-specific disease progression within the latent dynamic system. This extension increases the realism and interpretability of continuous-time latent process models, particularly for biomedical applications in which progression rates differ systematically across patients. Finally, we established structural and inferential theory for the proposed framework. We derived sufficient conditions that guarantee well-posed continuous-time latent dynamics, established conditional parameter identifiability given explicit ordinal-measurement, temporal-design, and covariate-separation assumptions, and demonstrated the asymptotic behavior of the CLOUD model. These theoretical results support the use of CLOUD as a longitudinal analysis tool.

Our simulation studies demonstrate that each methodological contribution provides measurable inferential benefits. Across six scenarios, the simulations provided evidence of CLOUD's ability to recover parameters under coupled drift structures with real and complex eigenvalues and at different latent dimensions, as well as the necessity of each component. In the matched two-dimensional computational comparison, the complete CLOUD implementation attained the sampling criterion more reliably and reached it faster than the directly constrained LOU implementation, with comparable drift recovery.
The PPMI analysis also illustrates how these methodological advances translate into practical scientific inference. The estimated progression patterns correspond to the clinical domain differences and scale structure reported in earlier studies \citep{Bartl2022MotorProgression,vanDenBerg2025ActionTremor,Goetz2008MDSUPDRS}. The directional temporal associations are findings of the present model. In the exploratory forecasting comparison, CLOUD improved on the GP-IRT and MLTLMM comparators while its NLPD was comparable to the restricted OU models (Table~\ref{tab:ppmi_forecast}). Although PPMI served as a motivating example, CLOUD is broadly applicable to longitudinal studies involving irregular observations. By accommodating general latent spaces of fixed dimension and heterogeneous temporal dynamics, CLOUD offers a modeling approach for systematically characterizing the dynamics of functional domains.

Several limitations warrant consideration. First, the current formulation assumes a linear moving mean function. Although this assumption provides interpretability and analytical tractability, a more flexible mean function may be needed for nonlinear or stage-dependent average progression.
Second, posterior estimation remains computationally demanding because Bayesian sampling requires repeated evaluation of matrix exponentials and transition covariance matrices, as well as repeated matrix factorizations, at irregular observation intervals. Further gains may come from structured matrix-function computations such as those studied by \citet{Monti2026}, as well as scalable Bayesian computation or approximate inference.

CLOUD provides a general statistical framework for continuous-time latent process modeling of multivariate longitudinal categorical data.
The proposed methodology is applicable to a broad range of biomedical, behavioral, and social science studies characterized by irregular observation schedules, heterogeneous progression, and interacting latent processes. The theoretical developments introduced here may facilitate broader adoption of continuous-time latent variable models in modern longitudinal research.

\section*{Acknowledgments}
This study was supported by the Ralph W. and Grace M. Showalter Research Trust and the Indiana University School of Medicine (award to X.H.); an Indiana University Seed Grant (X.H.); National Institutes of Health grants R35GM147241 and R01AG098161 (Y.W. and collaborators); and Lilly Endowment, Inc., through its high-performance computing support for the Indiana University Pervasive Technology Institute. The funders had no role in study design, analysis, interpretation, manuscript preparation, or the decision to submit the manuscript.

\section*{Data availability}
De-identified PPMI longitudinal study data are available to approved investigators through the PPMI data access portal (\url{https://www.ppmi-info.org/access-data-specimens/download-data}). 

\section*{Supplementary material}
The Appendices, Tables, and Figures referenced in Sections~\ref{sec:model}--\ref{sec:ppmi} are included in the supplementary material of this preprint. Appendix~\ref{appd:tables_figures} collects the supplementary tables and figures. The associated simulation code is available through the \href{https://github.com/xiaoqinghuanglab/CLOUD}{CLOUD repository} (\url{https://github.com/xiaoqinghuanglab/CLOUD}).

\clearpage
\bibliographystyle{biom}
\bibliography{ref}
\clearpage

\appendix
\section*{Supplementary Material}
\setcounter{section}{0}
\setcounter{equation}{0}
\setcounter{figure}{0}
\setcounter{table}{0}

\renewcommand{\thesection}{\Alph{section}}
\renewcommand{\theequation}{\thesection.\arabic{equation}}
\renewcommand{\thefigure}{\thesection.\arabic{figure}}
\renewcommand{\thetable}{\thesection.\arabic{table}}

\makeatletter
\newcommand{\CloudResetSupplementStatementCounters}{%
  \@for\CloudStmt:=thrm,lem,prop,cor,assum,theorem,lemma,corollary,assumption,remark\do{%
    \@ifundefined{c@\CloudStmt}{}{\setcounter{\CloudStmt}{0}}%
  }%
}
\newcommand{\CloudFormatSupplementStatementCounters}{%
  \@for\CloudStmt:=thrm,lem,prop,cor,assum,theorem,lemma,corollary,assumption,remark\do{%
    \@ifundefined{c@\CloudStmt}{}{%
      \expandafter\CloudFormatOneSupplementStatement\expandafter{\CloudStmt}%
    }%
  }%
}
\newcommand{\CloudFormatOneSupplementStatement}[1]{%
  \expandafter\def\csname the#1\endcsname{\thesection.\arabic{#1}}%
  \expandafter\def\csname theH#1\endcsname{supp.\Alph{section}.\arabic{#1}}%
}
\makeatother
\CloudFormatSupplementStatementCounters
\CloudResetSupplementStatementCounters

\renewcommand{\theHsection}{supp.\Alph{section}}
\renewcommand{\theHequation}{supp.\Alph{section}.\arabic{equation}}
\renewcommand{\theHfigure}{supp.\Alph{section}.\arabic{figure}}
\renewcommand{\theHtable}{supp.\Alph{section}.\arabic{table}.\arabic{CloudContinuedFloat}}

% Include the appendix prefix in headings, without modifying reference numbers.
\makeatletter
\let\CloudOriginalSectionNumberFormat\@seccntformat
\renewcommand{\@seccntformat}[1]{%
  \ifstrequal{#1}{section}{\CloudSupplementAppendixName\space\thesection\quad}%
    {\CloudOriginalSectionNumberFormat{#1}}}
\makeatother

\newcommand{\suppsection}[2]{%
  \ifnum\value{section}>0\relax\clearpage\fi
  \setcounter{equation}{0}%
  \setcounter{figure}{0}%
  \setcounter{table}{0}%
  \setcounter{CloudContinuedFloat}{0}%
  \CloudResetSupplementStatementCounters%
  \section{#1}%
  \label{#2}%
}

\suppsection{Centered OU representation and stability}{appd:tiou_proof}
The moving mean introduces no stability conditions beyond those of the
standard multivariate Ornstein--Uhlenbeck (OU) process. Conditional on the baseline covariates,
write $\mathbf v_i=\boldsymbol\alpha+\boldsymbol\Phi\mathbf x_i^{(2)}$,
$\boldsymbol\mu_i(t)=\mathbf v_it$, and
$\mathbf r_i(t)=\boldsymbol\xi_i(t)-\boldsymbol\mu_i(t)$.
Let $\mathbf W_i$ be an $R$-dimensional standard Brownian motion and let $\mathbf B$ satisfy $\mathbf B\mathbf B^\top=\boldsymbol\Sigma$. The equivalent stochastic differential equations (SDEs) are
\begin{align}
\mathrm d\mathbf r_i(t)
 &=-\boldsymbol\Gamma\mathbf r_i(t)\,\mathrm dt
   +\mathbf B\,\mathrm d\mathbf W_i(t), \label{eq:app_centered_ou}\\
\mathrm d\boldsymbol\xi_i(t)
 &=\bigl\{\mathbf v_i-\boldsymbol\Gamma[\boldsymbol\xi_i(t)-\mathbf v_it]\bigr\}\,\mathrm dt
   +\mathbf B\,\mathrm d\mathbf W_i(t).\notag
\end{align}
Thus, the mean derivative $\dot{\boldsymbol\mu}_i(t)=\mathbf v_i$
enters the drift explicitly. Stability concerns the centered process;
$\boldsymbol\xi_i(t)$ itself need not have a fixed equilibrium.

\begin{lem}[Centered OU transitions and stability]
\label{lemma:conditional_kernel_stability}
Suppose $\boldsymbol\Omega\succ0$,
$\boldsymbol\Sigma=\boldsymbol\Gamma\boldsymbol\Omega+
\boldsymbol\Omega\boldsymbol\Gamma^\top\succ0$,
and $\boldsymbol\Gamma$ is positive stable. For $h>0$, let
$\mathbf E_h=e^{-\boldsymbol\Gamma h}$. The conditional transition distribution
of $\boldsymbol\xi_i$ is given by Equation~\eqref{eq:ou}, with nondegenerate
covariance
\begin{equation}
\mathbf Q(h)=\boldsymbol\Omega-\mathbf E_h\boldsymbol\Omega\mathbf E_h^\top
 =\int_0^h e^{-\boldsymbol\Gamma u}\boldsymbol\Sigma
 e^{-\boldsymbol\Gamma^\top u}\,\mathrm du\succ0.
\label{eq:app_q_integral}
\end{equation}
The transitions satisfy the Chapman--Kolmogorov composition rule: integrating over an intermediate latent state yields the transition over the combined interval. Under
$\mathbf r_i(0)\sim\mathcal N(\mathbf0,\boldsymbol\Omega)$,
$\mathbf r_i$ is stationary; from any proper initial law
independent of future Brownian increments, it converges in distribution
to $\mathcal N(\mathbf0,\boldsymbol\Omega)$ as $t\to\infty$.
\end{lem}
\begin{proof}
The solution of Equation~\eqref{eq:app_centered_ou} is
$\mathbf r_i(t)=\mathbf E_{t-s}\mathbf r_i(s)+
\int_s^te^{-\boldsymbol\Gamma(t-u)}\mathbf B\,\mathrm d\mathbf W_i(u)$.
Differentiating $\mathbf E_h\boldsymbol\Omega\mathbf E_h^\top$
proves Equation~\eqref{eq:app_q_integral}; composing the Gaussian
transitions verifies the Chapman--Kolmogorov property. Stationarity follows from
$\mathbf E_h\boldsymbol\Omega\mathbf E_h^\top+\mathbf Q(h)
=\boldsymbol\Omega$.
Finally, positive stability implies $\mathbf E_h\to\mathbf0$
and $\mathbf Q(h)\to\boldsymbol\Omega$, giving the stated weak limit.
\end{proof}
These centered OU properties follow from the Gaussian transition law and the Lyapunov covariance identity \citep{Tran2020LatentOM,geissert2009asymptotic}.

\suppsection{Likelihood}{appd:likelihood}
Let $\mathcal O_i=\{(j,k):Y_{ijk}\text{ is observed}\}$ denote the observed item--visit pairs for subject $i$. Write $\boldsymbol\xi_{i,1:n_i}=(\boldsymbol\xi_i(t_{i1})^\top,\ldots,\boldsymbol\xi_i(t_{in_i})^\top)^\top$ for the latent states at the observation times and $\mathbf b_i=(b_{i1},\ldots,b_{iK})^\top$ for the persistent item-specific random intercepts. We condition on covariates and observation times throughout. The population parameters are
\begin{equation}
\boldsymbol\Psi=\{\boldsymbol\theta,\boldsymbol\Lambda,\boldsymbol\beta,\boldsymbol\sigma_b,
\boldsymbol\Gamma,\boldsymbol\Omega,\boldsymbol\Phi,\boldsymbol\alpha\}.
\label{eq:parameter_set}
\end{equation}
The initial latent state has the distribution
\begin{equation}
\boldsymbol\xi_i(t_{i1})\sim\mathcal N_R\{\boldsymbol\mu_i(t_{i1}),\boldsymbol\Omega\},
\label{eq:initial_latent_state}
\end{equation}
so the centered residual $\mathbf r_i(t)=\boldsymbol\xi_i(t)-\boldsymbol\mu_i(t)$ is stationary.

\subsection{Conditional response probabilities}
Define the item-specific linear predictor and its cutpoint-specific cumulative-link predictor by
\[
\ell_{ijk}=\boldsymbol\Lambda_k^\top\boldsymbol\xi_i(t_{ij})+
\boldsymbol\beta_k^\top\mathbf x_{ij}^{(1)}+b_{ik},\qquad
\nu_{ijkm}=\theta_{km}-\ell_{ijk}.
\]
Here $\mathbf x_{ij}^{(1)}$ denotes measurement-level covariates, and $h^{-1}$ is the inverse link. Let $F_{ijkm}=h^{-1}(\nu_{ijkm})$ for $m=0,\ldots,c_k-2$, with boundary values $F_{ijk,-1}=0$ and $F_{ijk,c_k-1}=1$. The category probabilities are
\begin{equation}
\pi_{ijkm}=\Pr_\Psi(Y_{ijk}=m\mid\boldsymbol\xi_{i,1:n_i},\mathbf b_i)
=F_{ijkm}-F_{ijk,m-1},\qquad m=0,\ldots,c_k-1.
\label{eq:category_prob}
\end{equation}
Conditional independence of the responses given the latent states and item intercepts gives
\begin{equation}
L_i(\boldsymbol\Psi;\mathbf y_i^{\mathrm{obs}}\mid\boldsymbol\xi_{i,1:n_i},\mathbf b_i)
=\prod_{(j,k)\in\mathcal O_i}\prod_{m=0}^{c_k-1}
\pi_{ijkm}^{\mathbf1\{y_{ijk}=m\}}.
\label{eq:cond_likelihood}
\end{equation}
Missing responses are integrated out; summing their category probabilities gives one. The observed-response likelihood therefore includes all available item responses from a partially recorded assessment. The recording-process conditions for likelihood ignorability are stated in Appendices~\ref{appd:id_proof} and \ref{appd:inf_proof}.

\subsection{Latent trajectory density and marginal likelihood}
For $j\ge2$, let $\Delta t_{ij}=t_{ij}-t_{i,j-1}>0$. The Gaussian transition has mean and covariance
\begin{align}
\mathbf m_{ij}&=\boldsymbol\mu_i(t_{ij})+
\mathbf T(\Delta t_{ij})\{\boldsymbol\xi_i(t_{i,j-1})-\boldsymbol\mu_i(t_{i,j-1})\},
\label{eq:latent_transition_mean}\\
\mathbf Q_{ij}&=\boldsymbol\Omega-\mathbf T(\Delta t_{ij})\boldsymbol\Omega\mathbf T(\Delta t_{ij})^\top,
\qquad \mathbf T(\tau)=e^{-\boldsymbol\Gamma\tau}.
\label{eq:latent_transition_cov}
\end{align}
Using $\varphi_d(\cdot;\mathbf m,\mathbf V)$ to denote a $d$-variate normal density, we write
\begin{align}
p_\Psi\{\boldsymbol\xi_i(t_{i1})\}
&=\varphi_R\{\boldsymbol\xi_i(t_{i1});\boldsymbol\mu_i(t_{i1}),\boldsymbol\Omega\},
\label{eq:initial_density}\\
p_\Psi\{\boldsymbol\xi_i(t_{ij})\mid\boldsymbol\xi_i(t_{i,j-1})\}
&=\varphi_R\{\boldsymbol\xi_i(t_{ij});\mathbf m_{ij},\mathbf Q_{ij}\}.
\label{eq:latent_transition_density}
\end{align}
The Markov factorization is
\begin{equation}
p_\Psi(\boldsymbol\xi_{i,1:n_i})
=p_\Psi\{\boldsymbol\xi_i(t_{i1})\}
\prod_{j=2}^{n_i}p_\Psi\{\boldsymbol\xi_i(t_{ij})\mid\boldsymbol\xi_i(t_{i,j-1})\}.
\label{eq:latent_path_density}
\end{equation}
The random-effect density is $p_\Psi(\mathbf b_i)=\prod_{k=1}^K\varphi_1(b_{ik};0,\sigma_{bk}^2)$. Integrating over the subject-specific quantities gives the marginal likelihood
\begin{equation}
L(\boldsymbol\Psi;\mathbf y^{\mathrm{obs}})=\prod_{i=1}^N
\int L_i(\boldsymbol\Psi;\mathbf y_i^{\mathrm{obs}}\mid\boldsymbol\xi_{i,1:n_i},\mathbf b_i)
 p_\Psi(\boldsymbol\xi_{i,1:n_i})p_\Psi(\mathbf b_i)
 \,d\boldsymbol\xi_{i,1:n_i}\,d\mathbf b_i.
\label{eq:marginal_likelihood}
\end{equation}
For categorical measurements, this integral generally has no closed-form expression. Bayesian inference proceeds by sampling the population parameters jointly with the latent states and random effects.

\subsection{Joint posterior}
Let $\boldsymbol\kappa_{\mathrm{prior}}$ collect any prior hyperparameters, and let $p(\boldsymbol\Psi,\boldsymbol\kappa_{\mathrm{prior}})$ denote the joint prior on the population parameters and these hyperparameters. The augmented posterior is
\begin{equation}
\begin{split}
p(\boldsymbol\Psi,\boldsymbol\kappa_{\mathrm{prior}},\{\boldsymbol\xi_{i,1:n_i},\mathbf b_i\}_{i=1}^N
\mid\mathbf y^{\mathrm{obs}})
\propto{}&p(\boldsymbol\Psi,\boldsymbol\kappa_{\mathrm{prior}})\\
&\times\prod_{i=1}^N
 L_i(\boldsymbol\Psi;\mathbf y_i^{\mathrm{obs}}\mid\boldsymbol\xi_{i,1:n_i},\mathbf b_i)
 p_\Psi(\boldsymbol\xi_{i,1:n_i})p_\Psi(\mathbf b_i).
\end{split}
\label{eq:joint_posterior}
\end{equation}
Under the structural parameterization, a joint prior on $(\mathbf S,\mathbf A,\boldsymbol\Omega)$ induces a prior on the drift--covariance pair. Regularization of the derived drift can couple these coordinates. The prior specifications are given in Appendix~\ref{appd:simu_supp}.

\suppsection{Exact parameterization of the admissible CLOUD drift--covariance space}{appd:helm_proof}
\subsection{Lyapunov stability theorem}
We recall both directions of the continuous-time Lyapunov criterion used below.
\begin{lem}[Continuous-time Lyapunov theorem]
    \label{lemma:lyapunov_stability}
    Let $\mathbf{H} \in \mathbb{R}^{n \times n}$ be a real square matrix, and let $\mathbf{V} \in \mathbb{R}^{n \times n}$ be a symmetric positive semidefinite (PSD) matrix. If $\mathbf{H}$ is a Hurwitz matrix (meaning every eigenvalue of $\mathbf{H}$ has a strictly negative real part), then the continuous-time Lyapunov equation
    \begin{align*}
        \mathbf{H} \mathbf{P} + \mathbf{P} \mathbf{H}^\top = -\mathbf{V}
    \end{align*}
    has a unique solution $\mathbf{P} \in \mathbb{R}^{n \times n}$. This solution $\mathbf{P}$ is also symmetric and positive semidefinite. It is given by the integral
    \begin{align*}
        \mathbf{P} = \int_0^\infty e^{\mathbf{H} t} \mathbf{V} e^{\mathbf{H}^\top t} \mathrm{d}t
    \end{align*}
    
    If $\mathbf{V}\succ0$, then $\mathbf{P}\succ0$. Conversely, if $\mathbf{P}\succ0$ and $\mathbf{H}\mathbf{P}+\mathbf{P}\mathbf{H}^{\top}\prec0$, then $\mathbf{H}$ is Hurwitz.
\end{lem}
For completeness, the exponential decay associated with a Hurwitz matrix ensures that the displayed integral converges. Differentiating $e^{\mathbf H t}\mathbf V e^{\mathbf H^{\top}t}$ and integrating over $[0,\infty)$ verifies the Lyapunov equation. Symmetry and positive semidefiniteness follow by integrating quadratic forms; if $\mathbf V\succ0$, the quadratic form is strictly positive for every nonzero vector.

For uniqueness, let $\mathbf D$ be the difference of two solutions. Since $\mathbf H\mathbf D+\mathbf D\mathbf H^{\top}=0$, differentiating $e^{\mathbf H t}\mathbf D e^{\mathbf H^{\top}t}$ shows that it equals $\mathbf D$ for all $t$. Both exponentials decay to zero, so $\mathbf D=0$.

For the converse, if $\mathbf P\succ0$ and $\mathbf H\mathbf P+\mathbf P\mathbf H^{\top}\prec0$, take a nonzero complex eigenvector $v$ of $\mathbf H^{\top}$ with eigenvalue $\lambda$. Then $2\operatorname{Re}(\lambda)\,v^*\mathbf P v=v^*(\mathbf H\mathbf P+\mathbf P\mathbf H^{\top})v<0$, so every eigenvalue of $\mathbf H$ has negative real part.
For CLOUD, the Lyapunov drift is $-\boldsymbol\Gamma$, its driving covariance is $\boldsymbol\Sigma$, and its solution is the stationary residual covariance $\boldsymbol\Omega$.

The following results provide the formal statements and proofs supporting Theorem~\ref{thm:main_generality} and Proposition~\ref{prop:main_direct_feasibility} in the main text. The preceding derivation includes the Lyapunov converse needed for the admissible-drift characterization. Throughout, $\mathbb S_{++}^R$ denotes the cone of $R\times R$ symmetric positive-definite matrices, and
\[
\mathfrak{so}(R)=\{\mathbf{A}\in\mathbb R^{R\times R}:\mathbf{A}^{\top}=-\mathbf{A}\}.
\]
After the latent-scale normalization in Equation~\eqref{eq:scal_norm}, let
\[
\mathcal C_R=\{\boldsymbol{\Omega}\in\mathbb S_{++}^R:\operatorname{diag}(\boldsymbol{\Omega})=\mathbf 1_R\}.
\]
If the scale normalization is not imposed, $\mathcal C_R$ below can be replaced by $\mathbb S_{++}^R$ without changing any argument.

\subsection{The valid CLOUD drift--covariance space}
Define the identified parameter space for valid dynamics by
\begin{equation}
\mathcal F_R
=\left\{(\boldsymbol{\Gamma},\boldsymbol{\Omega}):\boldsymbol{\Omega}\in\mathcal C_R,\;
\boldsymbol{\Gamma}\boldsymbol{\Omega}+\boldsymbol{\Omega}\boldsymbol{\Gamma}^{\top}\succ0\right\}.
\label{eq:C1}
\end{equation}
For every pair in $\mathcal F_R$, $\boldsymbol{\Gamma}$ is positive stable. Indeed, with $\mathbf{B}=-\boldsymbol{\Gamma}$, Equation~\eqref{eq:C1} is equivalent to
\[
\mathbf{B}\boldsymbol{\Omega}+\boldsymbol{\Omega}\mathbf{B}^{\top}\prec0,
\]
and the converse part of Lemma~\ref{lemma:lyapunov_stability} implies that $\mathbf{B}$ is Hurwitz. Thus the separate positive-stability requirement in the model specification is implied by $\boldsymbol{\Omega}\succ0$ and the strict Lyapunov covariance condition.

Define the structural parameter space
\begin{equation}
\mathcal Z_R=\mathbb S_{++}^R\times\mathfrak{so}(R)\times\mathcal C_R
\label{eq:C2}
\end{equation}
and the CLOUD map
\begin{equation}
\mathcal T(\mathbf{S},\mathbf{A},\boldsymbol{\Omega})=\bigl((\mathbf{S}+\mathbf{A})\boldsymbol{\Omega}^{-1},\boldsymbol{\Omega}\bigr).
\label{eq:C3}
\end{equation}

\begin{thrm}[Exact full-dimensional representation of valid CLOUD dynamics]\label{thm:C1}
The map $\mathcal T$ in Equation~\eqref{eq:C3} is a smooth bijection from $\mathcal Z_R$ onto $\mathcal F_R$. Its inverse is
\begin{equation}
\mathbf{S}=\frac{\boldsymbol{\Gamma}\boldsymbol{\Omega}+\boldsymbol{\Omega}\boldsymbol{\Gamma}^{\top}}{2},
\qquad
\mathbf{A}=\frac{\boldsymbol{\Gamma}\boldsymbol{\Omega}-\boldsymbol{\Omega}\boldsymbol{\Gamma}^{\top}}{2}.
\label{eq:C4}
\end{equation}
Consequently, the CLOUD representation covers exactly the valid drift--covariance pairs and imposes no additional restriction on the latent interaction matrix. Moreover,
\[
\dim\{\mathbf{S}\}+\dim\{\mathbf{A}\}=\frac{R(R+1)}2+\frac{R(R-1)}2=R^2,
\]
so the structural drift coordinates retain all the degrees of freedom of an unrestricted $R\times R$ drift matrix.
\end{thrm}

\begin{proof}
Take $(\mathbf{S},\mathbf{A},\boldsymbol{\Omega})\in\mathcal Z_R$ and set $\boldsymbol{\Gamma}=(\mathbf{S}+\mathbf{A})\boldsymbol{\Omega}^{-1}$. Then
\[
\boldsymbol{\Gamma}\boldsymbol{\Omega}+\boldsymbol{\Omega}\boldsymbol{\Gamma}^{\top}
=(\mathbf{S}+\mathbf{A})+(\mathbf{S}+\mathbf{A})^{\top}=2\mathbf{S}\succ0,
\]
so $(\boldsymbol{\Gamma},\boldsymbol{\Omega})\in\mathcal F_R$.

Conversely, let $(\boldsymbol{\Gamma},\boldsymbol{\Omega})\in\mathcal F_R$ and define $\mathbf{S}$ and $\mathbf{A}$ by Equation~\eqref{eq:C4}. Then $\mathbf{S}=\mathbf{S}^{\top}\succ0$, $\mathbf{A}^{\top}=-\mathbf{A}$, and
\[
\mathbf{S}+\mathbf{A}=\boldsymbol{\Gamma}\boldsymbol{\Omega},
\]
which gives $\boldsymbol{\Gamma}=(\mathbf{S}+\mathbf{A})\boldsymbol{\Omega}^{-1}$. The decomposition of the matrix $\boldsymbol{\Gamma}\boldsymbol{\Omega}$ into symmetric and skew-symmetric parts is unique, so the inverse is unique. Matrix multiplication and inversion are smooth on $\mathbb S_{++}^R$, establishing smoothness in both directions. The dimension statement follows from the dimensions of the symmetric and skew-symmetric matrix spaces.
\end{proof}

\begin{remark}[Model interpretation of the structural components]\label{rem:C1}
Equation~\eqref{eq:C4} implies
\[
\boldsymbol{\Sigma}:=\boldsymbol{\Gamma}\boldsymbol{\Omega}+\boldsymbol{\Omega}\boldsymbol{\Gamma}^{\top}=2\mathbf{S}.
\]
Thus $\mathbf{S}$ is one half of the infinitesimal innovation covariance of the centered latent OU process, while $\mathbf{A}$ is the skew component of $\boldsymbol{\Gamma}\boldsymbol{\Omega}$ not determined by that covariance. These components are therefore not introduced solely as computational auxiliaries.
\end{remark}

\subsection{Variation independence}
A parameterization is variation independent when its joint parameter space is the Cartesian product of the ranges of its components. This distinction is useful here because the direct OU validity restriction couples $\boldsymbol{\Gamma}$ and $\boldsymbol{\Omega}$.

\begin{prop}[Variation independence of the CLOUD structural coordinates]\label{prop:C2}
The structural coordinates $(\mathbf{S},\mathbf{A},\boldsymbol{\Omega})$ are variation independent on $\mathcal Z_R$. In contrast, the directly parameterized pair $(\boldsymbol{\Gamma},\boldsymbol{\Omega})$ is variation dependent for every $R\ge2$: a drift matrix that is admissible for one $\boldsymbol{\Omega}$ need not be admissible for another.
\end{prop}

\begin{proof}
The first claim follows immediately from the product form in Equation~\eqref{eq:C2}. For the second claim, a two-dimensional example suffices. Let
\[
\boldsymbol{\Omega}_1=\begin{pmatrix}1&-4/5\\-4/5&1\end{pmatrix},\qquad
\mathbf{S}_1=\mathbf{I}_2,\qquad
\mathbf{A}_1=\begin{pmatrix}0&-2\\2&0\end{pmatrix},
\]
and define
\[
\boldsymbol{\Gamma}_1=(\mathbf{S}_1+\mathbf{A}_1)\boldsymbol{\Omega}_1^{-1}
=\begin{pmatrix}-5/3&-10/3\\70/9&65/9\end{pmatrix}.
\]
Then $\boldsymbol{\Gamma}_1\boldsymbol{\Omega}_1+\boldsymbol{\Omega}_1\boldsymbol{\Gamma}_1^{\top}=2\mathbf{I}_2\succ0$, so $(\boldsymbol{\Gamma}_1,\boldsymbol{\Omega}_1)$ is valid. However, with $\boldsymbol{\Omega}=\mathbf{I}_2$, the $(1,1)$ entry of $\boldsymbol{\Gamma}_1+\boldsymbol{\Gamma}_1^{\top}$ is $-10/3$, so $(\boldsymbol{\Gamma}_1,\mathbf{I}_2)$ is invalid. For $R>2$, take $\boldsymbol{\Gamma}=\operatorname{diag}(\boldsymbol{\Gamma}_1,\mathbf I_{R-2})$ and $\boldsymbol{\Omega}=\operatorname{diag}(\boldsymbol{\Omega}_1,\mathbf I_{R-2})$. The valid pair has a positive-definite Lyapunov matrix, whereas replacing $\boldsymbol{\Omega}$ by $\mathbf I_R$ leaves the invalid leading principal block unchanged.
\end{proof}

\subsection{Dimension-dependent feasibility under direct ambient priors}
The preceding result is geometric. The next result quantifies one consequence of the common strategy of assigning simple centered priors or random initializations directly to the entries of $\boldsymbol{\Gamma}$ and then enforcing OU validity.

\begin{assumption}[Centered row-independent direct ambient distribution]\label{ass:C1}
Conditional on a fixed $\boldsymbol{\Omega}\succ0$, the rows $\boldsymbol{\gamma}_1^{\top},\ldots,\boldsymbol{\gamma}_R^{\top}$ of $\boldsymbol{\Gamma}$ are independent; each row is centrally symmetric in distribution, $\boldsymbol{\gamma}_r\overset d=-\boldsymbol{\gamma}_r$; and $\boldsymbol{\gamma}_r^{\top}\boldsymbol{\Omega}\mathbf{e}_r$ has no atom at zero. Independent centered Gaussian, Student, or Laplace row distributions satisfy this condition.
\end{assumption}

\begin{thrm}[Feasible mass under direct drift coordinates]\label{thm:C3}
{\samepage
Under Assumption~\ref{ass:C1}, for every fixed $\boldsymbol{\Omega}\succ0$,
\begin{equation}
Z_R(\boldsymbol{\Omega})
:=\Pr\{\boldsymbol{\Gamma}\boldsymbol{\Omega}+\boldsymbol{\Omega}\boldsymbol{\Gamma}^{\top}\succ0\mid\boldsymbol{\Omega}\}
\le 2^{-R}.
\label{eq:C5}
\end{equation}
\par}
If $\boldsymbol{\Omega}$ is random and Assumption~\ref{ass:C1} holds conditionally on $\boldsymbol{\Omega}$, then the unconditional probability of feasibility is also at most $2^{-R}$. By contrast, every structural draw with $\mathbf{S}\succ0$ and $\boldsymbol{\Omega}\succ0$ is valid by construction.
\end{thrm}

\begin{proof}
Let
\[
\mathbf{H}=\boldsymbol{\Gamma}\boldsymbol{\Omega}+\boldsymbol{\Omega}\boldsymbol{\Gamma}^{\top}.
\]
Positive definiteness requires $\mathbf{H}_{rr}>0$ for every $r$. Because $\boldsymbol{\Omega}$ is symmetric,
\[
\mathbf{H}_{rr}=2(\boldsymbol{\Gamma}\boldsymbol{\Omega})_{rr}=2\boldsymbol{\gamma}_r^{\top}\boldsymbol{\Omega}\mathbf{e}_r.
\]
Central symmetry and the no-atom condition imply
\[
\Pr(\mathbf{H}_{rr}>0\mid\boldsymbol{\Omega})=\frac12.
\]
The $R$ diagonal events are independent because they depend on independent rows of $\boldsymbol{\Gamma}$. Hence
\[
\Pr(\mathbf{H}\succ0\mid\boldsymbol{\Omega})
\le \Pr(\mathbf{H}_{11}>0,\ldots,\mathbf{H}_{RR}>0\mid\boldsymbol{\Omega})
=2^{-R}.
\]
Integrating over a random $\boldsymbol{\Omega}$ preserves the bound. For the structural representation,
$\boldsymbol{\Gamma}\boldsymbol{\Omega}+\boldsymbol{\Omega}\boldsymbol{\Gamma}^{\top}=2\mathbf{S}\succ0$ identically.
\end{proof}

\begin{remark}[Scope of Theorem~\ref{thm:C3}]\label{rem:C2}
Theorem~\ref{thm:C3} characterizes feasibility for direct drift coordinates under a broad class of centrally symmetric ambient priors or initializations. Geometry-aware direct methods can instead use the push-forward of a structural prior, projection, barriers, or other constrained proposals. The bijection in Theorem~\ref{thm:C1} ensures that such equivalent coordinate systems remain possible; the practical advantage of the CLOUD construction is to encode admissibility directly in the sampling coordinates.
\end{remark}

\subsection{Bayesian consequence: support-induced prior coupling}
\begin{prop}[Constraint-induced prior coupling in direct coordinates]\label{prop:C4}
Suppose a nominal direct prior factorizes as
\[
q(\boldsymbol{\Gamma},\boldsymbol{\Omega})=q_{\boldsymbol{\Gamma}}(\boldsymbol{\Gamma})q_{\boldsymbol{\Omega}}(\boldsymbol{\Omega})
\]
and is conditioned on the valid event $\mathcal F_R$. Assume this event has strictly positive probability under the nominal prior. Then the resulting marginal prior for $\boldsymbol{\Omega}$ is
\begin{equation}
q(\boldsymbol{\Omega}\mid\mathcal F_R)
=\frac{q_{\boldsymbol{\Omega}}(\boldsymbol{\Omega})Z_R(\boldsymbol{\Omega})}
{\int q_{\boldsymbol{\Omega}}(\widetilde{\boldsymbol{\Omega}})Z_R(\widetilde{\boldsymbol{\Omega}})\,d\widetilde{\boldsymbol{\Omega}}}.
\label{eq:C6}
\end{equation}

In particular, whenever $Z_R(\boldsymbol{\Omega})$ varies on a set of positive $q_{\boldsymbol{\Omega}}$-mass, the validity restriction changes the marginal prior for $\boldsymbol{\Omega}$ and prevents the two direct coordinates from remaining independent. Constancy of $Z_R$ alone would not guarantee independence; that would additionally require the conditional truncated distribution of $\boldsymbol{\Gamma}$ to be the same for almost every $\boldsymbol{\Omega}$.
\end{prop}

\begin{proof}
The conditioned density is proportional to
\[
q_{\boldsymbol{\Gamma}}(\boldsymbol{\Gamma})q_{\boldsymbol{\Omega}}(\boldsymbol{\Omega})
\mathbf 1\{(\boldsymbol{\Gamma},\boldsymbol{\Omega})\in\mathcal F_R\}.
\]
Integrating over $\boldsymbol{\Gamma}$ gives Equation~\eqref{eq:C6}. Indeed, conditional on a fixed $\boldsymbol{\Omega}$, the accepted drift has density proportional to $q_{\boldsymbol{\Gamma}}(\boldsymbol{\Gamma})\mathbf1_{\mathcal F_R}(\boldsymbol{\Gamma},\boldsymbol{\Omega})/Z_R(\boldsymbol{\Omega})$ whenever $Z_R(\boldsymbol{\Omega})>0$. Its dependence on $\boldsymbol{\Omega}$ establishes the claimed prior coupling when the feasible mass varies.
\end{proof}

\begin{remark}[Loading restrictions and prior support]\label{rem:C3}
Parameter constraints and prior regularization play distinct roles. In particular, positivity of the anchor loadings in Equation~\eqref{eq:struct} is required to fix factor orientation. Non-anchor loadings need not be positive for the general identification argument. In the fitted model, loading priors truncated to positive values impose an additional prior-support restriction. The large-sample result in \CloudSupplementAppendixName~\ref{appd:inf_proof} therefore requires only that the induced prior be positive near the true parameter within the actual parameter space used in a given fit.
\end{remark}

\subsection{Gaussian specialization of the feasibility probability}
The general bound above is sufficient for Proposition~\ref{prop:main_direct_feasibility}. A stronger specialization illustrates how small the feasible mass can be under direct parameterization.

\begin{corollary}[Gaussian benchmark at $\boldsymbol{\Omega}=\mathbf{I}_R$]\label{cor:C5}
If $\boldsymbol{\Omega}=\mathbf{I}_R$ and the entries of $\boldsymbol{\Gamma}$ are iid $N(0,\tau^2)$ with $\tau>0$, then OU validity is equivalent to positive definiteness of the symmetric part $(\boldsymbol{\Gamma}+\boldsymbol{\Gamma}^{\top})/2$, which is a scaled Gaussian orthogonal ensemble (GOE) matrix. Consequently,
\[
\log Z_R(\mathbf{I}_R)
=-\frac{\log 3}{4}R^2+o(R^2),\qquad R\to\infty.
\]
\end{corollary}

\begin{proof}
The diagonal entries of $(\boldsymbol{\Gamma}+\boldsymbol{\Gamma}^{\top})/2$ follow $N(0,\tau^2)$, the distinct off-diagonal entries follow $N(0,\tau^2/2)$, and the entries on and above the diagonal are mutually independent. Define $\mathbf G_R=(\sqrt{2}/\tau)(\boldsymbol{\Gamma}+\boldsymbol{\Gamma}^{\top})/2$. This matrix is symmetric, with independent $N(0,2)$ diagonal entries and independent $N(0,1)$ entries strictly above the diagonal; entries below the diagonal are determined by symmetry. This is the Gaussian orthogonal ensemble (GOE) in the unnormalized convention \citep{AndersonGuionnetZeitouni2010}. Positive scaling does not affect definiteness. Because $\boldsymbol{\Omega}=\mathbf I_R$, OU validity is equivalent to $\mathbf G_R\succ0$; hence $Z_R(\mathbf I_R)=\Pr(\mathbf G_R\succ0)$. The stated logarithmic asymptotic is the large-deviation result for all-positive eigenvalues established by \citet{DeanMajumdar2006}, with Dyson index $\beta=1$ for the real symmetric GOE.
\end{proof}

\clearpage

\suppsection{Noncentered parameterization}{appd:ncp}
The noncentered parameterization (NCP) represents each latent state as a deterministic transformation of independent standard normal variables and the population parameters \citep{papaspiliopoulos2007general,Betancourt2013HamiltonianMC}. For subject $i$, define
\[
\mathbf z_{ij}\overset{\mathrm{ind}}{\sim}\mathcal N_R(\mathbf0,\mathbf I_R),\qquad j=1,\ldots,n_i.
\]
The vector $\mathbf z_{i1}$ is the standardized initial deviation, and $\mathbf z_{ij}$ for $j\ge2$ is a standardized transition innovation. These variables are independent in the prior; the resulting centered OU residuals $\mathbf r_i(t_{ij})$ are temporally correlated.

\subsection{Exact noncentered recursion}
Let $\mathbf L_\Omega\mathbf L_\Omega^\top=\boldsymbol\Omega$ and, for each positive visit gap $\Delta t_{ij}$, define
\begin{align}
\mathbf T_{ij}&=e^{-\boldsymbol\Gamma\Delta t_{ij}},\\
\mathbf Q_{ij}&=\boldsymbol\Omega-\mathbf T_{ij}\boldsymbol\Omega\mathbf T_{ij}^\top,
\qquad \mathbf L_{Q,ij}\mathbf L_{Q,ij}^\top=\mathbf Q_{ij}.
\end{align}
The Cholesky factors have positive diagonal entries. The initial state and subsequent states are constructed as
\begin{align}
\boldsymbol\xi_i(t_{i1})&=\boldsymbol\mu_i(t_{i1})+\mathbf L_\Omega\mathbf z_{i1},\\
\boldsymbol\xi_i(t_{ij})&=\boldsymbol\mu_i(t_{ij})+
\mathbf T_{ij}\{\boldsymbol\xi_i(t_{i,j-1})-\boldsymbol\mu_i(t_{i,j-1})\}
+\mathbf L_{Q,ij}\mathbf z_{ij},\quad j\ge2.
\end{align}
Equivalently, the centered residual recursion is
\[
\mathbf r_i(t_{i1})=\mathbf L_\Omega\mathbf z_{i1},\qquad
\mathbf r_i(t_{ij})=\mathbf T_{ij}\mathbf r_i(t_{i,j-1})+\mathbf L_{Q,ij}\mathbf z_{ij}.
\]
For fixed population parameters, this triangular transformation is one-to-one. Integrating over the independent Gaussian variables recovers exactly the initial distribution and transition law in Appendix~\ref{appd:likelihood}. The sampling density is specified in the $\mathbf z_{ij}$ coordinates, with the latent states obtained deterministically. The persistent item intercepts have the analogous representation $b_{ik}=\sigma_{bk}b_{ik}^{\mathrm{std}}$, where $b_{ik}^{\mathrm{std}}\sim\mathcal N(0,1)$ independently.

\subsection{Numerical evaluation}
For an admissible drift--covariance pair and $\Delta t_{ij}>0$, $\mathbf Q_{ij}$ is positive definite. In floating-point arithmetic, replacing its computed value by $(\mathbf Q_{ij}+\mathbf Q_{ij}^\top)/2$ removes numerical asymmetry before Cholesky factorization. Diagonal regularization $\mathbf Q_{ij}+\epsilon_{\mathrm{num}}\mathbf I_R$, with $\epsilon_{\mathrm{num}}>0$, instead increases every transition variance by $\epsilon_{\mathrm{num}}$ and therefore represents a perturbed transition covariance. The exact NCP above uses $\mathbf Q_{ij}$ itself. Repeated measurements at the same time share a latent state, so no positive-variance transition is assigned to a zero-length interval.

The NCP separates standardized latent randomness from the population covariance and scale parameters. Its computational behavior depends on the information in the observations and the posterior geometry. The matched computational experiment evaluated the complete CLOUD implementation, including this representation and the admissible drift coordinates.

\suppsection{Identifiability of the CLOUD framework}{appd:id_proof}

This appendix establishes sufficient conditions for identifying CLOUD population parameters from the distribution of observed ordinal responses. The assumptions specify the measurement link, factor-coordinate constraints, covariate support, follow-up times, and recording process. Under these assumptions, the latent covariances and transitions are identified as consequences of the model. We present the argument for the ordered-logistic responses used in the simulations and PPMI analysis, along with the corresponding ordered-probit argument. Identification depends on informative population designs; finite-sample posterior precision is addressed separately.

\subsection{Latent-coordinate invariance and identifying restrictions}
The loading--state product in Equation~\eqref{eq:irt} is unchanged by an invertible $R\times R$ change of latent coordinates. Specifically, for any nonsingular $\mathbf P$, define
\[
 \boldsymbol\xi_i^*(t)=\mathbf P\boldsymbol\xi_i(t),\quad
 \boldsymbol\Lambda^*=\boldsymbol\Lambda\mathbf P^{-1},\quad
 \boldsymbol\Omega^*=\mathbf P\boldsymbol\Omega\mathbf P^\top,
\]
\[
 \boldsymbol\Gamma^*=\mathbf P\boldsymbol\Gamma\mathbf P^{-1},\quad
 \boldsymbol\Phi^*=\mathbf P\boldsymbol\Phi,\quad
 \boldsymbol\alpha^*=\mathbf P\boldsymbol\alpha.
\]
Both the categorical measurement predictor and the conditional OU transition law are therefore invariant under this simultaneous transformation. Consequently, without identifying restrictions, the individual loadings, covariance, drift, and covariate-dependent latent trends cannot be distinguished from transformed versions solely through their induced observation law.

There are $R^2$ continuous degrees of freedom in this latent-coordinate transformation. The $R$ diagonal restrictions on $\boldsymbol\Omega$ impose a unit-variance convention: Equation~\eqref{eq:scal_norm} maps any SPD working covariance $\boldsymbol\Omega'$ to a correlation matrix using its diagonal standard deviations. Under the centered stationary initialization in Equation~\eqref{eq:initial_latent_state}, $\boldsymbol\Omega$ is also the covariance at the initial latent time. When the initial measurement is indexed at $t=0$, it is the covariance of that first latent observation. This standardization fixes the variance scale but does not by itself remove rotations or mixing of the latent axes.

The remaining $R(R-1)$ equality restrictions are supplied by the known anchor set in Equation~\eqref{eq:struct}: its $R\times R$ loading block is diagonal, while its diagonal entries are strictly positive. This leaves every non-anchor loading row unrestricted. To see why the restrictions fix the residual coordinate freedom, suppose two parameterizations both satisfy the anchor and unit-variance conventions. Write their anchor blocks as $\mathbf D=\operatorname{diag}(d_1,\ldots,d_R)$ and $\mathbf D^*=\operatorname{diag}(d_1^*,\ldots,d_R^*)$, with all diagonal entries positive. Since $\mathbf D^*=\mathbf D\mathbf P^{-1}$, the matrix $\mathbf P$ must be diagonal with positive diagonal entries. The two correlation-matrix conventions then imply $1=(\boldsymbol\Omega^*)_{rr}=P_{rr}^2\Omega_{rr}=P_{rr}^2$, so $P_{rr}=1$ for every $r$, and hence $\mathbf P=\mathbf I_R$. Thus no nontrivial transformation in this family preserves both restrictions; positivity also removes the sign ambiguity. The constraints identify the latent coordinate system but do not by themselves identify every population model parameter. Full parameter identification additionally requires the link, covariate, temporal-design, and recording assumptions stated next; it is established in Theorem~\ref{thm:E1}.

\subsection{Measurement representation and explicit design conditions}
Let $\boldsymbol\Psi=(\boldsymbol\theta,\boldsymbol\Lambda,\mathbf B_{\boldsymbol\beta},\boldsymbol\Sigma_b,\boldsymbol\Gamma,\boldsymbol\Omega,\boldsymbol\Phi,\boldsymbol\alpha)$ denote the population parameters of the measurement and dynamic models, with $\mathbf B_{\boldsymbol\beta}$ having rows $\boldsymbol\beta_k^\top$ and $\boldsymbol\Sigma_b=\operatorname{diag}(\sigma_{b1}^2,\ldots,\sigma_{bK}^2)$. Define the observable regression-design vector and its item-specific coefficient vector by
\begin{equation}
\mathbf w_{ij}=\begin{pmatrix}\mathbf x_{ij}^{(1)}\\ t_{ij}\mathbf x_i^{(2)}\\t_{ij}\end{pmatrix},\qquad
\mathbf g_k=\begin{pmatrix}\boldsymbol\beta_k\\ \boldsymbol\Phi^\top\boldsymbol\Lambda_k\\ \boldsymbol\Lambda_k^\top\boldsymbol\alpha\end{pmatrix},
\qquad d_w=p+q+1.
\label{eq:E1}
\end{equation}
Here $\boldsymbol\Lambda_k\in\mathbb R^R$ is the transpose of row $k$ of $\boldsymbol\Lambda$. Thus $\mathbf g_k^\top\mathbf w_{ij}=\boldsymbol\beta_k^\top\mathbf x_{ij}^{(1)}+\boldsymbol\Lambda_k^\top(\boldsymbol\Phi\mathbf x_i^{(2)}+\boldsymbol\alpha)t_{ij}$. The time origin and the absence of an unrestricted item intercept are important: a free additive item intercept would be confounded with the thresholds.

\begin{assumption}[Specified cumulative-link Gaussian measurement model]\label{ass:E1}
The link in Equation~\eqref{eq:irt} is fixed and known, either the standard logit or the standard probit. The thresholds for each item are finite and strictly ordered. Equivalently, for every distinct item--visit measurement, introduce independent errors $\epsilon_{ijk}$ with the standard logistic distribution (logit link) or $\mathcal N(0,1)$ (probit link), independent of the latent process and random effects, and write
\begin{equation}
Y_{ijk}\le m \quad\Longleftrightarrow\quad
\epsilon_{ijk}+\boldsymbol\Lambda_k^\top\mathbf r_i(t_{ij})+b_{ik}
\le \theta_{km}-\mathbf g_k^\top\mathbf w_{ij},
\qquad \mathbf r_i(t)=\boldsymbol\xi_i(t)-\boldsymbol\mu_i(t).
\label{eq:E2}
\end{equation}
The centered process $\mathbf r_i$ is a stationary, mean-zero Gaussian OU process with covariance $\boldsymbol\Omega\succ0$ and positive-stable drift $\boldsymbol\Gamma$, as induced by the stationary initialization in Equation~\eqref{eq:initial_latent_state}. The time-constant item random effects are mutually independent Gaussian random variables $b_{ik}\sim\mathcal N(0,\sigma_{bk}^2)$, independent of $\mathbf r_i$, with $\sigma_{bk}^2>0$. The response errors are independent across distinct item--visit pairs. The covariate and time design is exogenous for the centered latent process and random effects.
\end{assumption}
The cumulative-logit and cumulative-probit threshold representations are standard ordinal regression constructions \citep{McCullagh1980Ordinal,de2013theory}; their combination with Gaussian latent trajectories and item effects is the CLOUD model specification. In particular, the logistic error scale is fixed at one, defining the response-predictor scale.

\begin{assumption}[Latent scale and anchor orientation]\label{ass:E2}
$\operatorname{diag}(\boldsymbol\Omega)=\mathbf1_R$; there is a known anchor set $\mathcal A=\{a_1,\ldots,a_R\}$ with
\[
\boldsymbol\Lambda_{\mathcal A}=\mathbf D=\operatorname{diag}(d_1,\ldots,d_R),\qquad d_r>0.
\]
The remaining loading rows can contain arbitrary signed entries. These model restrictions fix latent-coordinate scale, orientation, and signs.
\end{assumption}

\begin{assumption}[Observable covariate variation]\label{ass:E3}
The affine span of the population support $\mathcal W_{\mathrm{reg}}$ of $\mathbf w$ is all of $\mathbb R^{d_w}$:
\[
\operatorname{span}\{\mathbf w-\mathbf w':\mathbf w,\mathbf w'\in\mathcal W_{\mathrm{reg}}\}=\mathbb R^{d_w}.
\]
For the logit link, we additionally require that, for every item $k$, there exist $\mathbf w_0$, a nonzero direction $\mathbf h_k$, and $\delta>0$ such that $\mathbf w_0+s\mathbf h_k\in\mathcal W_{\mathrm{reg}}$ for $|s|<\delta$ and $\mathbf g_k^\top\mathbf h_k\ne0$. This requires a continuous predictor that varies for each item. The extra itemwise condition is not needed for the probit link under Assumption~\ref{ass:E4}.
\end{assumption}

\begin{assumption}[Informative observation lags]\label{ass:E4}
There exists a nonempty open interval $J\subset(0,\infty)$ such that, for every ordered item pair $k,l$ and each $\tau\in J$, the population visit design has support for observations of item $k$ at time $s$ and item $l$ at a distinct visit $s+\tau$, together with their covariates. The design supports the conditional item and item-pair category probabilities appearing below. Neither arbitrarily long follow-up nor visits with exactly zero separation are required. An interval within the positive-lag support permits analytic uniqueness.
\end{assumption}

\begin{assumption}[Complete recording or sequentially ignorable recording]\label{ass:E5}
The item--visit outcomes needed in Assumption~\ref{ass:E4} are fully recorded, or the observation indicators $R_1,\ldots,R_L$ in a fixed within-subject item--visit ordering have conditional observation probabilities
\[
\mathbb P(R_j=1\mid \mathbf Y^{\mathrm{full}},R_{1:j-1},\mathbf w)
=\pi_j(\mathbf Y_{1:j-1}^{\mathrm{obs}},R_{1:j-1},\mathbf w)>0,
\]
where each $\pi_j$ is a function only of recorded history and design. These probabilities are determined from the observed recording process and do not depend on unrecorded current or future responses; nuisance missingness parameters are variation independent of $\boldsymbol\Psi$. This is an explicit sequential version of the ignorability and positivity conditions associated with missing-at-random (MAR) analyses \citep{Rubin1976Missing}. It is stronger than generic MAR and is not assumed automatically for informative dropout.
\end{assumption}

\subsection{Identification from categorical probabilities}
We first derive identification from observable marginal and pairwise response probabilities. Let $\mathsf F$ be the known cumulative distribution function (CDF) of a standard logistic or standard normal error, and define
\[
 v_k=\boldsymbol\Lambda_k^\top\boldsymbol\Omega\boldsymbol\Lambda_k+\sigma_{bk}^2,
 \qquad G_v(u)=\mathbb E\{\mathsf F(u-\sqrt v Z)\},\quad Z\sim\mathcal N(0,1).
\]
Marginalizing Equation~\eqref{eq:E2} yields the directly observable marginal probability
\begin{equation}
\mathbb P(Y_{ijk}\le m\mid\mathbf w_{ij}=\mathbf w)
=G_{v_k}(\theta_{km}-\mathbf g_k^\top\mathbf w).
\label{eq:E3}
\end{equation}

\begin{lem}[Item scale and threshold recovery from ordinal marginals]\label{lem:E1}
Under As\-sump\-tions~\ref{ass:E1} and \ref{ass:E3}, for the ordered-logistic link, the marginal categorical probabilities in Equation~\eqref{eq:E3} uniquely determine $v_k$, every $\theta_{km}$, and $\mathbf g_k$, separately for each item. For the probit link, they determine $\theta_{km}/s_k$ and $\mathbf g_k/s_k$, where $s_k=\sqrt{1+v_k}$; the scale $s_k$ is recovered from cross-visit probabilities in Lemma~\ref{lem:E2} below.
\end{lem}
\begin{proof}
For the logit link, $G_v$ is the distribution function of $W_v=\epsilon+\sqrt v Z$, where $\epsilon$ is standard logistic and independent of $Z$. It is strictly increasing and real analytic in its argument. Its fourth cumulant is fixed and nonzero, independently of $v$:
\begin{equation}
\kappa_4(W_v)=\kappa_4(\epsilon)=\frac{2\pi^4}{15}>0,
\qquad \operatorname{Var}(W_v)=\frac{\pi^2}{3}+v.
\label{eq:Elogit}
\end{equation}
These identities follow directly from the logistic characteristic function $\varphi_\epsilon(u)=\pi u/\sinh(\pi u)$, since a Gaussian contributes no cumulants beyond order two. Consider two parameter sets producing the same probabilities on a segment $\mathbf w_0+s\mathbf h_k$ with nonzero predictor slope. Strict monotonicity shows that both slopes are nonzero and have the same sign. Writing the corresponding single-index arguments as $a-bs$ and $a'-b's$, equality on an open interval and real analyticity give $G_v(u)=G_{v'}(cu+d)$ for every real $u$, where $c=b'/b>0$. Consequently, $W_v$ is distributed as $(W_{v'}-d)/c$. Comparing fourth cumulants in Equation~\eqref{eq:Elogit} gives $c^4=1$, hence $c=1$; comparing means gives $d=0$, and comparing variances gives $v=v'$.

With $v$ now known, apply the inverse function $G_v^{-1}$ to Equation~\eqref{eq:E3}. Equality of observed marginal functions under two parameter sets yields
\[
\theta_{km}-\mathbf g_k^\top\mathbf w
=\theta'_{km}-(\mathbf g'_k)^\top\mathbf w
\quad\text{for every }\mathbf w\in\mathcal W_{\mathrm{reg}}.
\]
The full affine-span condition forces $\mathbf g_k=\mathbf g_k'$ and then $\theta_{km}=\theta'_{km}$, for all thresholds.

For the probit link, convolution of independent normal random variables gives $G_v(u)=\Phi(u/\sqrt{1+v})$. Applying $\Phi^{-1}$ to each marginal cumulative response probability and using the same affine-span argument identifies precisely the normalized coefficients and thresholds. No claim about $s_k$ follows from those univariate probabilities alone.
\end{proof}

Next, define the covariance of the Gaussian item-level predictors at distinct visits a lag $\tau>0$ apart:
\begin{equation}
\mathbf K(\tau)
=\boldsymbol\Lambda\boldsymbol\Omega e^{-\boldsymbol\Gamma^\top\tau}\boldsymbol\Lambda^\top+\boldsymbol\Sigma_b.
\label{eq:Ecov}
\end{equation}
This is a property of the Gaussian OU process and the time-constant random effects, derived by taking covariances. It is not the covariance of numerical ordinal scores.

\begin{lem}[Recovery of temporal predictor covariances from ordinal pairs]\label{lem:E2}
Under Assumptions~\ref{ass:E1}, \ref{ass:E3}, and \ref{ass:E4}, the conditional distribution of observed ordinal pairs determines $\mathbf K(\tau)$ for all $\tau\in J$. This holds for the ordered-logistic and ordered-probit links, without assuming recovery of a latent-predictor covariance.
\end{lem}
\begin{proof}
For the logistic case, the univariate variances $v_k$, threshold values, and predictor locations are determined by Lemma~\ref{lem:E1}. For any finite thresholds $u=\theta_{km}-\mathbf g_k^\top\mathbf w$ and $v=\theta_{ln}-\mathbf g_l^\top\mathbf w'$, the observable joint cumulative probability equals
\[
H(c)=\mathbb E\{\mathsf F(u-U)\mathsf F(v-V)\},
\quad (U,V)\sim\mathcal N\!\left(\mathbf0,
\begin{pmatrix}v_k&c\\c&v_l\end{pmatrix}\right),
\]
with $c=K_{kl}(\tau)$. Gaussian differentiation with respect to the off-diagonal covariance, followed by integration by parts, gives
\begin{equation}
\frac{\mathrm d H(c)}{\mathrm dc}
=\mathbb E\{\mathsf f(u-U)\mathsf f(v-V)\}>0,
\label{eq:Ederiv}
\end{equation}
where $\mathsf f=\mathsf F'$ is strictly positive. This is the Gaussian covariance-derivative identity known as Price's theorem; a general rigorous formulation is given by \citet{Voigtlaender2021Price}. Here, differentiation is justified by the bounded logistic density and the nonsingular Gaussian covariance. Thus the observed joint cumulative response probability uniquely determines $c$ throughout its admissible interval. Positive random-effect and transition variances ensure the relevant distinct measurements are nondegenerate.

For the probit link, the pair of latent threshold variables, including independent unit-variance response errors, is bivariate normal with marginal standard deviations $s_k=\sqrt{1+v_k}$. The univariate response probabilities determine normalized cutpoints $a$ and $b$, and the joint probability has the form $\Phi_2(a,b;\rho_{kl}(\tau))$, where $\rho_{kl}(\tau)=K_{kl}(\tau)/(s_k s_l)$. By the positive derivative $\partial\Phi_2(a,b;\rho)/\partial\rho=\varphi_2(a,b;\rho)>0$ \citep{Plackett1954Normal}, the observable joint probability uniquely determines $\rho_{kl}(\tau)$ for $\tau\in J$. This recovers the entire normalized covariance function on $J$.

Each entry of $\mathbf K(\tau)$ is real analytic in $\tau>0$, because it is a matrix-exponential entry plus a constant. The same is true of the normalized probit covariance function. Equality on the open interval $J$ therefore implies equality for every $\tau>0$ (identity theorem for real-analytic functions). In the probit case, for the same item observed twice at distinct, arbitrarily close times,
\[
\lim_{\tau\downarrow0}\rho_{kk}(\tau)=\frac{v_k}{1+v_k}.
\]
Hence $v_k=\rho_{kk}(0^+)/[1-\rho_{kk}(0^+)]$ and $s_k=(1-\rho_{kk}(0^+))^{-1/2}$ are identified. The normalized thresholds and coefficients from Lemma~\ref{lem:E1} are thereby converted to their original scales, and all $K_{kl}(\tau)=s_ks_l\rho_{kl}(\tau)$ are identified. Analytic continuation supplies the population-level uniqueness argument.
\end{proof}

\begin{lem}[Recovery of response probabilities under sequential recording]\label{lem:E3}
Under Assumption~\ref{ass:E5}, the complete-response joint probabilities needed in Lemmas~\ref{lem:E1}--\ref{lem:E2} are determined by the recorded-data distribution and the specified design.
\end{lem}
\begin{proof}
The fully recorded case is immediate. Otherwise, let $\mathbf R=(R_1,\ldots,R_L)$, and let $\pi_{\mathrm{all}}(\mathbf y,\mathbf w)$ be the product of the observable sequential response-recording probabilities along the path $\mathbf R=\mathbf1$, each evaluated at the previously recorded response history. By the sequential conditional-independence restriction in Assumption~\ref{ass:E5},
\[
\mathbb P(\mathbf R=\mathbf1\mid\mathbf Y^{\mathrm{full}}=\mathbf y,\mathbf w)
=\pi_{\mathrm{all}}(\mathbf y,\mathbf w)>0.
\]
The complete-case joint law $\mathbb P(\mathbf Y^{\mathrm{obs}}=\mathbf y,\mathbf R=\mathbf1\mid\mathbf w)$ and each factor in $\pi_{\mathrm{all}}$ are functionals of the observed-data distribution. Dividing this joint law by $\pi_{\mathrm{all}}(\mathbf y,\mathbf w)$ therefore recovers $\mathbb P(\mathbf Y^{\mathrm{full}}=\mathbf y\mid\mathbf w)$. Positivity therefore identifies the required full-response probabilities from recorded data.
\end{proof}

\subsection{Identification of the complete CLOUD parameter vector}
\begin{thrm}[Identifiability of the CLOUD model parameters]\label{thm:E1}
Under Assumptions~\ref{ass:E1}--\ref{ass:E5}, and for either the specified standard ordered-logistic or ordered-probit link, the joint distribution of the recorded longitudinal categorical responses and the observable covariate and time design uniquely determines
\[
\boldsymbol\Psi=
\{\boldsymbol\theta,\boldsymbol\Lambda,\mathbf B_{\boldsymbol\beta},
\boldsymbol\Sigma_b,\boldsymbol\Gamma,\boldsymbol\Omega,
\boldsymbol\Phi,\boldsymbol\alpha\}.
\]
For the logistic case, the item-specific continuous-predictor condition in Assumption~\ref{ass:E3} is required; for the probit case, it is not. When the Lyapunov parameterization in Equation~\eqref{eq:gamma_param} is used, its symmetric and skew-symmetric structural matrices are also uniquely determined. The result applies to the stated observation-design class and latent-coordinate constraints.
\end{thrm}
\begin{proof}
Lemmas~\ref{lem:E1}--\ref{lem:E3} recover the thresholds, the full coefficient vectors $\mathbf g_k$, and the covariance function $\mathbf K(\tau)$ on $J$ from the recorded ordinal law. Analytic continuation as in Lemma~\ref{lem:E2} makes $\mathbf K(\tau)$ unique for every positive lag. Positive stability implies $e^{-\boldsymbol\Gamma^\top\tau}\to\mathbf0$ as $\tau\to\infty$; hence Equation~\eqref{eq:Ecov} gives the identifiable random-effect covariance and centered latent-predictor covariances
\begin{equation}
\boldsymbol\Sigma_b=\lim_{\tau\to\infty}\mathbf K(\tau),
\qquad \mathbf C(\tau)=\mathbf K(\tau)-\boldsymbol\Sigma_b,
\qquad \mathbf C_0=\lim_{\tau\downarrow0}\mathbf C(\tau)
=\boldsymbol\Lambda\boldsymbol\Omega\boldsymbol\Lambda^\top.
\label{eq:E4}
\end{equation}

Consider the anchor--anchor block $\mathbf C_{0,\mathcal A\mathcal A}=\mathbf D\boldsymbol\Omega\mathbf D$. Since $\operatorname{diag}(\boldsymbol\Omega)=\mathbf1_R$ and $d_r>0$,
\begin{equation}
 d_r=\sqrt{(\mathbf C_{0,\mathcal A\mathcal A})_{rr}},
 \qquad \boldsymbol\Omega=\mathbf D^{-1}\mathbf C_{0,\mathcal A\mathcal A}\mathbf D^{-1},
 \qquad \boldsymbol\Lambda_k^\top
 =\mathbf C_{0,k\mathcal A}\mathbf D^{-1}\boldsymbol\Omega^{-1}.
\label{eq:E5}
\end{equation}
Thus all loadings and $\boldsymbol\Omega$ are unique, with $\boldsymbol\Lambda$ of full column rank. Let $\boldsymbol\Lambda^-=(\boldsymbol\Lambda^\top\boldsymbol\Lambda)^{-1}\boldsymbol\Lambda^\top$. From $\mathbf C(\tau)=\boldsymbol\Lambda\boldsymbol\Omega e^{-\boldsymbol\Gamma^\top\tau}\boldsymbol\Lambda^\top$,
\begin{equation}
\boldsymbol\Gamma^\top
=-\boldsymbol\Omega^{-1}
 \left.\frac{\mathrm d}{\mathrm d\tau}
 \big\{\boldsymbol\Lambda^-\mathbf C(\tau)
 (\boldsymbol\Lambda^-)^\top\big\}\right|_{\tau=0^+}.
\label{eq:E6}
\end{equation}
This derivative exists because the analytic matrix-exponential extension is unique. Equation~\eqref{eq:E6} identifies the continuous-time drift under the open-lag design in Assumption~\ref{ass:E4}.

Finally, stack the uniquely determined vectors $\mathbf g_k^\top$ into a $K\times d_w$ matrix $\mathbf G$. By their definition in Equation~\eqref{eq:E1},
\[
\mathbf G=\big[\,\mathbf B_{\boldsymbol\beta}\ \big|\ \boldsymbol\Lambda\boldsymbol\Phi\ \big|\ \boldsymbol\Lambda\boldsymbol\alpha\,\big].
\]
Its first $p$ columns identify $\mathbf B_{\boldsymbol\beta}$. Applying $\boldsymbol\Lambda^-$ to the next $q$ columns and the final column identifies $\boldsymbol\Phi$ and $\boldsymbol\alpha$, respectively. The full affine-span condition supplies the regression-identification step in Lemma~\ref{lem:E1}. The structural components follow uniquely from Theorem~\ref{thm:C1}:
\[
\mathbf S=\tfrac12(\boldsymbol\Gamma\boldsymbol\Omega+\boldsymbol\Omega\boldsymbol\Gamma^\top),
\qquad
\mathbf A=\tfrac12(\boldsymbol\Gamma\boldsymbol\Omega-\boldsymbol\Omega\boldsymbol\Gamma^\top).
\]
This completes identification of every listed population parameter.
\end{proof}

\subsection{Role of the observation-time design}
\begin{prop}[Generator identification from transitions at two incommensurate lags]\label{prop:E2}
If $\tau_1,\tau_2>0$ have an irrational ratio and two real matrices satisfy $e^{-\boldsymbol\Gamma\tau_j}=e^{-\widetilde{\boldsymbol\Gamma}\tau_j}$ for $j=1,2$, then $\boldsymbol\Gamma=\widetilde{\boldsymbol\Gamma}$.
\end{prop}
\begin{proof}
For each candidate, the two matrix exponentials commute, and their integer powers and inverses yield the exponential at any time $m\tau_1+n\tau_2$, $m,n\in\mathbb Z$. Because the ratio is irrational, this additive group is dense in $\mathbb R$. Equality at the two given lags and continuity of the exponential therefore imply equality for every real time. Differentiation at time zero gives $\boldsymbol\Gamma=\widetilde{\boldsymbol\Gamma}$.
\end{proof}

\begin{remark}[Common-grid aliasing and limits of the result]\label{rem:E1}
The open interval in Assumption~\ref{ass:E4} provides sufficient information for covariance extrapolation and generator recovery; merely assuming two incommensurate lags identifies the generator \emph{after} the transition matrices are known, but need not separately identify the item random-effect covariance from a small finite set of observed ordinal lags. Conversely, common-grid sampling can alias an OU generator, reflecting the continuous-time identification and matrix-logarithm issues discussed by \citet{phillips1973problem,Higham2008FunctionsOM}. For example, with $R=2$, $\boldsymbol\Omega=\mathbf I_2$, $a>0$, $\Delta>0$, and $\mathbf J=\left(\begin{smallmatrix}0&-1\\1&0\end{smallmatrix}\right)$, the distinct drifts $\boldsymbol\Gamma_1=a\mathbf I_2+\omega\mathbf J$ and $\boldsymbol\Gamma_2=a\mathbf I_2+(\omega+2\pi/\Delta)\mathbf J$ have the same transition matrices at all multiples of $\Delta$. Both are positive stable and satisfy $\boldsymbol\Gamma_j\boldsymbol\Omega+\boldsymbol\Omega\boldsymbol\Gamma_j^\top=2a\mathbf I_2\succ0$. With the same remaining CLOUD parameters, their categorical observation laws on that grid coincide. Thus common-grid sampling permits distinct admissible generators with identical observed laws.\par
\end{remark}

\suppsection{Large-sample inferential properties of CLOUD}{appd:inf_proof}
This appendix establishes fixed-dimensional large-sample properties for the CLOUD population parameters. The number of independent subjects $N$ increases while the latent dimension, number of items, covariate dimensions, and number of visits per subject remain bounded. The inferential targets are the population parameters and smooth functions of them, including latent mean slopes and finite-lag transitions.

\subsection{Subject-level marginal likelihood and asymptotic regime}
Let $\mathcal W_i$ contain subject $i$'s visit count, observation times, and covariates, and let $\mathbf M_i$ denote the response-observation mask. Write $\mathcal D_i=(\mathcal W_i,\mathbf M_i)$, let $\mathbf Y_i^{\mathrm{obs}}$ denote the recorded responses, and let $\mathbf O_i=(\mathcal D_i,\mathbf Y_i^{\mathrm{obs}})$ denote the observed subject cluster. The finite latent vector is $\boldsymbol\xi_{i,1:n_i}=(\boldsymbol\xi_i(t_{i1})^\top,\ldots,\boldsymbol\xi_i(t_{in_i})^\top)^\top$. For a fixed mask, the pattern-specific response probability is
\begin{equation}
\begin{split}
p_{\boldsymbol\Psi}^{\mathrm{pat}}(\mathbf Y_i^{\mathrm{obs}};\mathbf M_i,\mathcal W_i)
={}&\int\left\{\prod_{(j,k)\in\mathcal O_i}
 \Pr_\Psi(Y_{ijk}=y_{ijk}\mid\boldsymbol\xi_i(t_{ij}),b_{ik},\mathcal W_i)\right\}\\
&\quad\times p_\Psi(\boldsymbol\xi_{i,1:n_i}\mid\mathcal W_i)p_\Psi(\mathbf b_i)
\,d\boldsymbol\xi_{i,1:n_i}\,d\mathbf b_i,
\end{split}
\label{eq:J1}
\end{equation}
where $\mathcal O_i$ is the set of recorded item--visit pairs. Equivalently, for mask $m$ and design $w$, $p_\Psi^{\mathrm{pat}}(y^{\mathrm{obs}};m,w)=\sum_{y^{\mathrm{mis}}}p_\Psi(y^{\mathrm{obs}},y^{\mathrm{mis}}\mid w)$.

Assume that the pre-observation designs are independent and identically distributed (iid) with an exogenous law $g_0$ that is independent of $\boldsymbol\Psi$. Exogeneity here means that the visit schedule, conditional on the modeled covariates, is independent of the centered latent process and response errors. Let $r_{\boldsymbol\zeta}(m\mid y^{\mathrm{obs}},w)$ be a missing-at-random (MAR) recording law with distinct nuisance parameters $\boldsymbol\zeta$. Its dependence on the full response record is restricted to the values observed under mask $m$. The joint observed-data density is
\[
f_{\boldsymbol\Psi,\boldsymbol\zeta}(o)
=g_0(w)r_{\boldsymbol\zeta}(m\mid y^{\mathrm{obs}},w)
 p_{\boldsymbol\Psi}^{\mathrm{pat}}(y^{\mathrm{obs}};m,w).
\]
Thus $p_\Psi^{\mathrm{pat}}$ is a response-pattern likelihood factor; the recording law supplies the mask probability. For response-dependent recording, it is the product of these factors, rather than $p_\Psi^{\mathrm{pat}}$ alone, that defines the joint law of recorded responses and mask. Normalization follows from
\begin{align}
\sum_m\sum_{y^{\mathrm{obs}}}r_\zeta(m\mid y^{\mathrm{obs}},w)
 p_\Psi^{\mathrm{pat}}(y^{\mathrm{obs}};m,w)
&=\sum_{y^{\mathrm{full}}}p_\Psi(y^{\mathrm{full}}\mid w)
 \sum_m r_\zeta(m\mid y_m^{\mathrm{obs}},w)\\
&=1.
\label{eq:F_cluster_normalization}
\end{align}
Integration against $g_0$ gives the normalized cluster law. With variation-independent missingness parameters and independent priors, $g_0$ and $r_{\boldsymbol\zeta}$ cancel from likelihood ratios and the marginal posterior for $\boldsymbol\Psi$ \citep{Rubin1976Missing}. Prior hyperparameters are integrated out to obtain the induced prior on these population parameters.

Let $\boldsymbol\vartheta\in\Theta\subset\mathbb R^{d_\vartheta}$ denote nonredundant free coordinates for the fitted parameter space, with fixed dimension $d_\vartheta$. Near the truth, the coordinate map and its inverse are continuously differentiable. Examples include log-diagonal Cholesky coordinates, transformed partial correlations, free skew-symmetric entries, log-threshold spacings, and log-anchor loadings. Fixed structural-zero and sign restrictions are incorporated in $\Theta$. The observation-design vector $\mathbf w_{ij}$ in Appendix~\ref{appd:id_proof} has the separate dimension $d_w=p+q+1$. Throughout this appendix, $\tau$ denotes a prediction lag and $\mathbf u\in\mathbb R^{d_\vartheta}$ a local parameter displacement.

Write $P_0=P_{\boldsymbol\vartheta_0,\boldsymbol\zeta_0}$ for the true observed-cluster law, with expectation $E_0$ and variance $\operatorname{Var}_0$. Define the subject log likelihood, score, and total log likelihood by
\[
\ell_i(\boldsymbol\vartheta)=\log p_{\boldsymbol\vartheta}^{\mathrm{pat}}(\mathbf Y_i^{\mathrm{obs}};\mathbf M_i,\mathcal W_i),\quad
\mathbf s_i(\boldsymbol\vartheta)=\nabla_{\boldsymbol\vartheta}\ell_i(\boldsymbol\vartheta),\quad
L_N(\boldsymbol\vartheta)=\sum_{i=1}^N\ell_i(\boldsymbol\vartheta).
\]
The conditions below combine this smooth local representation with identification, information, and global separation of the observed-data likelihood \citep[Chapters~5, 7 and 10]{vanDerVaart1998}.

\subsection{Regularity conditions}
\begin{assumption}[Sparse independent-subject design]\label{ass:J1}
The observed subject clusters are iid. The dimensions $R,K,p,q$ are fixed, and each subject contributes $1\le n_i\le n_{\max}<\infty$ visits. Observation times and covariates have bounded support; for successive distinct occasions, $0<\delta\le t_{ij}-t_{i,j-1}\le H<\infty$. The population lag support satisfies Assumption~\ref{ass:E4}. This open-lag condition is compatible with a bounded number of visits per subject, for example, when each subject contributes two visits with a continuously distributed gap. The bounded design and positive minimum lag provide uniform bounds for the Gaussian integrations below.
\end{assumption}

\begin{assumption}[Regular interior of the fitted CLOUD model]\label{ass:J2}
The true parameter $\boldsymbol{\vartheta}_0$ is an interior point of the identified parameter space used to fit the model. In a neighborhood of $\boldsymbol{\vartheta}_0$, threshold spacings, anchor loadings, random-effect standard deviations, and the smallest eigenvalues of the required positive-definite matrices are bounded away from zero. The inverse link $h^{-1}$ is three times continuously differentiable, and its first three derivatives have at most polynomial growth in the linear predictor, uniformly over the allowed covariate support. Logistic and probit links satisfy this condition. The Cholesky, threshold-spacing, and matrix-exponential maps are smooth at such interior points. Requiring the diffusion covariance $\boldsymbol\Sigma=\boldsymbol\Gamma\boldsymbol\Omega+\boldsymbol\Omega\boldsymbol\Gamma^\top$ to remain positive definite implies nondegenerate OU transitions at every positive lag; this is satisfied by the strict Lyapunov parameterization in Equation~\eqref{eq:gamma_param}. These bounds are local (in a closed neighborhood of the truth) and do not require the entire unconstrained parameter space to be compact.
\end{assumption}

\begin{assumption}[Population identification and information]\label{ass:J3}
The population support of the design satisfies the conditions of Theorem~\ref{thm:E1}. The Fisher information
\[
\mathcal I(\boldsymbol{\vartheta}_0)
=\operatorname{Var}_{0}\{\mathbf s_i(\boldsymbol\vartheta_0)\}
\]
is finite and positive definite. Positive definiteness is a \emph{local} distinguishability condition, separate from the global population identification established by Theorem~\ref{thm:E1}. Specifically, for any nonzero direction $\mathbf a\in\mathbb R^{d_\vartheta}$, it requires the observed-data score $\mathbf a^\top\mathbf s_i(\boldsymbol\vartheta_0)$ to have positive variance; equivalently, some observable response/mask probability must have a nonzero directional derivative on a design set of positive probability. A globally injective parameterization need not have a full-rank derivative at every point, so this condition is not implied by identifiability alone. This condition excludes nonzero parameter directions with zero observed-data score variance at the true parameter.
\end{assumption}

\begin{assumption}[Global uniform likelihood separation]\label{ass:J4}

Let $L_N(\boldsymbol\vartheta)=\sum_{i=1}^N\ell_i(\boldsymbol\vartheta)$. For every fixed $\varepsilon>0$, there is a constant $c_\varepsilon>0$ such that
\begin{equation}
P_0\!\left\{\sup_{\boldsymbol\vartheta\in\Theta:\,\|\boldsymbol\vartheta-\boldsymbol\vartheta_0\|\geq\varepsilon}
\frac{L_N(\boldsymbol\vartheta)-L_N(\boldsymbol\vartheta_0)}{N}\leq -c_\varepsilon\right\}\longrightarrow1.
\label{eq:F_global_sep}
\end{equation}
This global separation condition excludes distant parameter values with asymptotically indistinguishable likelihoods. The theorem assumes a measurable global maximizer exists; the prior on $\Theta$ is proper.
One sufficient setting is a \emph{regular compact fitted parameter space}: $\Theta$ is compact, the bounds in Assumption~\ref{ass:J2} hold throughout $\Theta$, and the observed-data model is identified on $\Theta$ under the design in Assumption~\ref{ass:J3}. Then the expected criterion is continuous and, by identification, has a strictly positive Kullback--Leibler gap outside each neighborhood of the truth. The bounded subject-level derivatives give a uniform law of large numbers, which implies Equation~\eqref{eq:F_global_sep}. Compactness is one sufficient restriction on the fitted parameter space. For an unbounded fitted space, Equation~\eqref{eq:F_global_sep} must instead be established from additional global tail/coercivity bounds; neither bounded visit counts nor local Fisher information alone proves it.
\end{assumption}

\begin{assumption}[Ignorable schedule and missingness with local prior support]\label{ass:J5}
The observation schedule is exogenous conditional on the modeled covariates. There are no missing responses, or responses are missing at random with missingness parameters variation independent of $\boldsymbol{\Psi}$, and the prior for those nuisance parameters is independent of the prior for $\boldsymbol{\Psi}$. MAR means that, for each fixed observed pattern, the conditional probability of that pattern given the full response record depends only on responses observed under that pattern and on $\mathcal W_i$; this permits dependence on observed responses but not unobserved ones. Under these assumptions, the mechanism is ignorable for likelihood and Bayesian inference on $\boldsymbol{\Psi}$. The induced prior density for $\boldsymbol{\vartheta}$ is continuous and strictly positive at $\boldsymbol{\vartheta}_0$. The missing-data conditions are satisfied when, given the modeled covariates and observed response history, the decision to record subsequent items contains no further information about their unobserved values or about latent states. They may fail under informative dropout or clinically triggered follow-up. Assumption~\ref{ass:E5} additionally specifies the recording support needed to identify population parameters; MAR alone does not create otherwise unobserved design information. Independence of the priors for $\boldsymbol\zeta$ and $\boldsymbol\Psi$ is needed for the Bayesian ignorability factorization, whereas likelihood ignorability requires distinct parameter spaces \citep{Rubin1976Missing}.
\end{assumption}

\subsection{CLOUD-specific likelihood regularity}
\begin{lemma}[Smoothness of the integrated subject likelihood]\label{lem:J1}
Under Assumptions~\ref{ass:J1}--\ref{ass:J2}, there exists a neighborhood $U$ of $\boldsymbol{\vartheta}_0$ such that the subject-level probability $p_{\boldsymbol\vartheta}^{\mathrm{pat}}(\mathbf Y_i^{\mathrm{obs}};\mathbf M_i,\mathcal W_i)$ is strictly positive and three times continuously differentiable in $\boldsymbol{\vartheta}\in U$. Differentiation may be interchanged with the integrations over the subject random effects and latent states, and the first, second, and third derivatives of the integrated log likelihood (the score, Hessian, and third derivative, respectively) admit integrable envelopes uniformly on $U$. With the bounded design of Assumption~\ref{ass:J1}, these envelopes may be chosen as finite constants for all possible recorded response/mask patterns and allowable designs.
\end{lemma}

\begin{proof}

Fix a closed ball $U$ strictly inside the regular neighborhood of $\boldsymbol\vartheta_0$. Write $\mathbf u_{\mathrm{lat}}=(\boldsymbol\xi_{i,1:n_i},\mathbf b_i)$ for the stacked latent trajectory and item random effects. The number of coordinates of $\mathbf u_{\mathrm{lat}}$ is at most $R n_{\max}+K$. Conditional on $\mathcal W_i$, the joint latent law is Gaussian, with a bounded mean and covariance eigenvalues uniformly bounded above and away from zero on $U$. To see the nontrivial lower bound for successive OU innovations, use the Lyapunov identity
\begin{equation}
\mathbf Q(\tau)=\int_0^\tau e^{-\boldsymbol\Gamma u}\,
\boldsymbol\Sigma\,e^{-\boldsymbol\Gamma^\top u}\,\mathrm du,
\qquad \boldsymbol\Sigma=\boldsymbol\Gamma\boldsymbol\Omega+
\boldsymbol\Omega\boldsymbol\Gamma^\top\succ0.
\label{eq:F_Q_integral}
\end{equation}
For every $\mathbf a\ne0$ and $\tau>0$, the quadratic form of the integrand is strictly positive, because a matrix exponential is invertible. Hence $\mathbf Q(\tau)\succ0$. Its smallest eigenvalue varies continuously with $\tau$ and the parameters. Compactness of $U\times[\delta,H]$ therefore supplies a common positive lower bound. The initial covariance $\boldsymbol\Omega$ and random-effect covariance have the same property under Assumption~\ref{ass:J2}.

For the complete Gaussian state vector, center the stacked latent trajectory and random effects at their conditional mean and let $\mathbf e_i$ collect the independent innovations: the first centered state, $n_i-1$ OU transition innovations, and the $K$ independent item effects. The corresponding linear transformation is $\mathbf u_{\mathrm{lat}}-\mathbf m_i=\mathbf T_i(\boldsymbol\vartheta)\mathbf e_i$, where $\mathbf T_i$ is block lower triangular with identity diagonal blocks, and $\mathbf e_i$ has a block-diagonal covariance $\mathbf V_{\mathrm{inn},i}\succeq c_0\mathbf I$. Bounded $n_i$, continuous matrix exponentials, and compact neighborhoods of the design and parameters ensure $\sup_i\|\mathbf T_i^{-1}\|_{\mathrm{op}}\le C_T<\infty$. For every vector $a$, $\|\mathbf T_i^\top a\|\ge \|a\|/\|\mathbf T_i^{-1}\|$, hence
\begin{equation}
\lambda_{\min}\{\operatorname{Cov}(\mathbf u_{\mathrm{lat}}\mid\mathcal W_i)\}
=\lambda_{\min}(\mathbf T_i\mathbf V_{\mathrm{inn},i}\mathbf T_i^\top)
\ge\frac{\lambda_{\min}(\mathbf V_{\mathrm{inn},i})}{\|\mathbf T_i^{-1}\|_{\mathrm{op}}^2}
\ge\frac{c_0}{C_T^2}>0.
\label{eq:F_stacked_cov_bound}
\end{equation}
This gives a uniform lower eigenvalue bound for the full stacked Gaussian covariance.

Let $\phi_{\boldsymbol\vartheta}(\mathbf u_{\mathrm{lat}}\mid\mathcal W_i)$ denote this joint Gaussian density. A derivative with respect to an entry of its mean or covariance is a polynomial in $\mathbf u_{\mathrm{lat}}$ multiplied by the same density. Repeated differentiation and the matrix chain rule thus give, for every multi-index $\mathbf a$ with $|\mathbf a|\le3$, constants $C,c>0$ and an integer $m$ independent of the subject design such that
\begin{equation}
\sup_{\boldsymbol\vartheta\in U}
\left|\partial_{\boldsymbol\vartheta}^{\mathbf a}
\phi_{\boldsymbol\vartheta}(\mathbf u_{\mathrm{lat}}\mid\mathcal W_i)\right|
\le C(1+\|\mathbf u_{\mathrm{lat}}\|^m)\exp(-c\|\mathbf u_{\mathrm{lat}}\|^2).
\label{eq:F_gaussian_envelope}
\end{equation}
For completeness, if $\phi_d(z;m,V)$ is a nondegenerate $d$-variate normal density and $\partial$ is a scalar parameter derivative, differentiating the log determinant and quadratic form gives
\begin{align*}
\partial\log\phi_d(z;m,V)
={}&-\tfrac12\operatorname{tr}(V^{-1}\partial V)
  +(\partial m)^\top V^{-1}(z-m)\\
 &+\tfrac12(z-m)^\top V^{-1}(\partial V)V^{-1}(z-m).
\end{align*}
The identities $\partial V^{-1}=-V^{-1}(\partial V)V^{-1}$ and $\partial\phi=\phi\,\partial\log\phi$ show inductively that derivatives through order three are products of a Gaussian density and polynomials in $z$, with coefficients uniformly bounded under Equation~\eqref{eq:F_stacked_cov_bound}. 

We call the right-hand side a \emph{Gaussian dominating envelope}: it bounds the absolute derivatives by a polynomial times a Gaussian-shaped function. Its integral is finite. The constants can be chosen uniformly because there are only finitely many possible visit counts and the bounded time and covariate design lies in a compact closure. If the Gaussian dimension varies with visit count, apply this bound separately to each of the finitely many dimensions.

For a particular recorded ordinal pattern, let $a_{\boldsymbol\vartheta}(\mathbf u_{\mathrm{lat}};\mathbf Y_i^{\mathrm{obs}},\mathcal D_i)$ denote the product of its itemwise category probabilities. This product is between zero and one. Because ordered-logit and ordered-probit category probabilities are smooth differences of the corresponding strictly increasing distribution functions, their first three parameter derivatives are bounded by polynomials in $\|\mathbf u_{\mathrm{lat}}\|$, uniformly on $U$ and the design support; bounded visit counts keep the number of product-rule terms finite. Consequently, Leibniz's rule bounds every derivative through order three of the integrand $a_{\boldsymbol\vartheta}\phi_{\boldsymbol\vartheta}$ by another integrable function of the form in Equation~\eqref{eq:F_gaussian_envelope}. Dominated convergence therefore permits differentiation under the integral in Equation~\eqref{eq:J1} and also establishes continuity of those derivatives.

It remains to justify taking logarithms. For any finite predictor and strictly ordered finite thresholds, every category probability is strictly positive. The Gaussian density is positive everywhere, so each integrated pattern probability is positive. The collection of allowable designs has compact closure, and the number of possible visit counts, masks, and categorical patterns is finite. Continuity consequently gives a number $c_{\mathrm{prob}}>0$ with
\begin{equation}
\inf_{\boldsymbol\vartheta\in U,\,\mathcal W,\,\mathbf M,\,
\mathbf y^{\mathrm{obs}}}
 p_{\boldsymbol\vartheta}^{\mathrm{pat}}(\mathbf y^{\mathrm{obs}};\mathbf M,\mathcal W)\ge c_{\mathrm{prob}}.
\label{eq:F_positive_likelihood}
\end{equation}
The first three derivatives of $p_{\boldsymbol\vartheta}$ are uniformly bounded by the preceding integral inequalities. The identities $\partial\log p=(\partial p)/p$ and $\partial^2\log p=(\partial^2p)/p-(\partial p)(\partial p)^\top/p^2$ (and their third-order analogs), together with Equation~\eqref{eq:F_positive_likelihood}, bound the score, Hessian, and third derivative of the log likelihood. In particular, their expectations under the observed-data law are finite, and they support the local uniform laws of large numbers used below.
\end{proof}

Let $\Pi(d\boldsymbol\vartheta)$ be the prior measure and $\Pi(\cdot\mid\mathbf O_{1:N})$ the posterior measure. For two probability measures $P$ and $Q$, their total-variation distance is $\|P-Q\|_{\mathrm{TV}}=\sup_{B\in\mathcal B(\mathbb R^{d_\vartheta})}|P(B)-Q(B)|$, where $\mathcal B(\mathbb R^{d_\vartheta})$ is the Borel $\sigma$-algebra. The notation $\xrightarrow{P_0}$ denotes convergence in probability under the true observed-data law.

\begin{thrm}[Population-parameter consistency and asymptotic normality]\label{thm:J2}
Under Assumptions~\ref{ass:J1}--\ref{ass:J5}, as $N\to\infty$:
\begin{enumerate}
\renewcommand{\labelenumi}{(\alph{enumi})}
\item Let $\widehat{\boldsymbol{\vartheta}}_N$ be a measurable global maximizer of the marginal log likelihood in $\Theta$. It is consistent in probability, lies in the interior of the fitted parameter space with probability tending to one, and satisfies
\[
\sqrt N(\widehat{\boldsymbol{\vartheta}}_N-\boldsymbol{\vartheta}_0)
\xrightarrow{d}
\text{}\mathcal{N}\{\mathbf{0},\mathcal I(\boldsymbol{\vartheta}_0)^{-1}\};
\]
\item The Bayesian posterior for the population parameter is consistent and satisfies the finite-dimensional Bernstein--von Mises approximation
\[
\sup_B\left|
\Pi\!\left(\sqrt N(\boldsymbol{\vartheta}-\widehat{\boldsymbol{\vartheta}}_N)\in B\mid \mathbf{O}_{1:N}\right)
-\text{}\mathcal{N}\{\mathbf{0},\mathcal I(\boldsymbol{\vartheta}_0)^{-1}\}(B)
\right|\xrightarrow{P_0}0;
\]

\item For any continuously differentiable functional $g(\boldsymbol{\vartheta})$ whose first derivative is nonzero in a scalar component (or, for joint nonsingular inference, whose Jacobian has full row rank at the truth), both the plug-in estimator and posterior inherit the corresponding root-$N$ Gaussian limit. More generally, differentiability alone gives a possibly singular multivariate Gaussian delta-method limit; first-order scalar credible-interval coverage requires positive limiting variance. In particular, this applies, for fixed $\mathbf{x}$ and $\tau>0$, to
\[
\mathbf{v}(\mathbf{x})=\boldsymbol{\Phi}\mathbf{x}+\boldsymbol{\alpha},\qquad
\mathbf{T}(\tau)=e^{-\boldsymbol{\Gamma}\tau},\qquad
\mathbf{Q}(\tau)=\boldsymbol{\Omega}-\mathbf{T}(\tau)\boldsymbol{\Omega} \mathbf{T}(\tau)^{\top}.
\]
Therefore, regular scalar posterior credible intervals for components of these quantities of scientific interest have asymptotically correct first-order frequentist coverage.
\end{enumerate}
\end{thrm}

\begin{proof}

{\samepage
We give the arguments in detail for the observed subject clusters. In this proof, $L_N(\boldsymbol\vartheta)=\sum_{i=1}^N\ell_i(\boldsymbol\vartheta)$, where $\ell_i(\boldsymbol\vartheta)=\log p_{\boldsymbol\vartheta}^{\mathrm{pat}}(\mathbf Y_i^{\mathrm{obs}};\mathbf M_i,\mathcal W_i)$. The observations $\mathbf O_i$ are iid by Assumptions~\ref{ass:J1} and \ref{ass:J5}. The missingness factor and the design density do not depend on $\boldsymbol\vartheta$, and hence cancel from likelihood ratios.
\par}

\emph{Identification of the limiting criterion and consistency.}
Fix the true missingness parameter $\boldsymbol\zeta_0$. In the observed-data law $f_{\boldsymbol\vartheta,\boldsymbol\zeta_0}$ defined above, the factors $g_0$ and $r_{\boldsymbol\zeta_0}$ do not depend on $\boldsymbol\vartheta$. Thus, with expectation under the true joint observed-data law,
\begin{equation}
E_0\{\ell_i(\boldsymbol\vartheta)-\ell_i(\boldsymbol\vartheta_0)\}
=-\operatorname{KL}\!\left(
 f_{\boldsymbol\vartheta_0,\boldsymbol\zeta_0},
 f_{\boldsymbol\vartheta,\boldsymbol\zeta_0}\right)\le0.
\label{eq:F_KL_contrast}
\end{equation}
Specifically, $f_{\boldsymbol\vartheta,\boldsymbol\zeta_0}$ denotes the joint law $f_{\boldsymbol\Psi(\boldsymbol\vartheta),\boldsymbol\zeta_0}$ defined above, using the free-coordinate parameter $\boldsymbol\vartheta$ and holding the missingness parameter at its true value $\boldsymbol\zeta_0$. Thus both arguments of the KL divergence are distributions of the \emph{same observable subject cluster}, differing only in the CLOUD parameter. Their common design density $g_0$ and missingness law $r_{\boldsymbol\zeta_0}$ cancel inside the log density ratio, leaving $\ell_i(\boldsymbol\vartheta_0)-\ell_i(\boldsymbol\vartheta)$. The KL integral includes integration over continuous design variables and summation over discrete mask and categorical-response patterns. Although $g_0$ and $r_{\boldsymbol\zeta_0}$ cancel inside the logarithm, they still determine the distribution under which $E_0$ is taken.
Here $\operatorname{KL}(f,g)=\int \log(f/g)\,f\,\mathrm d\mu$ is the Kullback--Leibler divergence. Equality means equality of the observed-data distributions. Theorem~\ref{thm:E1}, together with the design/recording support in Assumption~\ref{ass:J3}, makes this equality possible only at $\boldsymbol\vartheta_0$. This is the population target of \emph{M-estimation}: maximizing the empirical average criterion $N^{-1}L_N$ rather than assuming the target parameters are known \citep[Chapter~5]{vanDerVaart1998}.
For any $\varepsilon>0$, Assumption~\ref{ass:J4} implies, with probability tending to one, that every parameter outside the $\varepsilon$-ball has a smaller log likelihood than the one at $\boldsymbol\vartheta_0$. A global maximizer therefore lies inside that ball with probability tending to one. Since $\varepsilon$ is arbitrary, $\widehat{\boldsymbol\vartheta}_N\to\boldsymbol\vartheta_0$ in probability. In the sufficient compact setting described under Assumption~\ref{ass:J4}, the required uniform separation follows directly from the KL gap in Equation~\eqref{eq:F_KL_contrast} and the uniform law of large numbers furnished by Lemma~\ref{lem:J1}.

The same separation also establishes posterior consistency without requiring it as an additional assumption. To see this, choose a small neighborhood $B$ of $\boldsymbol\vartheta_0$ with positive prior mass and expected log-likelihood loss less than $c_\varepsilon/4$. On a sufficiently small closed $B$, Lemma~\ref{lem:J1} gives a uniform law of large numbers, so eventually $\inf_{\boldsymbol\vartheta\in B}\{L_N(\boldsymbol\vartheta)-L_N(\boldsymbol\vartheta_0)\}\ge -Nc_\varepsilon/2$ with high probability. The posterior numerator outside the $\varepsilon$-ball is at most $\exp(-Nc_\varepsilon)$ times the likelihood at the truth, by Equation~\eqref{eq:F_global_sep}; the posterior denominator is at least $\Pi(B)\exp(-Nc_\varepsilon/2)$ times that likelihood. Their ratio tends to zero. This argument uses only a proper prior and positive local prior mass.

More explicitly, on the high-probability event just described, for $B\subset\{\|\boldsymbol\vartheta-\boldsymbol\vartheta_0\|<\varepsilon\}$ chosen small enough,
\[
\Pi(B_{\varepsilon}^{c}\mid\mathbf O_{1:N})
\le\frac{e^{-Nc_\varepsilon}}{\Pi(B)\,e^{-Nc_\varepsilon/2}}
=\Pi(B)^{-1}e^{-Nc_\varepsilon/2}\longrightarrow0,
\]
where $B_\varepsilon$ denotes the parameter ball and $\Pi(B)>0$ follows from local prior support. This bound makes the normalization in the posterior consistency argument explicit.

\emph{Score, information, and local quadratic expansion.}
Put $\mathbf s_i=\partial_{\boldsymbol\vartheta}\ell_i(\boldsymbol\vartheta_0)$ and $\mathcal H_i(\boldsymbol\vartheta)=\partial_{\boldsymbol\vartheta}^2\ell_i(\boldsymbol\vartheta)$. Differentiating the normalized observed-data density and using the domination from Lemma~\ref{lem:J1} gives the information identities
\begin{equation}
E_0(\mathbf s_i)=\mathbf0,\qquad
E_0\{\mathcal H_i(\boldsymbol\vartheta_0)\}
=-E_0(\mathbf s_i\mathbf s_i^\top)
=-\mathcal I(\boldsymbol\vartheta_0).
\label{eq:F_info_identity}
\end{equation}
For example, the first identity follows from $\int\partial_{\boldsymbol\vartheta} f_{\boldsymbol\vartheta,\boldsymbol\zeta_0}(o)\,\mathrm do=\partial_{\boldsymbol\vartheta}1=0$. Differentiating this equality once more yields the Hessian identity. The nuisance missingness factor contributes no $\boldsymbol\vartheta$ derivative, although expectations are taken under the actual observed-data law. Since the scores are iid with finite covariance, the multivariate central limit theorem gives
\begin{equation}
\boldsymbol\Delta_N=\frac1{\sqrt N}\sum_{i=1}^N\mathbf s_i
\;\xrightarrow{d}\;\text{}\mathcal{N}\!\left(\mathbf0,\mathcal I(\boldsymbol\vartheta_0)\right).
\label{eq:F_score_clt}
\end{equation}
A uniform law of large numbers for the Hessians holds on a sufficiently small ball about $\boldsymbol\vartheta_0$: Lemma~\ref{lem:J1} bounds the third derivatives, so the Hessians are uniformly Lipschitz in the parameter; finite nets and the ordinary law of large numbers make the convergence uniform. The limiting expected Hessian is continuous and equals $-\mathcal I(\boldsymbol\vartheta_0)$ at the truth. With $\mathcal I(\boldsymbol\vartheta_0)\succ0$, we can shrink the ball so that, with probability tending to one, $N^{-1}\partial^2 L_N(\boldsymbol\vartheta)$ is negative definite throughout it.
To make uniformity explicit, on a compact neighborhood $U_0$ let $\{\boldsymbol\vartheta_1,\ldots,\boldsymbol\vartheta_{m(\delta)}\}$ be a finite $\delta$-net and put $F_i(\vartheta)=\partial_\vartheta^2\ell_i(\vartheta)$. The third-derivative envelope from Lemma~\ref{lem:J1} gives a Lipschitz bound $\|F_i(\vartheta)-F_i(\vartheta')\|\le G_i\|\vartheta-\vartheta'\|$ with $E_0G_i<\infty$. Thus
\[
\sup_{\vartheta\in U_0}\|\mathbb P_NF(\vartheta)-P_0F(\vartheta)\|
\le\max_{1\le j\le m(\delta)}\|\mathbb P_NF(\vartheta_j)-P_0F(\vartheta_j)\|
+\delta\{\mathbb P_NG+P_0G\}.
\]
The finite maximum tends to zero by the ordinary law of large numbers; the envelope average is bounded in probability. Let $N\to\infty$ and then $\delta\downarrow0$ to obtain the uniform Hessian law, which ensures the local curvature bound used for both maximum-likelihood estimation and posterior tail bounds.

For a fixed $\mathbf u$ (uniformly for $\mathbf u$ in each bounded set), third-order Taylor expansion at $\boldsymbol\vartheta_0$ yields the \emph{local asymptotic normality} (LAN) expansion; see \citet[Section~7.2]{vanDerVaart1998} for the definition and general theory.
\begin{equation}
L_N(\boldsymbol\vartheta_0+\mathbf u/\sqrt N)-L_N(\boldsymbol\vartheta_0)
=\mathbf u^\top\boldsymbol\Delta_N
-\tfrac12\mathbf u^\top\mathcal I(\boldsymbol\vartheta_0)\mathbf u
+o_{P_0}(1).
\label{eq:F_LAN}
\end{equation}

We justify the Taylor expansion and its stochastic remainder explicitly. Let $\mathcal I_0=\mathcal I(\boldsymbol\vartheta_0)$, $\mathbf J_N(\boldsymbol\vartheta)=N^{-1}\partial_{\boldsymbol\vartheta}^2L_N(\boldsymbol\vartheta)$, and $\mathbf a_N=\mathbf u/\sqrt N$. Let $U$ be a fixed closed ball about $\boldsymbol\vartheta_0$ on which Lemma~\ref{lem:J1} holds. For $\boldsymbol\vartheta_0+\mathbf a_N\in U$, the ordinary multivariate Taylor formula with integral remainder (or its third-order mean-value bound) gives
\begin{align*}
L_N(\boldsymbol\vartheta_0+\mathbf a_N)-L_N(\boldsymbol\vartheta_0)
&= \partial_{\boldsymbol\vartheta}L_N(\boldsymbol\vartheta_0)^\top\mathbf a_N
 +\tfrac12\mathbf a_N^\top\partial_{\boldsymbol\vartheta}^2L_N(\boldsymbol\vartheta_0)\mathbf a_N
 +R_N(\mathbf a_N),\\
|R_N(\mathbf a_N)|
&\le \frac{\|\mathbf a_N\|^3}{6}
\sum_{i=1}^N\sup_{\boldsymbol\vartheta\in U}
 \|D^3\ell_i(\boldsymbol\vartheta)\|_{\mathrm{op}},
\end{align*}
where $D^3\ell_i$ is the third derivative viewed as a trilinear form, and $\|\cdot\|_{\mathrm{op}}$ is the corresponding operator norm. This bound follows by integrating the third directional derivative along the line segment from $\boldsymbol\vartheta_0$ to $\boldsymbol\vartheta_0+\mathbf a_N$. Write $A_N=N^{-1}\sum_i\sup_{\boldsymbol\vartheta\in U}\|D^3\ell_i(\boldsymbol\vartheta)\|_{\mathrm{op}}$. The uniform third-derivative envelope in Lemma~\ref{lem:J1} implies $A_N=O_{P_0}(1)$ (indeed, the compact-design version is uniformly bounded). Therefore, for each finite $M$,
\begin{align*}
\sup_{\|\mathbf u\|\le M}|R_N(\mathbf u/\sqrt N)|
&\le \frac{M^3 A_N}{6\sqrt N}=o_{P_0}(1),\\
\sup_{\|\mathbf u\|\le M}\left|
\tfrac12\mathbf u^\top\{\mathbf J_N(\boldsymbol\vartheta_0)+\mathcal I_0\}\mathbf u\right|
&\le\tfrac12 M^2\|\mathbf J_N(\boldsymbol\vartheta_0)+\mathcal I_0\|_{\mathrm{op}}
=o_{P_0}(1).
\end{align*}
For the second equality, the law of large numbers and the information identity in Equation~\eqref{eq:F_info_identity} give $\mathbf J_N(\boldsymbol\vartheta_0)\xrightarrow{P_0}-\mathcal I_0$. Finally, the linear Taylor term equals $\mathbf u^\top\boldsymbol\Delta_N$ exactly by the definition in Equation~\eqref{eq:F_score_clt}. Adding these three terms proves Equation~\eqref{eq:F_LAN} with an $o_{P_0}(1)$ remainder uniform on every bounded set of local displacements. These bounds establish the claimed uniform LAN expansion.

To interpret this expansion, let $\mathcal I_0=\mathcal I(\boldsymbol\vartheta_0)$ and consider an idealized Gaussian observation $\mathbf Z\sim \text{}\mathcal{N}(\mathcal I_0\mathbf u,\mathcal I_0)$ under local parameter $\mathbf u$. Its log-likelihood ratio relative to $\mathbf u=\mathbf0$ is $\mathbf u^\top\mathbf Z-\tfrac12\mathbf u^\top\mathcal I_0\mathbf u$, the same limiting form as Equation~\eqref{eq:F_LAN}. Thus ``normality'' refers to the \emph{local likelihood-ratio experiment}, not to a Gaussian distribution for the observed ordinal responses. LAN alone describes this root-$N$ neighborhood and does not guarantee global estimator consistency or posterior concentration; those are established separately above and below. This is the finite-dimensional LAN formulation of \citet[Section~7.2]{vanDerVaart1998}.

\emph{Asymptotic normality of the marginal maximum-likelihood estimator.}
Consistency places the global maximizer in the preceding interior ball with probability tending to one, where its score is zero. Taylor expansion of that score between the truth and the maximizer gives

\begin{equation}
\begin{aligned}
\sqrt N(\widehat{\boldsymbol\vartheta}_N-\boldsymbol\vartheta_0)
&=-\left\{\int_0^1 \mathbf J_N\bigl(\boldsymbol\vartheta_0+
 t(\widehat{\boldsymbol\vartheta}_N-\boldsymbol\vartheta_0)\bigr)\,\mathrm dt\right\}^{-1}
\boldsymbol\Delta_N\\
&=\mathcal I_0^{-1}\boldsymbol\Delta_N+o_{P_0}(1).
\end{aligned}
\label{eq:F_MLE_expansion}
\end{equation}

To justify both equalities in Equation~\eqref{eq:F_MLE_expansion}, let $\mathbf d_N=\widehat{\boldsymbol\vartheta}_N-\boldsymbol\vartheta_0$. With probability tending to one, consistency implies that $\widehat{\boldsymbol\vartheta}_N$ lies in the regular interior neighborhood and satisfies $\partial_{\boldsymbol\vartheta}L_N(\widehat{\boldsymbol\vartheta}_N)=\mathbf0$. Apply the fundamental theorem of calculus to the \emph{vector-valued score} along the segment $\boldsymbol\vartheta_0+t\mathbf d_N$:
\[
\mathbf0=\partial_{\boldsymbol\vartheta}L_N(\boldsymbol\vartheta_0)
 +\left\{\int_0^1\partial_{\boldsymbol\vartheta}^2L_N(\boldsymbol\vartheta_0+t\mathbf d_N)
\,\mathrm dt\right\}\mathbf d_N.
\]
Unlike a single ``intermediate Hessian'' matrix, this exact integral identity is valid for multivariate scores. By the uniform Hessian law on the neighborhood and $\mathbf d_N=o_{P_0}(1)$,
\[
\left\|\int_0^1\mathbf J_N(\boldsymbol\vartheta_0+t\mathbf d_N)\,\mathrm dt
  +\mathcal I_0\right\|_{\mathrm{op}}
\le\sup_{\boldsymbol\vartheta:\,\|\boldsymbol\vartheta-\boldsymbol\vartheta_0\|
 \le\|\mathbf d_N\|}
\|\mathbf J_N(\boldsymbol\vartheta)+\mathcal I_0\|_{\mathrm{op}}
=o_{P_0}(1).
\]
The limit $-\mathcal I_0$ is invertible by Assumption~\ref{ass:J3}. Continuity of matrix inversion therefore gives the inverse of the integral Hessian as $-\mathcal I_0^{-1}+o_{P_0}(1)$ in operator norm. Multiplication by the tight random vector $\boldsymbol\Delta_N=O_{P_0}(1)$ from Equation~\eqref{eq:F_score_clt} produces the remainder $o_{P_0}(1)$ in Equation~\eqref{eq:F_MLE_expansion}. The multivariate central limit theorem and Slutsky's lemma then establish part~(a).

\emph{Posterior localization and Gaussian approximation.}

For part~(b), it is not sufficient to know that a local likelihood is approximately quadratic: the posterior could otherwise retain appreciable mass far from the true parameter. We prove the required root-$N$ localization before approximating its density. Let $\lambda_0=\lambda_{\min}(\mathcal I_0)>0$. By the uniform Hessian law and continuity of the expected Hessian, there exists a fixed $\delta_0>0$ such that, with probability tending to one, $\mathbf J_N(\boldsymbol\vartheta)\preceq-(\lambda_0/2)\mathbf I_{d_\vartheta}$ for all $\|\boldsymbol\vartheta-\boldsymbol\vartheta_0\|\le\delta_0$. For $\mathbf a_N=\mathbf u/\sqrt N$, the second-order Taylor formula in integral form yields
\[
L_N(\boldsymbol\vartheta_0+\mathbf a_N)-L_N(\boldsymbol\vartheta_0)
=\partial_{\boldsymbol\vartheta}L_N(\boldsymbol\vartheta_0)^\top\mathbf a_N
+\int_0^1(1-t)\,\mathbf a_N^\top
\partial_{\boldsymbol\vartheta}^2L_N(\boldsymbol\vartheta_0+t\mathbf a_N)\mathbf a_N\,\mathrm dt.
\]
Using the Hessian bound and $\int_0^1(1-t)\,\mathrm dt=1/2$ proves Equation~\eqref{eq:F_local_tail} with $a=\lambda_0/4$, uniformly for $\|\mathbf u\|\le\delta_0\sqrt N$.
\begin{equation}
L_N(\boldsymbol\vartheta_0+\mathbf u/\sqrt N)-L_N(\boldsymbol\vartheta_0)
\le \|\mathbf u\|\,\|\boldsymbol\Delta_N\|-a\|\mathbf u\|^2.
\label{eq:F_local_tail}
\end{equation}

We next bound the posterior normalizing constant and the local and global posterior tails. Let $\pi$ be the proper prior density in the free coordinates, continuous and strictly positive at $\boldsymbol\vartheta_0$. Define $\mathcal U_N=\sqrt N(\Theta-\boldsymbol\vartheta_0)$ and $\ell_N^{\mathrm{loc}}(\mathbf u)=L_N(\boldsymbol\vartheta_0+\mathbf u/\sqrt N)-L_N(\boldsymbol\vartheta_0)$. With $\boldsymbol\vartheta=\boldsymbol\vartheta_0+\mathbf u/\sqrt N$, the posterior density of $\mathbf u$ is
\[
\pi_N^{\mathrm{loc}}(\mathbf u)=\frac{\mathbf1_{\mathcal U_N}(\mathbf u)
\,\pi(\boldsymbol\vartheta_0+\mathbf u/\sqrt N)e^{\ell_N^{\mathrm{loc}}(\mathbf u)}}{Z_N},\qquad
Z_N=\int_{\mathcal U_N}\pi(\boldsymbol\vartheta_0+\mathbf u/\sqrt N)
 e^{\ell_N^{\mathrm{loc}}(\mathbf u)}\,\mathrm d\mathbf u.
\]
Fix a small radius $r>0$ and choose $K<\infty$ so that the score event $\{\|\boldsymbol\Delta_N\|\le K\}$ has probability arbitrarily close to one for large $N$, using the central limit theorem. By the uniform LAN expansion in Equation~\eqref{eq:F_LAN}, on another event whose probability tends to one, the remainder is at most one for $\|\mathbf u\|\le r$. By continuity and positivity of $\pi$, there is a constant $\pi_->0$ such that $\pi(\boldsymbol\vartheta_0+\mathbf u/\sqrt N)\ge\pi_-$ on this ball for sufficiently large $N$. Therefore, on these events,
\[
Z_N\ \ge\ \pi_-\operatorname{vol}(B_r)
\exp\!\left\{-Kr-\tfrac12r^2\lambda_{\max}(\mathcal I_0)-1\right\}
=:c_{K,r}>0,\qquad B_r=\{\mathbf u:\|\mathbf u\|\le r\}.
\]
Next, choose a constant $\pi_+<\infty$ that bounds $\pi$ on the fixed $\delta_0$-ball. Equation~\eqref{eq:F_local_tail} and the lower bound for $Z_N$ give, on the same high-probability events, for any $M>r$,
\[
\Pi\!\left(M<\sqrt N\|\boldsymbol\vartheta-\boldsymbol\vartheta_0\|
\le\delta_0\sqrt N\mid\mathbf O_{1:N}\right)
\le \frac{\pi_+}{c_{K,r}}
\int_{\|\mathbf u\|>M}\exp\{K\|\mathbf u\|-a\|\mathbf u\|^2\}
\,\mathrm d\mathbf u.
\]
The last integral decreases to zero as $M\to\infty$, since its negative quadratic term dominates the linear term. For the region $\|\boldsymbol\vartheta-\boldsymbol\vartheta_0\|>\delta_0$, the global likelihood separation in Equation~\eqref{eq:F_global_sep} and the fact that $\int_\Theta\pi(\boldsymbol\vartheta)\,\mathrm d\boldsymbol\vartheta=1$ imply
\[
\Pi\!\left(\sqrt N\|\boldsymbol\vartheta-\boldsymbol\vartheta_0\|>
\delta_0\sqrt N\mid\mathbf O_{1:N}\right)
\le c_{K,r}^{-1}N^{d_\vartheta/2}e^{-Nc_{\delta_0}}\ \longrightarrow\ 0.
\]

More precisely, for any $\epsilon,\delta_{\mathrm{tail}}>0$, first choose $K$ so that the score event has limiting probability at least $1-\delta_{\mathrm{tail}}$; then choose $M$ so that the integral upper bound is below $\epsilon$. The global-tail term vanishes as $N\to\infty$. Thus $\limsup_{N\to\infty}P_0\{\Pi(\sqrt N\|\boldsymbol\vartheta-\boldsymbol\vartheta_0\|>M\mid\mathbf O_{1:N})>\epsilon\}\le\delta_{\mathrm{tail}}$. As $\delta_{\mathrm{tail}}$ is arbitrary, this is root-$N$ posterior tightness in $P_0$-probability.

We now prove the total-variation approximation. On each fixed $B_M$, Equation~\eqref{eq:F_LAN} holds uniformly and $\pi(\boldsymbol\vartheta_0+\mathbf u/\sqrt N)\to\pi(\boldsymbol\vartheta_0)>0$ uniformly. Together with $\boldsymbol\Delta_N=O_{P_0}(1)$, this implies that the unnormalized density in the numerator of $\pi_N^{\mathrm{loc}}$ differs in $L^1(B_M)$ by $o_{P_0}(1)$ from
\[
\pi(\boldsymbol\vartheta_0)
\exp\!\left\{\mathbf u^\top\boldsymbol\Delta_N
-\tfrac12\mathbf u^\top\mathcal I_0\mathbf u\right\}.
\]
We now supply the normalization step. Let $u_N(\mathbf u)$ be the nonnegative unnormalized local posterior density in the numerator of $\pi_N^{\mathrm{loc}}$, and let $v_N(\mathbf u)=\pi(\boldsymbol\vartheta_0)\exp\{\mathbf u^\top\boldsymbol\Delta_N-\tfrac12\mathbf u^\top\mathcal I_0\mathbf u\}$. Define $Z_N=\int u_N$, $C_N=\int v_N$, and for $B_M=\{\|\mathbf u\|\le M\}$ set
\begin{align*}
\delta_{N,M}&=\int_{B_M}|u_N-v_N|\,\mathrm d\mathbf u=o_{P_0}(1),\\
a_{N,M}&=Z_N^{-1}\int_{B_M^c}u_N\,\mathrm d\mathbf u,\qquad
b_{N,M}=C_N^{-1}\int_{B_M^c}v_N\,\mathrm d\mathbf u.
\end{align*}
On a high-probability event where $\|\boldsymbol\Delta_N\|\le K$, the Gaussian normalization constant is
\begin{align*}
C_N={}&\pi(\boldsymbol\vartheta_0)(2\pi)^{d_\vartheta/2}
|\mathcal I_0|^{-1/2}\\
&\quad\times\exp\!\left\{\tfrac12\boldsymbol\Delta_N^\top
\mathcal I_0^{-1}\boldsymbol\Delta_N\right\}.
\end{align*}
It is bounded above and below by positive constants. The earlier lower bound $Z_N\ge c_{K,r}$, together with the posterior and Gaussian tail estimates, ensures that $a_{N,M}$ and $b_{N,M}$ can be made arbitrarily small by first choosing $M$ sufficiently large. Split each normalization integral over $B_M$ and its complement to obtain
\begin{align*}
|Z_N-C_N|&\le\delta_{N,M}+Z_Na_{N,M}+C_Nb_{N,M},\\
\left\|\frac{u_N}{Z_N}-\frac{v_N}{C_N}\right\|_1
&\le\frac{\delta_{N,M}+Z_Na_{N,M}+C_Nb_{N,M}}{Z_N}
+\frac{|Z_N-C_N|}{Z_N}.
\end{align*}
To avoid a circular upper bound on $Z_N$, rearrange the first inequality as $Z_N(1-a_{N,M})\le C_N(1+b_{N,M})+\delta_{N,M}$. Since the tails are small and $C_N$ is bounded, $Z_N$ is also bounded in probability. Consequently $Z_N/C_N\to1$ and the displayed normalized $L^1$ difference tends to zero in $P_0$-probability after taking $N\to\infty$ and then $M\to\infty$. This establishes the connection between LAN on bounded balls and total-variation convergence of the entire posterior.

Completing the square shows that the normalized version of the latter density is $\phi_{\mathbf m_N,\mathcal I_0^{-1}}$, the multivariate normal density with $\mathbf m_N=\mathcal I_0^{-1}\boldsymbol\Delta_N$ and covariance $\mathcal I_0^{-1}$. Its tails are tight in probability because $\mathbf m_N=O_{P_0}(1)$. The posterior-tail bounds above establish the corresponding tightness. Therefore the $L^1$ approximation on $B_M$ extends to the whole space by first making both tails small and then taking $N\to\infty$:
\[
\int_{\mathbb R^{d_\vartheta}}\left|\pi_N^{\mathrm{loc}}(\mathbf u)
-\phi_{d_\vartheta}(\mathbf u;\mathbf m_N,\mathcal I_0^{-1})\right|\,\mathrm d\mathbf u
\xrightarrow{P_0}0.
\]
For probability densities, total variation is one half of this $L^1$ distance. Equation~\eqref{eq:F_MLE_expansion} gives $\sqrt N(\widehat{\boldsymbol\vartheta}_N-\boldsymbol\vartheta_0)-\mathbf m_N=o_{P_0}(1)$. Two normal laws with common covariance $\mathcal I_0^{-1}$ and mean difference $\mathbf e_N$ have Kullback--Leibler divergence $\tfrac12\mathbf e_N^\top\mathcal I_0\mathbf e_N$; Pinsker's inequality therefore bounds their total-variation distance by $\tfrac12\|\mathcal I_0^{1/2}\mathbf e_N\|=o_{P_0}(1)$. Shifting the posterior by $\sqrt N(\widehat{\boldsymbol\vartheta}_N-\boldsymbol\vartheta_0)$ and using the triangle inequality now gives exactly the statement expressed as a supremum over Borel sets in Theorem~\ref{thm:J2}(b). This establishes the finite-dimensional Bernstein--von Mises conclusion; see \citet[Chapter~10]{vanDerVaart1998}.

\emph{Smooth scientific functionals and credible intervals.}

For part~(c), fix a continuously differentiable scalar functional $g$ and write $\mathbf a_g=\nabla g(\boldsymbol\vartheta_0)$ and $V_g=\mathbf a_g^\top\mathcal I_0^{-1}\mathbf a_g$. Asymptotic normality in part~(a) and the first-order Taylor expansion at the truth yield
\begin{align*}
\sqrt N\{g(\widehat{\boldsymbol\vartheta}_N)-g(\boldsymbol\vartheta_0)\}
&=\mathbf a_g^\top\sqrt N(\widehat{\boldsymbol\vartheta}_N-\boldsymbol\vartheta_0)
+o_{P_0}(1)\\
&\xrightarrow{d}\mathcal N(0,V_g).
\end{align*}
Indeed, differentiability gives a remainder $o(\|\widehat{\boldsymbol\vartheta}_N-\boldsymbol\vartheta_0\|)$; Equation~\eqref{eq:F_MLE_expansion} makes its product with $\sqrt N$ equal to $o_{P_0}(1)$. To propagate the posterior approximation, define $\mathbf u=\sqrt N(\boldsymbol\vartheta-\widehat{\boldsymbol\vartheta}_N)$. For each fixed $M$, the fundamental theorem of calculus and continuity of $\nabla g$ imply
\[
\begin{aligned}
&\sup_{\|\mathbf u\|\le M}\bigg|
\sqrt N\Big\{g\big(\widehat{\boldsymbol\vartheta}_N+\mathbf u/\sqrt N\big)
-g(\widehat{\boldsymbol\vartheta}_N)\Big\}
-\mathbf a_g^\top\mathbf u\bigg|\\
&\quad\le M\!\sup_{\|\boldsymbol\vartheta-\boldsymbol\vartheta_0\|\le r_{N,M}}
\|\nabla g(\boldsymbol\vartheta)-\mathbf a_g\|=o_{P_0}(1),
\end{aligned}
\]
where $r_{N,M}=\|\widehat{\boldsymbol\vartheta}_N-
\boldsymbol\vartheta_0\|+M/\sqrt N$.
Part~(b) concentrates the posterior of $\mathbf u$ on bounded sets and approximates it in total variation by $\mathcal N(\mathbf0,\mathcal I_0^{-1})$. The preceding uniform remainder therefore gives convergence of the posterior distribution of $\sqrt N\{g(\boldsymbol\vartheta)-g(\widehat{\boldsymbol\vartheta}_N)\}$ to $\mathcal N(0,V_g)$ in the usual weak sense (in $P_0$-probability). For vector-valued $g$, replace $\mathbf a_g^\top$ by its Jacobian $J_g$ to obtain covariance $J_g\mathcal I_0^{-1}J_g^\top$, possibly singular if the Jacobian has deficient rank.
When $V_g>0$, the normal limit has a continuous strictly increasing distribution function, so the posterior $u$-quantile $q_{u,N}$ of $g(\boldsymbol\vartheta)$ satisfies, for every $0<u<1$,
\[
q_{u,N}=g(\widehat{\boldsymbol\vartheta}_N)
+N^{-1/2}z_{u,N}\sqrt{V_g}+o_{P_0}(N^{-1/2}),
\qquad z_{u,N}=\Phi_N^{-1}(u),
\]
where $\Phi_N$ denotes the standard normal cumulative distribution function. Fix $\alpha_{\mathrm{CI}}\in(0,1)$ as the intended noncoverage level. The equal-tailed $(1-\alpha_{\mathrm{CI}})$ credible interval $[q_{\alpha_{\mathrm{CI}}/2,N},q_{1-\alpha_{\mathrm{CI}}/2,N}]$ consequently covers $g(\boldsymbol\vartheta_0)$ with probability tending to $1-\alpha_{\mathrm{CI}}$: substitute these two quantile expansions, use the sampling limit for $g(\widehat{\boldsymbol\vartheta}_N)$ derived above, and apply Slutsky's lemma. If $V_g=0$, the stated first-order quantile and coverage conclusions do not follow and a higher-order expansion would be required. The same reasoning applies componentwise to smooth vector functionals whenever the relevant component has positive limiting variance.
For the CLOUD quantities displayed in part~(c), $\mathbf v(\mathbf x)$ is affine in the mean coefficients. Differentiability of $\mathbf T(\tau)=e^{-\boldsymbol\Gamma\tau}$ follows from the Fr\'echet derivative below, valid for any matrix direction $\Delta\boldsymbol\Gamma$.
\begin{equation}
D_{\boldsymbol{\Gamma}}e^{-\boldsymbol{\Gamma}\tau}[\Delta\boldsymbol\Gamma]
=-\int_0^\tau e^{-\boldsymbol{\Gamma}(\tau-r)}\Delta\boldsymbol\Gamma e^{-\boldsymbol{\Gamma}r}\,\mathrm{d}r,
\label{eq:J2}
\end{equation}
For completeness, this Fr\'echet formula follows from the exact Duhamel identity, valid without requiring $\boldsymbol\Gamma$ to commute with its perturbation: for $\varepsilon\ne0$,
\[
e^{-(\boldsymbol\Gamma+\varepsilon\Delta\boldsymbol\Gamma)\tau}
-e^{-\boldsymbol\Gamma\tau}
=-\varepsilon\int_0^\tau e^{-(\boldsymbol\Gamma+\varepsilon\Delta\boldsymbol\Gamma)(\tau-r)}
  \Delta\boldsymbol\Gamma\,e^{-\boldsymbol\Gamma r}\,\mathrm dr.
\]
Divide by $\varepsilon$ and let $\varepsilon\to0$; uniform convergence of the matrix exponentials on $[0,\tau]$ yields Equation~\eqref{eq:J2}. Its integral cannot in general be simplified to $-\tau e^{-\boldsymbol\Gamma\tau}\Delta\boldsymbol\Gamma$ because drift matrices and perturbations need not commute.

The product rule then establishes continuous differentiability of $\mathbf Q(\tau)=\boldsymbol\Omega-\mathbf T(\tau)\boldsymbol\Omega\mathbf T(\tau)^\top$. Thus all three functionals meet the smoothness requirement at fixed $\mathbf x$ and $\tau>0$. Their scalar components inherit the credible-interval coverage conclusion precisely when their derivative directions yield $V_g>0$. This completes part~(c).
\end{proof}

Theorem~\ref{thm:J2} is conditional on fixed $R,K,p,q$, bounded numbers of visits per subject, the observable support and missingness assumptions, uniform global separation, and positive-definite local Fisher information. The result concerns a fixed-dimensional population model; finite-sample behavior was evaluated in the simulation study.

\begin{prop}[Persistent individual-level uncertainty]\label{prop:J3}
Fix a subject whose number of observations remains bounded as $N\to\infty$. Under Theorem~\ref{thm:J2} and the same regular-interior and positive-lag conditions as in Lemma~\ref{lem:J1}, the subject's marginal latent-state posterior converges in total variation to
\[
p_{\boldsymbol{\Psi}_0}(\boldsymbol\xi_{i,1:n_i}\mid \mathbf{Y}_i^{\mathrm{obs}},\mathcal D_i),
\]
the oracle posterior obtained if the population parameters were known. It does not generally collapse to a point mass when the subject contributes only finitely many noisy categorical observations.
\end{prop}

\begin{proof}
Fix the subject's realized response pattern and design. Let $k_{\boldsymbol\Psi}(\boldsymbol\xi_{i,1:n_i})$ be the unnormalized conditional latent-state density: the joint Gaussian latent-state density multiplied by the likelihood of the observed item responses after integrating out the item random effects. The same domination argument as in Lemma~\ref{lem:J1} bounds $k_{\boldsymbol\Psi}$ by an integrable Gaussian envelope uniformly for $\boldsymbol\Psi$ close to $\boldsymbol\Psi_0$. Pointwise continuity and dominated convergence give $\int|k_{\boldsymbol\Psi}-k_{\boldsymbol\Psi_0}|\,\mathrm d\boldsymbol\xi_{i,1:n_i}\to0$. Since the normalizing constants are the strictly positive probabilities in Equation~\eqref{eq:J1}, the normalized conditional densities also converge in $L^1$, and therefore in total variation. This establishes continuity of the conditional latent-state kernel.

Writing $Z_{\boldsymbol\Psi}=\int k_{\boldsymbol\Psi}(\boldsymbol\xi_{i,1:n_i})\,\mathrm d\boldsymbol\xi_{i,1:n_i}$, the normalizing constants are positive and $|Z_{\boldsymbol\Psi}-Z_{\boldsymbol\Psi_0}|\le\|k_{\boldsymbol\Psi}-k_{\boldsymbol\Psi_0}\|_1$. Therefore
\[
\left\|\frac{k_{\boldsymbol\Psi}}{Z_{\boldsymbol\Psi}}-\frac{k_{\boldsymbol\Psi_0}}{Z_{\boldsymbol\Psi_0}}\right\|_1
\le\frac{\|k_{\boldsymbol\Psi}-k_{\boldsymbol\Psi_0}\|_1}{Z_{\boldsymbol\Psi}}
+\frac{|Z_{\boldsymbol\Psi}-Z_{\boldsymbol\Psi_0}|}{Z_{\boldsymbol\Psi}}
\longrightarrow0.
\]
Here, the convergence follows by dominated convergence and $Z_{\boldsymbol\Psi_0}>0$; it establishes total-variation continuity of the normalized conditional posterior.

Conditional on the population parameters, the joint distribution factorizes across independent subject clusters, so the subject's latent-state posterior given all subjects is the mixture
\begin{equation*}
\Pi(\mathrm{d}\boldsymbol\xi_{i,1:n_i}\mid \mathbf{O}_{1:N})
=\int p_{\boldsymbol\Psi}(\mathrm{d}\boldsymbol\xi_{i,1:n_i}\mid \mathbf{Y}_i^{\mathrm{obs}},\mathcal D_i)
\Pi(\mathrm{d}\boldsymbol\Psi\mid \mathbf{O}_{1:N}).
\end{equation*}
Write $\mathcal P_{\boldsymbol\Psi}^{\mathrm{lat}}$ for the conditional latent-state kernel. The triangle inequality for total variation implies
\begin{equation}
\left\|\Pi(\cdot\mid\mathbf O_{1:N})-\mathcal P_{\boldsymbol\Psi_0}^{\mathrm{lat}}\right\|_{\mathrm{TV}}
\le\int\left\|\mathcal P_{\boldsymbol\Psi}^{\mathrm{lat}}-\mathcal P_{\boldsymbol\Psi_0}^{\mathrm{lat}}\right\|_{\mathrm{TV}}
\Pi(\mathrm d\boldsymbol\Psi\mid\mathbf O_{1:N}).
\label{eq:F_oracle_TV}
\end{equation}
The integrand is at most one and tends to zero as $\boldsymbol\Psi\to\boldsymbol\Psi_0$. Posterior consistency from Theorem~\ref{thm:J2} makes the right-hand side vanish, proving the asserted total-variation convergence. Finally, the oracle posterior is a nondegenerate Gaussian latent prior multiplied by strictly positive categorical likelihood factors. It has a positive Lebesgue density for a finite number of visits and therefore cannot be a point mass. Increasing the number of \emph{other} subjects removes population-parameter uncertainty but does not eliminate the remaining observation noise for this subject.
\end{proof}

\suppsection{Simulation details}{appd:simu_supp}
\subsection{Prior specification}
Supplementary Table~\ref{tab:priors} specifies the priors used for the $R=2$ (S1--S2), $R=4$ (S3--S4), and $R=8$ (S5--S6) simulation designs. The prior families are shared across dimensions. All models use hierarchical measurement priors and noncentered subject intercepts and latent innovations. In $\mathcal N(\mu,\sigma^2)$, the second argument is the variance; a superscript $+$ denotes truncation below zero.

\subsubsection{Common measurement and moving-mean priors}
For each binary item $k$, the binary-response logit intercept $\theta_k^{\mathrm{bin}}$ has a normal prior $\mathcal N(\mu_\theta,\sigma_\theta^2)$. For each four-category ordinal item, the three thresholds have the same normal prior, conditional on being ordered. The hierarchical hyperparameters satisfy $\mu_\theta\sim\mathcal N(0,5^2)$ and $\sigma_\theta\sim\mathcal N^+(0,2^2)$. Positive factor loadings obey $\lambda_k\sim\mathcal N^+(1,\sigma_\lambda^2)$ with $\sigma_\lambda\sim\mathcal N^+(0,2^2)$, and all item-level regression coefficients obey $\beta_{kp}\sim\mathcal N(0,5^2)$. Subject-by-item random effects are written $b_{ik}=\sigma_{bk}b_{ik}^{\mathrm{std}}$, where $b_{ik}^{\mathrm{std}}\sim\mathcal N(0,1)$ is a standardized auxiliary variable and $\sigma_{bk}\sim\mathcal N^+(0,1^2)$. The noncentered latent-state construction uses independent standardized innovations $\mathbf z_{ij}\sim\mathcal N_R(\mathbf0,\mathbf I_R)$, defined in Appendix~\ref{appd:ncp}.

For the covariate-dependent latent mean, $\boldsymbol\Phi$ represents covariate-associated progression slopes and $\boldsymbol\alpha$ represents the population slope: $\Phi_{rq}\sim\mathcal N(0,2^2)$ and $\alpha_r\sim\mathcal N(0,0.5^2)$. CLOUD and DiagOU estimate both; TrendyOU estimates $\boldsymbol\alpha$ but fixes $\boldsymbol\Phi=\mathbf 0$; StationaryOU fixes both to zero and therefore has no priors for these omitted parameters.

The half-normal prior on each $\sigma_{bk}$ regularizes the item-specific random-effect standard deviation. Some generating values, including 6.1, lie far into its upper tail. Parameter-specific bias and repeated-sample interval coverage under this prior are reported in the recovery tables. Prior sensitivity is relevant for variance parameters \citep{gelman2006prior}. Repeated-sampling coverage at the fixed generating values used here is distinct from both posterior credible probability and the prior-predictive posterior-quantile calibration used for software validation \citep{CookGelmanRubin2006}.

\subsubsection{Full multivariate drift and stationary covariance}
For the three models with coupled latent OU residuals (CLOUD, TrendyOU, and StationaryOU), we use the following admissible parameterization:
\begin{align}
 \tau_{S,r}&=\exp(u_r),\quad u_r\sim\mathcal N(0,0.5^2),\\
 L_{R_S}&\sim\operatorname{LKJ\text{-}Cholesky}(2),\qquad
 S=\operatorname{diag}(\boldsymbol\tau_S)L_{R_S}L_{R_S}^{\top}\operatorname{diag}(\boldsymbol\tau_S),\\
 A_{rs}&=\frac{2}{\sqrt{R-1}}z_{A,rs}\quad(r<s),\qquad A_{sr}=-A_{rs},\qquad z_{A,rs}\sim\mathcal N(0,1),\\
 L_\Omega&\sim\operatorname{LKJ\text{-}Cholesky}(8),\qquad
 \boldsymbol\Omega=L_\Omega L_\Omega^{\top},\qquad
 \boldsymbol\Gamma=(S+A)\boldsymbol\Omega^{-1}.
\end{align}
The skew scale $\tau_A=2$ is fixed, not estimated. In addition to the priors on the constructive coordinates, the Stan model adds the log-density corresponding to $\operatorname{vec}(\boldsymbol\Gamma)\sim\mathcal N(0,2.5^2)$ entrywise as a soft penalty on the \emph{derived} drift. This factor couples $S$, $A$, and $\boldsymbol\Omega$ in the joint prior, rather than defining an independent set of free drift parameters. The Lewandowski--Kurowicka--Joe (LKJ) correlation prior \citep{Lewandowski2009} with shape 8 concentrates more strongly around weak correlations than a prior with shape 2; the resulting joint prior is regularizing. The construction guarantees $S\succ0$, $\boldsymbol\Omega\succ0$, and $\boldsymbol\Gamma\boldsymbol\Omega+\boldsymbol\Omega\boldsymbol\Gamma^\top=2S\succ0$ without restricting the allowable interaction pattern.

\subsubsection{Diagonal baseline and replication design}
DiagOU fixes the skew component at $A=0$, sets $\boldsymbol\Omega=I_R$, and uses independent parameters $d_r\sim\mathcal N^+(0,2^2)$ with $\Gamma_{rr}=d_r^2$, $\Gamma_{rs}=0$ for $r\ne s$. Thus, the half-normal prior is placed on $\sqrt{\Gamma_{rr}}$; DiagOU has no $L_{R_S}$, $L_\Omega$, skew, or derived-$\Gamma$ priors. All model variants used the same $N=600$ subjects and the same independently simulated dataset within each scenario--replicate pair. Each subject had $n_i\in\{2,\ldots,12\}$ visits with probabilities $(0.10,0.25,0.22,0.18,0.10,0.05,0.05,0.02,0.015,0.008,0.007)$, beginning at $t_{i1}=0$ with independent follow-up gaps drawn from $\mathcal U(0.5,1.5)$. At both the measurement and latent levels, the first covariate was drawn from a Bernoulli$(0.5)$ distribution and the second from a standard normal distribution; the latent covariates were fixed by subject, whereas measurement covariates varied by visit.

Baseline responses were complete. For each subsequent visit, the item missingness indicator was generated from a logistic model of current measurement covariates and the \emph{previous simulated complete} item response (before masking). This design combines irregular observation times with severity-related item recording, as may occur in longitudinal health studies \citep{DiggleKenward1994}. In contrast to a sequential-MAR mechanism based strictly on recorded history, a previously generated response can be absent from the observed dataset. Thus, ignorability under the formal definitions of \citet{Rubin1976Missing} and \citet{Seaman2013MAR} is not assured. The fitted IRT likelihood retained observed items without estimating the recording process; the simulated recovery results consequently include any effect of this potential observation-process misspecification. Accordingly, the recovery experiment evaluates the stated mask-generating process without invoking the theoretical MAR conditions. Exact item-specific coefficients are given for $R=2$ and $R=4$ below and were reused at $R=8$.

The common MCMC configuration used three chains, 2,000 warm-up and 2,000 post-warm-up iterations per chain, a target acceptance probability of 0.95, a maximum tree depth of 12, an initialization radius of 0.1, and eight threads per chain. Retention required $\max\widehat R<1.10$ and zero divergent transitions; tree-depth saturation and the energy Bayesian fraction of missing information (E-BFMI) were recorded as diagnostics but were not additional exclusion criteria. The separate matched-prior LOU/CLOUD computational benchmark in Appendix~\ref{appd:efficiency_comparison} has its own prior protocol and is not described by Supplementary Table~\ref{tab:priors}.

\subsection{Two-dimensional simulations}

This supplementary section provides the exact parameterization and distributional assumptions used to generate the two-dimensional latent-space simulation. In this setting, the latent process has dimension $R=2$, the measurement model contains $K=7$ categorical items, and the latent mean structure includes both a population-level linear time trend and covariate-dependent linear trends.

\subsubsection{Observation scheme and covariates}

For each of the $N=600$ individuals, the number of measurement occasions $n_i$ was drawn from a discrete distribution on $\{2,\ldots,12\}$:
\begin{align}
    \mathbb{P}(n_i = k)
    =
    \{0.10, 0.25, 0.22, 0.18, 0.10, 0.05, 0.05, 0.02, 0.015, 0.008, 0.007\},
\end{align}
for $k \in \{2,3,\ldots,12\}$. The first observation time was set to $t_{i1}=0$. For subsequent observations, the time intervals
$\Delta t_{ij}=t_{ij}-t_{i,j-1}$, $j=2,\ldots,n_i$, were generated independently from
\begin{align}
    \Delta t_{ij} \sim \mathcal{U}(0.5,1.5).
\end{align}
The resulting observation times therefore varied irregularly within and across individuals.

Two sets of covariates were generated. The time-invariant latent-level covariates were
\[
    \boldsymbol{x}^{(2)}_i =
    \left(x^{(2)}_{i1}, x^{(2)}_{i2}\right)^\top,
\]
where
\[
    x^{(2)}_{i1} \sim \text{Bernoulli}(0.5),
    \qquad
    x^{(2)}_{i2} \sim \mathcal{N}(0,1).
\]
The time-varying measurement-level covariates were
\[
    \boldsymbol{x}^{(1)}_{ij} =
    \left(x^{(1)}_{ij1}, x^{(1)}_{ij2}\right)^\top,
\]
where
\[
    x^{(1)}_{ij1} \sim \text{Bernoulli}(0.5),
    \qquad
    x^{(1)}_{ij2} \sim \mathcal{N}(0,1).
\]

\subsubsection{The two-dimensional latent Ornstein--Uhlenbeck process}

The latent process was generated from a stationary two-dimensional Ornstein--Uhlenbeck process with an added linear mean structure. Specifically, we first generated a zero-mean stationary OU process
$\boldsymbol{\xi}^*_{ij}=\boldsymbol{\xi}^*_i(t_{ij})$ and then defined the latent process entering the measurement model as
\begin{align}
    \boldsymbol{\xi}_{ij}
    =
    \boldsymbol{\xi}^*_{ij}
    +
    \left(
        \boldsymbol{\alpha}
        +
        \boldsymbol{\Phi}\boldsymbol{x}^{(2)}_i
    \right)t_{ij}.
\end{align}
Thus, $\boldsymbol{\alpha}$ represents the population-level linear trend, while $\boldsymbol{\Phi}$ captures how time-invariant covariates modify the latent trajectory over time. For $R=2$, the exact data-generating drift matrices are
\begin{align}
\boldsymbol\Gamma_{S1}&=\begin{bmatrix}0.70&-0.25\\-0.35&0.65\end{bmatrix},
&\boldsymbol\Gamma_{S2}&=\begin{bmatrix}-0.45&1.17\\-1.46&1.28\end{bmatrix}.
\end{align}
Their eigenvalues are $(0.9719,0.3781)$ for S1 and $(0.415\pm0.979\mathrm{i})$ for S2. The fitted S1--S2 models share the same measurement design and latent covariance specification.

The parameter values for the real-eigenvalue S1 and complex-eigenvalue S2 scenarios are reported in Supplementary Tables~\ref{tab:s1_comparison_1} and \ref{tab:s2_comparison_1}, respectively, including the continued portions of each table. The matrix $\boldsymbol\Omega$ is shared between the two scenarios, with $\rho$ denoting its off-diagonal entry.

The initial latent state was generated from the stationary distribution,
\begin{align}
    \boldsymbol{\xi}^*_{i1}
    \sim
    \mathcal{N}_2(\boldsymbol{0},\boldsymbol{\Omega}).
\end{align}
For $j>1$, conditional on the previous latent state, the transition distribution was
\begin{align}
    \boldsymbol{\xi}^*_{ij}
    \mid
    \boldsymbol{\xi}^*_{i,j-1}
    \sim
    \mathcal{N}_2
    \left\{
        \boldsymbol{E}_{ij}\boldsymbol{\xi}^*_{i,j-1},
        \boldsymbol{\Omega}
        -
        \boldsymbol{E}_{ij}\boldsymbol{\Omega}\boldsymbol{E}_{ij}^{\top}
    \right\},
\end{align}
where
\begin{align}
    \boldsymbol{E}_{ij}
    =
    \exp\left(-\boldsymbol{\Gamma}\Delta t_{ij}\right).
\end{align}

\subsubsection{The IRT measurement model}

The measurement model maps the $R=2$ latent variables to $K=7$ categorical items. Items 1--3 are binary, and items 4--7 are ordinal with four response categories. Let $Y_{ijk}$ denote the response for individual $i$ at visit $j$ on item $k$. Define the item-level linear predictor $\ell_{ijk}=\boldsymbol\beta_k^\top\mathbf x^{(1)}_{ij}+\boldsymbol\Lambda_k^\top\boldsymbol\xi_{ij}+b_{ik}$ and distinct response conventions for binary and ordinal items:
\begin{align}
   \mathbb P(Y_{ijk}=1)&=\operatorname{expit}(\theta_k^{\mathrm{bin}}+\ell_{ijk}), &&k=1,2,3,\\
   \mathbb P(Y_{ijk}\le m)&=\operatorname{expit}(\theta_{km}^{\mathrm{ord}}-\ell_{ijk}), &&k=4,\ldots,7,\quad m=1,2,3.
\end{align}
Here, $\operatorname{expit}(x)=1/(1+e^{-x})$, $b_{ik}\sim\mathcal N(0,\sigma_{bk}^2)$, and the ordinal outcomes are recorded in categories $1,2,3,4$ (binary outcomes use $0,1$). For binary items, the binary-response logit intercept $\theta_k^{\mathrm{bin}}$ parameterizes $\mathbb P(Y=1)$; in the cumulative $\mathbb P(Y\le0)$ convention of Equation~\ref{eq:irt}, this is equivalent to using the cutpoint $-\theta_k^{\mathrm{bin}}$. The ordinal implementation uses the cumulative convention directly.

For an ordinal item with $c_k=4$ response categories, let $F_m=\mathbb P(Y_{ijk}\le m)$ for $m=1,2,3$. The category probabilities are $p_1=F_1$, $p_m=F_m-F_{m-1}$ for $m=2,3$, and $p_4=1-F_3$. The simulator drew from these probabilities using the complete data before imposing missingness.

A simple loading structure was used: the first three items loaded on the first latent dimension, and the remaining four items loaded on the second latent dimension.
The full identified nonzero loading pattern is
\begin{align}
\boldsymbol\Lambda_{S1,S2}^\top=\begin{bmatrix}
\lambda_{1,1}&\lambda_{2,1}&\lambda_{3,1}&0&0&0&0\\
0&0&0&\lambda_{4,2}&\lambda_{5,2}&\lambda_{6,2}&\lambda_{7,2}
\end{bmatrix}.
\end{align}
The loading parameters and measurement-covariate effects are shared between S1 and S2; their data-generating values are reported in Supplementary Tables~\ref{tab:s1_comparison_1} and \ref{tab:s2_comparison_1}, including the continued portions of each table.

For binary items, one threshold was specified. For ordinal items with four categories, three ordered thresholds were specified. The item-specific thresholds and random-effect standard deviations are shared between S1 and S2 and are reported in the continued portions of Supplementary Tables~\ref{tab:s1_comparison_2} and~\ref{tab:s2_comparison_2}.

\subsubsection{Lagged-response-dependent missing-data mechanism}
Baseline measurements were fully observed. For visits $j>1$, item-level missingness depended on the current measurement-level covariates and the previous \emph{complete simulated response} for the same item, before masking. Let $M_{ijk}$ indicate whether item $k$ is missing for individual $i$ at visit $j$. The missingness probability was specified as
\begin{align}
    \mathbb{P}(M_{ijk}=1)
    =
    \operatorname{expit}
    \left(
        \kappa_{k0}
        +
        \kappa_{k1}x^{(1)}_{ij1}
        +
        \kappa_{k2}x^{(1)}_{ij2}
        +
        \kappa_{k3}Y^{\mathrm{full}}_{i,j-1,k}
    \right).
\end{align}
The item-specific missingness coefficients were
\begin{align}
    \boldsymbol{\kappa}
    =
    \begin{bmatrix}
        -0.25 & -0.27 & -0.6 & -1.8 \\
        -0.24 &  0.12 & -0.4 & -1.8 \\
        -0.27 & -0.11 & -0.7 & -1.5 \\
        -0.22 & -0.20 & -0.7 & -1.7 \\
         0.25 & -0.16 & -0.7 & -1.8 \\
         0.20 & -0.13 & -0.8 & -1.9 \\
         0.25 &  0.06 & -1.1 & -1.5
    \end{bmatrix}.
\end{align}
These coefficients were chosen to induce moderate item-level missingness after baseline while preserving complete baseline observations.

The predictor used the previous complete simulated response even when its value was masked at the preceding visit. Such history-dependent recording is motivated by clinical settings in which item availability can be associated with earlier symptom severity; nevertheless, the dependence on an intermittently unobserved outcome prevents a general claim of observed-history MAR \citep{Rubin1976Missing,Seaman2013MAR}. The fitted model integrated over missing responses but did not estimate their recording mechanism. The resulting recovery study therefore includes possible effects of nonignorability and does not provide a calibration result under a purely ignorable MAR mechanism.

\subsection{Four-dimensional simulations}
This section specifies the data-generating parameters and distributions for the four-dimensional simulations.

\subsubsection{Observation scheme and covariates}
For each of the $N=600$ individuals, the number of measurement occasions $n_i$ was drawn from a discrete distribution ranging from 2 to 12. To mimic sparse clinical follow-up, the visit-count distribution placed most of its mass on three to five observations:
\begin{align}
    \mathbb{P}(n_i = k) = \{0.10, 0.25, 0.22, 0.18, 0.10, 0.05, 0.05, 0.02, 0.015, 0.008, 0.007\}
\end{align}
for $k \in \{2, 3, \dots, 12\}$. The time of the first observation was set to $t_{i1} = 0$. Subsequent time intervals $\Delta t_{ij} = t_{ij} - t_{i,j-1}$ were drawn from a uniform distribution $\Delta t_{ij} \sim \mathcal{U}(0.5, 1.5)$.
The covariate distributions match those in the two-dimensional setting:
\begin{itemize}
    \item Time-invariant latent-level covariates ($\boldsymbol{x}^{(2)}_i$): generated as $\boldsymbol{x}^{(2)}_i = [\boldsymbol{x}^{(2)}_{i1}, \boldsymbol{x}^{(2)}_{i2}]^\top$, where $\boldsymbol{x}^{(2)}_{i1} \sim \text{Bernoulli}(0.5)$ and $\boldsymbol{x}^{(2)}_{i2} \sim \mathcal{N}(0, 1)$.
    \item Time-varying measurement-level covariates ($\boldsymbol{x}^{(1)}_{ij}$): generated as $\boldsymbol{x}^{(1)}_{ij} = [\boldsymbol{x}^{(1)}_{ij1}, \boldsymbol{x}^{(1)}_{ij2}]^\top$, where $\boldsymbol{x}^{(1)}_{ij1} \sim \text{Bernoulli}(0.5)$ and $\boldsymbol{x}^{(1)}_{ij2} \sim \mathcal{N}(0, 1)$.
\end{itemize}

\subsubsection{The multivariate Ornstein--Uhlenbeck process}
The evolution of the four-dimensional ($R=4$) latent process $\boldsymbol{\xi}_i(t)$ follows a multivariate time-inhomogeneous Ornstein--Uhlenbeck process. The exact drift matrices are given below; the mean parameters $\boldsymbol\Phi$ link baseline covariates to latent progression trends.
The final design uses the same identified correlation matrix for S3 and S4,
\[
\boldsymbol\Omega_{S3,S4}=\begin{bmatrix}
1&-0.020&0.040&0.060\\
-0.020&1&-0.048&-0.060\\
0.040&-0.048&1&0\\
0.060&-0.060&0&1
\end{bmatrix}.
\]
The two scenarios share the stationary correlation matrix above but differ in the drift matrices specified below. For each scenario, $S=(\boldsymbol\Gamma\boldsymbol\Omega+\boldsymbol\Omega\boldsymbol\Gamma^\top)/2$ is strictly positive definite; hence the infinitesimal diffusion covariance $\Sigma=2S$ is valid. We generated the latent process as $\boldsymbol\xi_i(t)=\boldsymbol\xi_i^*(t)+(\boldsymbol\Phi\boldsymbol x_i^{(2)}+\boldsymbol\alpha)t$, with $\boldsymbol\xi_i^*(0)\sim\mathcal N_4(0,\boldsymbol\Omega)$ and centered transition matrix $\exp(-\boldsymbol\Gamma\Delta t_{ij})$.

The data-generating drift matrices for S3 and S4 are
\begin{align}
\boldsymbol\Gamma_{S3}&=\begin{bmatrix}
0.79&-0.29&0.17&-0.12\\
-0.19&0.89&-0.10&0.17\\
-0.11&0.12&0.71&0.19\\
-0.38&0.45&-0.19&1.04
\end{bmatrix},\nonumber\\
\boldsymbol\Gamma_{S4}&=\begin{bmatrix}
0.50&2.07&-0.08&0.00\\
-1.93&0.30&0.16&0.08\\
-0.12&0.25&0.40&-2.47\\
0.00&0.12&2.43&0.90
\end{bmatrix}.
\end{align}
The S3 matrix has positive real eigenvalues $(1.2885,0.8588,0.6081,0.6745)$, while the S4 matrix has stable complex eigenvalues $(0.3751\pm2.0252\mathrm{i},\,0.6749\pm2.4037\mathrm{i})$. Both cases use the same stationary correlation matrix specified above.

\subsubsection{The IRT model}
The IRT model maps the $R=4$ latent variables to $K=12$ items: items 1--5 are binary, and items 6--12 are ordinal with four categories each. Let $Y_{ijk}$ be the response of subject $i$ to item $k$ at visit $j$. Define the item-level linear predictor $\ell_{ijk}=\boldsymbol\beta_k^\top\mathbf x^{(1)}_{ij}+\boldsymbol\Lambda_k^\top\boldsymbol\xi_i(t_{ij})+b_{ik}$, with $b_{ik}\sim\mathcal N(0,\sigma_{bk}^2)$. The binary and ordinal responses were generated according to
\begin{align}
\mathbb P(Y_{ijk}=1)&=\operatorname{expit}(\theta_k^{\mathrm{bin}}+\ell_{ijk}), && k=1,\ldots,5,\\
\mathbb P(Y_{ijk}\le m)&=\operatorname{expit}(\theta_{km}^{\mathrm{ord}}-\ell_{ijk}), && k=6,\ldots,12,\ m=1,2,3.
\end{align}
The binary outcomes are coded $0,1$, whereas the ordinal outcomes are coded $1,2,3,4$. As for $R=2$, the binary-response logit intercepts have the opposite sign from a cutpoint expressed through $\mathbb P(Y\le0)$. For ordinal items, the three thresholds are strictly ordered.

To support model identification, the factor loadings ($\boldsymbol{\Lambda}$) were constrained so that each item loaded on a single factor (simple structure).
The complete loading pattern (all entries not shown are constrained to zero) is
\begin{align}
\boldsymbol\Lambda_{S3,S4}^\top=\begin{bmatrix}
\lambda_{1,1}&\lambda_{2,1}&\lambda_{3,1}&0&0&0&0&0&0&0&0&0\\
0&0&0&\lambda_{4,2}&\lambda_{5,2}&\lambda_{6,2}&0&0&0&0&0&0\\
0&0&0&0&0&0&\lambda_{7,3}&\lambda_{8,3}&\lambda_{9,3}&0&0&0\\
0&0&0&0&0&0&0&0&0&\lambda_{10,4}&\lambda_{11,4}&\lambda_{12,4}
\end{bmatrix}.
\end{align}
The loading and measurement-covariate parameters are common to S3 and S4; their data-generating values are reported in Supplementary Tables~\ref{tab:s3_comparison_1} and \ref{tab:s4_comparison_1}, including the continued portions of each table.
Each binary item has one threshold, whereas each four-category ordinal item has three ordered thresholds. The threshold parameters and random-effect variances ($\sigma_b^2$) are also shared between S3 and S4; the continued portions of Supplementary Tables~\ref{tab:s3_comparison_2} and \ref{tab:s4_comparison_2} report the thresholds and corresponding random-effect standard deviations.

\subsubsection{Lagged-response-dependent missing-data mechanism}
To mimic incomplete item-level recording during clinical follow-up, baseline ($j=1$) measurements were fully observed. For $j > 1$, missingness was simulated at the item level. Let $M_{ijk}$ be an indicator that item $k$ is missing for individual $i$ at visit $j$. The probability of missingness followed a logistic model of the current measurement-level covariates $\boldsymbol{x}^{(1)}_{ij}$ and the previous complete response generated for that item (before any missingness mask was applied):
\begin{align}
    P(M_{ijk} = 1) = \text{expit}\left(\kappa_{k0} + \kappa_{k1}x^{(1)}_{ij1} + \kappa_{k2}x^{(1)}_{ij2} + \kappa_{k3}Y^{\mathrm{full}}_{i,j-1,k}\right)
\end{align}
The item-specific coefficients $\boldsymbol{\kappa}_k = [\kappa_{k0}, \kappa_{k1}, \kappa_{k2}, \kappa_{k3}]$ were specified as follows:
\begin{align}
\boldsymbol{\kappa} = \begin{bmatrix}
-0.25 & -0.27 & -0.6 & -1.8 \\
-0.24 &  0.12 & -0.4 & -1.8 \\
-0.27 & -0.11 & -0.7 & -1.5 \\
-0.22 & -0.20 & -0.7 & -1.7 \\
 0.10 & -0.10 & -0.5 & -1.6 \\
 0.25 & -0.16 & -0.7 & -1.8 \\
 0.20 & -0.13 & -0.8 & -1.9 \\
 0.25 &  0.06 & -1.1 & -1.5 \\
-0.15 &  0.05 & -0.6 & -1.4 \\
 0.18 & -0.08 & -0.9 & -1.7 \\
-0.20 &  0.15 & -0.5 & -1.5 \\
 0.22 & -0.11 & -0.8 & -1.8
\end{bmatrix}
\end{align}
The missingness mask was applied only after the complete responses were generated, and the lagged predictor could already be absent from the recorded history. The scheme captures response-dependent intermittency of the kind that complicates longitudinal studies \citep{DiggleKenward1994,Seaman2013MAR}, but it is not necessarily MAR relative to the observed responses. The fitted IRT likelihood used observed item scores without an explicit missingness submodel. These recovery results therefore measure the behavior of the fitted model under the stated mask-generating process, rather than demonstrating ignorable likelihood inference.

\clearpage
\subsection{Eight-dimensional simulations}\label{appd:simu8d_supp}

The high-dimensional extension uses $R=8$ latent factors and $K=24$ categorical items for each of $N=600$ subjects. Scenarios S5 and S6 extend the $R=4$ real-eigenvalue S3 and oscillatory S4 designs, respectively: S5 has strictly positive real drift eigenvalues and S6 has complex drift eigenvalues with positive real parts. Each scenario comprised 100 independently generated datasets.

\subsubsection{Observation, measurement, and mean structure}
The visit-count distribution, irregular time gaps, measurement-level covariates, latent-level covariates, and lagged-response-dependent missingness scheme matched those in the $R=4$ design in the preceding subsection. The $K=24$ measurement model duplicates the 12-item $R=4$ template: items 1--12 and 13--24 share the same loading magnitudes, measurement covariate effects, thresholds, and random-effect scales. Items 1--5 and 13--17 are binary; the remaining 14 items are ordinal with four response categories. Three consecutive items load on each of the eight latent factors. The mean parameters similarly duplicate the $R=4$ values, giving $\boldsymbol{\alpha}=(0.5,-0.3,0.2,-0.4,0.5,-0.3,0.2,-0.4)^\top$ and
\[
\boldsymbol{\Phi}
=
\begin{bmatrix}
0.4&-0.2\\
-0.3&0.5\\
0.2&0.1\\
-0.1&-0.3\\
0.4&-0.2\\
-0.3&0.5\\
0.2&0.1\\
-0.1&-0.3
\end{bmatrix}.
\]

\subsubsection{Eight-dimensional dynamic construction}
The eight-dimensional scenarios were constructed from four-dimensional covariance-matched reference pairs $(\boldsymbol{\Gamma}^{\mathrm{ref}}_4,\boldsymbol{\Omega}^{\mathrm{ref}}_4)$ with the real and complex spectra of S3 and S4, respectively. These reference pairs preserve the spectral regimes but use different drift coordinates from the S3 and S4 generating matrices. We formed two identical blocks, $\operatorname{blockdiag}(\boldsymbol{\Gamma}^{\mathrm{ref}}_4,\boldsymbol{\Gamma}^{\mathrm{ref}}_4)$ and $\operatorname{blockdiag}(\boldsymbol{\Omega}^{\mathrm{ref}}_4,\boldsymbol{\Omega}^{\mathrm{ref}}_4)$, applied a fixed orthogonal rotation $\boldsymbol V$, and rescaled the rotated covariance to a unit diagonal. The resulting systems have dense $R=8$ interactions, identified stationary correlations, and the duplicated four-dimensional spectra. Their data-generating matrices are

\begin{center}
\begingroup

\resizebox{0.96\linewidth}{!}{$
\boldsymbol{\Gamma}_{S5}=
\begin{bmatrix}
0.824 & -0.218 & -0.121 & -0.078 & 0.150 & -0.104 & 0.089 & 0.035 \\
0.002 & 0.754 & 0.069 & 0.027 & -0.186 & -0.101 & -0.005 & 0.057 \\
-0.144 & 0.256 & 0.935 & -0.077 & -0.203 & 0.211 & -0.059 & -0.038 \\
-0.027 & 0.075 & 0.037 & 0.779 & 0.061 & 0.134 & 0.038 & -0.026 \\
0.029 & -0.178 & -0.048 & -0.145 & 0.803 & 0.104 & 0.050 & -0.001 \\
0.087 & 0.106 & 0.429 & -0.134 & 0.165 & 1.157 & 0.232 & -0.044 \\
0.261 & -0.225 & -0.110 & -0.086 & 0.352 & 0.080 & 0.935 & 0.007 \\
0.067 & -0.047 & 0.006 & 0.017 & 0.138 & 0.113 & 0.110 & 0.671
\end{bmatrix}
$}
\par\medskip
\resizebox{0.96\linewidth}{!}{$
\boldsymbol{\Gamma}_{S6}=
\begin{bmatrix}
0.262 & 1.349 & -0.759 & 0.095 & -0.374 & 0.464 & 1.492 & -0.626 \\
-1.692 & 1.015 & -0.236 & -0.538 & -0.785 & -0.392 & -0.636 & -1.433 \\
0.006 & 0.469 & 1.136 & 0.615 & -0.152 & 1.650 & -1.162 & -0.626 \\
-0.613 & 0.091 & -0.545 & -0.008 & -1.806 & -0.988 & -1.266 & 1.641 \\
1.127 & -0.074 & 0.197 & 1.068 & 0.356 & -1.204 & -0.492 & -0.653 \\
-0.071 & 0.126 & -1.664 & 0.559 & 1.455 & 0.538 & 0.456 & 0.679 \\
-1.702 & 0.849 & 1.421 & 0.605 & 0.191 & 0.914 & 0.394 & 0.546 \\
-0.127 & 1.406 & 0.173 & -1.023 & 0.194 & -0.612 & -0.811 & 0.508
\end{bmatrix}
$}
\par\medskip
\resizebox{0.96\linewidth}{!}{$
\boldsymbol{\Omega}_{S5,S6}=
\begin{bmatrix}
1.000 & 0.288 & -0.072 & -0.152 & -0.215 & 0.258 & -0.241 & -0.019 \\
0.288 & 1.000 & -0.105 & -0.102 & 0.128 & -0.089 & -0.086 & 0.087 \\
-0.072 & -0.105 & 1.000 & 0.002 & 0.128 & -0.279 & 0.054 & 0.051 \\
-0.152 & -0.102 & 0.002 & 1.000 & -0.175 & -0.082 & -0.000 & -0.007 \\
-0.215 & 0.128 & 0.128 & -0.175 & 1.000 & -0.311 & -0.242 & 0.077 \\
0.258 & -0.089 & -0.279 & -0.082 & -0.311 & 1.000 & -0.082 & 0.157 \\
-0.241 & -0.086 & 0.054 & -0.000 & -0.242 & -0.082 & 1.000 & 0.008 \\
-0.019 & 0.087 & 0.051 & -0.007 & 0.077 & 0.157 & 0.008 & 1.000
\end{bmatrix}
$}
\endgroup
\end{center}

The eigenvalues of $\boldsymbol{\Gamma}_{S5}$ are the duplicated real S3 values $0.6081$, $0.6745$, $0.8588$, and $1.2885$ (each occurring twice). The eigenvalues of $\boldsymbol{\Gamma}_{S6}$ are the duplicated complex S4 values $0.3751\pm2.0252i$ and $0.6749\pm2.4037i$ (each conjugate pair appearing twice).

\subsubsection{Eight-dimensional priors, fitted models, and convergence screening}
Four aligned models were fitted to each simulated dataset: CLOUD, TrendyOU, StationaryOU, and DiagOU. CLOUD uses the full dense drift and covariance system and $\boldsymbol{\mu}_i(t)=(\boldsymbol{\Phi}\mathbf{x}^{(2)}_i+\boldsymbol{\alpha})t$; TrendyOU retains the full dynamic system but fixes $\boldsymbol{\Phi}=\mathbf 0$; StationaryOU fixes both $\boldsymbol{\Phi}=\mathbf 0$ and $\boldsymbol{\alpha}=\mathbf 0$; and DiagOU retains the covariate-dependent moving mean while constraining the latent dynamic dependence to be diagonal.

The models with coupled OU residual dynamics in $R=8$ use the same constructive priors as those in $R=2$ and $R=4$: $\log\tau_{S,r}\sim\mathcal N(0,0.5^2)$, $L_{R_S}\sim\operatorname{LKJ\text{-}Cholesky}(2)$, fixed $\tau_A=2$ with dimension-balanced standard-normal skew coordinates, $L_\Omega\sim\operatorname{LKJ\text{-}Cholesky}(8)$, and an additional $\mathcal N(0,2.5^2)$ density factor on each derived drift entry. Three chains were run with 2,000 warm-up and 2,000 post-warm-up iterations per chain, an initialization radius of 0.1, a target acceptance probability of 0.95, and a maximum tree depth of 12. Recovery summaries retain technically complete fits satisfying the prespecified $\widehat R<1.10$ screen with zero divergent transitions. For $\boldsymbol{\Omega}$, recovery was evaluated using the 28 unique off-diagonal correlations only; the fixed unit diagonal is excluded. Structural restrictions were evaluated at their implied zero values for the relevant ablation models. Complete parameter-level results for all 264 reported parameters under each fitted model are provided in Supplementary Tables~\ref{tab:r8_s5_full_recovery}--\ref{tab:r8_s6_full_recovery}.

\subsection{Computational efficiency benchmark}\label{appd:efficiency_comparison}

\paragraph{Benchmark design and inferential matching}
This benchmark compared the complete computational implementations of CLOUD and directly constrained latent Ornstein--Uhlenbeck (LOU) inference in the low-dimensional setting for which \citet{Tran2020LatentOM} derived explicit drift constraints. For each of two $R=2$ data-generating mechanisms, one oscillatory and one nonoscillatory, we generated 100 independent datasets with $N=600$ subjects and $K=7$ categorical items (three binary and four ordinal). The data-generating drift and covariance matrices were
\[
\boldsymbol\Gamma_{\mathrm{osc}}=
\begin{pmatrix}0.18&-0.07\\0.10&0.15\end{pmatrix},\qquad
\boldsymbol\Gamma_{\mathrm{non}}=
\begin{pmatrix}0.18&-0.07\\-0.10&0.15\end{pmatrix},\qquad
\boldsymbol\Omega=\begin{pmatrix}1&0.6\\0.6&1\end{pmatrix},
\]
so that the oscillatory drift has eigenvalues $0.165\pm0.082i$ whereas the nonoscillatory drift has eigenvalues $0.25$ and $0.08$. The measurement parameters, random-effect scales, visit-count distribution, and missingness mechanism reproduced the published $R=2$ simulation target of \citet{Tran2020LatentOM}. Both methods were fitted to the same simulated datasets using the same measurement likelihood. Both methods used a stationary latent mean, matching the data-generating specification.

The inferential target was also matched at the prior level. In both implementations, the drift elements $\Gamma_{rs}$ received the same diffuse $N(0,10^2)$ prior. CLOUD expressed that density in the constructive coordinates $\boldsymbol\Gamma=(\mathbf{S}+\mathbf{A})\boldsymbol\Omega^{-1}$ through the corresponding Jacobian, whereas LOU sampled $\Gamma$ directly under the algebraic admissibility constraints of the published direct formulation. For $R=2$, the CLOUD LKJ(1) prior on the correlation matrix induces the same uniform prior on $\rho\in(-1,1)$ as the direct LOU declaration. The nuisance-parameter prior family was likewise matched. CLOUD used the noncentered innovation representation for the latent trajectory, while LOU retained the centered latent-state representation of the source formulation. The comparison evaluated the combined matrix coordinates and latent-state representation; centered and noncentered representations can have different posterior geometry under the same statistical model \citep{papaspiliopoulos2007general,Betancourt2013HamiltonianMC}.

\paragraph{Sampling protocol and strict computational target}
Each fit used four No-U-Turn Sampler (NUTS) chains, 2,500 warm-up iterations per chain, up to 50,000 post-warm-up iterations, a target acceptance probability of 0.95, and a maximum tree depth of 12. A fit was declared to have reached the strict computational target only after two consecutive checkpoints satisfied all of the following: $\widehat R\leq1.01$, minimum bulk and tail effective sample sizes (ESS) of at least 400, a maximum ratio of Monte Carlo standard error to posterior standard deviation of at most 0.05, zero divergent transitions, and a maximum-tree-depth fraction no greater than 0.01. The $\widehat R$ and bulk/tail ESS diagnostics follow \citet{Vehtari2021Rhat}; the stated numerical cutoffs define this study's computational target. The energy Bayesian fraction of missing information (E-BFMI) was recorded as a warning diagnostic but was not used as a hard stopping criterion. Fits that did not meet the target were allowed to continue until 50,000 post-warm-up draws or a 90-hour runtime limit.

\paragraph{Reliability and time to the strict target}
The primary counts of fits reaching the target and median times are reported in Table~\ref{tab:sim_computational_summary} in the main text; the supplementary outcome diagnostics (Supplementary Table~\ref{tab:outcomes}) focus on intervals for success rates, upper percentiles of time to target, and reasons for nonattainment. Excluding the two nonoscillatory LOU fits terminated by computing-system limits, the paired success comparison was 198/198 for CLOUD versus 91/198 for LOU (two-sided exact McNemar $p<1.3\times10^{-32}$). Times to target were also substantially shorter and more concentrated under CLOUD. Among the 91 datasets for which both methods succeeded, every LOU fit took longer than its paired CLOUD fit. The median paired LOU/CLOUD ratio of times to target was 4.11 in the oscillatory scenario and 3.93 in the nonoscillatory scenario. Because these ratios condition on LOU success, they exclude its most difficult fits and are conservative with respect to the full computational burden.

The fixed-budget analysis compared target attainment at the same elapsed times. By 8 hours, CLOUD had reached the target in 199/200 datasets, versus 1/200 for LOU; by 16 hours, the counts were 200/200 versus 45/200; by 24 hours, 200/200 versus 78/200; and by 48 hours, 200/200 versus 90/200. One additional LOU fit succeeded at approximately 60.4 hours. These summaries include runs that had not attained the sampling criterion at the corresponding time.

\paragraph{Matched-draw efficiency and Hamiltonian Monte Carlo geometry}
To compare sampler behavior before the stopping rule could induce selection, we examined the common checkpoint at 1,000 post-warm-up draws, which was reached by both methods in all 200 paired datasets. At this matched posterior draw count, CLOUD was faster, had a larger minimum bulk ESS, used fewer leapfrog steps, and achieved a higher minimum bulk ESS per second of total elapsed time in every paired dataset (Supplementary Figure~\ref{fig:geometry}; Supplementary Table~\ref{tab:geometry}). The median paired LOU/CLOUD elapsed-time ratios were 2.58 and 2.77 for the oscillatory and nonoscillatory scenarios, respectively; the median CLOUD/LOU minimum-bulk-ESS ratios were 3.48 and 2.66; the median LOU/CLOUD leapfrog ratios were 4.02 in both scenarios; and the resulting CLOUD/LOU ESS-per-second ratios were 9.11 and 8.59.

The direct HMC diagnostics show that the efficiency difference corresponded to substantially more favorable sampling geometry rather than wall-clock speed alone. At 1,000 draws, the median minimum E-BFMI was 0.500 for CLOUD versus 0.111 for LOU in the oscillatory scenario and 0.503 versus 0.151 in the nonoscillatory scenario. All 200 LOU fits had minimum E-BFMI below 0.3, whereas none of the 200 CLOUD fits did. CLOUD also had zero divergences and zero maximum-tree-depth events in all 200 fits at this checkpoint; LOU had at least one divergence in 30/100 oscillatory and 15/100 nonoscillatory fits, and at least one maximum-tree-depth event in 40/100 and 38/100 fits, respectively. Consistent with the longer trajectories, the median LOU step size was smaller in both scenarios. E-BFMI was not part of the stopping rule, so these diagnostics are descriptive rather than components of the success criterion.

Drift recovery was comparable at the matched draw count. Median drift-matrix root mean squared error (RMSE) was 0.0420 for CLOUD versus 0.0422 for LOU in the oscillatory scenario and 0.0431 versus 0.0427 in the nonoscillatory scenario. The paired RMSE ratios are centered near one in Supplementary Figure~\ref{fig:geometry}D. Thus, the efficiency and geometry differences were already present at a fixed number of posterior draws, before differential early stopping, while drift recovery remained comparable.

\paragraph{Recovery conditional on successful computation}
For the 91 paired datasets for which both samplers satisfied the strict target, posterior recovery was essentially indistinguishable (Supplementary Figure~\ref{fig:common}; Supplementary Table~\ref{tab:recovery}). Pooling the two scenarios, the median per-dataset RMSE for the four drift elements was 0.0368 for CLOUD and 0.0358 for LOU, with pooled 95\% credible-interval coverage of 95.6\% and 95.3\%. For the single nontrivial $R=2$ correlation, the median absolute error was 0.0310 versus 0.0319 and coverage was 95.6\% versus 96.7\%. Recovery of factor loadings, measurement-covariate effects, thresholds, and random-effect standard deviations was similarly close. The common-success comparison therefore combines similar parameter recovery with the higher target-attainment rate and computational efficiency of CLOUD.

\suppsection{PPMI cohort, covariate screening, and prediction protocol}{appd:data_prep}
\subsection{Clinical observations and latent occasions}
\label{sec:ppmi_source}
The Parkinson's Progression Markers Initiative (PPMI) analysis included 1,679 participants, 19,101 examination records, and 33 five-category Movement Disorder Society--Unified Parkinson's Disease Rating Scale (MDS-UPDRS) Part III item scores over the first five years from the analysis time origin. The seven prespecified motor domains are listed in Supplementary Table~\ref{tab:ppmi_item_domain_mapping}. Of 630,333 potential item cells, 620,221 were observed; observed cells were modeled individually through the ordered-logistic likelihood, retaining assessments with other items missing. Empirical domain scores were calculated only for assessments with complete item data in the corresponding domain. Each Part III item had observations in all five categories, and every unordered pair of motor items was recorded together in at least one assessment.

For repeated examinations sharing the same participant and recorded time, the likelihood assigns a common latent occasion while preserving each recorded item response. The continuous-time innovation law is applied between distinct occasion times; no separate stochastic transition is introduced between simultaneous records. This construction accommodates multiple observed examinations without assuming that distinct clinical contexts at one recorded time are identical.

\subsection{Participant partition and baseline covariates}
\label{sec:ppmi_split}
The fixed split comprised 1,343 training participants and 336 test participants, with no overlap. The forecasting evaluation is exploratory because outcomes in the test partition had been examined during earlier model development. The 12 candidate predictors considered in the development set were age at the analysis time origin, sex, diagnosis duration, symptom duration, body mass index (BMI), Montreal Cognitive Assessment (MoCA), the Rapid Eye Movement Sleep Behavior Disorder Screening Questionnaire (RBDSQ) score, University of Pennsylvania Smell Identification Test (UPSIT), the Scales for Outcomes in Parkinson's Disease--Autonomic (SCOPA-AUT) assessment, Epworth Sleepiness Scale, Geriatric Depression Scale (GDS), and State--Trait Anxiety Inventory (STAI). The sleep and autonomic assessments are identified in the PPMI study documentation \citep{PPMIDataGuide2024}. The RBDSQ is a 10-item symptom-screening instrument with a total score from 0 to 13 \citep{StiasnyKolster2007RBDSQ}; SCOPA-AUT assesses autonomic symptoms across gastrointestinal, urinary, cardiovascular, thermoregulatory, pupillomotor, and sexual domains \citep{Visser2004SCOPAAUT}. Screening considered predictor availability, redundancy, empirical domain and ordinal associations, multiple-testing adjustment, and development-set cross-validation. Age and the RBDSQ score were entered into the covariate-dependent latent mean; inference on their coefficients is conditional on this selection. Age was standardized using the mean of 63.3453837067 and standard deviation of 9.9532114703 among training participants; the RBDSQ score (0--13) was standardized using a mean of 4.1146545178 and a standard deviation of 3.0115894641. For the 32 participants without a recorded RBDSQ score, the standardized score was set to zero and a separate missingness indicator was included in the measurement design.

Measurement-level covariates enter the ordered-logistic predictor in Equation~\ref{eq:irt} alongside item loadings and persistent participant-by-item random intercepts. The analysis used a single demographic missingness indicator for the coincident sex- and age-missingness patterns to avoid linear dependence. Empirical annual within-participant domain-change regressions used cluster-robust confidence intervals with Benjamini--Hochberg adjustment for seven trend and 14 age/RBDSQ interaction tests \citep{BenjaminiHochberg1995}. These native-score quantities are distinct from the coefficients of the latent moving mean.

\subsection{Latent measurement and comparison models}
\label{sec:ppmi_model_protocol}
The CLOUD measurement model contains 33 item-specific sets of ordered cutpoints, a fixed $33\times7$ single-domain loading support with anchor constraints, direct measurement covariates, and persistent item-specific participant intercepts. The centered latent initial state is $\mathbf r_i(0)\sim\mathcal N(0,\boldsymbol\Omega)$, and $\boldsymbol\mu_i(t)=t(\boldsymbol\alpha+\boldsymbol\Phi\mathbf x_i^{(2)})$. OU drift and stationary covariance parameters are shared across individuals. The seven comparator specifications, including the mean and dependence restrictions in TrendyOU, StationaryOU, and DiagOU, are summarized in Supplementary Table~\ref{tab:ppmi_v174_models}. In this study, GP-IRT places Mat\'ern-3/2 Gaussian processes on the seven latent trajectories over time \citep{RasmussenWilliams2006}; the measurement component retains ordinal response functions.

For the CLOUD population fit, sampling diagnostics for the development data showed a maximum $\widehat R=1.0447$, minimum bulk and tail effective sample sizes of approximately 119 and 235, respectively, and no recorded divergent transitions or maximum-tree-depth events.

\subsection{Forecast landmarks and target selection}
\label{sec:ppmi_forecast_protocol}
For each test participant, an eligible motor examination occasion is a distinct assessment time with at least one recorded motor item. The forecast landmark was the third such occasion. Candidate target examinations were restricted to 0.75--1.25 years after this landmark and ranked by their absolute distance from the one-year target. Ties were resolved first by the earlier assessment time and then by the lexicographically smaller assessment identifier. Selection depended only on timestamps and identifiers, not future item responses; items missing at the selected target were omitted during scoring. Of the 336 test participants, 58 lacked three distinct qualifying occasions and 41 had no examination in the target window. The final 237 selected participants contributed one target each and 7,787 observed target-item outcomes; their target gaps had a median of 1.0 year and a range of 0.75--1.25 years.

The permitted history comprised records at or before the landmark, including simultaneous examinations at the forecast landmark; all later assessments were excluded. There were 864 historical assessment rows, representing 711 distinct assessment times with usable motor data (three per selected participant). History records and target timestamps were stored separately from the target item responses. Each model produced its categorical predictive probabilities before the scoring stage paired predictions with observed target categories.

\subsection{Posterior predictive computation and examination context}
\label{sec:ppmi_v174_prediction}
Prediction integrated over 64 joint population-parameter draws selected from each model's stored development-set posterior samples. For each selected global draw, the ordinal history available at the forecast landmark defined a participant-conditional distribution over latent-history innovations and 33 persistent item intercepts for the OU and GP-IRT formulations. A locally initialized elliptical-slice sampler \citep{Murray2010EllipticalSlice} approximated this conditional distribution; the posterior predictive law subsequently propagated latent state uncertainty and averaged ordinal response probabilities over global and participant-specific draws. The local procedure used three chains, 256 warm-up iterations, and initially 512 retained draws per chain, with adaptive continuation governed by predictive-probability precision. These were participant-conditional sampling operations; the population parameters were not refitted using test participants.

Future clinical examination context was marginalized under a transition distribution estimated from development participants. This distribution contained 27 combinations of examination status relative to medication (recorded OFF, ON, or unclassified), Parkinson's disease medication status (yes, no, or unknown), and deep brain stimulation (DBS) status (yes, no, or unknown). Predictors included historical examination context, age/RBDSQ covariates, current motor severity, forecast-landmark time, and forecast interval. Estimation assigned comparable weights to training participants and smoothed sparse context categories. The forecast integrated over this context distribution rather than using the target examination's realized context as a predictor. The analyzed panel contained 9,213 ON, 5,459 OFF, and 4,429 unclassified examination rows.

\subsection{Scoring rules and paired comparison}
\label{sec:ppmi_v174_scoring}

Let $\mathcal I_{\mathrm{test}}$ denote the $N_{\mathrm{test}}=237$ scored participants. For participant $i$, let $\mathcal O_i^{\mathrm{tar}}$ contain the $m_i$ observed items at the selected target visit, let $y_{ik}^{\mathrm{tar}}\in\{0,\ldots,4\}$ be the target response for item $k$, and let $p_{ikc}^{\mathrm{tar}}$ be its predicted category probability. The negative log predictive density (NLPD) and normalized ranked probability score (RPS) are
\[
\mathrm{NLPD}_{ik}=-\log p_{ik,y_{ik}^{\mathrm{tar}}}^{\mathrm{tar}},\qquad
\mathrm{RPS}_{ik}=\frac14\sum_{c=0}^{3}
\left\{\sum_{a=0}^{c}p_{ika}^{\mathrm{tar}}-\mathbf1(y_{ik}^{\mathrm{tar}}\le c)\right\}^2.
\]
Mean-score mean absolute error (MAE) uses the probability-weighted mean category as the point prediction. Its participant-balanced value is
\begin{equation}
\mathrm{MAE}_{\mathrm{mean}}
=\frac{1}{N_{\mathrm{test}}}\sum_{i\in\mathcal I_{\mathrm{test}}}\frac{1}{m_i}
\sum_{k\in\mathcal O_i^{\mathrm{tar}}}
\left|\sum_{c=0}^4c\,p_{ikc}^{\mathrm{tar}}-y_{ik}^{\mathrm{tar}}\right|.
\label{eq:ppmi_mean_score_mae}
\end{equation}
The same participant-balanced aggregation is used for NLPD and RPS. Median-category MAE uses the predictive median category, and modal exact accuracy uses $\mathbf1\{\operatorname{argmax}_c p_{ikc}^{\mathrm{tar}}=y_{ik}^{\mathrm{tar}}\}$. Quadratic weighted kappa (QWK) summarizes pooled item-category decisions using the weighted-agreement framework of \citet{cohen1968weighted}. Paired score differences were evaluated by resampling whole participants. All eight models were scored on the same 7,787 observed items from the same 237 participants, preserving within-participant clustering in the comparisons.

\subsection{Observed-scale contrasts and latent-domain interpretation}
\label{sec:ppmi_independent_aggregate}
The empirical age/RBDSQ regressions assessed effects on native 0--4 observed-category changes, with the target-domain baseline score, recent historical change, temporal gap, and other covariates included in the adjustment. Across 48 contrasts between subgroups with high and low covariate values in observed-category calibration, none had a nominal 95\% interval excluding zero. Separately, 22 off-diagonal standardized one-year latent transitions selected by point-estimate magnitude involved seven distinct target domains; in the target-domain prediction comparison, no CLOUD--DiagOU NLPD advantage survived simultaneous uncertainty adjustment. These analyses complement the model's latent conditional associations with changes defined directly on the clinical rating scale.

\suppsection{Tables and figures}{appd:tables_figures}
\label{appd:tables}
\label{appd:figures}
\begin{table}[!htbp]
\centering
\caption{Unified simulation priors for $R=2$, $R=4$, and $R=8$; model-specific exceptions are indicated.}
\label{tab:priors}
\CloudStatTableStyle
\begin{threeparttable}
\begin{tabular*}{\linewidth}{@{\extracolsep{\fill}}p{0.35\linewidth}p{0.60\linewidth}@{}}
\toprule
Model parameter & Prior or fixed construction \\
\midrule
\multicolumn{2}{@{}l}{\textit{All fitted models: measurement and noncentered latent quantities}} \\
Binary/ordinal cutpoints $\theta$ & $\mathcal N(\mu_\theta,\sigma_\theta^2)$; ordinal triplets strictly ordered \\
Threshold mean and standard deviation $\mu_\theta,\sigma_\theta$ & $\mathcal N(0,5^2)$; $\mathcal N^+(0,2^2)$ \\
Loadings $\lambda_k$, loading standard deviation $\sigma_\lambda$ & $\mathcal N^+(1,\sigma_\lambda^2)$; $\mathcal N^+(0,2^2)$ \\
Item effects $\beta_{kp}$ & $\mathcal N(0,5^2)$ \\
Random effects $b_{ik}$ & $b_{ik}=\sigma_{bk}b_{ik}^{\mathrm{std}}$, $b_{ik}^{\mathrm{std}}\sim\mathcal N(0,1)$, $\sigma_{bk}\sim\mathcal N^+(0,1^2)$ \\
Standardized innovations $\mathbf z_{ij}$ & Independent $\mathcal N_R(\mathbf0,\mathbf I_R)$ \\
\midrule
\multicolumn{2}{@{}l}{\textit{Moving mean: CLOUD and DiagOU (TrendyOU: $\alpha$ only; StationaryOU: neither)}} \\
Covariate-dependent slopes $\Phi_{rq}$ & $\mathcal N(0,2^2)$ when included \\
Population slope $\alpha_r$ & $\mathcal N(0,0.5^2)$ when included \\
\midrule
\multicolumn{2}{@{}l}{\textit{Coupled residual dynamics: CLOUD, TrendyOU, StationaryOU}} \\
Symmetric scales $u_r=\log\tau_{S,r}$ & $u_r\sim\mathcal N(0,0.5^2)$, $\tau_{S,r}=e^{u_r}$ \\
Symmetric correlation factor $L_{R_S}$ & $\operatorname{LKJ\text{-}Cholesky}(2)$; $S=\operatorname{diag}(\tau_S)L_{R_S}L_{R_S}^{\top}\operatorname{diag}(\tau_S)$ \\
Skew coordinates $z_{A,rs}$ & $\mathcal N(0,1)$; $A_{rs}=2z_{A,rs}/\sqrt{R-1}$ for $r<s$; $A_{sr}=-A_{rs}$ \\
Skew amplitude $\tau_A$ & Fixed at $2.0$ (not a sampled hyperparameter) \\
Latent correlation factor $L_\Omega$ & $\operatorname{LKJ\text{-}Cholesky}(8)$, $\Omega=L_\Omega L_\Omega^\top$ \\
Induced drift $\Gamma=(S+A)\Omega^{-1}$ & Additional entrywise $\mathcal N(0,2.5^2)$ \emph{log-density factor} on the derived $\Gamma$ \\
\midrule
\multicolumn{2}{@{}l}{\textit{Independent residual dynamics: DiagOU only}} \\
Square root of diagonal drift $d_r=\sqrt{\Gamma_{rr}}$ & $\mathcal N^+(0,2^2)$, $\Gamma=\operatorname{diag}(d_1^2,\ldots,d_R^2)$ \\
Off-diagonal dependence, covariance & $\Gamma_{rs}=0$ for $r\ne s$; $\Omega=I_R$ (both fixed) \\
\bottomrule
\end{tabular*}
\begin{tablenotes}
\footnotesize
\item $\mathcal N(m,s^2)$ denotes a normal distribution with mean $m$ and standard deviation $s$; $\mathcal N^+$ is truncated to positive support. The additional factor on the derived $\Gamma$ is not an independent prior on free matrix elements. The same prior families are used at all three latent dimensions. The separate directly constrained LOU efficiency benchmark uses matched priors and is not covered by this table.
\end{tablenotes}
\end{threeparttable}
\end{table}

\begin{landscape}
\begin{table}[!htbp]
\centering
\caption{Full parameter-level simulation results: S1 (CLOUD, StationaryOU, and DiagOU).}
\label{tab:s1_comparison_1}

\footnotesize

\sisetup{
  mode = math
}

\setlength{\tabcolsep}{0.35pt}
\renewcommand{\arraystretch}{0.72}

\begin{threeparttable}
\begin{tabular*}{\linewidth}{@{\extracolsep{\fill}}l S[table-format=-1.2] *{3}{S[table-format=-2.3] S[table-format=2.3] S[table-format=3.1] S[table-format=4.1] S[table-format=1.2]}}
\toprule
Parameter & {True} & \multicolumn{5}{c}{CLOUD} & \multicolumn{5}{c}{StationaryOU} & \multicolumn{5}{c}{DiagOU} \\
\cmidrule(lr){3-7}\cmidrule(lr){8-12}\cmidrule(l){13-17}
 &  & {RB} & {MSE} & {CP} & {ESS} & {$\hat{R}$} & {RB} & {MSE} & {CP} & {ESS} & {$\hat{R}$} & {RB} & {MSE} & {CP} & {ESS} & {$\hat{R}$} \\
\midrule
\multicolumn{17}{@{}l}{\textit{$\boldsymbol\Gamma$ parameters}} \\
\midrule
$\Gamma_{1,1}$ & 0.7 & -0.059 & 0.003 & 95.6 & 4766.7 & 1.02 & -0.476 & 0.112 & 0.0 & 642.1 & 1.02 & 13.645 & 96.576 & 0.0 & 1887.6 & 1.00 \\
$\Gamma_{1,2}$ & -0.25 & 0.002 & 0.001 & 92.2 & 4475.0 & 1.02 & -1.060 & 0.075 & 0.0 & 295.1 & 1.07 & -1.000 & 0.0625 & 0.0 & \multicolumn{1}{c}{--} & \multicolumn{1}{c}{--} \\
$\Gamma_{2,1}$ & -0.35 & 0.003 & 0.001 & 88.9 & 4531.8 & 1.02 & -1.108 & 0.152 & 0.0 & 338.8 & 1.04 & -1.000 & 0.1225 & 0.0 & \multicolumn{1}{c}{--} & \multicolumn{1}{c}{--} \\
$\Gamma_{2,2}$ & 0.65 & 0.088 & 0.003 & 92.2 & 4222.2 & 1.01 & -0.631 & 0.168 & 0.0 & 798.2 & 1.03 & 12.794 & 71.263 & 0.0 & 1710.8 & 1.01 \\
\addlinespace[0.25em]
\multicolumn{17}{@{}l}{\textit{$\boldsymbol\Lambda$ parameters}} \\
\midrule
$\lambda_{1}$ & 1.2 & 0.031 & 0.012 & 96.7 & 3318.1 & 1.00 & 1.176 & 2.077 & 0.0 & 1013.7 & 1.01 & 0.150 & 0.055 & 72.0 & 2210.4 & 1.00 \\
$\lambda_{2}$ & 4.0 & 0.068 & 0.704 & 93.3 & 1957.1 & 1.01 & 1.607 & 46.593 & 3.3 & 292.5 & 1.04 & 0.321 & 3.253 & 90.0 & 563.1 & 1.02 \\
$\lambda_{3}$ & 4.1 & 0.105 & 0.905 & 91.1 & 1018.2 & 1.01 & 1.623 & 50.464 & 1.1 & 311.3 & 1.03 & 0.368 & 4.217 & 85.0 & 642.8 & 1.01 \\
$\lambda_{4}$ & 3.1 & 0.003 & 0.032 & 93.3 & 1921.2 & 1.00 & 0.688 & 4.664 & 0.0 & 942.3 & 1.01 & 0.092 & 0.129 & 77.0 & 1469.4 & 1.01 \\
$\lambda_{5}$ & 5.2 & 0.006 & 0.192 & 94.4 & 1064.5 & 1.01 & 0.758 & 16.625 & 0.0 & 417.4 & 1.03 & 0.106 & 0.642 & 89.0 & 708.9 & 1.01 \\
$\lambda_{6}$ & 3.0 & 0.015 & 0.032 & 94.4 & 1970.1 & 1.01 & 0.779 & 5.589 & 0.0 & 907.1 & 1.02 & 0.119 & 0.173 & 70.0 & 1531.3 & 1.01 \\
$\lambda_{7}$ & 1.7 & 0.003 & 0.007 & 94.4 & 2777.5 & 1.01 & 0.737 & 1.596 & 0.0 & 1429.9 & 1.01 & 0.084 & 0.029 & 56.0 & 2500.0 & 1.01 \\
\addlinespace[0.25em]
\multicolumn{17}{@{}l}{\textit{$\boldsymbol\Phi$ parameters}} \\
\midrule
$\Phi_{1,1}$ & 0.4 & 0.020 & 0.001 & 98.9 & 4281.6 & 1.00 & \multicolumn{1}{c}{--} & \multicolumn{1}{c}{--} & \multicolumn{1}{c}{--} & \multicolumn{1}{c}{--} & \multicolumn{1}{c}{--} & -0.073 & 0.002 & 88.0 & 3800.7 & 1.00 \\
$\Phi_{1,2}$ & -0.2 & 0.013 & 0.000 & 95.6 & 4058.6 & 1.00 & \multicolumn{1}{c}{--} & \multicolumn{1}{c}{--} & \multicolumn{1}{c}{--} & \multicolumn{1}{c}{--} & \multicolumn{1}{c}{--} & -0.078 & 0.001 & 90.0 & 3623.3 & 1.00 \\
$\Phi_{2,1}$ & -0.3 & 0.010 & 0.000 & 94.4 & 2775.7 & 1.01 & \multicolumn{1}{c}{--} & \multicolumn{1}{c}{--} & \multicolumn{1}{c}{--} & \multicolumn{1}{c}{--} & \multicolumn{1}{c}{--} & -0.071 & 0.001 & 89.0 & 2410.0 & 1.00 \\
$\Phi_{2,2}$ & 0.5 & 0.000 & 0.000 & 96.7 & 2478.9 & 1.01 & \multicolumn{1}{c}{--} & \multicolumn{1}{c}{--} & \multicolumn{1}{c}{--} & \multicolumn{1}{c}{--} & \multicolumn{1}{c}{--} & -0.076 & 0.002 & 37.0 & 2504.3 & 1.00 \\
\addlinespace[0.25em]
\multicolumn{17}{@{}l}{\textit{$\boldsymbol\alpha$ parameters}} \\
\midrule
$\alpha_{1}$ & 0.5 & 0.016 & 0.001 & 96.7 & 3426.9 & 1.01 & \multicolumn{1}{c}{--} & \multicolumn{1}{c}{--} & \multicolumn{1}{c}{--} & \multicolumn{1}{c}{--} & \multicolumn{1}{c}{--} & -0.075 & 0.002 & 67.0 & 2689.6 & 1.00 \\
$\alpha_{2}$ & -0.3 & -0.008 & 0.000 & 93.3 & 2993.9 & 1.00 & \multicolumn{1}{c}{--} & \multicolumn{1}{c}{--} & \multicolumn{1}{c}{--} & \multicolumn{1}{c}{--} & \multicolumn{1}{c}{--} & -0.085 & 0.001 & 75.0 & 2344.3 & 1.00 \\
\bottomrule
\end{tabular*}
\begin{tablenotes}
\footnotesize
\item RB: relative bias, defined as
$
\mathrm{RB}
=
M^{-1}\sum_{m=1}^{M}
\frac{\hat{\theta}^{(m)}-\theta_0}{\theta_0},
$
where $m=1,\ldots,M$ indexes the simulation replications, $M$ is the number of retained replications included in the reported summary for the corresponding model, $\hat{\theta}^{(m)}$ is the estimate from replication $m$, and $\theta_0$ is the true parameter value.
\item MSE: mean squared error, defined as
$
\mathrm{MSE}
=
M^{-1}\sum_{m=1}^{M}
\left(\hat{\theta}^{(m)}-\theta_0\right)^2.
$
\item CP: coverage probability, defined as
$
\mathrm{CP}
=
100 \times M^{-1}\sum_{m=1}^{M}
\mathbb{I}\!\left\{\theta_0 \in \mathrm{CI}^{(m)}\right\},
$
where $\mathrm{CI}^{(m)}$ is the credible interval from replication $m$.
\item ESS: effective sample size, measuring the amount of independent information in the posterior draws after accounting for autocorrelation.
\item $\hat{R}$: Gelman--Rubin diagnostic, measuring convergence across Markov chains; values close to 1 indicate good mixing and convergence.
\item A dash indicates that a parameter is not included in a model, RB is undefined because the true parameter value is zero, or a diagnostic is unavailable.
\end{tablenotes}
\end{threeparttable}
\end{table}
\clearpage

\begin{table}[!htbp]
\ContinuedFloat
\centering
\caption{Full parameter-level simulation results: S1 (continued).}
\label{tab:s1_comparison_2}

\footnotesize

\sisetup{
  mode = math
}

\setlength{\tabcolsep}{0.35pt}
\renewcommand{\arraystretch}{0.72}

\begin{threeparttable}
\begin{tabular*}{\linewidth}{@{\extracolsep{\fill}}l S[table-format=-1.2] *{3}{S[table-format=-2.3] S[table-format=2.3] S[table-format=3.1] S[table-format=4.1] S[table-format=1.2]}}
\toprule
Parameter & {True} & \multicolumn{5}{c}{CLOUD} & \multicolumn{5}{c}{StationaryOU} & \multicolumn{5}{c}{DiagOU} \\
\cmidrule(lr){3-7}\cmidrule(lr){8-12}\cmidrule(l){13-17}
 &  & {RB} & {MSE} & {CP} & {ESS} & {$\hat{R}$} & {RB} & {MSE} & {CP} & {ESS} & {$\hat{R}$} & {RB} & {MSE} & {CP} & {ESS} & {$\hat{R}$} \\
\midrule
\multicolumn{17}{@{}l}{\textit{$\boldsymbol\beta$ parameters}} \\
\midrule
$\beta_{1,1}$ & 0.3 & 0.075 & 0.035 & 97.8 & 8568.3 & 1.00 & -2.139 & 0.454 & 13.0 & 3086.0 & 1.01 & 0.056 & 0.039 & 97.0 & 7307.8 & 1.00 \\
$\beta_{1,2}$ & 0.5 & 0.012 & 0.012 & 92.2 & 7890.8 & 1.00 & -2.112 & 1.126 & 0.0 & 2612.8 & 1.01 & 0.028 & 0.014 & 95.0 & 6284.3 & 1.00 \\
$\beta_{2,1}$ & 0.1 & -0.109 & 0.089 & 94.4 & 5905.4 & 1.00 & -2.538 & 0.299 & 95.7 & 1556.6 & 1.01 & 0.052 & 0.129 & 96.0 & 3688.4 & 1.00 \\
$\beta_{2,2}$ & 0.2 & 0.100 & 0.027 & 92.2 & 5479.8 & 1.00 & -2.198 & 0.264 & 57.6 & 1375.7 & 1.01 & 0.317 & 0.058 & 91.0 & 2992.3 & 1.01 \\
$\beta_{3,1}$ & -0.1 & -0.080 & 0.121 & 96.7 & 5611.1 & 1.00 & -2.927 & 0.394 & 95.7 & 1613.8 & 1.01 & 0.055 & 0.183 & 97.0 & 3776.8 & 1.00 \\
$\beta_{3,2}$ & 0.2 & 0.199 & 0.027 & 93.3 & 5180.5 & 1.00 & -1.910 & 0.221 & 72.8 & 1537.8 & 1.01 & 0.463 & 0.054 & 96.0 & 3197.4 & 1.00 \\
$\beta_{4,1}$ & -0.2 & -0.088 & 0.025 & 93.3 & 5640.9 & 1.00 & -1.971 & 0.184 & 41.3 & 1772.4 & 1.01 & -0.154 & 0.034 & 96.0 & 2795.7 & 1.00 \\
$\beta_{4,2}$ & 0.4 & -0.013 & 0.007 & 93.3 & 5363.1 & 1.00 & -2.109 & 0.718 & 0.0 & 1735.5 & 1.02 & -0.005 & 0.008 & 93.0 & 2814.9 & 1.00 \\
$\beta_{5,1}$ & 0.3 & 0.028 & 0.057 & 96.7 & 4653.3 & 1.00 & -2.028 & 0.441 & 39.1 & 1450.1 & 1.01 & 0.088 & 0.090 & 95.0 & 2352.9 & 1.01 \\
$\beta_{5,2}$ & -0.3 & -0.036 & 0.011 & 98.9 & 4239.3 & 1.00 & -1.879 & 0.336 & 2.2 & 1392.6 & 1.01 & -0.036 & 0.018 & 97.0 & 2315.4 & 1.01 \\
$\beta_{6,1}$ & -0.1 & -0.304 & 0.024 & 95.6 & 5879.9 & 1.00 & -1.988 & 0.061 & 83.7 & 1918.5 & 1.00 & -0.379 & 0.032 & 98.0 & 3030.0 & 1.00 \\
$\beta_{6,2}$ & -0.2 & 0.013 & 0.007 & 95.6 & 5763.4 & 1.00 & -1.883 & 0.151 & 2.2 & 1879.8 & 1.00 & -0.003 & 0.009 & 97.0 & 3070.7 & 1.00 \\
$\beta_{7,1}$ & -0.2 & 0.039 & 0.012 & 97.8 & 7353.6 & 1.00 & -1.953 & 0.167 & 8.7 & 2320.3 & 1.01 & -0.008 & 0.016 & 96.0 & 3880.9 & 1.00 \\
$\beta_{7,2}$ & -0.1 & -0.021 & 0.003 & 96.7 & 7530.5 & 1.00 & -1.785 & 0.036 & 19.6 & 2352.7 & 1.00 & -0.062 & 0.003 & 98.0 & 4090.0 & 1.00 \\
\addlinespace[0.25em]
\multicolumn{17}{@{}l}{\textit{$\boldsymbol{\sigma}_{\text{bk}}$ parameters}} \\
\midrule
$\sigma_{\text{bk},1}$ & 3.7 & 0.021 & 0.084 & 97.8 & 2107.8 & 1.00 & 0.092 & 0.267 & 85.9 & 867.0 & 1.02 & 0.026 & 0.115 & 96.0 & 1893.0 & 1.01 \\
$\sigma_{\text{bk},2}$ & 3.4 & 0.068 & 0.523 & 94.4 & 2910.9 & 1.01 & 0.359 & 2.793 & 96.7 & 260.5 & 1.05 & 0.049 & 0.837 & 95.0 & 550.1 & 1.02 \\
$\sigma_{\text{bk},3}$ & 4.8 & 0.117 & 1.351 & 94.4 & 3974.8 & 1.01 & 0.527 & 9.762 & 93.5 & 323.4 & 1.02 & 0.199 & 2.794 & 97.0 & 637.9 & 1.01 \\
$\sigma_{\text{bk},4}$ & 3.1 & -0.003 & 0.047 & 92.2 & 1607.7 & 1.01 & 0.002 & 0.057 & 96.7 & 729.0 & 1.01 & -0.046 & 0.073 & 89.0 & 1284.7 & 1.01 \\
$\sigma_{\text{bk},5}$ & 6.1 & 0.004 & 0.216 & 98.9 & 1054.7 & 1.01 & 0.035 & 0.509 & 98.9 & 446.0 & 1.03 & -0.014 & 0.345 & 98.0 & 719.5 & 1.01 \\
$\sigma_{\text{bk},6}$ & 5.1 & 0.005 & 0.095 & 95.6 & 1578.6 & 1.01 & 0.015 & 0.165 & 94.6 & 819.4 & 1.02 & 0.004 & 0.144 & 97.0 & 1378.8 & 1.01 \\
$\sigma_{\text{bk},7}$ & 1.7 & 0.010 & 0.013 & 90.0 & 2067.0 & 1.01 & -0.004 & 0.016 & 93.5 & 673.1 & 1.02 & -0.036 & 0.018 & 90.0 & 1661.2 & 1.01 \\
\addlinespace[0.25em]
\multicolumn{17}{@{}l}{\textit{$\boldsymbol\theta$ parameters}} \\
\midrule
$\theta_{1}$ & 2.3 & 0.015 & 0.060 & 96.7 & 2879.1 & 1.00 & -1.388 & 10.249 & 0.0 & 1253.7 & 1.01 & 0.021 & 0.089 & 94.0 & 2486.9 & 1.00 \\
$\theta_{2}$ & 2.6 & 0.060 & 0.256 & 94.4 & 1369.6 & 1.01 & 0.018 & 0.470 & 98.9 & 392.7 & 1.03 & 0.165 & 0.672 & 97.0 & 692.8 & 1.02 \\
$\theta_{3}$ & 2.9 & 0.084 & 0.419 & 96.7 & 1422.9 & 1.01 & -0.178 & 0.641 & 91.3 & 492.8 & 1.02 & 0.203 & 0.992 & 98.0 & 804.5 & 1.01 \\
$\theta_{4,1}$ & -4.0 & 0.020 & 0.077 & 92.2 & 2318.2 & 1.00 & -0.579 & 5.461 & 0.0 & 578.0 & 1.02 & 0.023 & 0.103 & 90.0 & 1867.8 & 1.00 \\
$\theta_{4,2}$ & -1.0 & 0.052 & 0.048 & 92.2 & 3129.4 & 1.00 & -2.296 & 5.373 & 0.0 & 598.1 & 1.02 & 0.058 & 0.059 & 92.0 & 2822.1 & 1.00 \\
$\theta_{4,3}$ & 2.7 & -0.020 & 0.053 & 94.4 & 3087.3 & 1.00 & 0.874 & 5.727 & 0.0 & 720.3 & 1.02 & -0.019 & 0.061 & 94.0 & 2711.1 & 1.00 \\
$\theta_{5,1}$ & -7.5 & 0.011 & 0.423 & 96.7 & 1210.7 & 1.01 & -0.492 & 14.065 & 1.1 & 501.1 & 1.02 & 0.024 & 0.675 & 96.0 & 839.7 & 1.01 \\
$\theta_{5,2}$ & -2.5 & 0.018 & 0.132 & 97.8 & 1585.6 & 1.01 & -1.545 & 15.189 & 0.0 & 581.5 & 1.03 & 0.040 & 0.185 & 97.0 & 1439.0 & 1.01 \\
$\theta_{5,3}$ & 2.6 & -0.011 & 0.146 & 97.8 & 1703.4 & 1.01 & 1.569 & 17.317 & 0.0 & 486.8 & 1.02 & -0.009 & 0.172 & 99.0 & 1575.8 & 1.01 \\
$\theta_{6,1}$ & -5.5 & 0.019 & 0.114 & 95.6 & 1626.8 & 1.01 & -0.408 & 5.174 & 0.0 & 667.5 & 1.01 & 0.032 & 0.178 & 95.0 & 1642.2 & 1.01 \\
$\theta_{6,2}$ & -2.7 & 0.029 & 0.064 & 96.7 & 1508.1 & 1.01 & -0.835 & 5.189 & 0.0 & 645.4 & 1.01 & 0.045 & 0.089 & 98.0 & 1823.2 & 1.01 \\
$\theta_{6,3}$ & 2.5 & 0.001 & 0.074 & 94.4 & 1637.4 & 1.01 & 0.957 & 5.903 & 0.0 & 747.5 & 1.02 & 0.003 & 0.082 & 97.0 & 2048.1 & 1.00 \\
$\theta_{7,1}$ & -4.3 & 0.003 & 0.038 & 92.2 & 3531.0 & 1.00 & -0.301 & 1.715 & 0.0 & 743.5 & 1.03 & 0.000 & 0.036 & 95.0 & 2800.5 & 1.00 \\
$\theta_{7,2}$ & -1.0 & 0.005 & 0.018 & 90.0 & 5045.7 & 1.00 & -1.254 & 1.604 & 0.0 & 647.9 & 1.03 & 0.001 & 0.018 & 96.0 & 3691.2 & 1.00 \\
$\theta_{7,3}$ & 1.4 & 0.016 & 0.019 & 92.2 & 5291.5 & 1.00 & 0.929 & 1.730 & 0.0 & 740.5 & 1.02 & 0.016 & 0.023 & 92.0 & 4006.7 & 1.00 \\
\addlinespace[0.25em]
\multicolumn{17}{@{}l}{\textit{Other parameters}} \\
\midrule
$\rho$ & 0.6 & -0.047 & 0.001 & 95.6 & 1936.9 & 1.01 & -1.506 & 0.818 & 0.0 & 584.8 & 1.02 & -1.000 & 0.36 & 0.0 & \multicolumn{1}{c}{--} & \multicolumn{1}{c}{--} \\
\bottomrule
\end{tabular*}
\begin{tablenotes}
\footnotesize
\item RB, relative bias; MSE, mean squared error; CP, coverage of the 95\% credible interval (\%); ESS, effective sample size; $\hat R$, Gelman--Rubin convergence diagnostic. Metric formulas are given in Supplementary Table~\ref{tab:s1_comparison_1}. A dash denotes an excluded parameter, undefined relative bias when the true value is zero, or an unavailable diagnostic.
\end{tablenotes}
\end{threeparttable}
\end{table}
\clearpage
\end{landscape}

\begin{landscape}
\begin{table}[p]
\centering
\caption{Full parameter-level simulation results: S2 (CLOUD, StationaryOU, and DiagOU).}
\label{tab:s2_comparison_1}

\footnotesize

\sisetup{
  mode = math
}

\setlength{\tabcolsep}{0.35pt}
\renewcommand{\arraystretch}{0.72}
\begin{threeparttable}
\begin{tabular*}{\linewidth}{@{\extracolsep{\fill}}l S[table-format=-1.2] *{3}{S[table-format=-2.3] S[table-format=2.3] S[table-format=3.1] S[table-format=4.1] S[table-format=1.2]}}
\toprule
Parameter & {True} & \multicolumn{5}{c}{CLOUD} & \multicolumn{5}{c}{StationaryOU} & \multicolumn{5}{c}{DiagOU} \\
\cmidrule(lr){3-7}\cmidrule(lr){8-12}\cmidrule(l){13-17}
 &  & {RB} & {MSE} & {CP} & {ESS} & {$\hat{R}$} & {RB} & {MSE} & {CP} & {ESS} & {$\hat{R}$} & {RB} & {MSE} & {CP} & {ESS} & {$\hat{R}$} \\
\midrule
\multicolumn{17}{@{}l}{\textit{$\Gamma$ parameters}} \\
\midrule
$\Gamma_{1,1}$ & -0.45 & -0.085 & 0.016 & 92.4 & 4478.5 & 1.05 & -2.115 & 0.907 & 0.0 & 894.0 & 1.04 & -2.951 & 1.798 & 0.0 & 1314.4 & 1.01 \\
$\Gamma_{1,2}$ & 1.17 & -0.003 & 0.007 & 95.7 & 4783.0 & 1.04 & -0.433 & 0.257 & 0.0 & 652.2 & 1.04 & -1.000 & 1.369 & 0.0 & \multicolumn{1}{c}{--} & \multicolumn{1}{c}{--} \\
$\Gamma_{2,1}$ & -1.46 & -0.004 & 0.016 & 96.7 & 4535.8 & 1.03 & -0.859 & 1.572 & 0.0 & 569.9 & 1.06 & -1.000 & 2.132 & 0.0 & \multicolumn{1}{c}{--} & \multicolumn{1}{c}{--} \\
$\Gamma_{2,2}$ & 1.28 & -0.015 & 0.013 & 97.8 & 4725.7 & 1.03 & -0.845 & 1.172 & 0.0 & 439.2 & 1.03 & -0.314 & 0.166 & 0.0 & 1396.6 & 1.01 \\
\addlinespace[0.35em]
\multicolumn{17}{@{}l}{\textit{$\Lambda$ parameters}} \\
\midrule
$\lambda_{1}$ & 1.2 & 0.043 & 0.019 & 89.1 & 2763.1 & 1.01 & 1.345 & 2.689 & 0.0 & 1123.0 & 1.01 & 0.106 & 0.041 & 80.0 & 2180.7 & 1.00 \\
$\lambda_{2}$ & 4.0 & 0.150 & 1.001 & 97.8 & 728.9 & 1.02 & 1.544 & 41.719 & 1.2 & 375.7 & 1.06 & 0.372 & 3.557 & 84.0 & 616.6 & 1.01 \\
$\lambda_{3}$ & 4.1 & 0.195 & 1.489 & 96.7 & 789.2 & 1.02 & 1.636 & 48.184 & 0.0 & 411.7 & 1.03 & 0.411 & 4.033 & 83.0 & 710.1 & 1.01 \\
$\lambda_{4}$ & 3.1 & 0.021 & 0.044 & 93.5 & 1926.5 & 1.01 & 0.786 & 6.064 & 0.0 & 783.9 & 1.01 & 0.079 & 0.104 & 76.0 & 1484.0 & 1.01 \\
$\lambda_{5}$ & 5.2 & -0.004 & 0.242 & 95.7 & 1026.0 & 1.01 & 0.823 & 19.512 & 0.0 & 374.8 & 1.02 & 0.078 & 0.524 & 97.0 & 761.1 & 1.01 \\
$\lambda_{6}$ & 3.0 & 0.018 & 0.050 & 92.4 & 1921.1 & 1.01 & 0.868 & 6.966 & 0.0 & 771.1 & 1.01 & 0.081 & 0.116 & 79.0 & 1520.7 & 1.01 \\
$\lambda_{7}$ & 1.7 & 0.006 & 0.009 & 92.4 & 2679.0 & 1.00 & 0.811 & 1.926 & 0.0 & 1127.3 & 1.01 & 0.061 & 0.021 & 76.0 & 2291.0 & 1.01 \\
\addlinespace[0.35em]
\multicolumn{17}{@{}l}{\textit{$\boldsymbol\Phi$ parameters}} \\
\midrule
$\boldsymbol\Phi_{1,1}$ & 0.4 & 0.014 & 0.001 & 95.7 & 3971.6 & 1.00 & \multicolumn{1}{c}{--} & \multicolumn{1}{c}{--} & \multicolumn{1}{c}{--} & \multicolumn{1}{c}{--} & \multicolumn{1}{c}{--} & -0.039 & 0.002 & 97.0 & 3428.2 & 1.00 \\
$\boldsymbol\Phi_{1,2}$ & -0.2 & 0.004 & 0.000 & 97.8 & 3879.8 & 1.00 & \multicolumn{1}{c}{--} & \multicolumn{1}{c}{--} & \multicolumn{1}{c}{--} & \multicolumn{1}{c}{--} & \multicolumn{1}{c}{--} & -0.046 & 0.000 & 91.0 & 3268.2 & 1.00 \\
$\boldsymbol\Phi_{2,1}$ & -0.3 & -0.005 & 0.001 & 90.2 & 3161.4 & 1.00 & \multicolumn{1}{c}{--} & \multicolumn{1}{c}{--} & \multicolumn{1}{c}{--} & \multicolumn{1}{c}{--} & \multicolumn{1}{c}{--} & -0.055 & 0.001 & 89.0 & 2354.7 & 1.00 \\
$\boldsymbol\Phi_{2,2}$ & 0.5 & 0.001 & 0.000 & 92.4 & 2470.2 & 1.00 & \multicolumn{1}{c}{--} & \multicolumn{1}{c}{--} & \multicolumn{1}{c}{--} & \multicolumn{1}{c}{--} & \multicolumn{1}{c}{--} & -0.051 & 0.001 & 66.0 & 2141.7 & 1.00 \\
\addlinespace[0.35em]
\multicolumn{17}{@{}l}{\textit{$\boldsymbol\alpha$ parameters}} \\
\midrule
$\alpha_{1}$ & 0.5 & 0.004 & 0.001 & 95.7 & 2736.8 & 1.00 & \multicolumn{1}{c}{--} & \multicolumn{1}{c}{--} & \multicolumn{1}{c}{--} & \multicolumn{1}{c}{--} & \multicolumn{1}{c}{--} & -0.048 & 0.001 & 85.0 & 2320.8 & 1.00 \\
$\alpha_{2}$ & -0.3 & -0.008 & 0.000 & 89.1 & 3182.4 & 1.00 & \multicolumn{1}{c}{--} & \multicolumn{1}{c}{--} & \multicolumn{1}{c}{--} & \multicolumn{1}{c}{--} & \multicolumn{1}{c}{--} & -0.058 & 0.001 & 80.0 & 2257.4 & 1.01 \\
\bottomrule
\end{tabular*}
\begin{tablenotes}
\footnotesize
\item RB, relative bias; MSE, mean squared error; CP, coverage of the 95\% credible interval (\%); ESS, effective sample size; $\hat R$, Gelman--Rubin convergence diagnostic. Metric formulas are given in Supplementary Table~\ref{tab:s1_comparison_1}. A dash denotes an excluded parameter, undefined relative bias when the true value is zero, or an unavailable diagnostic.
\end{tablenotes}
\end{threeparttable}
\end{table}
\clearpage
\end{landscape}

\begin{landscape}
\begin{table}[!htbp]
\ContinuedFloat
\centering
\caption{Full parameter-level simulation results: S2 (continued).}
\label{tab:s2_comparison_2}

\footnotesize

\sisetup{
  mode = math
}

\setlength{\tabcolsep}{0.35pt}
\renewcommand{\arraystretch}{0.72}

\begin{threeparttable}
\begin{tabular*}{\linewidth}{@{\extracolsep{\fill}}l S[table-format=-1.2] *{3}{S[table-format=-2.3] S[table-format=2.3] S[table-format=3.1] S[table-format=4.1] S[table-format=1.2]}}
\toprule
Parameter & {True} & \multicolumn{5}{c}{CLOUD} & \multicolumn{5}{c}{StationaryOU} & \multicolumn{5}{c}{DiagOU} \\
\cmidrule(lr){3-7}\cmidrule(lr){8-12}\cmidrule(l){13-17}
 &  & {RB} & {MSE} & {CP} & {ESS} & {$\hat{R}$} & {RB} & {MSE} & {CP} & {ESS} & {$\hat{R}$} & {RB} & {MSE} & {CP} & {ESS} & {$\hat{R}$} \\
\midrule
\multicolumn{17}{@{}l}{\textit{$\boldsymbol\beta$ parameters}} \\
\midrule
$\beta_{1,1}$ & 0.3 & 0.010 & 0.033 & 95.7 & 8792.1 & 1.00 & -2.021 & 0.413 & 18.3 & 3084.6 & 1.00 & 0.015 & 0.036 & 96.0 & 7306.3 & 1.00 \\
$\beta_{1,2}$ & 0.5 & 0.044 & 0.011 & 94.6 & 7871.5 & 1.00 & -2.092 & 1.104 & 0.0 & 2661.7 & 1.00 & 0.045 & 0.012 & 94.0 & 6422.1 & 1.00 \\
$\beta_{2,1}$ & 0.1 & 0.357 & 0.119 & 94.6 & 5290.0 & 1.00 & -2.135 & 0.242 & 93.9 & 1826.6 & 1.01 & 0.463 & 0.166 & 97.0 & 3772.6 & 1.00 \\
$\beta_{2,2}$ & 0.2 & 0.315 & 0.032 & 94.6 & 4519.2 & 1.01 & -1.970 & 0.189 & 58.5 & 1773.1 & 1.01 & 0.527 & 0.055 & 96.0 & 2895.8 & 1.01 \\
$\beta_{3,1}$ & -0.1 & 0.141 & 0.148 & 96.7 & 5197.3 & 1.00 & -2.958 & 0.305 & 95.1 & 1863.3 & 1.01 & 0.060 & 0.207 & 97.0 & 3915.5 & 1.00 \\
$\beta_{3,2}$ & 0.2 & 0.308 & 0.057 & 94.6 & 4567.6 & 1.01 & -2.049 & 0.220 & 69.5 & 1805.3 & 1.01 & 0.513 & 0.083 & 93.0 & 3298.3 & 1.00 \\
$\beta_{4,1}$ & -0.2 & -0.010 & 0.028 & 94.6 & 5266.6 & 1.00 & -2.127 & 0.211 & 32.9 & 1495.4 & 1.01 & 0.014 & 0.029 & 96.0 & 3095.0 & 1.00 \\
$\beta_{4,2}$ & 0.4 & 0.036 & 0.008 & 96.7 & 4938.0 & 1.00 & -2.102 & 0.716 & 0.0 & 1461.6 & 1.01 & 0.038 & 0.009 & 95.0 & 3020.0 & 1.00 \\
$\beta_{5,1}$ & 0.3 & -0.129 & 0.060 & 95.7 & 4486.3 & 1.00 & -1.984 & 0.447 & 40.2 & 1258.5 & 1.01 & -0.102 & 0.069 & 95.0 & 2606.8 & 1.00 \\
$\beta_{5,2}$ & -0.3 & -0.034 & 0.018 & 94.6 & 4094.3 & 1.00 & -1.903 & 0.344 & 1.2 & 1185.7 & 1.01 & -0.021 & 0.021 & 94.0 & 2399.9 & 1.01 \\
$\beta_{6,1}$ & -0.1 & 0.199 & 0.028 & 94.6 & 5526.3 & 1.00 & -2.209 & 0.072 & 85.4 & 1613.4 & 1.01 & 0.118 & 0.030 & 97.0 & 3372.5 & 1.00 \\
$\beta_{6,2}$ & -0.2 & -0.141 & 0.006 & 97.8 & 5537.5 & 1.00 & -1.857 & 0.145 & 1.2 & 1606.0 & 1.00 & -0.143 & 0.008 & 97.0 & 3404.2 & 1.00 \\
$\beta_{7,1}$ & -0.2 & 0.048 & 0.015 & 94.6 & 6861.1 & 1.00 & -2.076 & 0.189 & 6.1 & 1972.2 & 1.00 & 0.054 & 0.017 & 93.0 & 4343.1 & 1.00 \\
$\beta_{7,2}$ & -0.1 & -0.129 & 0.004 & 94.6 & 7062.0 & 1.00 & -1.760 & 0.035 & 22.0 & 2014.3 & 1.01 & -0.097 & 0.004 & 94.0 & 4520.5 & 1.00 \\
\addlinespace[0.25em]
\multicolumn{17}{@{}l}{\textit{$\boldsymbol{\sigma}_{\text{bk}}$ parameters}} \\
\midrule
$\sigma_{\text{bk},1}$ & 3.7 & 0.020 & 0.117 & 93.5 & 2143.4 & 1.01 & 0.109 & 0.346 & 82.9 & 864.5 & 1.01 & 0.022 & 0.133 & 94.0 & 1979.3 & 1.01 \\
$\sigma_{\text{bk},2}$ & 3.4 & 0.135 & 0.722 & 95.7 & 2710.7 & 1.02 & 0.376 & 2.651 & 90.2 & 357.6 & 1.07 & 0.197 & 1.230 & 98.0 & 596.4 & 1.02 \\
$\sigma_{\text{bk},3}$ & 4.8 & 0.202 & 2.011 & 97.8 & 2777.9 & 1.02 & 0.430 & 5.881 & 89.0 & 407.8 & 1.03 & 0.325 & 3.987 & 91.0 & 704.3 & 1.01 \\
$\sigma_{\text{bk},4}$ & 3.1 & 0.016 & 0.049 & 94.6 & 1746.9 & 1.00 & 0.009 & 0.061 & 95.1 & 683.2 & 1.01 & -0.010 & 0.052 & 98.0 & 1327.6 & 1.01 \\
$\sigma_{\text{bk},5}$ & 6.1 & -0.011 & 0.275 & 95.7 & 1086.1 & 1.01 & 0.027 & 0.536 & 97.6 & 401.2 & 1.02 & 0.000 & 0.408 & 96.0 & 789.5 & 1.01 \\
$\sigma_{\text{bk},6}$ & 5.1 & 0.005 & 0.170 & 91.3 & 1745.7 & 1.01 & 0.024 & 0.197 & 90.2 & 751.5 & 1.01 & -0.001 & 0.178 & 91.0 & 1449.4 & 1.01 \\
$\sigma_{\text{bk},7}$ & 1.7 & 0.014 & 0.012 & 93.5 & 1962.1 & 1.01 & 0.006 & 0.019 & 93.9 & 704.5 & 1.02 & -0.012 & 0.013 & 94.0 & 1676.7 & 1.01 \\
\addlinespace[0.25em]
\multicolumn{17}{@{}l}{\textit{$\boldsymbol\theta$ parameters}} \\
\midrule
$\theta_{1}$ & 2.3 & 0.017 & 0.083 & 94.6 & 2918.9 & 1.00 & -1.168 & 7.269 & 0.0 & 1323.0 & 1.01 & 0.015 & 0.089 & 94.0 & 2518.3 & 1.00 \\
$\theta_{2}$ & 2.6 & 0.122 & 0.331 & 100.0 & 1018.0 & 1.01 & 0.694 & 4.039 & 62.2 & 434.6 & 1.09 & 0.251 & 0.871 & 97.0 & 766.7 & 1.01 \\
$\theta_{3}$ & 2.9 & 0.160 & 0.580 & 97.8 & 1099.8 & 1.02 & 0.522 & 2.819 & 90.2 & 497.1 & 1.02 & 0.281 & 1.202 & 93.0 & 890.8 & 1.01 \\
$\theta_{4,1}$ & -4.0 & 0.029 & 0.109 & 93.5 & 2503.5 & 1.00 & -0.252 & 1.084 & 6.1 & 866.5 & 1.02 & 0.029 & 0.105 & 96.0 & 1898.2 & 1.01 \\
$\theta_{4,2}$ & -1.0 & 0.056 & 0.062 & 92.4 & 3575.0 & 1.00 & -1.029 & 1.111 & 1.2 & 921.2 & 1.02 & 0.060 & 0.061 & 94.0 & 2851.5 & 1.01 \\
$\theta_{4,3}$ & 2.7 & 0.001 & 0.070 & 93.5 & 3539.2 & 1.00 & 0.407 & 1.311 & 0.0 & 940.7 & 1.02 & 0.003 & 0.065 & 95.0 & 2789.6 & 1.00 \\
$\theta_{5,1}$ & -7.5 & -0.005 & 0.461 & 96.7 & 1276.6 & 1.01 & -0.206 & 2.971 & 45.1 & 482.3 & 1.03 & 0.020 & 0.667 & 96.0 & 912.5 & 1.01 \\
$\theta_{5,2}$ & -2.5 & -0.007 & 0.159 & 94.6 & 2019.6 & 1.01 & -0.656 & 2.891 & 8.5 & 882.3 & 1.02 & 0.019 & 0.183 & 99.0 & 1563.4 & 1.01 \\
$\theta_{5,3}$ & 2.6 & -0.003 & 0.160 & 96.7 & 2238.9 & 1.01 & 0.695 & 3.649 & 8.5 & 559.5 & 1.02 & 0.016 & 0.180 & 96.0 & 1738.8 & 1.01 \\
$\theta_{6,1}$ & -5.5 & 0.014 & 0.207 & 93.5 & 1949.5 & 1.00 & -0.166 & 0.963 & 36.6 & 834.9 & 1.01 & 0.018 & 0.212 & 94.0 & 1657.7 & 1.01 \\
$\theta_{6,2}$ & -2.7 & 0.018 & 0.112 & 92.4 & 1974.6 & 1.00 & -0.353 & 0.988 & 17.1 & 879.0 & 1.01 & 0.025 & 0.113 & 95.0 & 1827.1 & 1.01 \\
$\theta_{6,3}$ & 2.5 & 0.007 & 0.069 & 97.8 & 2134.8 & 1.01 & 0.441 & 1.352 & 13.4 & 932.1 & 1.01 & 0.005 & 0.073 & 97.0 & 2039.0 & 1.01 \\
$\theta_{7,1}$ & -4.3 & 0.002 & 0.039 & 95.7 & 3511.3 & 1.00 & -0.132 & 0.362 & 17.1 & 1107.4 & 1.02 & 0.003 & 0.041 & 94.0 & 2768.0 & 1.00 \\
$\theta_{7,2}$ & -1.0 & -0.017 & 0.019 & 93.5 & 4912.7 & 1.00 & -0.542 & 0.313 & 6.1 & 1051.2 & 1.02 & -0.014 & 0.020 & 94.0 & 3665.1 & 1.00 \\
$\theta_{7,3}$ & 1.4 & 0.018 & 0.020 & 96.7 & 5287.6 & 1.00 & 0.412 & 0.353 & 3.7 & 1142.0 & 1.02 & 0.014 & 0.021 & 95.0 & 4023.3 & 1.00 \\
\addlinespace[0.25em]
\multicolumn{17}{@{}l}{\textit{Other parameters}} \\
\midrule
$\rho$ & 0.6 & -0.021 & 0.001 & 94.6 & 1108.1 & 1.02 & -1.591 & 0.912 & 0.0 & 597.4 & 1.02 & -1.000 & 0.360 & 0.0 & \multicolumn{1}{c}{--} & \multicolumn{1}{c}{--} \\
\bottomrule
\end{tabular*}
\begin{tablenotes}
\footnotesize
\item RB, relative bias; MSE, mean squared error; CP, coverage of the 95\% credible interval (\%); ESS, effective sample size; $\hat R$, Gelman--Rubin convergence diagnostic. Metric formulas are given in Supplementary Table~\ref{tab:s1_comparison_1}. A dash denotes an excluded parameter, undefined relative bias when the true value is zero, or an unavailable diagnostic.
\end{tablenotes}
\end{threeparttable}
\end{table}

\end{landscape}

\begin{landscape}
\begin{table}[!htbp]
\centering
\caption{Full parameter-level simulation results: S3 (CLOUD, StationaryOU, and DiagOU).}
\label{tab:s3_comparison_1}

\footnotesize

\sisetup{
  mode = math
}

\setlength{\tabcolsep}{0.35pt}
\renewcommand{\arraystretch}{0.72}

\begin{threeparttable}
\begin{tabular*}{\linewidth}{@{\extracolsep{\fill}}l S[table-format=-1.3] *{3}{S[table-format=-2.3] S[table-format=2.3] S[table-format=3.1] S[table-format=4.1] S[table-format=1.2]}}
\toprule
Parameter & {True} & \multicolumn{5}{c}{CLOUD} & \multicolumn{5}{c}{StationaryOU} & \multicolumn{5}{c}{DiagOU} \\
\cmidrule(lr){3-7}\cmidrule(lr){8-12}\cmidrule(l){13-17}
 &  & {RB} & {MSE} & {CP} & {ESS} & {$\hat{R}$} & {RB} & {MSE} & {CP} & {ESS} & {$\hat{R}$} & {RB} & {MSE} & {CP} & {ESS} & {$\hat{R}$} \\
\midrule
\multicolumn{17}{@{}l}{\textit{$\boldsymbol\Phi$ parameters}} \\
\midrule
$\Phi_{1,1}$ & 0.4 & 0.058 & 0.002 & 100.0 & 2201.6 & 1.00 & \multicolumn{1}{c}{--} & \multicolumn{1}{c}{--} & \multicolumn{1}{c}{--} & \multicolumn{1}{c}{--} & \multicolumn{1}{c}{--} & 0.071 & 0.005 & 100.0 & 1120.7 & 1.01 \\
$\Phi_{1,2}$ & -0.2 & 0.044 & 0.000 & 100.0 & 2019.2 & 1.01 & \multicolumn{1}{c}{--} & \multicolumn{1}{c}{--} & \multicolumn{1}{c}{--} & \multicolumn{1}{c}{--} & \multicolumn{1}{c}{--} & 0.119 & 0.002 & 100.0 & 1088.7 & 1.02 \\
$\Phi_{2,1}$ & -0.3 & 0.030 & 0.002 & 85.0 & 1778.7 & 1.00 & \multicolumn{1}{c}{--} & \multicolumn{1}{c}{--} & \multicolumn{1}{c}{--} & \multicolumn{1}{c}{--} & \multicolumn{1}{c}{--} & 0.145 & 0.005 & 91.0 & 1107.1 & 1.01 \\
$\Phi_{2,2}$ & 0.5 & 0.014 & 0.000 & 95.0 & 1163.9 & 1.01 & \multicolumn{1}{c}{--} & \multicolumn{1}{c}{--} & \multicolumn{1}{c}{--} & \multicolumn{1}{c}{--} & \multicolumn{1}{c}{--} & 0.035 & 0.001 & 100.0 & 637.9 & 1.01 \\
$\Phi_{3,1}$ & 0.2 & 0.024 & 0.001 & 95.0 & 1780.5 & 1.01 & \multicolumn{1}{c}{--} & \multicolumn{1}{c}{--} & \multicolumn{1}{c}{--} & \multicolumn{1}{c}{--} & \multicolumn{1}{c}{--} & 0.002 & 0.001 & 91.0 & 1305.9 & 1.00 \\
$\Phi_{3,2}$ & 0.1 & -0.020 & 0.000 & 100.0 & 2005.2 & 1.00 & \multicolumn{1}{c}{--} & \multicolumn{1}{c}{--} & \multicolumn{1}{c}{--} & \multicolumn{1}{c}{--} & \multicolumn{1}{c}{--} & -0.086 & 0.000 & 100.0 & 1463.0 & 1.01 \\
$\Phi_{4,1}$ & -0.1 & 0.125 & 0.001 & 100.0 & 1679.5 & 1.00 & \multicolumn{1}{c}{--} & \multicolumn{1}{c}{--} & \multicolumn{1}{c}{--} & \multicolumn{1}{c}{--} & \multicolumn{1}{c}{--} & 0.070 & 0.002 & 91.0 & 1413.1 & 1.01 \\
$\Phi_{4,2}$ & -0.3 & 0.007 & 0.000 & 90.0 & 1560.9 & 1.01 & \multicolumn{1}{c}{--} & \multicolumn{1}{c}{--} & \multicolumn{1}{c}{--} & \multicolumn{1}{c}{--} & \multicolumn{1}{c}{--} & -0.010 & 0.000 & 100.0 & 653.6 & 1.01 \\
\addlinespace[0.25em]
\multicolumn{17}{@{}l}{\textit{$\boldsymbol\Gamma$ parameters}} \\
\midrule
$\Gamma_{1,1}$ & 0.79 & -0.074 & 0.050 & 90.0 & 2276.6 & 1.04 & 1.333 & 1.169 & 0.0 & 572.5 & 1.01 & 0.632 & 0.524 & 0.0 & 614.6 & 1.03 \\
$\Gamma_{1,2}$ & -0.29 & 0.019 & 0.034 & 95.0 & 2242.9 & 1.04 & -0.377 & 0.373 & 16.7 & 398.8 & 1.02 & -1.000 & 0.0841 & 0.0 & \multicolumn{1}{c}{--} & \multicolumn{1}{c}{--} \\
$\Gamma_{1,3}$ & 0.17 & -0.058 & 0.016 & 100.0 & 2468.2 & 1.03 & 14.079 & 6.318 & 0.0 & 492.8 & 1.01 & -1.000 & 0.0289 & 0.0 & \multicolumn{1}{c}{--} & \multicolumn{1}{c}{--} \\
$\Gamma_{1,4}$ & -0.12 & 0.003 & 0.053 & 85.0 & 2478.2 & 1.03 & -0.237 & 0.161 & 33.3 & 431.5 & 1.01 & -1.000 & 0.0144 & 0.0 &\multicolumn{1}{c}{--} & \multicolumn{1}{c}{--} \\
$\Gamma_{2,1}$ & -0.19 & 0.014 & 0.032 & 90.0 & 2278.8 & 1.02 & -0.640 & 0.809 & 0.0 & 534.2 & 1.01 & -1.000 & 0.0361 & 0.0 & \multicolumn{1}{c}{--} & \multicolumn{1}{c}{--} \\
$\Gamma_{2,2}$ & 0.89 & 0.018 & 0.034 & 90.0 & 2333.1 & 1.02 & -0.784 & 0.723 & 0.0 & 471.2 & 1.01 & -0.292 & 0.093 & 49.0 & 668.8 & 1.01 \\
$\Gamma_{2,3}$ & -0.1 & -0.078 & 0.022 & 100.0 & 2484.2 & 1.01 & -0.547 & 0.030 & 83.3 & 471.7 & 1.02 & -1.000 & 0.01 & 0.0 & \multicolumn{1}{c}{--} & \multicolumn{1}{c}{--} \\
$\Gamma_{2,4}$ & 0.17 & 0.078 & 0.020 & 100.0 & 2480.6 & 1.02 & -3.186 & 0.311 & 16.7 & 524.2 & 1.01 & -1.000 & 0.0289 & 0.0 & \multicolumn{1}{c}{--} & \multicolumn{1}{c}{--} \\
$\Gamma_{3,1}$ & -0.11 & 0.042 & 0.013 & 100.0 & 2394.2 & 1.03 & -7.513 & 0.923 & 0.0 & 563.5 & 1.01 & -1.000 & 0.0121 & 0.0 & \multicolumn{1}{c}{--} & \multicolumn{1}{c}{--} \\
$\Gamma_{3,2}$ & 0.12 & -0.085 & 0.020 & 95.0 & 2433.4 & 1.02 & 3.182 & 0.672 & 0.0 & 538.0 & 1.01 & -1.000 & 0.0144 & 0.0 & \multicolumn{1}{c}{--} & \multicolumn{1}{c}{--} \\
$\Gamma_{3,3}$ & 0.71 & 0.060 & 0.015 & 95.0 & 2617.0 & 1.01 & 0.942 & 0.467 & 16.7 & 501.7 & 1.01 & 4.327 & 9.856 & 0.0 & 600.8 & 1.02 \\
$\Gamma_{3,4}$ & 0.19 & 0.006 & 0.008 & 100.0 & 2616.3 & 1.02 & 0.026 & 0.025 & 100.0 & 520.3 & 1.01 & -1.000 & 0.0361 & 0.0 & \multicolumn{1}{c}{--} & \multicolumn{1}{c}{--} \\
$\Gamma_{4,1}$ & -0.38 & 0.011 & 0.033 & 85.0 & 2428.2 & 1.02 & 2.820 & 1.279 & 0.0 & 534.8 & 1.01 & -1.000 & 0.1444 & 0.0 & \multicolumn{1}{c}{--} & \multicolumn{1}{c}{--} \\
$\Gamma_{4,2}$ & 0.45 & -0.036 & 0.009 & 100.0 & 2501.2 & 1.02 & 2.590 & 1.516 & 16.7 & 544.2 & 1.01 & -1.000 & 0.2025 & 0.0 & \multicolumn{1}{c}{--} & \multicolumn{1}{c}{--} \\
$\Gamma_{4,3}$ & -0.19 & 0.015 & 0.015 & 100.0 & 2657.7 & 1.01 & -0.466 & 0.987 & 0.0 & 491.5 & 1.01 & -1.000 & 0.0361 & 0.0 & \multicolumn{1}{c}{--} & \multicolumn{1}{c}{--} \\
$\Gamma_{4,4}$ & 1.04 & 0.008 & 0.031 & 90.0 & 2635.4 & 1.01 & -0.790 & 0.683 & 0.0 & 465.2 & 1.01 & 0.728 & 0.857 & 0.0 & 484.2 & 1.02 \\
\addlinespace[0.25em]
\multicolumn{17}{@{}l}{\textit{$\boldsymbol\alpha$ parameters}} \\
\midrule
$\alpha_{1}$ & 0.5 & 0.045 & 0.002 & 90.0 & 1423.0 & 1.01 & \multicolumn{1}{c}{--} & \multicolumn{1}{c}{--} & \multicolumn{1}{c}{--} & \multicolumn{1}{c}{--} & \multicolumn{1}{c}{--} & 0.072 & 0.002 & 77.0 & 1074.9 & 1.01 \\
$\alpha_{2}$ & -0.3 & -0.009 & 0.001 & 90.0 & 1621.7 & 1.01 & \multicolumn{1}{c}{--} & \multicolumn{1}{c}{--} & \multicolumn{1}{c}{--} & \multicolumn{1}{c}{--} & \multicolumn{1}{c}{--} & -0.286 & 0.008 & 2.0 & 1184.2 & 1.01 \\
$\alpha_{3}$ & 0.2 & -0.015 & 0.000 & 100.0 & 1783.5 & 1.01 & \multicolumn{1}{c}{--} & \multicolumn{1}{c}{--} & \multicolumn{1}{c}{--} & \multicolumn{1}{c}{--} & \multicolumn{1}{c}{--} & 0.335 & 0.005 & 1.0 & 1223.7 & 1.01 \\
$\alpha_{4}$ & -0.4 & -0.016 & 0.000 & 100.0 & 1543.2 & 1.01 & \multicolumn{1}{c}{--} & \multicolumn{1}{c}{--} & \multicolumn{1}{c}{--} & \multicolumn{1}{c}{--} & \multicolumn{1}{c}{--} & -0.198 & 0.007 & 2.0 & 1006.7 & 1.01 \\
\addlinespace[0.25em]
\multicolumn{17}{@{}l}{\textit{$\boldsymbol\Lambda$ parameters}} \\
\midrule
$\lambda_{1}$ & 0.8 & -0.046 & 0.016 & 95.0 & 1565.0 & 1.00 & 0.890 & 1.174 & 0.0 & 1088.7 & 1.00 & 0.162 & 0.057 & 89.0 & 1124.5 & 1.01 \\
$\lambda_{2}$ & 1.2 & -0.061 & 0.056 & 85.0 & 1606.0 & 1.02 & 0.210 & 0.903 & 83.3 & 615.3 & 1.01 & -0.089 & 0.320 & 90.0 & 518.6 & 1.02 \\
$\lambda_{3}$ & 1.1 & -0.083 & 0.085 & 95.0 & 1750.2 & 1.01 & 0.145 & 0.385 & 100.0 & 675.2 & 1.01 & -0.167 & 0.590 & 74.0 & 618.2 & 1.02 \\
$\lambda_{4}$ & 0.9 & -0.024 & 0.072 & 100.0 & 1872.8 & 1.01 & 0.605 & 3.585 & 0.0 & 824.0 & 1.01 & 0.389 & 1.596 & 19.0 & 739.9 & 1.01 \\
$\lambda_{5}$ & 1.4 & -0.027 & 0.054 & 90.0 & 1813.2 & 1.01 & 0.185 & 1.084 & 100.0 & 778.8 & 1.01 & 0.031 & 0.257 & 100.0 & 649.2 & 1.02 \\
$\lambda_{6}$ & 1 & -0.015 & 0.036 & 85.0 & 1625.9 & 1.01 & 0.487 & 2.209 & 0.0 & 469.5 & 1.01 & 0.353 & 1.210 & 5.0 & 436.4 & 1.02 \\
$\lambda_{7}$ & 0.7 & -0.038 & 0.016 & 90.0 & 1269.5 & 1.01 & 0.269 & 0.230 & 0.0 & 1019.7 & 1.00 & 0.038 & 0.017 & 93.0 & 829.1 & 1.02 \\
$\lambda_{8}$ & 1.1 & 0.024 & 0.074 & 90.0 & 1502.0 & 1.02 & 0.253 & 0.428 & 33.3 & 438.8 & 1.01 & 0.087 & 0.092 & 94.0 & 312.2 & 1.03 \\
$\lambda_{9}$ & 0.9 & 0.031 & 0.012 & 85.0 & 1167.2 & 1.01 & 0.263 & 0.162 & 0.0 & 893.3 & 1.01 & 0.115 & 0.041 & 69.0 & 718.7 & 1.02 \\
$\lambda_{10}$ & 1.2 & -0.003 & 0.058 & 90.0 & 1382.0 & 1.01 & 0.175 & 0.922 & 83.3 & 307.7 & 1.01 & -0.010 & 0.147 & 97.0 & 258.4 & 1.03 \\
$\lambda_{11}$ & 0.8 & -0.010 & 0.014 & 100.0 & 1809.2 & 1.01 & 0.386 & 1.160 & 0.0 & 594.7 & 1.01 & 0.209 & 0.356 & 20.0 & 492.2 & 1.02 \\
$\lambda_{12}$ & 1 & 0.009 & 0.005 & 95.0 & 1442.6 & 1.00 & 0.532 & 0.562 & 0.0 & 1087.0 & 1.01 & 0.232 & 0.115 & 8.0 & 995.8 & 1.01 \\
\bottomrule
\end{tabular*}
\begin{tablenotes}
\footnotesize
\item RB, relative bias; MSE, mean squared error; CP, coverage of the 95\% credible interval (\%); ESS, effective sample size; $\hat R$, Gelman--Rubin convergence diagnostic. Metric formulas are given in Supplementary Table~\ref{tab:s1_comparison_1}. A dash denotes an excluded parameter, undefined relative bias when the true value is zero, or an unavailable diagnostic.
\end{tablenotes}
\end{threeparttable}
\end{table}

\begin{table}[!htbp]
\ContinuedFloat
\centering
\caption{Full parameter-level simulation results: S3 (continued).}
\label{tab:s3_comparison_2}

\footnotesize

\sisetup{
  mode = math
}

\setlength{\tabcolsep}{0.35pt}
\renewcommand{\arraystretch}{0.72}

\begin{threeparttable}
\begin{tabular*}{\linewidth}{@{\extracolsep{\fill}}l S[table-format=-1.3] *{3}{S[table-format=-2.3] S[table-format=2.3] S[table-format=3.1] S[table-format=4.1] S[table-format=1.2]}}
\toprule
Parameter & {True} & \multicolumn{5}{c}{CLOUD} & \multicolumn{5}{c}{StationaryOU} & \multicolumn{5}{c}{DiagOU} \\
\cmidrule(lr){3-7}\cmidrule(lr){8-12}\cmidrule(l){13-17}
 &  & {RB} & {MSE} & {CP} & {ESS} & {$\hat{R}$} & {RB} & {MSE} & {CP} & {ESS} & {$\hat{R}$} & {RB} & {MSE} & {CP} & {ESS} & {$\hat{R}$} \\
\midrule
\multicolumn{17}{@{}l}{\textit{$\boldsymbol\beta$ parameters}} \\
\midrule
$\beta_{1,1}$ & 0.3 & -0.090 & 0.053 & 90.0 & 3766.3 & 1.00 & 0.066 & 0.011 & 100.0 & 3245.8 & 1.00 & -0.041 & 0.047 & 95.0 & 3018.8 & 1.00 \\
$\beta_{1,2}$ & 0.5 & -0.104 & 0.016 & 90.0 & 3684.8 & 1.00 & 0.086 & 0.033 & 83.3 & 2658.0 & 1.00 & -0.029 & 0.013 & 96.0 & 2804.8 & 1.00 \\
$\beta_{2,1}$ & 0.1 & 0.628 & 0.106 & 90.0 & 2621.4 & 1.00 & -1.307 & 0.094 & 100.0 & 2201.3 & 1.00 & 0.244 & 0.118 & 93.0 & 2022.1 & 1.00 \\
$\beta_{2,2}$ & 0.2 & -0.165 & 0.011 & 100.0 & 2836.1 & 1.00 & 0.497 & 0.039 & 83.3 & 2126.5 & 1.00 & -0.177 & 0.022 & 95.0 & 2177.6 & 1.00 \\
$\beta_{3,1}$ & -0.1 & -1.411 & 0.053 & 100.0 & 2744.3 & 1.00 & -0.605 & 0.073 & 100.0 & 2355.7 & 1.00 & -0.621 & 0.081 & 95.0 & 2005.6 & 1.00 \\
$\beta_{3,2}$ & 0.2 & -0.145 & 0.015 & 100.0 & 3011.4 & 1.00 & 0.212 & 0.028 & 100.0 & 2341.5 & 1.00 & -0.169 & 0.022 & 95.0 & 2256.7 & 1.01 \\
$\beta_{4,1}$ & 0.2 & -0.290 & 0.051 & 95.0 & 2661.0 & 1.00 & -0.297 & 0.077 & 83.3 & 1915.2 & 1.00 & -0.161 & 0.069 & 95.0 & 1770.4 & 1.01 \\
$\beta_{4,2}$ & -0.1 & 0.230 & 0.015 & 95.0 & 2910.9 & 1.00 & 0.060 & 0.004 & 100.0 & 2144.8 & 1.00 & 0.216 & 0.017 & 95.0 & 1890.6 & 1.01 \\
$\beta_{5,1}$ & 0.1 & 0.620 & 0.087 & 95.0 & 2498.6 & 1.00 & -0.868 & 0.141 & 100.0 & 1800.2 & 1.00 & -0.152 & 0.087 & 93.0 & 1645.8 & 1.00 \\
$\beta_{5,2}$ & 0.3 & -0.426 & 0.045 & 90.0 & 2633.5 & 1.00 & -0.106 & 0.012 & 100.0 & 1897.8 & 1.00 & -0.406 & 0.042 & 88.0 & 1670.2 & 1.01 \\
$\beta_{6,1}$ & -0.1 & -0.638 & 0.029 & 100.0 & 2631.2 & 1.00 & 1.208 & 0.065 & 100.0 & 1798.0 & 1.00 & 0.335 & 0.043 & 97.0 & 1590.1 & 1.00 \\
$\beta_{6,2}$ & -0.2 & 0.017 & 0.007 & 100.0 & 2695.5 & 1.00 & 0.152 & 0.008 & 100.0 & 1849.0 & 1.00 & 0.128 & 0.013 & 91.0 & 1587.4 & 1.01 \\
$\beta_{7,1}$ & -0.2 & 0.220 & 0.023 & 95.0 & 3284.8 & 1.00 & -0.314 & 0.008 & 100.0 & 2581.0 & 1.00 & 0.022 & 0.025 & 92.0 & 2179.0 & 1.01 \\
$\beta_{7,2}$ & -0.1 & 0.025 & 0.005 & 90.0 & 3297.2 & 1.00 & 0.025 & 0.004 & 100.0 & 2602.0 & 1.00 & -0.020 & 0.005 & 94.0 & 2325.9 & 1.00 \\
$\beta_{8,1}$ & 0.4 & -0.053 & 0.018 & 100.0 & 2640.6 & 1.00 & 0.122 & 0.031 & 83.3 & 1921.2 & 1.00 & 0.041 & 0.027 & 94.0 & 1573.6 & 1.01 \\
$\beta_{8,2}$ & 0.1 & 0.219 & 0.006 & 95.0 & 2789.2 & 1.00 & 0.283 & 0.003 & 100.0 & 2264.2 & 1.00 & -0.003 & 0.006 & 94.0 & 1727.8 & 1.01 \\
$\beta_{9,1}$ & 0.2 & 0.004 & 0.011 & 95.0 & 3190.0 & 1.00 & -0.080 & 0.013 & 100.0 & 2651.7 & 1.00 & -0.070 & 0.014 & 97.0 & 2041.2 & 1.00 \\
$\beta_{9,2}$ & -0.4 & -0.036 & 0.004 & 90.0 & 3046.9 & 1.00 & -0.125 & 0.007 & 83.3 & 2557.2 & 1.00 & 0.028 & 0.005 & 90.0 & 1940.8 & 1.00 \\
$\beta_{10,1}$ & -0.3 & -0.130 & 0.030 & 100.0 & 2021.6 & 1.00 & 0.308 & 0.089 & 83.3 & 1364.7 & 1.00 & -0.056 & 0.047 & 97.0 & 1089.9 & 1.01 \\
$\beta_{10,2}$ & 0.2 & -0.231 & 0.015 & 95.0 & 2016.4 & 1.00 & -0.795 & 0.049 & 50.0 & 1533.3 & 1.00 & -0.312 & 0.021 & 92.0 & 1100.9 & 1.01 \\
$\beta_{11,1}$ & 0.1 & 0.322 & 0.022 & 95.0 & 2395.0 & 1.00 & -1.493 & 0.043 & 83.3 & 1777.2 & 1.00 & 0.056 & 0.028 & 97.0 & 1323.7 & 1.01 \\
$\beta_{11,2}$ & -0.1 & 0.331 & 0.012 & 90.0 & 2432.8 & 1.00 & 0.450 & 0.008 & 83.3 & 1720.3 & 1.00 & 0.223 & 0.009 & 92.0 & 1353.4 & 1.01 \\
$\beta_{12,1}$ & 0.2 & 0.014 & 0.007 & 100.0 & 3129.2 & 1.00 & 0.004 & 0.002 & 100.0 & 2167.7 & 1.00 & 0.004 & 0.012 & 97.0 & 1845.0 & 1.00 \\
$\beta_{12,2}$ & 0.3 & -0.062 & 0.003 & 95.0 & 3101.8 & 1.00 & -0.063 & 0.012 & 66.7 & 2074.5 & 1.00 & -0.056 & 0.004 & 94.0 & 1822.1 & 1.00 \\
\addlinespace[0.25em]
\multicolumn{17}{@{}l}{\textit{$\boldsymbol{\sigma}_{\text{bk}}$ parameters}} \\
\midrule
$\sigma_{\text{bk},1}$ & 3.7 & -0.071 & 0.129 & 90.0 & 1113.0 & 1.01 & -0.025 & 0.032 & 100.0 & 962.8 & 1.00 & -0.063 & 0.129 & 88.0 & 997.3 & 1.01 \\
$\sigma_{\text{bk},2}$ & 3.4 & -0.189 & 0.525 & 85.0 & 1608.3 & 1.01 & -0.169 & 0.503 & 66.7 & 619.7 & 1.01 & -0.260 & 0.890 & 52.0 & 516.3 & 1.02 \\
$\sigma_{\text{bk},3}$ & 4.8 & -0.285 & 1.930 & 85.0 & 1736.1 & 1.01 & -0.256 & 1.532 & 0.0 & 725.2 & 1.01 & -0.323 & 2.470 & 2.0 & 653.1 & 1.02 \\
$\sigma_{\text{bk},4}$ & 3.1 & -0.033 & 0.082 & 100.0 & 1831.4 & 1.01 & -0.090 & 0.133 & 83.3 & 760.7 & 1.01 & -0.070 & 0.127 & 98.0 & 694.5 & 1.01 \\
$\sigma_{\text{bk},5}$ & 4.0 & -0.248 & 1.028 & 85.0 & 1778.4 & 1.01 & -0.299 & 1.449 & 0.0 & 691.2 & 1.01 & -0.296 & 1.481 & 15.0 & 609.6 & 1.02 \\
$\sigma_{\text{bk},6}$ & 6.1 & -0.178 & 1.242 & 85.0 & 1702.6 & 1.01 & -0.188 & 1.355 & 0.0 & 642.0 & 1.01 & -0.164 & 1.065 & 23.0 & 525.9 & 1.02 \\
$\sigma_{\text{bk},7}$ & 5.1 & -0.060 & 0.154 & 90.0 & 1918.4 & 1.01 & -0.047 & 0.206 & 66.7 & 848.8 & 1.00 & -0.074 & 0.200 & 67.0 & 752.7 & 1.02 \\
$\sigma_{\text{bk},8}$ & 1.7 & 0.018 & 0.046 & 85.0 & 1576.6 & 1.01 & 0.069 & 0.066 & 83.3 & 497.0 & 1.02 & -0.015 & 0.023 & 94.0 & 386.1 & 1.02 \\
$\sigma_{\text{bk},9}$ & 2.5 & 0.017 & 0.019 & 95.0 & 1973.2 & 1.01 & -0.020 & 0.019 & 100.0 & 817.5 & 1.00 & 0.000 & 0.022 & 94.0 & 753.7 & 1.01 \\
$\sigma_{\text{bk},10}$ & 3.3 & -0.119 & 0.230 & 90.0 & 1423.4 & 1.01 & -0.209 & 0.543 & 33.3 & 336.0 & 1.02 & -0.120 & 0.228 & 77.0 & 287.0 & 1.03 \\
$\sigma_{\text{bk},11}$ & 4.2 & -0.032 & 0.041 & 100.0 & 1743.8 & 1.01 & -0.074 & 0.112 & 100.0 & 683.5 & 1.01 & -0.018 & 0.057 & 97.0 & 540.1 & 1.02 \\
$\sigma_{\text{bk},12}$ & 2.8 & 0.003 & 0.020 & 100.0 & 1967.4 & 1.01 & -0.016 & 0.007 & 100.0 & 912.2 & 1.01 & 0.000 & 0.020 & 95.0 & 845.8 & 1.01 \\
\bottomrule
\end{tabular*}
\begin{tablenotes}
\footnotesize
\item RB, relative bias; MSE, mean squared error; CP, coverage of the 95\% credible interval (\%); ESS, effective sample size; $\hat R$, Gelman--Rubin convergence diagnostic. Metric formulas are given in Supplementary Table~\ref{tab:s1_comparison_1}. A dash denotes an excluded parameter, undefined relative bias when the true value is zero, or an unavailable diagnostic.
\end{tablenotes}
\end{threeparttable}
\end{table}

\begin{table}[!htbp]
\ContinuedFloat
\centering
\caption{Full parameter-level simulation results: S3 (continued).}
\label{tab:s3_comparison_3}

\footnotesize

\sisetup{
  mode = math
}

\setlength{\tabcolsep}{0.35pt}
\renewcommand{\arraystretch}{0.72}

\begin{threeparttable}
\begin{tabular*}{\linewidth}{@{\extracolsep{\fill}}l S[table-format=-1.3] *{3}{S[table-format=-2.3] S[table-format=2.3] S[table-format=3.1] S[table-format=4.1] S[table-format=1.2]}}
\toprule
Parameter & {True} & \multicolumn{5}{c}{CLOUD} & \multicolumn{5}{c}{StationaryOU} & \multicolumn{5}{c}{DiagOU} \\
\cmidrule(lr){3-7}\cmidrule(lr){8-12}\cmidrule(l){13-17}
 &  & {RB} & {MSE} & {CP} & {ESS} & {$\hat{R}$} & {RB} & {MSE} & {CP} & {ESS} & {$\hat{R}$} & {RB} & {MSE} & {CP} & {ESS} & {$\hat{R}$} \\
\midrule
\multicolumn{17}{@{}l}{\textit{$\boldsymbol\theta$ parameters}} \\
\midrule
$\theta_{1}$ & 2.3 & -0.072 & 0.058 & 90.0 & 1619.8 & 1.00 & 0.688 & 2.549 & 0.0 & 1085.7 & 1.00 & -0.018 & 0.055 & 96.0 & 1296.2 & 1.01 \\
$\theta_{2}$ & 1.5 & -0.242 & 0.211 & 95.0 & 1504.8 & 1.00 & 2.477 & 13.975 & 0.0 & 619.7 & 1.01 & -0.048 & 0.097 & 93.0 & 1036.8 & 1.01 \\
$\theta_{3}$ & 0.8 & -0.351 & 0.157 & 90.0 & 1959.2 & 1.00 & 4.465 & 12.784 & 0.0 & 781.8 & 1.01 & 0.026 & 0.087 & 94.0 & 1393.6 & 1.01 \\
$\theta_{4}$ & 1.2 & -0.026 & 0.089 & 90.0 & 1623.4 & 1.01 & -1.149 & 1.944 & 0.0 & 1012.2 & 1.01 & -0.286 & 0.186 & 77.0 & 1158.2 & 1.01 \\
$\theta_{5}$ & 2.1 & -0.248 & 0.374 & 90.0 & 1421.0 & 1.01 & -1.042 & 4.909 & 0.0 & 907.3 & 1.01 & -0.453 & 0.988 & 23.0 & 1006.7 & 1.01 \\
$\theta_{6,1}$ & -4.0 & -0.135 & 0.436 & 90.0 & 1788.7 & 1.01 & -0.413 & 2.793 & 0.0 & 768.8 & 1.01 & -0.159 & 0.570 & 59.0 & 658.5 & 1.02 \\
$\theta_{6,2}$ & -1.0 & -0.134 & 0.121 & 85.0 & 1891.8 & 1.01 & -1.300 & 1.758 & 0.0 & 875.7 & 1.01 & -0.282 & 0.203 & 79.0 & 947.4 & 1.01 \\
$\theta_{6,3}$ & 2.7 & -0.108 & 0.213 & 85.0 & 1861.1 & 1.01 & 0.333 & 0.937 & 33.3 & 727.5 & 1.01 & -0.036 & 0.140 & 92.0 & 855.6 & 1.01 \\
$\theta_{7,1}$ & -5.5 & -0.005 & 0.069 & 100.0 & 1741.0 & 1.01 & 0.184 & 1.152 & 16.7 & 721.2 & 1.01 & -0.029 & 0.110 & 91.0 & 627.0 & 1.02 \\
$\theta_{7,2}$ & -2.5 & 0.028 & 0.048 & 100.0 & 1638.6 & 1.02 & 0.433 & 1.207 & 0.0 & 639.8 & 1.01 & -0.017 & 0.052 & 95.0 & 564.4 & 1.02 \\
$\theta_{7,3}$ & 2.6 & -0.095 & 0.102 & 90.0 & 1723.4 & 1.02 & -0.452 & 1.419 & 0.0 & 636.3 & 1.01 & -0.053 & 0.086 & 92.0 & 616.1 & 1.02 \\
$\theta_{8,1}$ & -4.5 & 0.014 & 0.125 & 85.0 & 1736.5 & 1.01 & 0.350 & 2.640 & 0.0 & 509.3 & 1.01 & 0.017 & 0.091 & 97.0 & 431.5 & 1.02 \\
$\theta_{8,2}$ & -1.5 & 0.037 & 0.019 & 100.0 & 1635.1 & 1.01 & 1.046 & 2.524 & 0.0 & 671.5 & 1.01 & 0.020 & 0.028 & 97.0 & 987.5 & 1.02 \\
$\theta_{8,3}$ & 2.0 & 0.008 & 0.054 & 90.0 & 1064.3 & 1.01 & -0.755 & 2.310 & 0.0 & 1711.5 & 1.01 & 0.005 & 0.037 & 98.0 & 590.7 & 1.02 \\
$\theta_{9,1}$ & -3.0 & 0.012 & 0.029 & 90.0 & 1312.4 & 1.00 & 0.314 & 0.893 & 0.0 & 983.0 & 1.01 & 0.024 & 0.031 & 97.0 & 1021.6 & 1.01 \\
$\theta_{9,2}$ & 0.0 & \multicolumn{1}{c}{--} & 0.018 & 100.0 & 1484.2 & 1.00 & \multicolumn{1}{c}{--} & 0.868 & 0.0 & 1217.7 & 1.01 & \multicolumn{1}{c}{--} & 0.021 & 96.0 & 1281.5 & 1.01 \\
$\theta_{9,3}$ & 3.0 & 0.003 & 0.024 & 100.0 & 1416.5 & 1.01 & -0.324 & 0.985 & 0.0 & 1260.5 & 1.00 & -0.008 & 0.039 & 95.0 & 983.3 & 1.01 \\
$\theta_{10,1}$ & -6.0 & -0.095 & 0.608 & 90.0 & 1466.9 & 1.01 & -0.667 & 16.068 & 0.0 & 639.7 & 1.01 & -0.192 & 1.494 & 45.0 & 328.8 & 1.02 \\
$\theta_{10,2}$ & -2.0 & -0.083 & 0.079 & 90.0 & 1898.4 & 1.01 & -1.678 & 11.301 & 0.0 & 832.2 & 1.00 & -0.283 & 0.391 & 52.0 & 743.0 & 1.01 \\
$\theta_{10,3}$ & 1.5 & -0.079 & 0.037 & 100.0 & 1466.8 & 1.00 & 1.854 & 7.882 & 0.0 & 418.2 & 1.01 & 0.059 & 0.057 & 98.0 & 842.7 & 1.01 \\
$\theta_{11,1}$ & -5.0 & -0.017 & 0.096 & 100.0 & 1902.4 & 1.01 & -0.428 & 4.622 & 0.0 & 737.0 & 1.00 & -0.037 & 0.158 & 92.0 & 605.8 & 1.02 \\
$\theta_{11,2}$ & -1.0 & -0.049 & 0.061 & 90.0 & 1072.2 & 1.01 & -1.832 & 3.385 & 0.0 & 1024.2 & 1.00 & -0.153 & 0.075 & 91.0 & 935.2 & 1.01 \\
$\theta_{11,3}$ & 2.5 & -0.006 & 0.079 & 90.0 & 1071.8 & 1.01 & 0.677 & 2.903 & 0.0 & 856.8 & 1.01 & 0.066 & 0.091 & 93.0 & 912.9 & 1.01 \\
$\theta_{12,1}$ & -4.0 & -0.002 & 0.022 & 100.0 & 1264.5 & 1.00 & -0.278 & 1.261 & 0.0 & 1043.2 & 1.01 & -0.029 & 0.047 & 95.0 & 986.1 & 1.01 \\
$\theta_{12,2}$ & -0.5 & 0.061 & 0.010 & 100.0 & 1249.0 & 1.00 & -2.204 & 1.236 & 0.0 & 1130.7 & 1.01 & -0.173 & 0.031 & 91.0 & 1068.6 & 1.01 \\
$\theta_{12,3}$ & 2.0 & 0.010 & 0.017 & 100.0 & 1325.0 & 1.00 & 0.523 & 1.111 & 0.0 & 1189.0 & 1.01 & 0.059 & 0.045 & 90.0 & 1147.4 & 1.01 \\
\addlinespace[0.25em]
\multicolumn{17}{@{}l}{\textit{Other parameters}} \\
\midrule
$\rho_{1}$ & -0.02 & 0.016 & 0.002 & 100.0 & 1497.2 & 1.01 & 5.499 & 0.163 & 0.0 & 402.2 & 1.01 & -1.000 & 0.0004 & 0.0 & \multicolumn{1}{c}{--} & \multicolumn{1}{c}{--} \\
$\rho_{2}$ & 0.04 & -0.019 & 0.003 & 95.0 & 1373.8 & 1.02 & 1.584 & 0.075 & 0.0 & 535.8 & 1.01 & -1.000 & 0.0016 & 0.0 & \multicolumn{1}{c}{--} & \multicolumn{1}{c}{--} \\
$\rho_{3}$ & 0.06 & 0.062 & 0.001 & 100.0 & 1294.6 & 1.02 & -1.974 & 0.287 & 0.0 & 419.3 & 1.01 & -1.000 & 0.0036 & 0.0 & \multicolumn{1}{c}{--} & \multicolumn{1}{c}{--} \\
$\rho_{4}$ & -0.048 & -0.009 & 0.002 & 90.0 & 1663.9 & 1.02 & -0.187 & 0.003 & 66.7 & 953.0 & 1.00 & -1.000 & 0.0023 & 0.0 & \multicolumn{1}{c}{--} & \multicolumn{1}{c}{--} \\
$\rho_{5}$ & -0.06 & -0.017 & 0.001 & 100.0 & 1649.4 & 1.02 & -0.724 & 0.002 & 100.0 & 747.3 & 1.00 & -1.000 & 0.0036 & 0.0 & \multicolumn{1}{c}{--} & \multicolumn{1}{c}{--} \\
$\rho_{6}$ & 0.0 & \multicolumn{1}{c}{--} & 0.001 & 95.0 & 1018.6 & 1.01 & \multicolumn{1}{c}{--} & 0.164 & 0.0 & 611.5 & 1.02 & \multicolumn{1}{c}{--} & 0.0 & 100.0 & \multicolumn{1}{c}{--} & \multicolumn{1}{c}{--} \\
\bottomrule
\end{tabular*}
\begin{tablenotes}
\footnotesize
\item RB, relative bias; MSE, mean squared error; CP, coverage of the 95\% credible interval (\%); ESS, effective sample size; $\hat R$, Gelman--Rubin convergence diagnostic. Metric formulas are given in Supplementary Table~\ref{tab:s1_comparison_1}. A dash denotes an excluded parameter, undefined relative bias when the true value is zero, or an unavailable diagnostic.
\end{tablenotes}
\end{threeparttable}
\end{table}
\end{landscape}

\begin{landscape}
\begin{table}[p]
\centering
\caption{Full parameter-level simulation results: S4 (CLOUD, StationaryOU, and DiagOU).}
\label{tab:s4_comparison_1}

\footnotesize

\sisetup{
  mode = math
}

\setlength{\tabcolsep}{0.35pt}
\renewcommand{\arraystretch}{0.72}

\begin{threeparttable}
\begin{tabular*}{\linewidth}{@{\extracolsep{\fill}}l S[table-format=-1.3] *{3}{S[table-format=-2.3] S[table-format=3.3] S[table-format=3.1] S[table-format=4.1] S[table-format=1.2]}}
\toprule
Parameter & {True} & \multicolumn{5}{c}{CLOUD} & \multicolumn{5}{c}{StationaryOU} & \multicolumn{5}{c}{DiagOU} \\
\cmidrule(lr){3-7}\cmidrule(lr){8-12}\cmidrule(l){13-17}
 &  & {RB} & {MSE} & {CP} & {ESS} & {$\hat{R}$} & {RB} & {MSE} & {CP} & {ESS} & {$\hat{R}$} & {RB} & {MSE} & {CP} & {ESS} & {$\hat{R}$} \\
\midrule
\multicolumn{17}{@{}l}{\textit{$\boldsymbol\Phi$ parameters}} \\
\midrule
$\Phi_{1,1}$ & 0.4 & 0.063 & 0.003 & 100.0 & 1303.4 & 1.01 & \multicolumn{1}{c}{--} & \multicolumn{1}{c}{--} & \multicolumn{1}{c}{--} & \multicolumn{1}{c}{--} & \multicolumn{1}{c}{--} & 0.025 & 0.003 & 94.9 & 2262.8 & 1.00 \\
$\Phi_{1,2}$ & -0.2 & 0.086 & 0.002 & 100.0 & 1239.2 & 1.01 & \multicolumn{1}{c}{--} & \multicolumn{1}{c}{--} & \multicolumn{1}{c}{--} & \multicolumn{1}{c}{--} & \multicolumn{1}{c}{--} & 0.037 & 0.001 & 94.9 & 2255.3 & 1.00 \\
$\Phi_{2,1}$ & -0.3 & 0.181 & 0.005 & 85.7 & 1175.2 & 1.01 & \multicolumn{1}{c}{--} & \multicolumn{1}{c}{--} & \multicolumn{1}{c}{--} & \multicolumn{1}{c}{--} & \multicolumn{1}{c}{--} & 0.016 & 0.001 & 94.9 & 1893.3 & 1.00 \\
$\Phi_{2,2}$ & 0.5 & 0.084 & 0.003 & 85.7 & 1773.1 & 1.01 & \multicolumn{1}{c}{--} & \multicolumn{1}{c}{--} & \multicolumn{1}{c}{--} & \multicolumn{1}{c}{--} & \multicolumn{1}{c}{--} & 0.007 & 0.000 & 94.9 & 1408.8 & 1.01 \\
$\Phi_{3,1}$ & 0.2 & -0.039 & 0.001 & 100.0 & 1405.8 & 1.00 & \multicolumn{1}{c}{--} & \multicolumn{1}{c}{--} & \multicolumn{1}{c}{--} & \multicolumn{1}{c}{--} & \multicolumn{1}{c}{--} & 0.013 & 0.001 & 94.9 & 1893.8 & 1.00 \\
$\Phi_{3,2}$ & 0.1 & -0.036 & 0.021 & 90.5 & 1581.5 & 1.00 & \multicolumn{1}{c}{--} & \multicolumn{1}{c}{--} & \multicolumn{1}{c}{--} & \multicolumn{1}{c}{--} & \multicolumn{1}{c}{--} & -0.013 & 0.000 & 100.0 & 2145.2 & 1.00 \\
$\Phi_{4,1}$ & -0.1 & 0.156 & 0.001 & 100.0 & 1285.4 & 1.00 & \multicolumn{1}{c}{--} & \multicolumn{1}{c}{--} & \multicolumn{1}{c}{--} & \multicolumn{1}{c}{--} & \multicolumn{1}{c}{--} & 0.095 & 0.001 & 94.9 & 2041.6 & 1.00 \\
$\Phi_{4,2}$ & -0.3 & 0.007 & 0.023 & 95.2 & 1085.9 & 1.01 & \multicolumn{1}{c}{--} & \multicolumn{1}{c}{--} & \multicolumn{1}{c}{--} & \multicolumn{1}{c}{--} & \multicolumn{1}{c}{--} & 0.020 & 0.000 & 91.9 & 1376.4 & 1.01 \\
\addlinespace[0.25em]
\multicolumn{17}{@{}l}{\textit{$\boldsymbol\Gamma$ parameters}} \\
\midrule
$\Gamma_{1,1}$ & 0.5 & 0.076 & 0.074 & 100.0 & 2296.3 & 1.01 & 0.341 & 0.062 & 100.0 & 402.0 & 1.01 & 4.635 & 8.780 & 98.0 & 206.9 & 1.07 \\
$\Gamma_{1,2}$ & 2.07 & -0.006 & 0.093 & 100.0 & 2240.9 & 1.01 & -0.668 & 1.923 & 0.0 & 351.0 & 1.01 & -1.000 & 4.285 & 0.0 & \multicolumn{1}{c}{--} & \multicolumn{1}{c}{--} \\
$\Gamma_{1,3}$ & -0.08 & 0.013 & 0.056 & 100.0 & 2228.3 & 1.01 & 12.458 & 1.030 & 0.0 & 307.6 & 1.03 & -1.000 & 0.006 & 0.0 & \multicolumn{1}{c}{--} & \multicolumn{1}{c}{--} \\
$\Gamma_{1,4}$ & 0.0 & \multicolumn{1}{c}{--} & 0.077 & 100.0 & 2186.4 & 1.01 & \multicolumn{1}{c}{--} & 0.085 & 90.9 & 246.5 & 1.03 & \multicolumn{1}{c}{--} & 0.000 & 100.0 & \multicolumn{1}{c}{--} & \multicolumn{1}{c}{--} \\
$\Gamma_{2,1}$ & -1.93 & 0.034 & 0.082 & 95.2 & 2236.8 & 1.02 & -0.701 & 1.891 & 0.0 & 257.5 & 1.01 & -1.000 & 3.725 & 0.0 & \multicolumn{1}{c}{--} & \multicolumn{1}{c}{--} \\
$\Gamma_{2,2}$ & 0.3 & 0.043 & 0.076 & 100.0 & 2224.9 & 1.02 & -1.289 & 0.433 & 0.0 & 235.5 & 1.01 & 0.144 & 0.502 & 88.9 & 236.4 & 1.08 \\
$\Gamma_{2,3}$ & 0.16 & -0.015 & 0.042 & 100.0 & 2261.6 & 1.01 & -0.575 & 0.055 & 100.0 & 293.6 & 1.09 & -1.000 & 0.026 & 0.0 & \multicolumn{1}{c}{--} & \multicolumn{1}{c}{--} \\
$\Gamma_{2,4}$ & 0.08 & 0.021 & 0.084 & 100.0 & 2349.9 & 1.02 & -3.199 & 0.094 & 81.8 & 326.7 & 1.05 & -1.000 & 0.006 & 0.0 & \multicolumn{1}{c}{--} & \multicolumn{1}{c}{--} \\
$\Gamma_{3,1}$ & -0.12 & -0.013 & 0.094 & 100.0 & 2190.4 & 1.01 & -13.191 & 2.853 & 9.1 & 260.3 & 1.03 & -1.000 & 0.014 & 0.0 & \multicolumn{1}{c}{--} & \multicolumn{1}{c}{--} \\
$\Gamma_{3,2}$ & 0.25 & -0.015 & 0.051 & 100.0 & 2231.2 & 1.02 & 5.596 & 2.116 & 0.0 & 246.5 & 1.02 & -1.000 & 0.063 & 0.0 & \multicolumn{1}{c}{--} & \multicolumn{1}{c}{--} \\
$\Gamma_{3,3}$ & 0.4 & 0.038 & 0.048 & 100.0 & 1465.0 & 1.01 & 0.271 & 0.063 & 100.0 & 269.5 & 1.02 & 24.307 & 149.880 & 0.0 & 1269.7 & 1.01 \\
$\Gamma_{3,4}$ & -2.47 & 0.021 & 0.014 & 100.0 & 1452.6 & 1.02 & -0.215 & 0.326 & 50.0 & 244.5 & 1.01 & -1.000 & 6.101 & 0.0 & \multicolumn{1}{c}{--} & \multicolumn{1}{c}{--} \\
$\Gamma_{4,1}$ & 0.0 & \multicolumn{1}{c}{--} & 0.013 & 100.0 & 2145.7 & 1.02 & \multicolumn{1}{c}{--} & 0.848 & 0.0 & 346.1 & 1.02 & \multicolumn{1}{c}{--} & 0.000 & 100.0 & \multicolumn{1}{c}{--} & \multicolumn{1}{c}{--} \\
$\Gamma_{4,2}$ & 0.12 & -0.094 & 0.041 & 100.0 & 2246.5 & 1.02 & 2.402 & 0.137 & 54.5 & 312.5 & 1.02 & -1.000 & 0.014 & 0.0 & \multicolumn{1}{c}{--} & \multicolumn{1}{c}{--} \\
$\Gamma_{4,3}$ & 2.43 & 0.007 & 0.024 & 95.2 & 2441.7 & 1.01 & -0.340 & 0.683 & 0.0 & 335.5 & 1.02 & -1.000 & 5.905 & 0.0 & \multicolumn{1}{c}{--} & \multicolumn{1}{c}{--} \\
$\Gamma_{4,4}$ & 0.9 & 0.087 & 0.041 & 100.0 & 2462.0 & 1.01 & 0.196 & 0.068 & 100.0 & 340.5 & 1.01 & 18.863 & 301.534 & 1.0 & 447.3 & 1.05 \\
\addlinespace[0.25em]
\multicolumn{17}{@{}l}{\textit{$\boldsymbol\alpha$ parameters}} \\
\midrule
$\alpha_{1}$ & 0.5 & 0.123 & 0.009 & 100.0 & 1645.9 & 1.03 & \multicolumn{1}{c}{--} & \multicolumn{1}{c}{--} & \multicolumn{1}{c}{--} & \multicolumn{1}{c}{--} & \multicolumn{1}{c}{--} & 0.180 & 0.013 & 90.9 & 284.6 & 1.04 \\
$\alpha_{2}$ & -0.3 & 0.056 & 0.002 & 91.7 & 1285.3 & 1.01 & \multicolumn{1}{c}{--} & \multicolumn{1}{c}{--} & \multicolumn{1}{c}{--} & \multicolumn{1}{c}{--} & \multicolumn{1}{c}{--} & -0.189 & 0.004 & 42.4 & 572.4 & 1.02 \\
$\alpha_{3}$ & 0.2 & -0.022 & 0.001 & 100.0 & 1763.9 & 1.00 & \multicolumn{1}{c}{--} & \multicolumn{1}{c}{--} & \multicolumn{1}{c}{--} & \multicolumn{1}{c}{--} & \multicolumn{1}{c}{--} & 0.378 & 0.006 & 9.1 & 1013.5 & 1.01 \\
$\alpha_{4}$ & -0.4 & -0.051 & 0.001 & 100.0 & 1220.9 & 1.01 & \multicolumn{1}{c}{--} & \multicolumn{1}{c}{--} & \multicolumn{1}{c}{--} & \multicolumn{1}{c}{--} & \multicolumn{1}{c}{--} & -0.140 & 0.004 & 38.4 & 724.2 & 1.02 \\
\addlinespace[0.25em]
\multicolumn{17}{@{}l}{\textit{$\boldsymbol\Lambda$ parameters}} \\
\midrule
$\lambda_{1}$ & 0.8 & -0.009 & 0.005 & 100.0 & 1785.3 & 1.01 & 0.984 & 0.668 & 0.0 & 898.0 & 1.00 & 0.109 & 0.019 & 97.0 & 531.6 & 1.02 \\
$\lambda_{2}$ & 1.2 & -0.134 & 0.030 & 90.5 & 1766.3 & 1.01 & 0.585 & 0.497 & 0.0 & 779.0 & 1.00 & -0.017 & 0.021 & 99.0 & 429.5 & 1.03 \\
$\lambda_{3}$ & 1.1 & -0.126 & 0.024 & 95.2 & 1840.8 & 1.01 & 0.810 & 0.843 & 0.0 & 667.0 & 1.01 & 0.011 & 0.016 & 97.0 & 458.5 & 1.02 \\
$\lambda_{4}$ & 0.9 & -0.032 & 0.007 & 95.2 & 1764.8 & 1.01 & 1.036 & 0.887 & 0.0 & 817.5 & 1.00 & 0.524 & 0.257 & 18.2 & 583.8 & 1.02 \\
$\lambda_{5}$ & 1.4 & -0.200 & 0.081 & 89.3 & 1850.0 & 1.01 & 0.442 & 0.398 & 0.0 & 510.0 & 1.00 & 0.240 & 0.166 & 70.7 & 449.5 & 1.03 \\
$\lambda_{6}$ & 1.0 & -0.114 & 0.016 & 85.7 & 1575.1 & 1.01 & 0.922 & 0.850 & 0.0 & 558.5 & 1.01 & 0.477 & 0.258 & 17.2 & 365.3 & 1.02 \\
$\lambda_{7}$ & 0.7 & 0.066 & 0.010 & 85.7 & 1130.8 & 1.00 & 0.199 & 0.062 & 50.0 & 973.5 & 1.00 & 0.076 & 0.012 & 91.9 & 1051.2 & 1.01 \\
$\lambda_{8}$ & 1.1 & -0.011 & 0.005 & 100.0 & 1764.7 & 1.02 & 0.175 & 0.061 & 50.0 & 498.0 & 1.00 & 0.057 & 0.016 & 94.9 & 502.4 & 1.03 \\
$\lambda_{9}$ & 0.9 & 0.052 & 0.007 & 95.2 & 1900.5 & 1.01 & 0.268 & 0.058 & 0.0 & 624.5 & 1.00 & 0.111 & 0.019 & 85.9 & 635.6 & 1.01 \\
$\lambda_{10}$ & 1.2 & -0.031 & 0.006 & 95.2 & 1687.2 & 1.01 & 0.689 & 0.684 & 0.0 & 398.5 & 1.01 & 0.149 & 0.043 & 68.7 & 511.4 & 1.02 \\
$\lambda_{11}$ & 0.8 & 0.023 & 0.004 & 95.2 & 1044.2 & 1.01 & 0.705 & 0.324 & 0.0 & 765.0 & 1.01 & 0.195 & 0.031 & 53.5 & 907.0 & 1.01 \\
$\lambda_{12}$ & 1.0 & 0.000 & 0.004 & 95.2 & 1815.8 & 1.01 & 0.570 & 0.330 & 0.0 & 699.0 & 1.00 & 0.177 & 0.042 & 56.6 & 647.9 & 1.02 \\
\bottomrule
\end{tabular*}
\begin{tablenotes}
\footnotesize
\item RB, relative bias; MSE, mean squared error; CP, coverage of the 95\% credible interval (\%); ESS, effective sample size; $\hat R$, Gelman--Rubin convergence diagnostic. Metric formulas are given in Supplementary Table~\ref{tab:s1_comparison_1}. A dash denotes an excluded parameter, undefined relative bias when the true value is zero, or an unavailable diagnostic.
\end{tablenotes}
\end{threeparttable}
\end{table}
\clearpage
\end{landscape}

\begin{landscape}
\begin{table}[!htbp]
\ContinuedFloat
\centering
\caption{Full parameter-level simulation results: S4 (continued).}
\label{tab:s4_comparison_2}

\footnotesize

\sisetup{
  mode = math
}

\setlength{\tabcolsep}{0.35pt}
\renewcommand{\arraystretch}{0.72}

\begin{threeparttable}
\begin{tabular*}{\linewidth}{@{\extracolsep{\fill}}l S[table-format=-1.3] *{3}{S[table-format=-2.3] S[table-format=3.3] S[table-format=3.1] S[table-format=4.1] S[table-format=1.2]}}
\toprule
Parameter & {True} & \multicolumn{5}{c}{CLOUD} & \multicolumn{5}{c}{StationaryOU} & \multicolumn{5}{c}{DiagOU} \\
\cmidrule(lr){3-7}\cmidrule(lr){8-12}\cmidrule(l){13-17}
 &  & {RB} & {MSE} & {CP} & {ESS} & {$\hat{R}$} & {RB} & {MSE} & {CP} & {ESS} & {$\hat{R}$} & {RB} & {MSE} & {CP} & {ESS} & {$\hat{R}$} \\
\midrule
\multicolumn{17}{@{}l}{\textit{$\boldsymbol\beta$ parameters}} \\
\midrule
$\beta_{1,1}$ & 0.3 & -0.107 & 0.037 & 95.2 & 3427.9 & 1.00 & 0.075 & 0.002 & 100.0 & 4212.0 & 1.00 & -0.132 & 0.043 & 91.9 & 3736.3 & 1.00 \\
$\beta_{1,2}$ & 0.5 & 0.031 & 0.009 & 100.0 & 3111.1 & 1.00 & -0.111 & 0.010 & 100.0 & 3590.5 & 1.00 & -0.024 & 0.008 & 98.0 & 3352.5 & 1.00 \\
$\beta_{2,1}$ & 0.1 & -0.786 & 0.053 & 95.2 & 2990.0 & 1.00 & 2.520 & 0.083 & 50.0 & 3709.5 & 1.00 & -0.177 & 0.045 & 94.9 & 3420.7 & 1.00 \\
$\beta_{2,2}$ & 0.2 & 0.030 & 0.010 & 95.2 & 3408.0 & 1.00 & 0.055 & 0.006 & 100.0 & 3549.5 & 1.00 & -0.099 & 0.010 & 92.9 & 3702.5 & 1.01 \\
$\beta_{3,1}$ & -0.1 & 0.032 & 0.043 & 90.5 & 3091.4 & 1.00 & -0.330 & 0.044 & 100.0 & 2885.0 & 1.00 & -0.021 & 0.039 & 96.0 & 3439.2 & 1.00 \\
$\beta_{3,2}$ & 0.2 & -0.095 & 0.013 & 95.2 & 3299.1 & 1.00 & -0.115 & 0.041 & 50.0 & 3133.0 & 1.00 & -0.124 & 0.010 & 96.0 & 3796.5 & 1.00 \\
$\beta_{4,1}$ & 0.2 & -0.330 & 0.030 & 90.5 & 3227.6 & 1.00 & -0.560 & 0.016 & 100.0 & 3744.5 & 1.00 & -0.076 & 0.023 & 96.0 & 3744.5 & 1.00 \\
$\beta_{4,2}$ & -0.1 & -0.131 & 0.005 & 100.0 & 3703.6 & 1.00 & 1.135 & 0.013 & 100.0 & 3605.0 & 1.00 & 0.046 & 0.006 & 99.0 & 4088.3 & 1.00 \\
$\beta_{5,1}$ & 0.1 & -0.044 & 0.063 & 85.7 & 2971.6 & 1.00 & 1.580 & 0.043 & 100.0 & 3469.5 & 1.00 & -0.042 & 0.038 & 92.9 & 3543.3 & 1.00 \\
$\beta_{5,2}$ & 0.3 & -0.248 & 0.022 & 81.0 & 3124.8 & 1.00 & -0.418 & 0.016 & 100.0 & 3498.5 & 1.00 & -0.278 & 0.017 & 81.8 & 3573.8 & 1.00 \\
$\beta_{6,1}$ & -0.1 & -0.180 & 0.014 & 95.2 & 3245.7 & 1.00 & 1.820 & 0.034 & 100.0 & 3227.5 & 1.00 & 0.104 & 0.017 & 94.9 & 3710.6 & 1.00 \\
$\beta_{6,2}$ & -0.2 & -0.026 & 0.003 & 100.0 & 3153.4 & 1.00 & 0.182 & 0.004 & 100.0 & 3393.0 & 1.00 & 0.021 & 0.004 & 97.0 & 3669.9 & 1.00 \\
$\beta_{7,1}$ & -0.2 & 0.111 & 0.010 & 100.0 & 3476.0 & 1.00 & 0.480 & 0.009 & 100.0 & 3865.5 & 1.00 & 0.058 & 0.013 & 96.0 & 3656.2 & 1.00 \\
$\beta_{7,2}$ & -0.1 & 0.170 & 0.003 & 100.0 & 3541.0 & 1.00 & -0.270 & 0.008 & 100.0 & 3809.5 & 1.00 & 0.033 & 0.004 & 90.9 & 3885.8 & 1.00 \\
$\beta_{8,1}$ & 0.4 & -0.070 & 0.009 & 100.0 & 3196.8 & 1.00 & 0.215 & 0.008 & 100.0 & 3095.0 & 1.00 & -0.017 & 0.012 & 96.0 & 3166.7 & 1.00 \\
$\beta_{8,2}$ & 0.1 & -0.119 & 0.002 & 95.2 & 3390.6 & 1.00 & -0.010 & 0.000 & 100.0 & 3530.0 & 1.00 & -0.185 & 0.003 & 94.9 & 3434.1 & 1.00 \\
$\beta_{9,1}$ & 0.2 & -0.143 & 0.005 & 100.0 & 3147.8 & 1.00 & 0.010 & 0.000 & 100.0 & 3629.5 & 1.00 & -0.106 & 0.010 & 97.0 & 3285.9 & 1.00 \\
$\beta_{9,2}$ & -0.4 & 0.041 & 0.004 & 90.5 & 2961.3 & 1.00 & -0.039 & 0.001 & 100.0 & 3122.5 & 1.00 & 0.032 & 0.003 & 91.9 & 2848.9 & 1.00 \\
$\beta_{10,1}$ & -0.3 & -0.088 & 0.018 & 95.2 & 2930.4 & 1.00 & -0.380 & 0.021 & 100.0 & 3060.0 & 1.00 & -0.027 & 0.017 & 91.9 & 2870.8 & 1.00 \\
$\beta_{10,2}$ & 0.2 & -0.128 & 0.003 & 95.2 & 2989.2 & 1.00 & -0.250 & 0.004 & 100.0 & 2980.0 & 1.00 & -0.059 & 0.004 & 92.9 & 2849.4 & 1.00 \\
$\beta_{11,1}$ & 0.1 & 0.049 & 0.014 & 95.2 & 3429.0 & 1.00 & 0.270 & 0.001 & 100.0 & 3787.5 & 1.00 & 0.104 & 0.013 & 94.9 & 3416.6 & 1.00 \\
$\beta_{11,2}$ & -0.1 & 0.020 & 0.003 & 95.2 & 3389.5 & 1.00 & -0.395 & 0.003 & 100.0 & 3571.0 & 1.00 & 0.054 & 0.004 & 94.9 & 3525.5 & 1.00 \\
$\beta_{12,1}$ & 0.2 & -0.069 & 0.016 & 85.7 & 3004.4 & 1.00 & -0.275 & 0.039 & 50.0 & 3294.5 & 1.00 & 0.062 & 0.013 & 93.9 & 3124.6 & 1.00 \\
$\beta_{12,2}$ & 0.3 & -0.023 & 0.004 & 95.2 & 2921.0 & 1.00 & -0.277 & 0.008 & 50.0 & 3072.5 & 1.00 & -0.024 & 0.004 & 88.9 & 2811.7 & 1.00 \\
\addlinespace[0.25em]
\multicolumn{17}{@{}l}{\textit{$\boldsymbol{\sigma}_{\text{bk}}$ parameters}} \\
\midrule
$\sigma_{\text{bk},1}$ & 3.7 & -0.028 & 0.116 & 85.7 & 1040.9 & 1.01 & -0.057 & 0.048 & 100.0 & 1132.0 & 1.00 & -0.079 & 0.152 & 80.8 & 917.3 & 1.01 \\
$\sigma_{\text{bk},2}$ & 3.4 & -0.073 & 0.110 & 81.0 & 1958.4 & 1.01 & -0.110 & 0.142 & 100.0 & 975.0 & 1.00 & -0.135 & 0.268 & 69.7 & 577.3 & 1.02 \\
$\sigma_{\text{bk},3}$ & 4.8 & -0.134 & 0.481 & 86.9 & 1082.9 & 1.01 & -0.218 & 1.095 & 0.0 & 943.0 & 1.00 & -0.160 & 0.670 & 38.4 & 837.9 & 1.02 \\
$\sigma_{\text{bk},4}$ & 3.1 & -0.015 & 0.055 & 95.2 & 1991.0 & 1.00 & -0.138 & 0.184 & 50.0 & 999.0 & 1.00 & -0.084 & 0.109 & 77.8 & 892.2 & 1.01 \\
$\sigma_{\text{bk},5}$ & 4.0 & -0.137 & 0.367 & 88.1 & 1925.0 & 1.01 & -0.280 & 1.257 & 0.0 & 908.5 & 1.00 & -0.219 & 0.817 & 15.2 & 778.7 & 1.01 \\
$\sigma_{\text{bk},6}$ & 6.1 & -0.098 & 0.431 & 88.1 & 1708.2 & 1.01 & -0.171 & 1.124 & 0.0 & 762.5 & 1.00 & -0.150 & 0.894 & 10.1 & 792.4 & 1.01 \\
$\sigma_{\text{bk},7}$ & 5.1 & -0.043 & 0.091 & 90.5 & 1779.5 & 1.01 & -0.047 & 0.098 & 100.0 & 739.0 & 1.01 & -0.050 & 0.116 & 81.8 & 811.7 & 1.02 \\
$\sigma_{\text{bk},8}$ & 1.7 & -0.014 & 0.011 & 90.5 & 1844.7 & 1.01 & -0.140 & 0.058 & 0.0 & 682.5 & 1.01 & -0.010 & 0.014 & 96.0 & 746.6 & 1.02 \\
$\sigma_{\text{bk},9}$ & 2.5 & 0.012 & 0.011 & 100.0 & 1826.1 & 1.01 & 0.031 & 0.015 & 100.0 & 736.5 & 1.00 & -0.004 & 0.017 & 93.9 & 790.0 & 1.01 \\
$\sigma_{\text{bk},10}$ & 3.3 & -0.029 & 0.055 & 85.7 & 1734.8 & 1.01 & -0.017 & 0.007 & 100.0 & 742.5 & 1.01 & -0.030 & 0.042 & 91.9 & 644.1 & 1.01 \\
$\sigma_{\text{bk},11}$ & 4.2 & -0.015 & 0.029 & 95.2 & 1718.2 & 1.01 & -0.002 & 0.036 & 100.0 & 956.0 & 1.00 & -0.031 & 0.046 & 88.9 & 800.7 & 1.01 \\
$\sigma_{\text{bk},12}$ & 2.8 & -0.011 & 0.032 & 95.2 & 1809.6 & 1.01 & -0.059 & 0.029 & 100.0 & 899.5 & 1.00 & -0.009 & 0.026 & 94.9 & 771.3 & 1.01 \\
\bottomrule
\end{tabular*}
\begin{tablenotes}
\footnotesize
\item RB, relative bias; MSE, mean squared error; CP, coverage of the 95\% credible interval (\%); ESS, effective sample size; $\hat R$, Gelman--Rubin convergence diagnostic. Metric formulas are given in Supplementary Table~\ref{tab:s1_comparison_1}. A dash denotes an excluded parameter, undefined relative bias when the true value is zero, or an unavailable diagnostic.
\end{tablenotes}
\end{threeparttable}
\end{table}

\begin{table}[!htbp]
\ContinuedFloat
\centering
\caption{Full parameter-level simulation results: S4 (continued).}
\label{tab:s4_comparison_3}

\footnotesize

\sisetup{
  mode = math
}

\setlength{\tabcolsep}{0.35pt}
\renewcommand{\arraystretch}{0.72}

\begin{threeparttable}
\begin{tabular*}{\linewidth}{@{\extracolsep{\fill}}l S[table-format=-1.3] *{3}{S[table-format=-2.3] S[table-format=3.3] S[table-format=3.1] S[table-format=4.1] S[table-format=1.2]}}
\toprule
Parameter & {True} & \multicolumn{5}{c}{CLOUD} & \multicolumn{5}{c}{StationaryOU} & \multicolumn{5}{c}{DiagOU} \\
\cmidrule(lr){3-7}\cmidrule(lr){8-12}\cmidrule(l){13-17}
 &  & {RB} & {MSE} & {CP} & {ESS} & {$\hat{R}$} & {RB} & {MSE} & {CP} & {ESS} & {$\hat{R}$} & {RB} & {MSE} & {CP} & {ESS} & {$\hat{R}$} \\
\midrule
\multicolumn{17}{@{}l}{\textit{$\boldsymbol\theta$ parameters}} \\
\midrule
$\theta_{1}$ & 2.3 & -0.037 & 0.047 & 100.0 & 1236.4 & 1.01 & 0.307 & 0.511 & 0.0 & 1252.0 & 1.00 & -0.060 & 0.070 & 92.9 & 1332.5 & 1.01 \\
$\theta_{2}$ & 1.5 & -0.053 & 0.054 & 90.5 & 1385.5 & 1.00 & 0.593 & 0.804 & 0.0 & 1035.5 & 1.00 & -0.073 & 0.050 & 90.9 & 1319.3 & 1.01 \\
$\theta_{3}$ & 0.8 & -0.155 & 0.084 & 90.5 & 1056.0 & 1.01 & 1.463 & 1.375 & 0.0 & 1043.5 & 1.00 & -0.052 & 0.060 & 94.9 & 1264.0 & 1.01 \\
$\theta_{4}$ & 1.2 & 0.027 & 0.031 & 100.0 & 1213.7 & 1.01 & -0.324 & 0.152 & 50.0 & 1341.5 & 1.00 & -0.048 & 0.037 & 96.0 & 1476.3 & 1.01 \\
$\theta_{5}$ & 2.1 & -0.149 & 0.129 & 85.7 & 1101.3 & 1.01 & -0.571 & 1.442 & 0.0 & 1310.5 & 1.00 & -0.201 & 0.226 & 55.6 & 1322.9 & 1.01 \\
$\theta_{6,1}$ & -4.0 & -0.063 & 0.232 & 81.0 & 1343.7 & 1.02 & -0.159 & 0.412 & 0.0 & 633.0 & 1.01 & -0.087 & 0.246 & 73.7 & 610.9 & 1.03 \\
$\theta_{6,2}$ & -1.0 & -0.074 & 0.134 & 85.7 & 1295.0 & 1.02 & -0.394 & 0.161 & 100.0 & 594.0 & 1.01 & -0.059 & 0.100 & 84.8 & 614.5 & 1.02 \\
$\theta_{6,3}$ & 2.7 & -0.043 & 0.154 & 81.0 & 1331.6 & 1.02 & -0.006 & 0.008 & 100.0 & 510.0 & 1.01 & -0.106 & 0.176 & 79.8 & 652.4 & 1.02 \\
$\theta_{7,1}$ & -5.5 & 0.001 & 0.075 & 95.2 & 1471.0 & 1.02 & 0.007 & 0.012 & 100.0 & 476.5 & 1.00 & -0.014 & 0.101 & 89.9 & 531.0 & 1.02 \\
$\theta_{7,2}$ & -2.5 & 0.028 & 0.041 & 95.2 & 1363.6 & 1.03 & 0.094 & 0.055 & 100.0 & 396.0 & 1.01 & -0.008 & 0.059 & 94.9 & 428.4 & 1.03 \\
$\theta_{7,3}$ & 2.6 & -0.045 & 0.045 & 100.0 & 1395.0 & 1.02 & -0.139 & 0.143 & 50.0 & 430.5 & 1.00 & -0.022 & 0.060 & 96.0 & 475.3 & 1.02 \\
$\theta_{8,1}$ & -4.5 & 0.003 & 0.026 & 95.2 & 1287.6 & 1.01 & 0.099 & 0.221 & 50.0 & 822.0 & 1.00 & 0.007 & 0.043 & 97.0 & 1018.6 & 1.01 \\
$\theta_{8,2}$ & -1.5 & 0.001 & 0.017 & 90.5 & 1859.4 & 1.00 & 0.311 & 0.233 & 0.0 & 1101.0 & 1.00 & 0.003 & 0.012 & 96.0 & 1838.9 & 1.01 \\
$\theta_{8,3}$ & 2.0 & -0.020 & 0.022 & 85.7 & 1490.4 & 1.00 & -0.335 & 0.450 & 0.0 & 1408.5 & 1.00 & -0.010 & 0.022 & 92.9 & 1077.1 & 1.02 \\
$\theta_{9,1}$ & -3.0 & 0.012 & 0.027 & 95.2 & 1903.4 & 1.00 & 0.196 & 0.345 & 0.0 & 822.0 & 1.00 & 0.010 & 0.025 & 93.9 & 1034.8 & 1.01 \\
$\theta_{9,2}$ & 0.0 & \multicolumn{1}{c}{--} & 0.023 & 91.7 & 1166.9 & 1.01 & \multicolumn{1}{c}{--} & 0.430 & 0.0 & 1164.7 & 1.01 & \multicolumn{1}{c}{--} & 0.020 & 94.9 & 1096.7 & 1.01 \\
$\theta_{9,3}$ & 3.0 & 0.020 & 0.034 & 95.2 & 1067.6 & 1.00 & -0.227 & 0.465 & 0.0 & 884.5 & 1.00 & -0.006 & 0.036 & 93.9 & 968.9 & 1.01 \\
$\theta_{10,1}$ & -6.0 & -0.024 & 0.081 & 95.2 & 1760.0 & 1.01 & -0.221 & 1.798 & 0.0 & 528.0 & 1.01 & -0.033 & 0.123 & 89.9 & 689.0 & 1.01 \\
$\theta_{10,2}$ & -2.0 & -0.052 & 0.039 & 95.2 & 1733.3 & 1.01 & -0.633 & 1.606 & 0.0 & 936.0 & 1.01 & -0.044 & 0.043 & 92.9 & 860.8 & 1.01 \\
$\theta_{10,3}$ & 1.5 & 0.016 & 0.011 & 100.0 & 1789.7 & 1.01 & 0.759 & 1.314 & 0.0 & 866.0 & 1.00 & 0.008 & 0.023 & 97.0 & 1025.7 & 1.01 \\
$\theta_{11,1}$ & -5.0 & -0.006 & 0.091 & 85.7 & 1549.1 & 1.03 & -0.222 & 1.233 & 0.0 & 514.0 & 1.01 & -0.018 & 0.073 & 93.9 & 665.1 & 1.02 \\
$\theta_{11,2}$ & -1.0 & -0.051 & 0.064 & 81.0 & 1412.0 & 1.03 & -0.914 & 0.839 & 0.0 & 469.5 & 1.01 & -0.031 & 0.043 & 94.9 & 534.8 & 1.02 \\
$\theta_{11,3}$ & 2.5 & 0.013 & 0.075 & 81.0 & 1449.9 & 1.02 & 0.399 & 1.014 & 0.0 & 648.5 & 1.00 & -0.010 & 0.050 & 94.9 & 580.3 & 1.02 \\
$\theta_{12,1}$ & -4.0 & -0.005 & 0.031 & 90.5 & 1888.7 & 1.01 & -0.247 & 1.029 & 0.0 & 937.5 & 1.00 & -0.025 & 0.046 & 90.9 & 882.7 & 1.01 \\
$\theta_{12,2}$ & -0.5 & -0.021 & 0.012 & 100.0 & 1885.2 & 1.01 & -1.979 & 1.010 & 0.0 & 1092.5 & 1.00 & -0.082 & 0.022 & 94.9 & 1096.6 & 1.01 \\
$\theta_{12,3}$ & 2.0 & 0.001 & 0.017 & 100.0 & 1956.3 & 1.01 & 0.426 & 0.734 & 0.0 & 1080.5 & 1.00 & 0.012 & 0.029 & 91.9 & 1152.9 & 1.01 \\
\addlinespace[0.25em]
\multicolumn{17}{@{}l}{\textit{Other parameters}} \\
\midrule
$\rho_{1}$ & -0.02 & -0.078 & 0.012 & 95.2 & 1332.9 & 1.04 & 3.658 & 0.268 & 0.0 & 243.5 & 1.03 & -1.000 & 0.0004 & 0.0 & \multicolumn{1}{c}{--} & \multicolumn{1}{c}{--} \\
$\rho_{2}$ & 0.04 & -0.023 & 0.011 & 91.7 & 1175.8 & 1.05 & 1.996 & 0.121 & 0.0 & 306.4 & 1.02 & -1.000 & 0.0016 & 0.0 & \multicolumn{1}{c}{--} & \multicolumn{1}{c}{--} \\
$\rho_{3}$ & 0.06 & 0.037 & 0.020 & 91.7 & 1134.2 & 1.07 & -2.475 & 0.455 & 0.0 & 312.5 & 1.02 & -1.000 & 0.0036 & 0.0 & \multicolumn{1}{c}{--} & \multicolumn{1}{c}{--} \\
$\rho_{4}$ & -0.048 & -0.023 & 0.008 & 100.0 & 1303.1 & 1.02 & 0.173 & 0.006 & 81.8 & 489.3 & 1.03 & -1.000 & 0.0023 & 0.0 & \multicolumn{1}{c}{--} & \multicolumn{1}{c}{--} \\
$\rho_{5}$ & -0.06 & -0.067 & 0.005 & 100.0 & 1350.9 & 1.02 & -8.375 & 0.353 & 81.8 & 420.2 & 1.02 & -1.000 & 0.0036 & 0.0 & \multicolumn{1}{c}{--} & \multicolumn{1}{c}{--} \\
$\rho_{6}$ & 0.0 & \multicolumn{1}{c}{--} & 0.002 & 100.0 & 1776.4 & 1.01 & \multicolumn{1}{c}{--} & 0.297 & 0.0 & 539.0 & 1.02 & \multicolumn{1}{c}{--} & 0.0 & 100.0 & \multicolumn{1}{c}{--} & \multicolumn{1}{c}{--} \\
\bottomrule
\end{tabular*}
\begin{tablenotes}
\footnotesize
\item RB, relative bias; MSE, mean squared error; CP, coverage of the 95\% credible interval (\%); ESS, effective sample size; $\hat R$, Gelman--Rubin convergence diagnostic. Metric formulas are given in Supplementary Table~\ref{tab:s1_comparison_1}. A dash denotes an excluded parameter, undefined relative bias when the true value is zero, or an unavailable diagnostic.
\end{tablenotes}
\end{threeparttable}
\end{table}
\end{landscape}

\begin{landscape}
\begin{table}[p]

\centering
\caption{Full parameter-level simulation results: S5 (CLOUD, TrendyOU, StationaryOU, and DiagOU).}
\label{tab:r8_s5_full_recovery}
\footnotesize
\sisetup{mode=math}
\setlength{\tabcolsep}{0.35pt}
\renewcommand{\arraystretch}{0.72}
\begin{threeparttable}
\begin{tabular*}{\linewidth}{@{\extracolsep{\fill}}l S[table-format=-1.3] *{4}{S[table-format=-4.3] S[table-format=3.3] S[table-format=3.1] S[table-format=5.1] S[table-format=1.2]}}
\toprule
Parameter & {True} & \multicolumn{5}{c}{CLOUD} & \multicolumn{5}{c}{TrendyOU} & \multicolumn{5}{c}{StationaryOU} & \multicolumn{5}{c}{DiagOU} \\
\cmidrule(lr){3-7}\cmidrule(lr){8-12}\cmidrule(lr){13-17}\cmidrule(l){18-22}
 & & {RB} & {MSE} & {CP} & {ESS} & {$\hat{R}$} & {RB} & {MSE} & {CP} & {ESS} & {$\hat{R}$} & {RB} & {MSE} & {CP} & {ESS} & {$\hat{R}$} & {RB} & {MSE} & {CP} & {ESS} & {$\hat{R}$} \\
\midrule
\multicolumn{22}{@{}l}{\textit{Drift $\boldsymbol\Gamma$}} \\
\midrule
$\Gamma_{1,1}$ & 0.824 & 0.076 & 0.513 & 100.0 & 360.1 & 1.04 & -0.028 & 0.267 & 97.1 & 244.0 & 1.10 & -0.479 & 0.451 & 70.3 & 238.6 & 1.06 & 0.487 & 1.430 & 93.0 & 265.0 & 1.05 \\
$\Gamma_{1,2}$ & -0.218 & 0.097 & 0.087 & 100.0 & 341.2 & 1.06 & -1.198 & 0.213 & 98.5 & 251.7 & 1.07 & -0.915 & 0.205 & 91.9 & 206.2 & 1.06 & -1.000 & 0.048 & 0.0 & \multicolumn{1}{c}{--} & \multicolumn{1}{c}{--} \\
$\Gamma_{1,3}$ & -0.121 & -0.094 & 0.097 & 100.0 & 365.7 & 1.06 & -1.376 & 0.126 & 100.0 & 325.6 & 1.05 & -1.165 & 0.069 & 100.0 & 323.3 & 1.05 & -1.000 & 0.015 & 0.0 & \multicolumn{1}{c}{--} & \multicolumn{1}{c}{--} \\
$\Gamma_{1,4}$ & -0.078 & -0.666 & 0.092 & 100.0 & 429.5 & 1.04 & -0.149 & 0.111 & 100.0 & 326.3 & 1.08 & 0.521 & 0.115 & 97.3 & 244.8 & 1.07 & -1.000 & 0.006 & 0.0 & \multicolumn{1}{c}{--} & \multicolumn{1}{c}{--} \\
$\Gamma_{1,5}$ & 0.150 & -0.067 & 0.080 & 100.0 & 287.3 & 1.05 & -1.128 & 0.117 & 100.0 & 243.7 & 1.06 & -4.075 & 0.536 & 43.2 & 247.7 & 1.06 & -1.000 & 0.022 & 0.0 & \multicolumn{1}{c}{--} & \multicolumn{1}{c}{--} \\
$\Gamma_{1,6}$ & -0.104 & 0.078 & 0.130 & 100.0 & 325.8 & 1.05 & -2.002 & 0.210 & 100.0 & 230.6 & 1.08 & 0.198 & 0.126 & 100.0 & 201.7 & 1.05 & -1.000 & 0.011 & 0.0 & \multicolumn{1}{c}{--} & \multicolumn{1}{c}{--} \\
$\Gamma_{1,7}$ & 0.089 & 0.096 & 0.102 & 100.0 & 358.8 & 1.04 & -1.255 & 0.085 & 100.0 & 310.9 & 1.09 & -1.969 & 0.103 & 97.3 & 273.4 & 1.04 & -1.000 & 0.008 & 0.0 & \multicolumn{1}{c}{--} & \multicolumn{1}{c}{--} \\
$\Gamma_{1,8}$ & 0.035 & -0.094 & 0.077 & 100.0 & 366.7 & 1.03 & -1.744 & 0.112 & 100.0 & 320.8 & 1.05 & -0.340 & 0.050 & 100.0 & 265.9 & 1.08 & -1.000 & 0.001 & 0.0 & \multicolumn{1}{c}{--} & \multicolumn{1}{c}{--} \\
$\Gamma_{2,1}$ & 0.002 & -0.037 & 0.117 & 100.0 & 309.9 & 1.07 & 75.690 & 0.125 & 97.1 & 263.1 & 1.07 & 141.205 & 0.254 & 78.4 & 227.2 & 1.06 & -1.000 & 0.000 & 0.0 & \multicolumn{1}{c}{--} & \multicolumn{1}{c}{--} \\
$\Gamma_{2,2}$ & 0.754 & 0.069 & 0.441 & 100.0 & 347.9 & 1.04 & 0.161 & 0.164 & 100.0 & 347.7 & 1.03 & 0.076 & 0.176 & 94.6 & 213.8 & 1.07 & 0.305 & 1.012 & 95.0 & 312.5 & 1.07 \\
$\Gamma_{2,3}$ & 0.069 & 0.033 & 0.103 & 99.0 & 399.1 & 1.05 & -2.259 & 0.067 & 100.0 & 360.6 & 1.07 & -1.849 & 0.078 & 94.6 & 367.6 & 1.08 & -1.000 & 0.005 & 0.0 & \multicolumn{1}{c}{--} & \multicolumn{1}{c}{--} \\
$\Gamma_{2,4}$ & 0.027 & 0.083 & 0.085 & 100.0 & 432.1 & 1.03 & 3.720 & 0.069 & 100.0 & 313.1 & 1.05 & 5.503 & 0.117 & 94.6 & 271.7 & 1.05 & -1.000 & 0.001 & 0.0 & \multicolumn{1}{c}{--} & \multicolumn{1}{c}{--} \\
$\Gamma_{2,5}$ & -0.186 & -0.056 & 0.098 & 100.0 & 300.3 & 1.05 & -1.245 & 0.173 & 95.6 & 227.6 & 1.09 & -1.763 & 0.229 & 78.4 & 244.1 & 1.05 & -1.000 & 0.034 & 0.0 & \multicolumn{1}{c}{--} & \multicolumn{1}{c}{--} \\
$\Gamma_{2,6}$ & -0.101 & -0.046 & 0.131 & 100.0 & 352.9 & 1.04 & 3.750 & 0.249 & 100.0 & 272.9 & 1.07 & 2.646 & 0.176 & 97.3 & 190.7 & 1.07 & -1.000 & 0.010 & 0.0 & \multicolumn{1}{c}{--} & \multicolumn{1}{c}{--} \\
$\Gamma_{2,7}$ & -0.005 & -0.028 & 0.079 & 100.0 & 372.5 & 1.04 & -22.106 & 0.061 & 100.0 & 327.9 & 1.07 & -34.576 & 0.069 & 100.0 & 299.2 & 1.08 & -1.000 & 0.000 & 0.0 & \multicolumn{1}{c}{--} & \multicolumn{1}{c}{--} \\
$\Gamma_{2,8}$ & 0.057 & -0.048 & 0.117 & 100.0 & 401.6 & 1.06 & 1.761 & 0.073 & 100.0 & 318.1 & 1.04 & 1.263 & 0.081 & 100.0 & 305.8 & 1.08 & -1.000 & 0.003 & 0.0 & \multicolumn{1}{c}{--} & \multicolumn{1}{c}{--} \\
$\Gamma_{3,1}$ & -0.144 & -0.081 & 0.131 & 100.0 & 300.6 & 1.06 & -1.121 & 0.140 & 100.0 & 298.6 & 1.07 & -0.469 & 0.226 & 97.3 & 306.8 & 1.04 & -1.000 & 0.021 & 0.0 & \multicolumn{1}{c}{--} & \multicolumn{1}{c}{--} \\
$\Gamma_{3,2}$ & 0.256 & -0.038 & 0.110 & 99.0 & 335.4 & 1.03 & -0.913 & 0.188 & 100.0 & 324.9 & 1.09 & -1.222 & 0.346 & 97.3 & 352.0 & 1.07 & -1.000 & 0.066 & 0.0 & \multicolumn{1}{c}{--} & \multicolumn{1}{c}{--} \\
$\Gamma_{3,3}$ & 0.935 & 0.047 & 0.409 & 100.0 & 423.7 & 1.03 & 0.460 & 0.479 & 98.5 & 411.9 & 1.05 & 0.636 & 0.920 & 91.9 & 350.3 & 1.04 & 0.063 & 0.522 & 84.0 & 431.5 & 1.02 \\
$\Gamma_{3,4}$ & -0.077 & -0.069 & 0.092 & 100.0 & 455.9 & 1.03 & -4.139 & 0.219 & 100.0 & 393.0 & 1.08 & -4.486 & 0.457 & 94.6 & 391.1 & 1.03 & -1.000 & 0.006 & 0.0 & \multicolumn{1}{c}{--} & \multicolumn{1}{c}{--} \\
$\Gamma_{3,5}$ & -0.203 & 0.040 & 0.122 & 100.0 & 301.6 & 1.05 & 0.409 & 0.138 & 100.0 & 269.1 & 1.07 & 1.310 & 0.278 & 94.6 & 351.4 & 1.05 & -1.000 & 0.041 & 0.0 & \multicolumn{1}{c}{--} & \multicolumn{1}{c}{--} \\
$\Gamma_{3,6}$ & 0.211 & 0.036 & 0.122 & 100.0 & 352.4 & 1.04 & -0.450 & 0.204 & 100.0 & 305.5 & 1.07 & -0.684 & 0.298 & 94.6 & 316.5 & 1.07 & -1.000 & 0.045 & 0.0 & \multicolumn{1}{c}{--} & \multicolumn{1}{c}{--} \\
$\Gamma_{3,7}$ & -0.059 & -0.067 & 0.093 & 100.0 & 399.0 & 1.05 & 3.193 & 0.120 & 100.0 & 423.7 & 1.05 & 3.591 & 0.185 & 100.0 & 412.6 & 1.07 & -1.000 & 0.003 & 0.0 & \multicolumn{1}{c}{--} & \multicolumn{1}{c}{--} \\
$\Gamma_{3,8}$ & -0.038 & 0.025 & 0.072 & 100.0 & 389.6 & 1.07 & -8.842 & 0.208 & 100.0 & 409.1 & 1.04 & -6.746 & 0.238 & 100.0 & 403.5 & 1.06 & -1.000 & 0.001 & 0.0 & \multicolumn{1}{c}{--} & \multicolumn{1}{c}{--} \\
$\Gamma_{4,1}$ & -0.027 & -0.034 & 0.114 & 100.0 & 325.4 & 1.04 & -1.846 & 0.112 & 100.0 & 284.5 & 1.06 & -8.815 & 0.182 & 86.5 & 227.2 & 1.06 & -1.000 & 0.001 & 0.0 & \multicolumn{1}{c}{--} & \multicolumn{1}{c}{--} \\
$\Gamma_{4,2}$ & 0.075 & 0.071 & 0.087 & 100.0 & 360.0 & 1.03 & 3.345 & 0.191 & 98.5 & 281.3 & 1.05 & 0.722 & 0.113 & 100.0 & 255.8 & 1.05 & -1.000 & 0.006 & 0.0 & \multicolumn{1}{c}{--} & \multicolumn{1}{c}{--} \\
$\Gamma_{4,3}$ & 0.037 & 0.075 & 0.079 & 100.0 & 447.3 & 1.03 & 4.797 & 0.087 & 100.0 & 387.4 & 1.06 & 3.766 & 0.089 & 100.0 & 356.0 & 1.04 & -1.000 & 0.001 & 0.0 & \multicolumn{1}{c}{--} & \multicolumn{1}{c}{--} \\
$\Gamma_{4,4}$ & 0.779 & 0.059 & 0.347 & 99.0 & 420.0 & 1.03 & 0.544 & 0.353 & 98.5 & 382.3 & 1.05 & 0.282 & 0.170 & 100.0 & 281.8 & 1.04 & 0.090 & 0.189 & 95.0 & 458.0 & 1.03 \\
$\Gamma_{4,5}$ & 0.061 & 0.044 & 0.141 & 100.0 & 316.5 & 1.05 & 1.313 & 0.164 & 100.0 & 258.4 & 1.05 & 5.013 & 0.163 & 83.8 & 240.8 & 1.05 & -1.000 & 0.004 & 0.0 & \multicolumn{1}{c}{--} & \multicolumn{1}{c}{--} \\
$\Gamma_{4,6}$ & 0.134 & 0.062 & 0.113 & 100.0 & 375.2 & 1.03 & 2.214 & 0.265 & 98.5 & 265.7 & 1.06 & 2.002 & 0.157 & 97.3 & 252.0 & 1.07 & -1.000 & 0.018 & 0.0 & \multicolumn{1}{c}{--} & \multicolumn{1}{c}{--} \\
$\Gamma_{4,7}$ & 0.038 & 0.046 & 0.094 & 100.0 & 415.7 & 1.03 & 4.506 & 0.100 & 100.0 & 373.4 & 1.06 & 4.645 & 0.099 & 94.6 & 297.9 & 1.04 & -1.000 & 0.001 & 0.0 & \multicolumn{1}{c}{--} & \multicolumn{1}{c}{--} \\
$\Gamma_{4,8}$ & -0.026 & -0.049 & 0.051 & 100.0 & 418.1 & 1.02 & 6.074 & 0.094 & 100.0 & 366.2 & 1.06 & 10.853 & 0.167 & 94.6 & 306.3 & 1.05 & -1.000 & 0.001 & 0.0 & \multicolumn{1}{c}{--} & \multicolumn{1}{c}{--} \\
$\Gamma_{5,1}$ & 0.029 & -0.059 & 0.080 & 100.0 & 283.0 & 1.06 & -8.007 & 0.132 & 100.0 & 300.5 & 1.05 & -9.188 & 0.212 & 81.1 & 208.0 & 1.06 & -1.000 & 0.001 & 0.0 & \multicolumn{1}{c}{--} & \multicolumn{1}{c}{--} \\
\bottomrule
\end{tabular*}
\begin{tablenotes}
\footnotesize
\item RB is mean relative bias across retained replications, MSE is mean squared error, CP is empirical coverage of 95\% posterior credible intervals, ESS is effective sample size, and $\hat{R}$ is the Gelman--Rubin convergence diagnostic. Summaries use the retained-fit denominator reported in Table~\ref{tab:simulation_summary} in the main text. A dash denotes an unavailable quantity, undefined RB when the data-generating value is zero, or an unavailable sampling diagnostic for a structurally fixed parameter.
\end{tablenotes}
\end{threeparttable}
\end{table}
\clearpage
\end{landscape}

\begin{landscape}
\begin{table}[p]
\ContinuedFloat

\centering
\caption{Full parameter-level simulation results: S5 (continued).}
\label{tab:r8_s5_full_recovery_2}
\footnotesize
\sisetup{mode=math}
\setlength{\tabcolsep}{0.35pt}
\renewcommand{\arraystretch}{0.72}
\begin{threeparttable}
\begin{tabular*}{\linewidth}{@{\extracolsep{\fill}}l S[table-format=-1.3] *{4}{S[table-format=-4.3] S[table-format=3.3] S[table-format=3.1] S[table-format=5.1] S[table-format=1.2]}}
\toprule
Parameter & {True} & \multicolumn{5}{c}{CLOUD} & \multicolumn{5}{c}{TrendyOU} & \multicolumn{5}{c}{StationaryOU} & \multicolumn{5}{c}{DiagOU} \\
\cmidrule(lr){3-7}\cmidrule(lr){8-12}\cmidrule(lr){13-17}\cmidrule(l){18-22}
 & & {RB} & {MSE} & {CP} & {ESS} & {$\hat{R}$} & {RB} & {MSE} & {CP} & {ESS} & {$\hat{R}$} & {RB} & {MSE} & {CP} & {ESS} & {$\hat{R}$} & {RB} & {MSE} & {CP} & {ESS} & {$\hat{R}$} \\
\midrule
\multicolumn{22}{@{}l}{\textit{Drift $\boldsymbol\Gamma$ (continued)}} \\
\midrule
$\Gamma_{5,2}$ & -0.178 & -0.063 & 0.107 & 100.0 & 320.5 & 1.05 & -1.620 & 0.235 & 98.5 & 267.3 & 1.06 & -0.454 & 0.117 & 94.6 & 206.3 & 1.07 & -1.000 & 0.032 & 0.0 & \multicolumn{1}{c}{--} & \multicolumn{1}{c}{--} \\
$\Gamma_{5,3}$ & -0.048 & 0.092 & 0.089 & 100.0 & 357.6 & 1.04 & 0.228 & 0.075 & 100.0 & 328.5 & 1.07 & -1.008 & 0.055 & 97.3 & 303.6 & 1.06 & -1.000 & 0.002 & 0.0 & \multicolumn{1}{c}{--} & \multicolumn{1}{c}{--} \\
$\Gamma_{5,4}$ & -0.145 & -0.034 & 0.124 & 100.0 & 399.9 & 1.03 & -1.226 & 0.179 & 100.0 & 347.6 & 1.06 & 1.308 & 0.098 & 97.3 & 254.4 & 1.05 & -1.000 & 0.021 & 0.0 & \multicolumn{1}{c}{--} & \multicolumn{1}{c}{--} \\
$\Gamma_{5,5}$ & 0.803 & 0.075 & 0.495 & 100.0 & 358.7 & 1.03 & 0.242 & 0.317 & 97.1 & 240.9 & 1.07 & -1.052 & 0.853 & 35.1 & 228.3 & 1.08 & 0.467 & 1.417 & 93.0 & 275.4 & 1.04 \\
$\Gamma_{5,6}$ & 0.104 & 0.089 & 0.108 & 100.0 & 311.0 & 1.06 & 1.995 & 0.242 & 98.5 & 240.0 & 1.07 & -3.701 & 0.269 & 78.4 & 184.4 & 1.06 & -1.000 & 0.011 & 0.0 & \multicolumn{1}{c}{--} & \multicolumn{1}{c}{--} \\
$\Gamma_{5,7}$ & 0.050 & 0.093 & 0.096 & 100.0 & 375.9 & 1.05 & -1.238 & 0.072 & 100.0 & 338.5 & 1.08 & -4.253 & 0.096 & 94.6 & 227.3 & 1.06 & -1.000 & 0.002 & 0.0 & \multicolumn{1}{c}{--} & \multicolumn{1}{c}{--} \\
$\Gamma_{5,8}$ & -0.001 & 0.031 & 0.075 & 100.0 & 387.8 & 1.04 & 8.394 & 0.118 & 100.0 & 356.0 & 1.08 & 35.232 & 0.068 & 97.3 & 213.9 & 1.06 & -1.000 & 0.000 & 0.0 & \multicolumn{1}{c}{--} & \multicolumn{1}{c}{--} \\
$\Gamma_{6,1}$ & 0.087 & -0.092 & 0.164 & 100.0 & 293.2 & 1.06 & 1.702 & 0.104 & 98.5 & 240.4 & 1.09 & 1.378 & 0.155 & 91.9 & 224.3 & 1.06 & -1.000 & 0.008 & 0.0 & \multicolumn{1}{c}{--} & \multicolumn{1}{c}{--} \\
$\Gamma_{6,2}$ & 0.106 & -0.054 & 0.085 & 100.0 & 346.9 & 1.05 & -5.444 & 0.420 & 92.6 & 250.6 & 1.07 & -4.626 & 0.406 & 86.5 & 185.6 & 1.08 & -1.000 & 0.011 & 0.0 & \multicolumn{1}{c}{--} & \multicolumn{1}{c}{--} \\
$\Gamma_{6,3}$ & 0.429 & -0.020 & 0.120 & 100.0 & 434.7 & 1.03 & -0.668 & 0.143 & 98.5 & 328.9 & 1.08 & -0.838 & 0.230 & 86.5 & 295.1 & 1.08 & -1.000 & 0.184 & 0.0 & \multicolumn{1}{c}{--} & \multicolumn{1}{c}{--} \\
$\Gamma_{6,4}$ & -0.134 & -0.064 & 0.093 & 100.0 & 438.0 & 1.03 & -2.202 & 0.219 & 100.0 & 282.6 & 1.06 & -1.806 & 0.203 & 91.9 & 278.2 & 1.06 & -1.000 & 0.018 & 0.0 & \multicolumn{1}{c}{--} & \multicolumn{1}{c}{--} \\
$\Gamma_{6,5}$ & 0.165 & 0.024 & 0.096 & 100.0 & 291.2 & 1.07 & 0.905 & 0.145 & 100.0 & 224.1 & 1.09 & 1.074 & 0.109 & 91.9 & 225.9 & 1.05 & -1.000 & 0.027 & 0.0 & \multicolumn{1}{c}{--} & \multicolumn{1}{c}{--} \\
$\Gamma_{6,6}$ & 1.157 & 0.029 & 0.270 & 100.0 & 394.2 & 1.05 & -0.059 & 0.182 & 95.6 & 303.9 & 1.06 & -0.181 & 0.338 & 89.2 & 214.9 & 1.07 & 0.321 & 1.847 & 89.0 & 335.2 & 1.06 \\
$\Gamma_{6,7}$ & 0.232 & -0.041 & 0.108 & 100.0 & 369.2 & 1.05 & -0.586 & 0.075 & 100.0 & 305.4 & 1.06 & -0.767 & 0.100 & 91.9 & 287.7 & 1.05 & -1.000 & 0.054 & 0.0 & \multicolumn{1}{c}{--} & \multicolumn{1}{c}{--} \\
$\Gamma_{6,8}$ & -0.044 & 0.083 & 0.120 & 100.0 & 404.1 & 1.05 & -4.344 & 0.110 & 100.0 & 296.0 & 1.07 & -4.479 & 0.154 & 86.5 & 295.2 & 1.08 & -1.000 & 0.002 & 0.0 & \multicolumn{1}{c}{--} & \multicolumn{1}{c}{--} \\
$\Gamma_{7,1}$ & 0.261 & 0.045 & 0.068 & 100.0 & 320.8 & 1.07 & -0.416 & 0.064 & 100.0 & 320.5 & 1.05 & -1.056 & 0.216 & 94.6 & 314.6 & 1.04 & -1.000 & 0.068 & 0.0 & \multicolumn{1}{c}{--} & \multicolumn{1}{c}{--} \\
$\Gamma_{7,2}$ & -0.225 & -0.078 & 0.117 & 100.0 & 346.8 & 1.04 & -0.881 & 0.227 & 98.5 & 339.5 & 1.06 & -0.877 & 0.256 & 97.3 & 300.9 & 1.07 & -1.000 & 0.051 & 0.0 & \multicolumn{1}{c}{--} & \multicolumn{1}{c}{--} \\
$\Gamma_{7,3}$ & -0.110 & 0.045 & 0.099 & 100.0 & 427.0 & 1.04 & 2.220 & 0.175 & 100.0 & 408.5 & 1.05 & 2.331 & 0.249 & 100.0 & 388.0 & 1.06 & -1.000 & 0.012 & 0.0 & \multicolumn{1}{c}{--} & \multicolumn{1}{c}{--} \\
$\Gamma_{7,4}$ & -0.086 & -0.057 & 0.063 & 100.0 & 440.1 & 1.03 & -2.284 & 0.147 & 100.0 & 388.5 & 1.06 & -2.804 & 0.176 & 100.0 & 348.5 & 1.04 & -1.000 & 0.007 & 0.0 & \multicolumn{1}{c}{--} & \multicolumn{1}{c}{--} \\
$\Gamma_{7,5}$ & 0.352 & -0.039 & 0.074 & 100.0 & 311.5 & 1.04 & -0.401 & 0.140 & 97.1 & 298.3 & 1.06 & -1.644 & 0.408 & 78.4 & 303.5 & 1.06 & -1.000 & 0.124 & 0.0 & \multicolumn{1}{c}{--} & \multicolumn{1}{c}{--} \\
$\Gamma_{7,6}$ & 0.080 & 0.019 & 0.114 & 100.0 & 330.5 & 1.04 & 1.033 & 0.190 & 100.0 & 298.3 & 1.05 & -1.054 & 0.230 & 100.0 & 289.7 & 1.07 & -1.000 & 0.006 & 0.0 & \multicolumn{1}{c}{--} & \multicolumn{1}{c}{--} \\
$\Gamma_{7,7}$ & 0.935 & 0.037 & 0.248 & 100.0 & 384.9 & 1.03 & 0.315 & 0.252 & 98.5 & 414.5 & 1.04 & 0.370 & 0.468 & 97.3 & 309.7 & 1.04 & -0.114 & 0.233 & 81.0 & 466.1 & 1.02 \\
$\Gamma_{7,8}$ & 0.007 & -0.067 & 0.067 & 100.0 & 405.5 & 1.03 & 21.635 & 0.109 & 100.0 & 421.4 & 1.05 & 34.214 & 0.228 & 94.6 & 375.2 & 1.04 & -1.000 & 0.000 & 0.0 & \multicolumn{1}{c}{--} & \multicolumn{1}{c}{--} \\
$\Gamma_{8,1}$ & 0.067 & -0.042 & 0.067 & 100.0 & 295.8 & 1.04 & -2.246 & 0.068 & 100.0 & 290.5 & 1.06 & 0.756 & 0.078 & 97.3 & 270.8 & 1.05 & -1.000 & 0.005 & 0.0 & \multicolumn{1}{c}{--} & \multicolumn{1}{c}{--} \\
$\Gamma_{8,2}$ & -0.047 & 0.033 & 0.084 & 100.0 & 340.3 & 1.03 & -3.448 & 0.099 & 100.0 & 301.2 & 1.05 & -3.656 & 0.086 & 97.3 & 264.6 & 1.07 & -1.000 & 0.002 & 0.0 & \multicolumn{1}{c}{--} & \multicolumn{1}{c}{--} \\
$\Gamma_{8,3}$ & 0.006 & -0.065 & 0.062 & 100.0 & 409.8 & 1.06 & 29.055 & 0.086 & 100.0 & 402.8 & 1.03 & 21.421 & 0.057 & 100.0 & 374.4 & 1.06 & -1.000 & 0.000 & 0.0 & \multicolumn{1}{c}{--} & \multicolumn{1}{c}{--} \\
$\Gamma_{8,4}$ & 0.017 & -0.079 & 0.066 & 100.0 & 430.0 & 1.03 & -10.720 & 0.119 & 100.0 & 358.9 & 1.07 & -10.552 & 0.124 & 97.3 & 329.1 & 1.06 & -1.000 & 0.000 & 0.0 & \multicolumn{1}{c}{--} & \multicolumn{1}{c}{--} \\
$\Gamma_{8,5}$ & 0.138 & -0.029 & 0.140 & 100.0 & 280.6 & 1.06 & -0.948 & 0.113 & 100.0 & 251.8 & 1.08 & 1.432 & 0.082 & 89.2 & 271.3 & 1.06 & -1.000 & 0.019 & 0.0 & \multicolumn{1}{c}{--} & \multicolumn{1}{c}{--} \\
$\Gamma_{8,6}$ & 0.113 & -0.094 & 0.155 & 100.0 & 333.7 & 1.06 & 1.359 & 0.125 & 100.0 & 276.4 & 1.09 & 1.612 & 0.104 & 97.3 & 253.1 & 1.06 & -1.000 & 0.013 & 0.0 & \multicolumn{1}{c}{--} & \multicolumn{1}{c}{--} \\
$\Gamma_{8,7}$ & 0.110 & -0.075 & 0.080 & 100.0 & 384.7 & 1.04 & 0.420 & 0.061 & 100.0 & 387.4 & 1.05 & 1.223 & 0.092 & 94.6 & 360.4 & 1.05 & -1.000 & 0.012 & 0.0 & \multicolumn{1}{c}{--} & \multicolumn{1}{c}{--} \\
$\Gamma_{8,8}$ & 0.671 & 0.062 & 0.293 & 100.0 & 372.4 & 1.05 & 0.329 & 0.148 & 100.0 & 452.2 & 1.03 & 0.155 & 0.089 & 100.0 & 381.4 & 1.04 & 0.054 & 0.059 & 97.0 & 428.2 & 1.03 \\
\addlinespace[0.25em]
\multicolumn{22}{@{}l}{\textit{Loadings $\boldsymbol\Lambda$}} \\
\midrule
$\lambda_{1}$ & 0.800 & -0.043 & 0.009 & 96.0 & 1025.9 & 1.01 & -0.052 & 0.017 & 95.6 & 568.7 & 1.05 & 0.748 & 0.397 & 2.7 & 1184.7 & 1.02 & 0.020 & 0.007 & 99.0 & 1379.7 & 1.01 \\
$\lambda_{2}$ & 1.200 & -0.129 & 0.029 & 89.0 & 999.6 & 1.01 & -0.126 & 0.043 & 89.7 & 445.4 & 1.07 & 0.558 & 0.476 & 5.4 & 767.8 & 1.02 & -0.088 & 0.016 & 97.0 & 1412.2 & 1.01 \\
\bottomrule
\end{tabular*}
\begin{tablenotes}
\footnotesize
\item RB is mean relative bias across retained replications, MSE is mean squared error, CP is empirical coverage of 95\% posterior credible intervals, ESS is effective sample size, and $\hat{R}$ is the Gelman--Rubin convergence diagnostic. Summaries use the retained-fit denominator reported in Table~\ref{tab:simulation_summary} in the main text. A dash denotes an unavailable quantity, undefined RB when the data-generating value is zero, or an unavailable sampling diagnostic for a structurally fixed parameter.
\end{tablenotes}
\end{threeparttable}
\end{table}
\clearpage
\end{landscape}

\begin{landscape}
\begin{table}[p]
\ContinuedFloat

\centering
\caption{Full parameter-level simulation results: S5 (continued).}
\label{tab:r8_s5_full_recovery_3}
\footnotesize
\sisetup{mode=math}
\setlength{\tabcolsep}{0.35pt}
\renewcommand{\arraystretch}{0.72}
\begin{threeparttable}
\begin{tabular*}{\linewidth}{@{\extracolsep{\fill}}l S[table-format=-1.3] *{4}{S[table-format=-4.3] S[table-format=3.3] S[table-format=3.1] S[table-format=5.1] S[table-format=1.2]}}
\toprule
Parameter & {True} & \multicolumn{5}{c}{CLOUD} & \multicolumn{5}{c}{TrendyOU} & \multicolumn{5}{c}{StationaryOU} & \multicolumn{5}{c}{DiagOU} \\
\cmidrule(lr){3-7}\cmidrule(lr){8-12}\cmidrule(lr){13-17}\cmidrule(l){18-22}
 & & {RB} & {MSE} & {CP} & {ESS} & {$\hat{R}$} & {RB} & {MSE} & {CP} & {ESS} & {$\hat{R}$} & {RB} & {MSE} & {CP} & {ESS} & {$\hat{R}$} & {RB} & {MSE} & {CP} & {ESS} & {$\hat{R}$} \\
\midrule
\multicolumn{22}{@{}l}{\textit{Loadings $\boldsymbol\Lambda$ (continued)}} \\
\midrule
$\lambda_{3}$ & 1.100 & -0.141 & 0.030 & 92.0 & 1023.6 & 1.01 & -0.103 & 0.032 & 91.2 & 484.6 & 1.06 & 0.668 & 0.573 & 0.0 & 1113.7 & 1.01 & -0.094 & 0.016 & 98.0 & 1332.5 & 1.01 \\
$\lambda_{4}$ & 0.900 & -0.107 & 0.016 & 96.0 & 742.4 & 1.02 & 0.448 & 0.182 & 19.1 & 1128.3 & 1.01 & 0.771 & 0.526 & 2.7 & 1041.4 & 1.01 & -0.032 & 0.006 & 100.0 & 1088.9 & 1.02 \\
$\lambda_{5}$ & 1.400 & -0.217 & 0.099 & 39.0 & 833.0 & 1.02 & 0.231 & 0.139 & 69.1 & 936.7 & 1.02 & 0.543 & 0.629 & 8.1 & 971.0 & 1.01 & -0.166 & 0.060 & 72.0 & 1196.0 & 1.02 \\
$\lambda_{6}$ & 1.000 & -0.159 & 0.032 & 67.0 & 588.1 & 1.02 & 0.459 & 0.234 & 8.8 & 800.7 & 1.01 & 0.827 & 0.723 & 0.0 & 768.7 & 1.02 & -0.084 & 0.013 & 92.0 & 801.8 & 1.02 \\
$\lambda_{7}$ & 0.700 & 0.026 & 0.008 & 94.0 & 977.3 & 1.01 & 0.058 & 0.012 & 95.6 & 1271.2 & 1.01 & 0.232 & 0.037 & 67.6 & 1343.5 & 1.01 & 0.059 & 0.010 & 93.0 & 1472.1 & 1.01 \\
$\lambda_{8}$ & 1.100 & -0.054 & 0.010 & 97.0 & 767.6 & 1.02 & 0.013 & 0.012 & 95.6 & 829.0 & 1.01 & 0.186 & 0.054 & 64.9 & 765.5 & 1.01 & -0.022 & 0.007 & 99.0 & 1204.4 & 1.01 \\
$\lambda_{9}$ & 0.900 & -0.017 & 0.008 & 95.0 & 800.1 & 1.01 & 0.055 & 0.014 & 94.1 & 943.7 & 1.01 & 0.268 & 0.071 & 43.2 & 927.1 & 1.01 & 0.013 & 0.008 & 98.0 & 1261.1 & 1.01 \\
$\lambda_{10}$ & 1.200 & -0.062 & 0.011 & 92.0 & 675.8 & 1.02 & 0.238 & 0.096 & 32.4 & 827.3 & 1.01 & 0.608 & 0.552 & 0.0 & 1010.6 & 1.02 & -0.032 & 0.007 & 97.0 & 1032.6 & 1.01 \\
$\lambda_{11}$ & 0.800 & -0.020 & 0.005 & 98.0 & 781.7 & 1.01 & 0.279 & 0.058 & 29.4 & 1149.3 & 1.02 & 0.654 & 0.293 & 0.0 & 1498.5 & 1.02 & 0.015 & 0.005 & 97.0 & 1210.6 & 1.01 \\
$\lambda_{12}$ & 1.000 & -0.036 & 0.008 & 90.0 & 687.0 & 1.01 & 0.255 & 0.079 & 27.9 & 936.4 & 1.02 & 0.595 & 0.373 & 0.0 & 1182.9 & 1.01 & -0.004 & 0.006 & 94.0 & 1083.3 & 1.01 \\
$\lambda_{13}$ & 0.800 & -0.041 & 0.010 & 99.0 & 1055.3 & 1.01 & 0.025 & 0.021 & 92.6 & 744.4 & 1.04 & 0.775 & 0.430 & 0.0 & 1696.7 & 1.01 & 0.027 & 0.008 & 98.0 & 1425.7 & 1.01 \\
$\lambda_{14}$ & 1.200 & -0.130 & 0.031 & 90.0 & 1018.5 & 1.01 & -0.036 & 0.034 & 89.7 & 542.8 & 1.06 & 0.566 & 0.509 & 16.2 & 1180.4 & 1.01 & -0.086 & 0.016 & 99.0 & 1434.0 & 1.01 \\
$\lambda_{15}$ & 1.100 & -0.140 & 0.029 & 87.0 & 1015.4 & 1.01 & -0.040 & 0.025 & 95.6 & 600.5 & 1.04 & 0.648 & 0.546 & 0.0 & 1502.1 & 1.01 & -0.089 & 0.014 & 98.0 & 1402.8 & 1.01 \\
$\lambda_{16}$ & 0.900 & -0.071 & 0.012 & 94.0 & 930.1 & 1.02 & 0.500 & 0.220 & 7.4 & 1238.8 & 1.02 & 0.781 & 0.528 & 2.7 & 1255.5 & 1.02 & -0.034 & 0.007 & 100.0 & 1234.0 & 1.01 \\
$\lambda_{17}$ & 1.400 & -0.196 & 0.083 & 46.0 & 1003.2 & 1.02 & 0.249 & 0.150 & 58.8 & 1082.0 & 1.01 & 0.544 & 0.625 & 5.4 & 1145.3 & 1.01 & -0.175 & 0.067 & 59.0 & 1373.1 & 1.01 \\
$\lambda_{18}$ & 1.000 & -0.124 & 0.021 & 85.0 & 719.8 & 1.04 & 0.486 & 0.254 & 2.9 & 861.1 & 1.01 & 0.847 & 0.749 & 0.0 & 864.6 & 1.02 & -0.087 & 0.013 & 96.0 & 900.1 & 1.02 \\
$\lambda_{19}$ & 0.700 & 0.028 & 0.007 & 95.0 & 1011.8 & 1.01 & 0.048 & 0.010 & 95.6 & 1326.8 & 1.01 & 0.260 & 0.047 & 62.2 & 1304.2 & 1.01 & 0.068 & 0.009 & 95.0 & 1526.5 & 1.01 \\
$\lambda_{20}$ & 1.100 & -0.029 & 0.008 & 98.0 & 791.0 & 1.01 & 0.041 & 0.016 & 95.6 & 831.5 & 1.01 & 0.198 & 0.067 & 56.8 & 840.8 & 1.01 & -0.002 & 0.006 & 100.0 & 1243.9 & 1.01 \\
$\lambda_{21}$ & 0.900 & 0.005 & 0.007 & 98.0 & 854.3 & 1.01 & 0.094 & 0.016 & 91.2 & 1015.1 & 1.01 & 0.316 & 0.099 & 37.8 & 939.7 & 1.01 & 0.033 & 0.007 & 96.0 & 1323.5 & 1.01 \\
$\lambda_{22}$ & 1.200 & -0.062 & 0.012 & 93.0 & 608.0 & 1.03 & 0.242 & 0.101 & 39.7 & 804.3 & 1.01 & 0.643 & 0.620 & 0.0 & 954.1 & 1.01 & -0.029 & 0.007 & 99.0 & 940.4 & 1.01 \\
$\lambda_{23}$ & 0.800 & -0.010 & 0.004 & 97.0 & 687.0 & 1.03 & 0.298 & 0.063 & 19.1 & 1034.1 & 1.02 & 0.688 & 0.312 & 0.0 & 1315.3 & 1.01 & 0.029 & 0.004 & 99.0 & 1085.1 & 1.01 \\
$\lambda_{24}$ & 1.000 & -0.031 & 0.006 & 97.0 & 617.6 & 1.02 & 0.271 & 0.085 & 29.4 & 868.0 & 1.02 & 0.616 & 0.398 & 0.0 & 1151.0 & 1.01 & 0.005 & 0.004 & 99.0 & 976.9 & 1.01 \\
\addlinespace[0.25em]
\multicolumn{22}{@{}l}{\textit{Covariate slopes $\boldsymbol\Phi$}} \\
\midrule
$\Phi_{1,1}$ & 0.400 & 0.127 & 0.010 & 91.0 & 1589.0 & 1.01 & -1.000 & 0.160 & 0.0 & \multicolumn{1}{c}{--} & \multicolumn{1}{c}{--} & -1.000 & 0.160 & 0.0 & \multicolumn{1}{c}{--} & \multicolumn{1}{c}{--} & 0.043 & 0.007 & 94.0 & 1921.4 & 1.01 \\
$\Phi_{1,2}$ & -0.200 & 0.108 & 0.002 & 91.0 & 1521.9 & 1.01 & -1.000 & 0.040 & 0.0 & \multicolumn{1}{c}{--} & \multicolumn{1}{c}{--} & -1.000 & 0.040 & 0.0 & \multicolumn{1}{c}{--} & \multicolumn{1}{c}{--} & 0.027 & 0.002 & 98.0 & 1993.5 & 1.01 \\
$\Phi_{2,1}$ & -0.300 & 0.141 & 0.005 & 94.0 & 1262.0 & 1.01 & -1.000 & 0.090 & 0.0 & \multicolumn{1}{c}{--} & \multicolumn{1}{c}{--} & -1.000 & 0.090 & 0.0 & \multicolumn{1}{c}{--} & \multicolumn{1}{c}{--} & 0.046 & 0.003 & 98.0 & 1849.3 & 1.02 \\
$\Phi_{2,2}$ & 0.500 & 0.151 & 0.009 & 84.0 & 593.4 & 1.02 & -1.000 & 0.250 & 0.0 & \multicolumn{1}{c}{--} & \multicolumn{1}{c}{--} & -1.000 & 0.250 & 0.0 & \multicolumn{1}{c}{--} & \multicolumn{1}{c}{--} & 0.054 & 0.003 & 97.0 & 716.5 & 1.03 \\
$\Phi_{3,1}$ & 0.200 & 0.042 & 0.002 & 94.0 & 1756.2 & 1.01 & -1.000 & 0.040 & 0.0 & \multicolumn{1}{c}{--} & \multicolumn{1}{c}{--} & -1.000 & 0.040 & 0.0 & \multicolumn{1}{c}{--} & \multicolumn{1}{c}{--} & 0.008 & 0.002 & 92.0 & 3186.9 & 1.00 \\
$\Phi_{3,2}$ & 0.100 & 0.034 & 0.000 & 99.0 & 1853.5 & 1.01 & -1.000 & 0.010 & 0.0 & \multicolumn{1}{c}{--} & \multicolumn{1}{c}{--} & -1.000 & 0.010 & 0.0 & \multicolumn{1}{c}{--} & \multicolumn{1}{c}{--} & 0.001 & 0.000 & 98.0 & 3441.3 & 1.00 \\
$\Phi_{4,1}$ & -0.100 & 0.059 & 0.001 & 98.0 & 2001.4 & 1.01 & -1.000 & 0.010 & 0.0 & \multicolumn{1}{c}{--} & \multicolumn{1}{c}{--} & -1.000 & 0.010 & 0.0 & \multicolumn{1}{c}{--} & \multicolumn{1}{c}{--} & 0.019 & 0.001 & 98.0 & 3872.2 & 1.00 \\
$\Phi_{4,2}$ & -0.300 & 0.034 & 0.001 & 94.0 & 770.6 & 1.01 & -1.000 & 0.090 & 0.0 & \multicolumn{1}{c}{--} & \multicolumn{1}{c}{--} & -1.000 & 0.090 & 0.0 & \multicolumn{1}{c}{--} & \multicolumn{1}{c}{--} & 0.001 & 0.001 & 95.0 & 1232.4 & 1.01 \\
$\Phi_{5,1}$ & 0.400 & 0.116 & 0.009 & 95.0 & 1604.5 & 1.01 & -1.000 & 0.160 & 0.0 & \multicolumn{1}{c}{--} & \multicolumn{1}{c}{--} & -1.000 & 0.160 & 0.0 & \multicolumn{1}{c}{--} & \multicolumn{1}{c}{--} & 0.027 & 0.005 & 100.0 & 1964.3 & 1.02 \\
$\Phi_{5,2}$ & -0.200 & 0.104 & 0.002 & 95.0 & 1557.5 & 1.01 & -1.000 & 0.040 & 0.0 & \multicolumn{1}{c}{--} & \multicolumn{1}{c}{--} & -1.000 & 0.040 & 0.0 & \multicolumn{1}{c}{--} & \multicolumn{1}{c}{--} & 0.019 & 0.001 & 96.0 & 2009.7 & 1.01 \\
$\Phi_{6,1}$ & -0.300 & 0.133 & 0.006 & 90.0 & 1454.0 & 1.02 & -1.000 & 0.090 & 0.0 & \multicolumn{1}{c}{--} & \multicolumn{1}{c}{--} & -1.000 & 0.090 & 0.0 & \multicolumn{1}{c}{--} & \multicolumn{1}{c}{--} & 0.086 & 0.004 & 92.0 & 2021.2 & 1.01 \\
\bottomrule
\end{tabular*}
\begin{tablenotes}
\footnotesize
\item RB is mean relative bias across retained replications, MSE is mean squared error, CP is empirical coverage of 95\% posterior credible intervals, ESS is effective sample size, and $\hat{R}$ is the Gelman--Rubin convergence diagnostic. Summaries use the retained-fit denominator reported in Table~\ref{tab:simulation_summary} in the main text. A dash denotes an unavailable quantity, undefined RB when the data-generating value is zero, or an unavailable sampling diagnostic for a structurally fixed parameter.
\end{tablenotes}
\end{threeparttable}
\end{table}
\clearpage
\end{landscape}

\begin{landscape}
\begin{table}[p]
\ContinuedFloat

\centering
\caption{Full parameter-level simulation results: S5 (continued).}
\label{tab:r8_s5_full_recovery_4}
\footnotesize
\sisetup{mode=math}
\setlength{\tabcolsep}{0.35pt}
\renewcommand{\arraystretch}{0.72}
\begin{threeparttable}
\begin{tabular*}{\linewidth}{@{\extracolsep{\fill}}l S[table-format=-1.3] *{4}{S[table-format=-4.3] S[table-format=3.3] S[table-format=3.1] S[table-format=5.1] S[table-format=1.2]}}
\toprule
Parameter & {True} & \multicolumn{5}{c}{CLOUD} & \multicolumn{5}{c}{TrendyOU} & \multicolumn{5}{c}{StationaryOU} & \multicolumn{5}{c}{DiagOU} \\
\cmidrule(lr){3-7}\cmidrule(lr){8-12}\cmidrule(lr){13-17}\cmidrule(l){18-22}
 & & {RB} & {MSE} & {CP} & {ESS} & {$\hat{R}$} & {RB} & {MSE} & {CP} & {ESS} & {$\hat{R}$} & {RB} & {MSE} & {CP} & {ESS} & {$\hat{R}$} & {RB} & {MSE} & {CP} & {ESS} & {$\hat{R}$} \\
\midrule
\multicolumn{22}{@{}l}{\textit{Covariate slopes $\boldsymbol\Phi$ (continued)}} \\
\midrule
$\Phi_{6,2}$ & 0.500 & 0.106 & 0.005 & 90.0 & 762.7 & 1.04 & -1.000 & 0.250 & 0.0 & \multicolumn{1}{c}{--} & \multicolumn{1}{c}{--} & -1.000 & 0.250 & 0.0 & \multicolumn{1}{c}{--} & \multicolumn{1}{c}{--} & 0.061 & 0.003 & 94.0 & 854.4 & 1.02 \\
$\Phi_{7,1}$ & 0.200 & 0.035 & 0.002 & 92.0 & 1810.1 & 1.01 & -1.000 & 0.040 & 0.0 & \multicolumn{1}{c}{--} & \multicolumn{1}{c}{--} & -1.000 & 0.040 & 0.0 & \multicolumn{1}{c}{--} & \multicolumn{1}{c}{--} & 0.004 & 0.001 & 93.0 & 3346.6 & 1.00 \\
$\Phi_{7,2}$ & 0.100 & 0.006 & 0.000 & 92.0 & 1941.3 & 1.00 & -1.000 & 0.010 & 0.0 & \multicolumn{1}{c}{--} & \multicolumn{1}{c}{--} & -1.000 & 0.010 & 0.0 & \multicolumn{1}{c}{--} & \multicolumn{1}{c}{--} & -0.027 & 0.000 & 94.0 & 3545.2 & 1.00 \\
$\Phi_{8,1}$ & -0.100 & -0.025 & 0.001 & 94.0 & 2029.1 & 1.01 & -1.000 & 0.010 & 0.0 & \multicolumn{1}{c}{--} & \multicolumn{1}{c}{--} & -1.000 & 0.010 & 0.0 & \multicolumn{1}{c}{--} & \multicolumn{1}{c}{--} & -0.065 & 0.001 & 95.0 & 3963.2 & 1.00 \\
$\Phi_{8,2}$ & -0.300 & 0.036 & 0.001 & 95.0 & 669.2 & 1.02 & -1.000 & 0.090 & 0.0 & \multicolumn{1}{c}{--} & \multicolumn{1}{c}{--} & -1.000 & 0.090 & 0.0 & \multicolumn{1}{c}{--} & \multicolumn{1}{c}{--} & -0.002 & 0.001 & 99.0 & 1110.4 & 1.01 \\
\addlinespace[0.25em]
\multicolumn{22}{@{}l}{\textit{Population slopes $\boldsymbol\alpha$}} \\
\midrule
$\alpha_{1}$ & 0.500 & 0.128 & 0.008 & 90.0 & 1031.8 & 1.01 & 0.444 & 0.060 & 22.1 & 440.9 & 1.06 & -1.000 & 0.250 & 0.0 & \multicolumn{1}{c}{--} & \multicolumn{1}{c}{--} & 0.048 & 0.004 & 98.0 & 1015.3 & 1.01 \\
$\alpha_{2}$ & -0.300 & 0.097 & 0.003 & 96.0 & 1039.8 & 1.01 & -0.073 & 0.002 & 82.4 & 1051.2 & 1.02 & -1.000 & 0.090 & 0.0 & \multicolumn{1}{c}{--} & \multicolumn{1}{c}{--} & 0.005 & 0.002 & 98.0 & 1413.0 & 1.02 \\
$\alpha_{3}$ & 0.200 & 0.069 & 0.001 & 95.0 & 1503.1 & 1.01 & 0.476 & 0.010 & 4.4 & 1069.1 & 1.01 & -1.000 & 0.040 & 0.0 & \multicolumn{1}{c}{--} & \multicolumn{1}{c}{--} & 0.039 & 0.001 & 95.0 & 2590.6 & 1.01 \\
$\alpha_{4}$ & -0.400 & 0.013 & 0.001 & 97.0 & 849.1 & 1.01 & -0.146 & 0.004 & 38.2 & 931.0 & 1.02 & -1.000 & 0.160 & 0.0 & \multicolumn{1}{c}{--} & \multicolumn{1}{c}{--} & -0.018 & 0.001 & 93.0 & 1356.6 & 1.01 \\
$\alpha_{5}$ & 0.500 & 0.106 & 0.006 & 96.0 & 1082.6 & 1.01 & 0.325 & 0.035 & 47.1 & 559.1 & 1.06 & -1.000 & 0.250 & 0.0 & \multicolumn{1}{c}{--} & \multicolumn{1}{c}{--} & 0.021 & 0.002 & 100.0 & 1126.2 & 1.02 \\
$\alpha_{6}$ & -0.300 & 0.035 & 0.002 & 99.0 & 1344.2 & 1.03 & -0.093 & 0.002 & 72.1 & 1256.8 & 1.01 & -1.000 & 0.090 & 0.0 & \multicolumn{1}{c}{--} & \multicolumn{1}{c}{--} & -0.007 & 0.001 & 98.0 & 1716.9 & 1.01 \\
$\alpha_{7}$ & 0.200 & 0.012 & 0.001 & 91.0 & 1649.0 & 1.01 & 0.418 & 0.008 & 11.8 & 1183.9 & 1.02 & -1.000 & 0.040 & 0.0 & \multicolumn{1}{c}{--} & \multicolumn{1}{c}{--} & -0.018 & 0.001 & 92.0 & 2827.3 & 1.00 \\
$\alpha_{8}$ & -0.400 & 0.024 & 0.001 & 98.0 & 730.0 & 1.02 & -0.156 & 0.005 & 39.7 & 902.3 & 1.02 & -1.000 & 0.160 & 0.0 & \multicolumn{1}{c}{--} & \multicolumn{1}{c}{--} & -0.013 & 0.001 & 99.0 & 1200.1 & 1.01 \\
\addlinespace[0.25em]
\multicolumn{22}{@{}l}{\textit{Measurement effects $\boldsymbol\beta$}} \\
\midrule
$\beta_{1,1}$ & 0.300 & 0.065 & 0.031 & 95.0 & 5381.6 & 1.00 & 0.053 & 0.034 & 95.6 & 7422.8 & 1.00 & 0.011 & 0.025 & 97.3 & 7360.7 & 1.00 & 0.049 & 0.030 & 95.0 & 9463.1 & 1.00 \\
$\beta_{1,2}$ & 0.500 & -0.084 & 0.012 & 95.0 & 5216.2 & 1.00 & -0.124 & 0.013 & 91.2 & 7061.2 & 1.00 & -0.114 & 0.015 & 91.9 & 6584.6 & 1.00 & -0.093 & 0.012 & 95.0 & 8325.5 & 1.00 \\
$\beta_{2,1}$ & 0.100 & 0.023 & 0.038 & 94.0 & 4976.9 & 1.00 & -0.124 & 0.040 & 94.1 & 6912.6 & 1.00 & -0.210 & 0.033 & 94.6 & 6863.6 & 1.00 & -0.018 & 0.037 & 93.0 & 8982.9 & 1.00 \\
$\beta_{2,2}$ & 0.200 & -0.112 & 0.011 & 91.0 & 5471.2 & 1.00 & -0.151 & 0.012 & 91.2 & 7301.5 & 1.00 & -0.124 & 0.013 & 86.5 & 7062.5 & 1.00 & -0.133 & 0.011 & 91.0 & 8977.6 & 1.00 \\
$\beta_{3,1}$ & -0.100 & -0.041 & 0.042 & 97.0 & 5004.0 & 1.00 & -0.117 & 0.034 & 98.5 & 7030.5 & 1.00 & 0.370 & 0.038 & 94.6 & 6693.2 & 1.00 & -0.038 & 0.040 & 96.0 & 9162.2 & 1.00 \\
$\beta_{3,2}$ & 0.200 & -0.206 & 0.009 & 96.0 & 5527.8 & 1.00 & -0.223 & 0.009 & 95.6 & 7581.7 & 1.00 & -0.220 & 0.008 & 97.3 & 7247.0 & 1.00 & -0.218 & 0.009 & 97.0 & 9244.6 & 1.00 \\
$\beta_{4,1}$ & 0.200 & -0.001 & 0.023 & 96.0 & 5288.5 & 1.00 & 0.036 & 0.017 & 97.1 & 7569.4 & 1.00 & -0.017 & 0.035 & 89.2 & 6871.1 & 1.00 & -0.003 & 0.023 & 96.0 & 9480.2 & 1.00 \\
$\beta_{4,2}$ & -0.100 & 0.186 & 0.007 & 96.0 & 6021.4 & 1.00 & 0.179 & 0.007 & 94.1 & 8238.4 & 1.00 & 0.219 & 0.008 & 91.9 & 8011.1 & 1.00 & 0.174 & 0.007 & 96.0 & 9954.3 & 1.00 \\
$\beta_{5,1}$ & 0.100 & 0.059 & 0.033 & 96.0 & 4985.2 & 1.00 & -0.095 & 0.030 & 94.1 & 7198.8 & 1.00 & -0.025 & 0.031 & 97.3 & 6633.9 & 1.00 & 0.051 & 0.032 & 95.0 & 9016.9 & 1.00 \\
$\beta_{5,2}$ & 0.300 & -0.140 & 0.010 & 92.0 & 5150.0 & 1.00 & -0.215 & 0.009 & 94.1 & 7135.9 & 1.00 & -0.316 & 0.015 & 94.6 & 6986.6 & 1.00 & -0.154 & 0.010 & 92.0 & 8480.2 & 1.00 \\
$\beta_{6,1}$ & -0.100 & 0.091 & 0.018 & 93.0 & 5504.2 & 1.00 & -0.132 & 0.016 & 97.1 & 7851.2 & 1.00 & -0.135 & 0.019 & 91.9 & 7191.8 & 1.00 & 0.086 & 0.018 & 94.0 & 9803.8 & 1.00 \\
$\beta_{6,2}$ & -0.200 & 0.032 & 0.005 & 94.0 & 5524.4 & 1.00 & 0.010 & 0.006 & 91.2 & 7558.3 & 1.00 & 0.081 & 0.006 & 91.9 & 6919.7 & 1.00 & 0.026 & 0.005 & 94.0 & 9384.8 & 1.00 \\
$\beta_{7,1}$ & -0.200 & -0.036 & 0.014 & 95.0 & 5780.9 & 1.00 & -0.037 & 0.013 & 94.1 & 7763.5 & 1.00 & -0.119 & 0.009 & 100.0 & 7495.8 & 1.00 & -0.037 & 0.014 & 94.0 & 9842.8 & 1.00 \\
$\beta_{7,2}$ & -0.100 & -0.020 & 0.003 & 97.0 & 5988.8 & 1.00 & 0.034 & 0.002 & 98.5 & 8152.9 & 1.00 & 0.006 & 0.003 & 97.3 & 7664.5 & 1.00 & -0.020 & 0.003 & 97.0 & 10084.8 & 1.00 \\
$\beta_{8,1}$ & 0.400 & -0.019 & 0.011 & 93.0 & 5399.2 & 1.00 & -0.004 & 0.012 & 92.6 & 7225.5 & 1.00 & -0.067 & 0.010 & 94.6 & 6921.8 & 1.00 & -0.018 & 0.011 & 93.0 & 9446.4 & 1.00 \\
$\beta_{8,2}$ & 0.100 & -0.001 & 0.003 & 93.0 & 5781.6 & 1.00 & 0.016 & 0.004 & 89.7 & 7597.0 & 1.00 & -0.011 & 0.002 & 97.3 & 7230.2 & 1.00 & 0.007 & 0.003 & 93.0 & 9681.0 & 1.00 \\
$\beta_{9,1}$ & 0.200 & 0.040 & 0.010 & 97.0 & 5415.6 & 1.00 & 0.033 & 0.010 & 95.6 & 7455.0 & 1.00 & 0.041 & 0.007 & 100.0 & 7096.4 & 1.00 & 0.040 & 0.010 & 97.0 & 9667.9 & 1.00 \\
$\beta_{9,2}$ & -0.400 & 0.002 & 0.002 & 96.0 & 5081.4 & 1.00 & 0.008 & 0.002 & 97.1 & 6560.8 & 1.00 & 0.014 & 0.003 & 94.6 & 5929.4 & 1.00 & 0.001 & 0.002 & 96.0 & 8446.3 & 1.00 \\
$\beta_{10,1}$ & -0.300 & -0.016 & 0.011 & 96.0 & 5199.0 & 1.00 & 0.003 & 0.013 & 92.6 & 7202.6 & 1.00 & 0.002 & 0.014 & 89.2 & 6665.1 & 1.00 & -0.014 & 0.011 & 95.0 & 8998.7 & 1.00 \\
$\beta_{10,2}$ & 0.200 & -0.046 & 0.004 & 94.0 & 5215.0 & 1.00 & -0.036 & 0.004 & 92.6 & 7062.9 & 1.00 & -0.074 & 0.004 & 94.6 & 6831.8 & 1.00 & -0.042 & 0.004 & 94.0 & 8909.1 & 1.00 \\
\bottomrule
\end{tabular*}
\begin{tablenotes}
\footnotesize
\item RB is mean relative bias across retained replications, MSE is mean squared error, CP is empirical coverage of 95\% posterior credible intervals, ESS is effective sample size, and $\hat{R}$ is the Gelman--Rubin convergence diagnostic. Summaries use the retained-fit denominator reported in Table~\ref{tab:simulation_summary} in the main text. A dash denotes an unavailable quantity, undefined RB when the data-generating value is zero, or an unavailable sampling diagnostic for a structurally fixed parameter.
\end{tablenotes}
\end{threeparttable}
\end{table}
\clearpage
\end{landscape}

\begin{landscape}
\begin{table}[p]
\ContinuedFloat

\centering
\caption{Full parameter-level simulation results: S5 (continued).}
\label{tab:r8_s5_full_recovery_5}
\footnotesize
\sisetup{mode=math}
\setlength{\tabcolsep}{0.35pt}
\renewcommand{\arraystretch}{0.72}
\begin{threeparttable}
\begin{tabular*}{\linewidth}{@{\extracolsep{\fill}}l S[table-format=-1.3] *{4}{S[table-format=-4.3] S[table-format=3.3] S[table-format=3.1] S[table-format=5.1] S[table-format=1.2]}}
\toprule
Parameter & {True} & \multicolumn{5}{c}{CLOUD} & \multicolumn{5}{c}{TrendyOU} & \multicolumn{5}{c}{StationaryOU} & \multicolumn{5}{c}{DiagOU} \\
\cmidrule(lr){3-7}\cmidrule(lr){8-12}\cmidrule(lr){13-17}\cmidrule(l){18-22}
 & & {RB} & {MSE} & {CP} & {ESS} & {$\hat{R}$} & {RB} & {MSE} & {CP} & {ESS} & {$\hat{R}$} & {RB} & {MSE} & {CP} & {ESS} & {$\hat{R}$} & {RB} & {MSE} & {CP} & {ESS} & {$\hat{R}$} \\
\midrule
\multicolumn{22}{@{}l}{\textit{Measurement effects $\boldsymbol\beta$ (continued)}} \\
\midrule
$\beta_{11,1}$ & 0.100 & 0.164 & 0.016 & 89.0 & 5730.8 & 1.00 & -0.055 & 0.015 & 89.7 & 7869.5 & 1.00 & -0.052 & 0.012 & 91.9 & 7480.0 & 1.00 & 0.165 & 0.016 & 89.0 & 9769.9 & 1.00 \\
$\beta_{11,2}$ & -0.100 & 0.011 & 0.002 & 98.0 & 5842.7 & 1.00 & 0.020 & 0.003 & 97.1 & 8034.1 & 1.00 & 0.169 & 0.003 & 97.3 & 7717.9 & 1.00 & 0.007 & 0.002 & 98.0 & 9799.5 & 1.00 \\
$\beta_{12,1}$ & 0.200 & -0.005 & 0.011 & 94.0 & 5437.6 & 1.00 & -0.051 & 0.011 & 92.6 & 7623.9 & 1.00 & -0.125 & 0.015 & 89.2 & 7084.1 & 1.00 & -0.005 & 0.011 & 95.0 & 9520.7 & 1.00 \\
$\beta_{12,2}$ & 0.300 & -0.037 & 0.003 & 96.0 & 5070.4 & 1.00 & -0.053 & 0.003 & 95.6 & 6871.7 & 1.00 & -0.067 & 0.003 & 97.3 & 6582.6 & 1.00 & -0.034 & 0.003 & 96.0 & 8472.5 & 1.00 \\
$\beta_{13,1}$ & 0.300 & -0.063 & 0.037 & 98.0 & 5437.7 & 1.00 & -0.087 & 0.038 & 97.1 & 7581.1 & 1.00 & -0.049 & 0.044 & 91.9 & 7168.8 & 1.00 & -0.072 & 0.036 & 98.0 & 9547.1 & 1.00 \\
$\beta_{13,2}$ & 0.500 & -0.015 & 0.013 & 91.0 & 5088.2 & 1.00 & -0.037 & 0.015 & 86.8 & 6755.3 & 1.00 & -0.037 & 0.013 & 91.9 & 6563.3 & 1.00 & -0.027 & 0.013 & 90.0 & 7964.1 & 1.00 \\
$\beta_{14,1}$ & 0.100 & -0.137 & 0.038 & 93.0 & 4998.4 & 1.00 & -0.213 & 0.044 & 88.2 & 6918.4 & 1.00 & -0.348 & 0.047 & 83.8 & 7200.3 & 1.00 & -0.136 & 0.036 & 93.0 & 9025.1 & 1.00 \\
$\beta_{14,2}$ & 0.200 & -0.030 & 0.009 & 98.0 & 5452.4 & 1.00 & -0.078 & 0.010 & 94.1 & 7361.6 & 1.00 & -0.006 & 0.009 & 94.6 & 7220.7 & 1.00 & -0.056 & 0.009 & 97.0 & 8922.8 & 1.00 \\
$\beta_{15,1}$ & -0.100 & -0.076 & 0.043 & 93.0 & 5035.6 & 1.00 & -0.128 & 0.046 & 89.7 & 7029.0 & 1.00 & -0.191 & 0.045 & 94.6 & 7038.9 & 1.00 & -0.100 & 0.041 & 93.0 & 9240.1 & 1.00 \\
$\beta_{15,2}$ & 0.200 & -0.158 & 0.013 & 91.0 & 5551.9 & 1.00 & -0.186 & 0.013 & 91.2 & 7411.0 & 1.00 & -0.043 & 0.015 & 83.8 & 7296.7 & 1.00 & -0.172 & 0.013 & 92.0 & 9248.6 & 1.00 \\
$\beta_{16,1}$ & 0.200 & 0.021 & 0.027 & 93.0 & 5245.7 & 1.00 & -0.063 & 0.028 & 95.6 & 7548.3 & 1.00 & -0.212 & 0.022 & 94.6 & 6875.6 & 1.00 & 0.018 & 0.027 & 93.0 & 9197.6 & 1.00 \\
$\beta_{16,2}$ & -0.100 & 0.105 & 0.006 & 98.0 & 5956.5 & 1.00 & 0.139 & 0.007 & 98.5 & 8165.0 & 1.00 & 0.187 & 0.006 & 100.0 & 7777.3 & 1.00 & 0.101 & 0.006 & 98.0 & 9823.6 & 1.00 \\
$\beta_{17,1}$ & 0.100 & -0.415 & 0.031 & 96.0 & 4937.8 & 1.00 & -0.393 & 0.029 & 98.5 & 7220.6 & 1.00 & -0.413 & 0.034 & 97.3 & 6674.7 & 1.00 & -0.422 & 0.031 & 97.0 & 8720.8 & 1.00 \\
$\beta_{17,2}$ & 0.300 & -0.125 & 0.009 & 97.0 & 5140.6 & 1.00 & -0.217 & 0.012 & 89.7 & 6962.9 & 1.00 & -0.277 & 0.013 & 91.9 & 6605.5 & 1.00 & -0.134 & 0.009 & 97.0 & 8209.6 & 1.00 \\
$\beta_{18,1}$ & -0.100 & 0.103 & 0.013 & 98.0 & 5385.6 & 1.00 & 0.199 & 0.012 & 98.5 & 7597.4 & 1.00 & 0.042 & 0.013 & 97.3 & 7323.5 & 1.00 & 0.106 & 0.013 & 97.0 & 9480.5 & 1.00 \\
$\beta_{18,2}$ & -0.200 & 0.030 & 0.004 & 97.0 & 5535.9 & 1.00 & 0.027 & 0.005 & 94.1 & 7537.1 & 1.00 & 0.054 & 0.005 & 94.6 & 7153.1 & 1.00 & 0.029 & 0.004 & 97.0 & 9144.8 & 1.00 \\
$\beta_{19,1}$ & -0.200 & 0.010 & 0.014 & 97.0 & 5770.4 & 1.00 & 0.011 & 0.015 & 95.6 & 7960.1 & 1.00 & -0.000 & 0.011 & 100.0 & 7649.4 & 1.00 & 0.012 & 0.014 & 97.0 & 9891.5 & 1.00 \\
$\beta_{19,2}$ & -0.100 & 0.094 & 0.003 & 96.0 & 5972.7 & 1.00 & 0.176 & 0.003 & 95.6 & 8064.4 & 1.00 & 0.029 & 0.004 & 89.2 & 7922.8 & 1.00 & 0.093 & 0.003 & 95.0 & 10118.5 & 1.00 \\
$\beta_{20,1}$ & 0.400 & 0.006 & 0.009 & 97.0 & 5491.0 & 1.00 & -0.004 & 0.009 & 95.6 & 7302.4 & 1.00 & 0.004 & 0.008 & 100.0 & 6860.8 & 1.00 & 0.003 & 0.009 & 97.0 & 9427.2 & 1.00 \\
$\beta_{20,2}$ & 0.100 & 0.066 & 0.002 & 98.0 & 5842.1 & 1.00 & 0.069 & 0.002 & 97.1 & 7697.5 & 1.00 & 0.106 & 0.003 & 94.6 & 7371.2 & 1.00 & 0.072 & 0.002 & 98.0 & 9547.3 & 1.00 \\
$\beta_{21,1}$ & 0.200 & -0.051 & 0.008 & 96.0 & 5454.1 & 1.00 & -0.071 & 0.009 & 95.6 & 7462.8 & 1.00 & -0.033 & 0.010 & 94.6 & 7078.2 & 1.00 & -0.055 & 0.008 & 97.0 & 9686.1 & 1.00 \\
$\beta_{21,2}$ & -0.400 & 0.024 & 0.003 & 96.0 & 5067.7 & 1.00 & 0.025 & 0.003 & 95.6 & 6487.7 & 1.00 & 0.028 & 0.002 & 100.0 & 5739.2 & 1.00 & 0.022 & 0.003 & 97.0 & 8304.1 & 1.00 \\
$\beta_{22,1}$ & -0.300 & 0.028 & 0.010 & 97.0 & 5151.2 & 1.00 & 0.012 & 0.011 & 95.6 & 7377.3 & 1.00 & 0.008 & 0.009 & 97.3 & 6810.9 & 1.00 & 0.030 & 0.010 & 97.0 & 9015.0 & 1.00 \\
$\beta_{22,2}$ & 0.200 & -0.055 & 0.004 & 93.0 & 5310.6 & 1.00 & -0.052 & 0.004 & 94.1 & 7115.2 & 1.00 & -0.056 & 0.005 & 91.9 & 6662.4 & 1.00 & -0.053 & 0.004 & 93.0 & 9036.6 & 1.00 \\
$\beta_{23,1}$ & 0.100 & -0.086 & 0.012 & 96.0 & 5751.0 & 1.00 & -0.024 & 0.014 & 98.5 & 8052.5 & 1.00 & -0.051 & 0.011 & 100.0 & 7611.1 & 1.00 & -0.092 & 0.012 & 95.0 & 9914.5 & 1.00 \\
$\beta_{23,2}$ & -0.100 & -0.037 & 0.003 & 95.0 & 5905.9 & 1.00 & -0.018 & 0.003 & 95.6 & 8077.5 & 1.00 & 0.034 & 0.005 & 91.9 & 7731.6 & 1.00 & -0.039 & 0.003 & 94.0 & 9967.8 & 1.00 \\
$\beta_{24,1}$ & 0.200 & 0.067 & 0.010 & 94.0 & 5404.9 & 1.00 & 0.034 & 0.010 & 95.6 & 7774.1 & 1.00 & 0.005 & 0.011 & 91.9 & 7140.1 & 1.00 & 0.065 & 0.010 & 94.0 & 9532.5 & 1.00 \\
$\beta_{24,2}$ & 0.300 & -0.023 & 0.003 & 96.0 & 5246.9 & 1.00 & -0.028 & 0.003 & 95.6 & 6967.9 & 1.00 & -0.021 & 0.003 & 94.6 & 6665.2 & 1.00 & -0.019 & 0.003 & 96.0 & 8672.2 & 1.00 \\
\addlinespace[0.25em]
\multicolumn{22}{@{}l}{\textit{Random-effect scales $\boldsymbol\sigma_b$}} \\
\midrule
$\sigma_{b,1}$ & 3.700 & -0.064 & 0.127 & 81.0 & 1642.6 & 1.00 & -0.086 & 0.171 & 76.5 & 1803.0 & 1.00 & -0.062 & 0.133 & 86.5 & 1755.2 & 1.01 & -0.079 & 0.156 & 79.0 & 1541.2 & 1.01 \\
$\sigma_{b,2}$ & 3.400 & -0.079 & 0.123 & 85.0 & 1326.9 & 1.01 & -0.103 & 0.173 & 73.5 & 1298.2 & 1.01 & -0.103 & 0.171 & 78.4 & 1321.2 & 1.01 & -0.116 & 0.209 & 75.0 & 790.6 & 1.02 \\
$\sigma_{b,3}$ & 4.800 & -0.115 & 0.359 & 64.0 & 1726.7 & 1.00 & -0.141 & 0.529 & 42.6 & 1624.7 & 1.01 & -0.131 & 0.458 & 59.5 & 1814.9 & 1.00 & -0.137 & 0.485 & 51.0 & 1244.7 & 1.01 \\
$\sigma_{b,4}$ & 3.100 & -0.035 & 0.053 & 93.0 & 1579.7 & 1.01 & -0.050 & 0.067 & 91.2 & 1795.4 & 1.01 & -0.055 & 0.067 & 83.8 & 1773.6 & 1.01 & -0.047 & 0.062 & 93.0 & 1547.6 & 1.02 \\
$\sigma_{b,5}$ & 4.000 & -0.127 & 0.320 & 59.0 & 1331.0 & 1.01 & -0.161 & 0.483 & 47.1 & 1613.7 & 1.01 & -0.163 & 0.488 & 43.2 & 1572.5 & 1.01 & -0.151 & 0.429 & 46.0 & 1211.6 & 1.03 \\
\bottomrule
\end{tabular*}
\begin{tablenotes}
\footnotesize
\item RB is mean relative bias across retained replications, MSE is mean squared error, CP is empirical coverage of 95\% posterior credible intervals, ESS is effective sample size, and $\hat{R}$ is the Gelman--Rubin convergence diagnostic. Summaries use the retained-fit denominator reported in Table~\ref{tab:simulation_summary} in the main text. A dash denotes an unavailable quantity, undefined RB when the data-generating value is zero, or an unavailable sampling diagnostic for a structurally fixed parameter.
\end{tablenotes}
\end{threeparttable}
\end{table}
\clearpage
\end{landscape}

\begin{landscape}
\begin{table}[p]
\ContinuedFloat

\centering
\caption{Full parameter-level simulation results: S5 (continued).}
\label{tab:r8_s5_full_recovery_6}
\footnotesize
\sisetup{mode=math}
\setlength{\tabcolsep}{0.35pt}
\renewcommand{\arraystretch}{0.72}
\begin{threeparttable}
\begin{tabular*}{\linewidth}{@{\extracolsep{\fill}}l S[table-format=-1.3] *{4}{S[table-format=-4.3] S[table-format=3.3] S[table-format=3.1] S[table-format=5.1] S[table-format=1.2]}}
\toprule
Parameter & {True} & \multicolumn{5}{c}{CLOUD} & \multicolumn{5}{c}{TrendyOU} & \multicolumn{5}{c}{StationaryOU} & \multicolumn{5}{c}{DiagOU} \\
\cmidrule(lr){3-7}\cmidrule(lr){8-12}\cmidrule(lr){13-17}\cmidrule(l){18-22}
 & & {RB} & {MSE} & {CP} & {ESS} & {$\hat{R}$} & {RB} & {MSE} & {CP} & {ESS} & {$\hat{R}$} & {RB} & {MSE} & {CP} & {ESS} & {$\hat{R}$} & {RB} & {MSE} & {CP} & {ESS} & {$\hat{R}$} \\
\midrule
\multicolumn{22}{@{}l}{\textit{Random-effect scales $\boldsymbol\sigma_b$ (continued)}} \\
\midrule
$\sigma_{b,6}$ & 6.100 & -0.098 & 0.413 & 48.0 & 1218.8 & 1.01 & -0.105 & 0.465 & 41.2 & 1641.2 & 1.01 & -0.115 & 0.547 & 35.1 & 1534.3 & 1.01 & -0.103 & 0.449 & 47.0 & 1578.5 & 1.01 \\
$\sigma_{b,7}$ & 5.100 & -0.051 & 0.115 & 85.0 & 1300.7 & 1.01 & -0.054 & 0.118 & 82.4 & 1677.5 & 1.01 & -0.052 & 0.121 & 81.1 & 1555.4 & 1.00 & -0.052 & 0.115 & 85.0 & 1961.4 & 1.01 \\
$\sigma_{b,8}$ & 1.700 & -0.012 & 0.009 & 98.0 & 914.0 & 1.01 & -0.009 & 0.008 & 98.5 & 1066.8 & 1.01 & 0.005 & 0.009 & 97.3 & 991.2 & 1.02 & -0.023 & 0.012 & 95.0 & 1055.1 & 1.01 \\
$\sigma_{b,9}$ & 2.500 & -0.005 & 0.017 & 97.0 & 1313.2 & 1.01 & -0.006 & 0.017 & 95.6 & 1571.7 & 1.01 & -0.001 & 0.012 & 100.0 & 1383.4 & 1.01 & -0.008 & 0.017 & 95.0 & 1706.7 & 1.00 \\
$\sigma_{b,10}$ & 3.300 & -0.044 & 0.041 & 88.0 & 1109.0 & 1.01 & -0.042 & 0.041 & 91.2 & 1450.7 & 1.01 & -0.053 & 0.057 & 83.8 & 1398.5 & 1.00 & -0.049 & 0.047 & 86.0 & 1601.5 & 1.01 \\
$\sigma_{b,11}$ & 4.200 & -0.034 & 0.059 & 86.0 & 1310.8 & 1.01 & -0.039 & 0.065 & 80.9 & 1735.0 & 1.01 & -0.033 & 0.067 & 83.8 & 1651.8 & 1.00 & -0.034 & 0.060 & 87.0 & 1967.5 & 1.00 \\
$\sigma_{b,12}$ & 2.800 & -0.023 & 0.027 & 90.0 & 1267.4 & 1.01 & -0.024 & 0.033 & 85.3 & 1534.7 & 1.00 & -0.045 & 0.045 & 81.1 & 1507.3 & 1.01 & -0.027 & 0.030 & 87.0 & 1711.4 & 1.01 \\
$\sigma_{b,13}$ & 3.700 & -0.052 & 0.106 & 89.0 & 1677.6 & 1.01 & -0.076 & 0.146 & 85.3 & 1874.9 & 1.01 & -0.076 & 0.145 & 81.1 & 1863.4 & 1.00 & -0.071 & 0.138 & 85.0 & 1602.0 & 1.01 \\
$\sigma_{b,14}$ & 3.400 & -0.088 & 0.139 & 84.0 & 1380.5 & 1.01 & -0.098 & 0.167 & 77.9 & 1323.8 & 1.01 & -0.110 & 0.180 & 78.4 & 1600.2 & 1.00 & -0.128 & 0.246 & 73.0 & 875.7 & 1.02 \\
$\sigma_{b,15}$ & 4.800 & -0.120 & 0.416 & 58.0 & 1775.2 & 1.00 & -0.144 & 0.542 & 44.1 & 1728.1 & 1.01 & -0.148 & 0.585 & 45.9 & 1968.7 & 1.01 & -0.144 & 0.554 & 51.0 & 1366.0 & 1.02 \\
$\sigma_{b,16}$ & 3.100 & -0.035 & 0.050 & 97.0 & 1659.0 & 1.01 & -0.047 & 0.072 & 89.7 & 1785.1 & 1.01 & -0.059 & 0.077 & 91.9 & 1817.6 & 1.01 & -0.041 & 0.060 & 94.0 & 1584.2 & 1.01 \\
$\sigma_{b,17}$ & 4.000 & -0.128 & 0.321 & 69.0 & 1404.1 & 1.01 & -0.161 & 0.474 & 44.1 & 1663.5 & 1.01 & -0.151 & 0.433 & 48.6 & 1613.7 & 1.01 & -0.142 & 0.384 & 58.0 & 1307.8 & 1.02 \\
$\sigma_{b,18}$ & 6.100 & -0.100 & 0.417 & 38.0 & 1235.5 & 1.01 & -0.114 & 0.538 & 26.5 & 1588.7 & 1.01 & -0.109 & 0.488 & 37.8 & 1465.0 & 1.01 & -0.102 & 0.446 & 38.0 & 1596.4 & 1.01 \\
$\sigma_{b,19}$ & 5.100 & -0.056 & 0.121 & 83.0 & 1310.1 & 1.01 & -0.063 & 0.143 & 75.0 & 1667.4 & 1.00 & -0.064 & 0.141 & 75.7 & 1589.8 & 1.01 & -0.055 & 0.120 & 83.0 & 2032.0 & 1.01 \\
$\sigma_{b,20}$ & 1.700 & -0.009 & 0.013 & 92.0 & 934.2 & 1.01 & -0.000 & 0.013 & 95.6 & 1104.2 & 1.01 & -0.014 & 0.011 & 97.3 & 934.5 & 1.01 & -0.023 & 0.017 & 91.0 & 1071.3 & 1.01 \\
$\sigma_{b,21}$ & 2.500 & -0.014 & 0.016 & 96.0 & 1368.5 & 1.01 & -0.017 & 0.017 & 92.6 & 1586.8 & 1.01 & -0.014 & 0.019 & 97.3 & 1419.9 & 1.01 & -0.019 & 0.017 & 94.0 & 1772.7 & 1.00 \\
$\sigma_{b,22}$ & 3.300 & -0.033 & 0.031 & 90.0 & 1122.8 & 1.01 & -0.032 & 0.030 & 92.6 & 1493.8 & 1.01 & -0.042 & 0.042 & 91.9 & 1430.8 & 1.00 & -0.039 & 0.037 & 89.0 & 1515.9 & 1.01 \\
$\sigma_{b,23}$ & 4.200 & -0.028 & 0.043 & 94.0 & 1295.9 & 1.01 & -0.032 & 0.046 & 91.2 & 1755.1 & 1.00 & -0.039 & 0.052 & 89.2 & 1648.1 & 1.01 & -0.029 & 0.044 & 95.0 & 1981.9 & 1.00 \\
$\sigma_{b,24}$ & 2.800 & -0.016 & 0.019 & 98.0 & 1234.2 & 1.01 & -0.023 & 0.023 & 92.6 & 1590.4 & 1.01 & -0.029 & 0.028 & 86.5 & 1643.5 & 1.00 & -0.021 & 0.021 & 95.0 & 1620.7 & 1.01 \\
\addlinespace[0.25em]
\multicolumn{22}{@{}l}{\textit{Thresholds $\boldsymbol\theta$}} \\
\midrule
$\theta_{1}$ & 2.300 & -0.068 & 0.072 & 93.0 & 2083.0 & 1.01 & -0.071 & 0.073 & 91.2 & 2650.4 & 1.00 & 0.375 & 0.804 & 5.4 & 1352.7 & 1.01 & -0.079 & 0.081 & 93.0 & 3160.5 & 1.00 \\
$\theta_{6,1}$ & -4.000 & -0.034 & 0.114 & 93.0 & 640.9 & 1.01 & -0.046 & 0.134 & 92.6 & 1157.6 & 1.01 & -0.269 & 1.228 & 0.0 & 976.6 & 1.01 & -0.036 & 0.119 & 93.0 & 1226.4 & 1.01 \\
$\theta_{6,2}$ & -1.000 & -0.007 & 0.077 & 95.0 & 563.0 & 1.02 & -0.030 & 0.083 & 92.6 & 1118.6 & 1.01 & -0.919 & 0.912 & 8.1 & 1048.8 & 1.01 & -0.009 & 0.077 & 95.0 & 1163.1 & 1.01 \\
$\theta_{6,3}$ & 2.700 & -0.064 & 0.099 & 91.0 & 608.5 & 1.02 & -0.069 & 0.102 & 91.2 & 1187.0 & 1.01 & 0.269 & 0.601 & 29.7 & 1144.9 & 1.01 & -0.068 & 0.102 & 91.0 & 1215.2 & 1.01 \\
$\theta_{2}$ & 1.500 & -0.057 & 0.060 & 94.0 & 2218.1 & 1.01 & -0.060 & 0.057 & 92.6 & 2602.0 & 1.00 & 0.819 & 1.595 & 0.0 & 852.0 & 1.02 & -0.079 & 0.066 & 93.0 & 2957.7 & 1.00 \\
$\theta_{7,1}$ & -5.500 & -0.021 & 0.078 & 91.0 & 849.4 & 1.02 & -0.020 & 0.070 & 94.1 & 1179.1 & 1.01 & 0.054 & 0.185 & 81.1 & 1084.2 & 1.01 & -0.020 & 0.077 & 92.0 & 1546.9 & 1.01 \\
$\theta_{7,2}$ & -2.500 & -0.025 & 0.060 & 95.0 & 690.8 & 1.02 & -0.024 & 0.047 & 97.1 & 962.0 & 1.01 & 0.136 & 0.192 & 67.6 & 864.3 & 1.01 & -0.025 & 0.059 & 94.0 & 1268.3 & 1.01 \\
$\theta_{7,3}$ & 2.600 & -0.004 & 0.056 & 95.0 & 758.9 & 1.02 & -0.010 & 0.058 & 95.6 & 1065.0 & 1.01 & -0.194 & 0.324 & 48.6 & 866.2 & 1.01 & -0.004 & 0.056 & 95.0 & 1386.0 & 1.01 \\
$\theta_{3}$ & 0.800 & -0.110 & 0.064 & 95.0 & 1773.1 & 1.00 & -0.080 & 0.059 & 95.6 & 2515.7 & 1.00 & 1.448 & 1.446 & 0.0 & 1102.2 & 1.01 & -0.124 & 0.064 & 95.0 & 3348.1 & 1.00 \\
$\theta_{8,1}$ & -4.500 & -0.009 & 0.033 & 97.0 & 1975.7 & 1.00 & -0.004 & 0.032 & 100.0 & 2390.6 & 1.01 & 0.148 & 0.477 & 13.5 & 1486.7 & 1.01 & -0.009 & 0.035 & 98.0 & 3169.3 & 1.00 \\
$\theta_{8,2}$ & -1.500 & -0.002 & 0.017 & 95.0 & 3095.4 & 1.00 & -0.005 & 0.014 & 97.1 & 4039.5 & 1.00 & 0.448 & 0.476 & 0.0 & 1673.6 & 1.00 & -0.003 & 0.017 & 94.0 & 5087.8 & 1.00 \\
$\theta_{8,3}$ & 2.000 & -0.018 & 0.023 & 93.0 & 2394.3 & 1.00 & -0.013 & 0.025 & 91.2 & 2822.5 & 1.00 & -0.354 & 0.524 & 0.0 & 1938.8 & 1.01 & -0.016 & 0.023 & 93.0 & 3885.7 & 1.00 \\
$\theta_{4}$ & 1.200 & -0.021 & 0.035 & 96.0 & 2041.3 & 1.00 & -0.047 & 0.031 & 98.5 & 3053.9 & 1.01 & -0.697 & 0.736 & 0.0 & 1706.7 & 1.01 & -0.026 & 0.034 & 96.0 & 3815.7 & 1.00 \\
$\theta_{9,1}$ & -3.000 & 0.006 & 0.029 & 91.0 & 1659.5 & 1.01 & 0.010 & 0.033 & 89.7 & 2273.4 & 1.00 & 0.199 & 0.387 & 5.4 & 1609.1 & 1.01 & 0.006 & 0.029 & 91.0 & 3088.8 & 1.00 \\
\bottomrule
\end{tabular*}
\begin{tablenotes}
\footnotesize
\item RB is mean relative bias across retained replications, MSE is mean squared error, CP is empirical coverage of 95\% posterior credible intervals, ESS is effective sample size, and $\hat{R}$ is the Gelman--Rubin convergence diagnostic. Summaries use the retained-fit denominator reported in Table~\ref{tab:simulation_summary} in the main text. A dash denotes an unavailable quantity, undefined RB when the data-generating value is zero, or an unavailable sampling diagnostic for a structurally fixed parameter.
\end{tablenotes}
\end{threeparttable}
\end{table}
\clearpage
\end{landscape}

\begin{landscape}
\begin{table}[p]
\ContinuedFloat

\centering
\caption{Full parameter-level simulation results: S5 (continued).}
\label{tab:r8_s5_full_recovery_7}
\footnotesize
\sisetup{mode=math}
\setlength{\tabcolsep}{0.35pt}
\renewcommand{\arraystretch}{0.72}
\begin{threeparttable}
\begin{tabular*}{\linewidth}{@{\extracolsep{\fill}}l S[table-format=-1.3] *{4}{S[table-format=-4.3] S[table-format=3.3] S[table-format=3.1] S[table-format=5.1] S[table-format=1.2]}}
\toprule
Parameter & {True} & \multicolumn{5}{c}{CLOUD} & \multicolumn{5}{c}{TrendyOU} & \multicolumn{5}{c}{StationaryOU} & \multicolumn{5}{c}{DiagOU} \\
\cmidrule(lr){3-7}\cmidrule(lr){8-12}\cmidrule(lr){13-17}\cmidrule(l){18-22}
 & & {RB} & {MSE} & {CP} & {ESS} & {$\hat{R}$} & {RB} & {MSE} & {CP} & {ESS} & {$\hat{R}$} & {RB} & {MSE} & {CP} & {ESS} & {$\hat{R}$} & {RB} & {MSE} & {CP} & {ESS} & {$\hat{R}$} \\
\midrule
\multicolumn{22}{@{}l}{\textit{Thresholds $\boldsymbol\theta$ (continued)}} \\
\midrule
$\theta_{9,2}$ & 0.000 & \multicolumn{1}{c}{--} & 0.014 & 98.0 & 1822.6 & 1.01 & \multicolumn{1}{c}{--} & 0.015 & 98.5 & 2651.8 & 1.00 & \multicolumn{1}{c}{--} & 0.334 & 2.7 & 1920.3 & 1.00 & \multicolumn{1}{c}{--} & 0.014 & 98.0 & 3550.2 & 1.00 \\
$\theta_{9,3}$ & 3.000 & -0.000 & 0.021 & 98.0 & 1795.2 & 1.01 & -0.003 & 0.025 & 98.5 & 2387.1 & 1.00 & -0.178 & 0.312 & 13.5 & 1893.5 & 1.00 & 0.000 & 0.021 & 98.0 & 3210.4 & 1.00 \\
$\theta_{5}$ & 2.100 & -0.109 & 0.104 & 81.0 & 1887.0 & 1.00 & -0.145 & 0.145 & 76.5 & 2812.3 & 1.00 & -0.636 & 1.834 & 0.0 & 1466.7 & 1.01 & -0.125 & 0.118 & 79.0 & 3025.9 & 1.00 \\
$\theta_{10,1}$ & -6.000 & -0.024 & 0.075 & 96.0 & 1318.4 & 1.01 & -0.026 & 0.084 & 95.6 & 1799.8 & 1.01 & -0.196 & 1.422 & 0.0 & 1427.5 & 1.01 & -0.025 & 0.074 & 95.0 & 2298.5 & 1.01 \\
$\theta_{10,2}$ & -2.000 & -0.031 & 0.031 & 96.0 & 1428.2 & 1.01 & -0.024 & 0.031 & 98.5 & 2288.4 & 1.01 & -0.545 & 1.211 & 0.0 & 1819.3 & 1.01 & -0.032 & 0.031 & 96.0 & 2724.1 & 1.00 \\
$\theta_{10,3}$ & 1.500 & -0.017 & 0.031 & 93.0 & 1622.3 & 1.01 & -0.012 & 0.033 & 92.6 & 2723.7 & 1.00 & 0.683 & 1.096 & 0.0 & 1981.3 & 1.00 & -0.016 & 0.031 & 93.0 & 3104.8 & 1.00 \\
$\theta_{13}$ & 2.300 & -0.032 & 0.073 & 94.0 & 2043.6 & 1.00 & -0.035 & 0.073 & 94.1 & 2671.0 & 1.00 & 0.297 & 0.538 & 24.3 & 1567.5 & 1.02 & -0.046 & 0.077 & 94.0 & 3112.9 & 1.00 \\
$\theta_{11,1}$ & -5.000 & -0.011 & 0.064 & 95.0 & 1003.3 & 1.01 & -0.012 & 0.070 & 92.6 & 1547.8 & 1.01 & -0.168 & 0.763 & 10.8 & 1311.9 & 1.00 & -0.010 & 0.064 & 95.0 & 1946.4 & 1.00 \\
$\theta_{11,2}$ & -1.000 & -0.008 & 0.049 & 95.0 & 791.8 & 1.01 & 0.008 & 0.050 & 94.1 & 1270.8 & 1.01 & -0.761 & 0.622 & 5.4 & 1136.9 & 1.01 & -0.006 & 0.049 & 94.0 & 1563.4 & 1.01 \\
$\theta_{11,3}$ & 2.500 & -0.012 & 0.047 & 94.0 & 879.9 & 1.01 & -0.023 & 0.049 & 92.6 & 1390.6 & 1.01 & 0.287 & 0.559 & 8.1 & 1354.0 & 1.01 & -0.011 & 0.047 & 95.0 & 1707.7 & 1.01 \\
$\theta_{14}$ & 1.500 & -0.053 & 0.046 & 96.0 & 2193.4 & 1.01 & -0.033 & 0.044 & 94.1 & 2661.8 & 1.01 & 0.699 & 1.166 & 0.0 & 1111.5 & 1.02 & -0.078 & 0.050 & 95.0 & 3057.3 & 1.00 \\
$\theta_{12,1}$ & -4.000 & -0.016 & 0.037 & 96.0 & 1507.6 & 1.01 & -0.016 & 0.033 & 97.1 & 2180.9 & 1.00 & -0.238 & 0.939 & 0.0 & 1682.9 & 1.01 & -0.016 & 0.036 & 94.0 & 2840.9 & 1.00 \\
$\theta_{12,2}$ & -0.500 & -0.050 & 0.023 & 93.0 & 1656.4 & 1.01 & -0.052 & 0.025 & 94.1 & 2679.6 & 1.00 & -1.759 & 0.793 & 0.0 & 2077.4 & 1.01 & -0.050 & 0.022 & 93.0 & 3343.2 & 1.00 \\
$\theta_{12,3}$ & 2.000 & 0.001 & 0.027 & 96.0 & 1792.2 & 1.00 & -0.002 & 0.029 & 94.1 & 2832.1 & 1.00 & 0.425 & 0.753 & 0.0 & 2180.3 & 1.00 & 0.002 & 0.027 & 96.0 & 3471.0 & 1.00 \\
$\theta_{15}$ & 0.800 & -0.079 & 0.063 & 93.0 & 1774.9 & 1.00 & -0.087 & 0.074 & 91.2 & 2501.5 & 1.00 & 1.205 & 1.029 & 5.4 & 1287.5 & 1.02 & -0.098 & 0.062 & 94.0 & 3424.1 & 1.00 \\
$\theta_{18,1}$ & -4.000 & -0.041 & 0.105 & 91.0 & 655.9 & 1.02 & -0.048 & 0.125 & 88.2 & 1170.6 & 1.01 & -0.258 & 1.150 & 8.1 & 942.2 & 1.02 & -0.043 & 0.109 & 89.0 & 1199.0 & 1.01 \\
$\theta_{18,2}$ & -1.000 & -0.020 & 0.077 & 94.0 & 577.3 & 1.02 & -0.010 & 0.081 & 94.1 & 1136.9 & 1.01 & -0.897 & 0.882 & 16.2 & 1021.7 & 1.02 & -0.020 & 0.077 & 94.0 & 1130.8 & 1.01 \\
$\theta_{18,3}$ & 2.700 & -0.059 & 0.104 & 89.0 & 627.5 & 1.02 & -0.078 & 0.128 & 83.8 & 1202.4 & 1.01 & 0.264 & 0.600 & 35.1 & 1150.6 & 1.01 & -0.061 & 0.110 & 89.0 & 1177.1 & 1.01 \\
$\theta_{16}$ & 1.200 & -0.004 & 0.032 & 94.0 & 2056.3 & 1.00 & -0.001 & 0.031 & 95.6 & 3067.8 & 1.00 & -0.667 & 0.684 & 8.1 & 1855.9 & 1.01 & -0.006 & 0.031 & 94.0 & 3777.1 & 1.00 \\
$\theta_{19,1}$ & -5.500 & -0.019 & 0.096 & 94.0 & 829.6 & 1.02 & -0.019 & 0.102 & 92.6 & 1184.3 & 1.01 & 0.060 & 0.189 & 86.5 & 1151.3 & 1.01 & -0.018 & 0.097 & 94.0 & 1577.6 & 1.01 \\
$\theta_{19,2}$ & -2.500 & -0.026 & 0.065 & 93.0 & 675.6 & 1.02 & -0.019 & 0.069 & 92.6 & 963.1 & 1.01 & 0.144 & 0.199 & 62.2 & 935.8 & 1.01 & -0.026 & 0.066 & 94.0 & 1288.4 & 1.01 \\
$\theta_{19,3}$ & 2.600 & -0.016 & 0.068 & 95.0 & 739.0 & 1.02 & -0.040 & 0.081 & 89.7 & 1052.6 & 1.01 & -0.187 & 0.299 & 51.4 & 957.3 & 1.01 & -0.014 & 0.067 & 94.0 & 1404.7 & 1.01 \\
$\theta_{17}$ & 2.100 & -0.090 & 0.092 & 87.0 & 1948.5 & 1.00 & -0.126 & 0.134 & 82.4 & 2831.1 & 1.00 & -0.612 & 1.729 & 0.0 & 1537.3 & 1.01 & -0.101 & 0.099 & 86.0 & 2964.3 & 1.01 \\
$\theta_{20,1}$ & -4.500 & -0.004 & 0.035 & 95.0 & 2015.3 & 1.01 & 0.001 & 0.033 & 97.1 & 2438.4 & 1.00 & 0.131 & 0.395 & 24.3 & 1566.7 & 1.01 & -0.005 & 0.037 & 95.0 & 3248.7 & 1.00 \\
$\theta_{20,2}$ & -1.500 & 0.004 & 0.016 & 96.0 & 3190.8 & 1.00 & 0.016 & 0.016 & 98.5 & 4000.9 & 1.00 & 0.417 & 0.413 & 0.0 & 1688.5 & 1.01 & 0.003 & 0.016 & 97.0 & 5263.2 & 1.00 \\
$\theta_{20,3}$ & 2.000 & -0.009 & 0.023 & 95.0 & 2456.3 & 1.00 & -0.008 & 0.025 & 92.6 & 2766.5 & 1.00 & -0.341 & 0.488 & 0.0 & 2153.7 & 1.01 & -0.010 & 0.023 & 93.0 & 3957.8 & 1.00 \\
$\theta_{21,1}$ & -3.000 & 0.008 & 0.032 & 94.0 & 1700.6 & 1.01 & 0.014 & 0.032 & 91.2 & 2351.0 & 1.01 & 0.198 & 0.385 & 5.4 & 1604.3 & 1.00 & 0.007 & 0.032 & 94.0 & 3187.3 & 1.00 \\
$\theta_{21,2}$ & 0.000 & \multicolumn{1}{c}{--} & 0.023 & 94.0 & 1875.2 & 1.00 & \multicolumn{1}{c}{--} & 0.025 & 91.2 & 2727.7 & 1.00 & \multicolumn{1}{c}{--} & 0.347 & 2.7 & 1913.4 & 1.00 & \multicolumn{1}{c}{--} & 0.023 & 93.0 & 3634.5 & 1.00 \\
$\theta_{21,3}$ & 3.000 & -0.009 & 0.022 & 99.0 & 1862.4 & 1.00 & -0.009 & 0.024 & 98.5 & 2443.4 & 1.00 & -0.183 & 0.324 & 8.1 & 1922.3 & 1.00 & -0.009 & 0.024 & 99.0 & 3312.8 & 1.00 \\
$\theta_{22,1}$ & -6.000 & -0.026 & 0.076 & 91.0 & 1316.5 & 1.01 & -0.033 & 0.095 & 92.6 & 1846.1 & 1.00 & -0.192 & 1.373 & 0.0 & 1544.4 & 1.01 & -0.027 & 0.078 & 90.0 & 2282.3 & 1.01 \\
$\theta_{22,2}$ & -2.000 & -0.027 & 0.032 & 94.0 & 1435.8 & 1.01 & -0.033 & 0.036 & 94.1 & 2310.4 & 1.00 & -0.527 & 1.149 & 0.0 & 2035.7 & 1.01 & -0.029 & 0.033 & 95.0 & 2725.7 & 1.00 \\
$\theta_{22,3}$ & 1.500 & -0.008 & 0.029 & 97.0 & 1598.9 & 1.00 & 0.001 & 0.031 & 95.6 & 2670.3 & 1.00 & 0.687 & 1.099 & 0.0 & 2099.8 & 1.01 & -0.008 & 0.029 & 96.0 & 3082.4 & 1.00 \\
$\theta_{23,1}$ & -5.000 & -0.002 & 0.054 & 97.0 & 1030.9 & 1.01 & -0.009 & 0.048 & 95.6 & 1557.4 & 1.01 & -0.166 & 0.749 & 5.4 & 1365.2 & 1.00 & -0.001 & 0.054 & 98.0 & 1960.1 & 1.00 \\
\bottomrule
\end{tabular*}
\begin{tablenotes}
\footnotesize
\item RB is mean relative bias across retained replications, MSE is mean squared error, CP is empirical coverage of 95\% posterior credible intervals, ESS is effective sample size, and $\hat{R}$ is the Gelman--Rubin convergence diagnostic. Summaries use the retained-fit denominator reported in Table~\ref{tab:simulation_summary} in the main text. A dash denotes an unavailable quantity, undefined RB when the data-generating value is zero, or an unavailable sampling diagnostic for a structurally fixed parameter.
\end{tablenotes}
\end{threeparttable}
\end{table}
\clearpage
\end{landscape}

\begin{landscape}
\begin{table}[p]
\ContinuedFloat

\centering
\caption{Full parameter-level simulation results: S5 (continued).}
\label{tab:r8_s5_full_recovery_8}
\footnotesize
\sisetup{mode=math}
\setlength{\tabcolsep}{0.35pt}
\renewcommand{\arraystretch}{0.72}
\begin{threeparttable}
\begin{tabular*}{\linewidth}{@{\extracolsep{\fill}}l S[table-format=-1.3] *{4}{S[table-format=-4.3] S[table-format=3.3] S[table-format=3.1] S[table-format=5.1] S[table-format=1.2]}}
\toprule
Parameter & {True} & \multicolumn{5}{c}{CLOUD} & \multicolumn{5}{c}{TrendyOU} & \multicolumn{5}{c}{StationaryOU} & \multicolumn{5}{c}{DiagOU} \\
\cmidrule(lr){3-7}\cmidrule(lr){8-12}\cmidrule(lr){13-17}\cmidrule(l){18-22}
 & & {RB} & {MSE} & {CP} & {ESS} & {$\hat{R}$} & {RB} & {MSE} & {CP} & {ESS} & {$\hat{R}$} & {RB} & {MSE} & {CP} & {ESS} & {$\hat{R}$} & {RB} & {MSE} & {CP} & {ESS} & {$\hat{R}$} \\
\midrule
\multicolumn{22}{@{}l}{\textit{Thresholds $\boldsymbol\theta$ (continued)}} \\
\midrule
$\theta_{23,2}$ & -1.000 & 0.008 & 0.029 & 97.0 & 815.0 & 1.01 & 0.000 & 0.028 & 98.5 & 1287.2 & 1.01 & -0.762 & 0.624 & 5.4 & 1216.8 & 1.00 & 0.009 & 0.029 & 98.0 & 1601.4 & 1.01 \\
$\theta_{23,3}$ & 2.500 & -0.016 & 0.044 & 96.0 & 899.7 & 1.01 & -0.020 & 0.047 & 94.1 & 1416.1 & 1.01 & 0.284 & 0.559 & 10.8 & 1439.0 & 1.00 & -0.016 & 0.044 & 95.0 & 1761.7 & 1.01 \\
$\theta_{24,1}$ & -4.000 & -0.004 & 0.037 & 97.0 & 1553.3 & 1.01 & -0.007 & 0.039 & 94.1 & 2257.9 & 1.00 & -0.227 & 0.856 & 0.0 & 1880.8 & 1.01 & -0.004 & 0.038 & 96.0 & 2786.9 & 1.00 \\
$\theta_{24,2}$ & -0.500 & -0.028 & 0.020 & 97.0 & 1703.3 & 1.01 & -0.034 & 0.019 & 98.5 & 2788.8 & 1.00 & -1.745 & 0.786 & 0.0 & 2216.0 & 1.00 & -0.031 & 0.021 & 97.0 & 3360.3 & 1.00 \\
$\theta_{24,3}$ & 2.000 & -0.006 & 0.026 & 98.0 & 1830.3 & 1.01 & -0.009 & 0.026 & 97.1 & 2977.1 & 1.00 & 0.421 & 0.745 & 0.0 & 2306.6 & 1.00 & -0.006 & 0.026 & 97.0 & 3487.8 & 1.00 \\
\addlinespace[0.25em]
\multicolumn{22}{@{}l}{\textit{Correlations $\boldsymbol\rho$}} \\
\midrule
$\rho_{1}$ & 0.288 & -0.518 & 0.033 & 93.0 & 316.0 & 1.07 & -1.839 & 0.289 & 0.0 & 377.0 & 1.04 & -2.717 & 0.618 & 0.0 & 386.5 & 1.02 & -1.000 & 0.083 & 0.0 & \multicolumn{1}{c}{--} & \multicolumn{1}{c}{--} \\
$\rho_{2}$ & -0.072 & -0.270 & 0.011 & 96.0 & 271.0 & 1.05 & -0.292 & 0.011 & 98.5 & 286.0 & 1.03 & -5.992 & 0.193 & 0.0 & 383.1 & 1.02 & -1.000 & 0.005 & 0.0 & \multicolumn{1}{c}{--} & \multicolumn{1}{c}{--} \\
$\rho_{3}$ & -0.152 & -0.330 & 0.013 & 97.0 & 244.1 & 1.05 & -1.834 & 0.087 & 32.4 & 309.4 & 1.04 & 1.707 & 0.070 & 0.0 & 453.7 & 1.03 & -1.000 & 0.023 & 0.0 & \multicolumn{1}{c}{--} & \multicolumn{1}{c}{--} \\
$\rho_{4}$ & -0.215 & -0.813 & 0.041 & 89.0 & 383.4 & 1.03 & -1.827 & 0.162 & 8.8 & 445.4 & 1.03 & -4.040 & 0.761 & 0.0 & 435.1 & 1.02 & -1.000 & 0.046 & 0.0 & \multicolumn{1}{c}{--} & \multicolumn{1}{c}{--} \\
$\rho_{5}$ & 0.258 & -0.461 & 0.022 & 96.0 & 350.6 & 1.03 & -1.966 & 0.265 & 0.0 & 381.6 & 1.04 & -2.840 & 0.542 & 0.0 & 362.1 & 1.03 & -1.000 & 0.067 & 0.0 & \multicolumn{1}{c}{--} & \multicolumn{1}{c}{--} \\
$\rho_{6}$ & -0.241 & -0.253 & 0.016 & 95.0 & 268.5 & 1.03 & -0.294 & 0.016 & 98.5 & 277.9 & 1.04 & -2.153 & 0.276 & 0.0 & 337.3 & 1.07 & -1.000 & 0.058 & 0.0 & \multicolumn{1}{c}{--} & \multicolumn{1}{c}{--} \\
$\rho_{7}$ & -0.019 & -0.734 & 0.013 & 99.0 & 231.7 & 1.05 & -10.677 & 0.052 & 63.2 & 279.1 & 1.07 & 17.721 & 0.122 & 0.0 & 387.0 & 1.03 & -1.000 & 0.000 & 0.0 & \multicolumn{1}{c}{--} & \multicolumn{1}{c}{--} \\
$\rho_{8}$ & -0.105 & -0.284 & 0.012 & 99.0 & 313.6 & 1.04 & -1.715 & 0.039 & 44.1 & 573.7 & 1.04 & 0.713 & 0.010 & 86.5 & 655.7 & 1.02 & -1.000 & 0.011 & 0.0 & \multicolumn{1}{c}{--} & \multicolumn{1}{c}{--} \\
$\rho_{9}$ & -0.102 & -0.205 & 0.010 & 99.0 & 289.9 & 1.04 & 3.628 & 0.141 & 0.0 & 441.4 & 1.04 & -0.319 & 0.004 & 91.9 & 717.4 & 1.01 & -1.000 & 0.010 & 0.0 & \multicolumn{1}{c}{--} & \multicolumn{1}{c}{--} \\
$\rho_{10}$ & 0.128 & -0.496 & 0.015 & 97.0 & 427.2 & 1.02 & -3.464 & 0.203 & 0.0 & 742.8 & 1.02 & -5.143 & 0.435 & 0.0 & 627.2 & 1.03 & -1.000 & 0.016 & 0.0 & \multicolumn{1}{c}{--} & \multicolumn{1}{c}{--} \\
$\rho_{11}$ & -0.089 & -0.607 & 0.012 & 100.0 & 384.0 & 1.04 & -7.663 & 0.472 & 0.0 & 457.1 & 1.04 & -8.534 & 0.585 & 0.0 & 354.9 & 1.04 & -1.000 & 0.008 & 0.0 & \multicolumn{1}{c}{--} & \multicolumn{1}{c}{--} \\
$\rho_{12}$ & -0.086 & -0.248 & 0.011 & 99.0 & 316.1 & 1.04 & -2.133 & 0.039 & 47.1 & 572.0 & 1.03 & 0.924 & 0.011 & 78.4 & 595.0 & 1.02 & -1.000 & 0.007 & 0.0 & \multicolumn{1}{c}{--} & \multicolumn{1}{c}{--} \\
$\rho_{13}$ & 0.087 & -0.042 & 0.010 & 98.0 & 267.2 & 1.03 & -5.535 & 0.235 & 0.0 & 448.0 & 1.02 & -1.228 & 0.016 & 56.8 & 592.0 & 1.02 & -1.000 & 0.008 & 0.0 & \multicolumn{1}{c}{--} & \multicolumn{1}{c}{--} \\
$\rho_{14}$ & 0.002 & -3.953 & 0.008 & 97.0 & 452.7 & 1.02 & -102.247 & 0.046 & 30.9 & 661.9 & 1.02 & -228.716 & 0.205 & 0.0 & 590.9 & 1.01 & -1.000 & 0.000 & 0.0 & \multicolumn{1}{c}{--} & \multicolumn{1}{c}{--} \\
$\rho_{15}$ & 0.128 & -0.229 & 0.012 & 98.0 & 532.9 & 1.04 & -0.473 & 0.014 & 92.6 & 756.7 & 1.02 & 2.343 & 0.094 & 2.7 & 825.0 & 1.03 & -1.000 & 0.016 & 0.0 & \multicolumn{1}{c}{--} & \multicolumn{1}{c}{--} \\
$\rho_{16}$ & -0.279 & -0.220 & 0.012 & 97.0 & 536.8 & 1.02 & -0.935 & 0.074 & 11.8 & 801.6 & 1.02 & -0.131 & 0.005 & 97.3 & 866.0 & 1.02 & -1.000 & 0.078 & 0.0 & \multicolumn{1}{c}{--} & \multicolumn{1}{c}{--} \\
$\rho_{17}$ & 0.054 & -0.017 & 0.008 & 97.0 & 510.7 & 1.02 & 2.541 & 0.025 & 66.2 & 647.3 & 1.02 & 6.497 & 0.126 & 0.0 & 584.5 & 1.02 & -1.000 & 0.003 & 0.0 & \multicolumn{1}{c}{--} & \multicolumn{1}{c}{--} \\
$\rho_{18}$ & 0.051 & -0.480 & 0.007 & 95.0 & 419.5 & 1.02 & -4.530 & 0.059 & 16.2 & 612.9 & 1.03 & -8.971 & 0.215 & 0.0 & 602.3 & 1.02 & -1.000 & 0.003 & 0.0 & \multicolumn{1}{c}{--} & \multicolumn{1}{c}{--} \\
$\rho_{19}$ & -0.175 & -0.386 & 0.015 & 96.0 & 491.4 & 1.03 & -1.710 & 0.098 & 16.2 & 792.5 & 1.02 & 1.185 & 0.047 & 10.8 & 1048.9 & 1.01 & -1.000 & 0.031 & 0.0 & \multicolumn{1}{c}{--} & \multicolumn{1}{c}{--} \\
$\rho_{20}$ & -0.082 & -0.135 & 0.008 & 98.0 & 498.4 & 1.03 & 4.554 & 0.144 & 0.0 & 639.4 & 1.02 & -0.316 & 0.004 & 97.3 & 900.1 & 1.01 & -1.000 & 0.007 & 0.0 & \multicolumn{1}{c}{--} & \multicolumn{1}{c}{--} \\
$\rho_{21}$ & -0.000 & 411.723 & 0.006 & 99.0 & 465.2 & 1.02 & 4198.437 & 0.049 & 23.5 & 682.0 & 1.01 & 8935.732 & 0.202 & 0.0 & 666.1 & 1.02 & -1.000 & 0.000 & 0.0 & \multicolumn{1}{c}{--} & \multicolumn{1}{c}{--} \\
$\rho_{22}$ & -0.007 & 0.076 & 0.006 & 99.0 & 390.3 & 1.03 & -55.585 & 0.167 & 0.0 & 538.9 & 1.02 & -85.319 & 0.386 & 0.0 & 429.6 & 1.02 & -1.000 & 0.000 & 0.0 & \multicolumn{1}{c}{--} & \multicolumn{1}{c}{--} \\
$\rho_{23}$ & -0.311 & -0.531 & 0.035 & 94.0 & 407.9 & 1.02 & 0.366 & 0.018 & 85.3 & 617.4 & 1.01 & 0.924 & 0.085 & 0.0 & 562.3 & 1.02 & -1.000 & 0.097 & 0.0 & \multicolumn{1}{c}{--} & \multicolumn{1}{c}{--} \\
$\rho_{24}$ & -0.242 & -0.329 & 0.014 & 99.0 & 311.0 & 1.05 & -0.409 & 0.018 & 91.2 & 351.0 & 1.03 & -2.246 & 0.300 & 0.0 & 403.5 & 1.03 & -1.000 & 0.059 & 0.0 & \multicolumn{1}{c}{--} & \multicolumn{1}{c}{--} \\
$\rho_{25}$ & 0.077 & -0.513 & 0.011 & 99.0 & 248.6 & 1.04 & 1.856 & 0.028 & 83.8 & 370.3 & 1.03 & -5.233 & 0.166 & 0.0 & 423.4 & 1.04 & -1.000 & 0.006 & 0.0 & \multicolumn{1}{c}{--} & \multicolumn{1}{c}{--} \\
$\rho_{26}$ & -0.082 & -0.043 & 0.008 & 99.0 & 388.4 & 1.02 & -1.892 & 0.030 & 60.3 & 671.3 & 1.02 & 1.162 & 0.014 & 75.7 & 714.3 & 1.02 & -1.000 & 0.007 & 0.0 & \multicolumn{1}{c}{--} & \multicolumn{1}{c}{--} \\
$\rho_{27}$ & 0.157 & -0.195 & 0.011 & 96.0 & 333.2 & 1.04 & -3.320 & 0.277 & 0.0 & 533.6 & 1.03 & -0.958 & 0.026 & 32.4 & 704.6 & 1.02 & -1.000 & 0.025 & 0.0 & \multicolumn{1}{c}{--} & \multicolumn{1}{c}{--} \\
$\rho_{28}$ & 0.008 & 0.064 & 0.006 & 98.0 & 460.5 & 1.03 & -24.477 & 0.042 & 29.4 & 690.3 & 1.04 & -55.026 & 0.189 & 0.0 & 603.8 & 1.03 & -1.000 & 0.000 & 0.0 & \multicolumn{1}{c}{--} & \multicolumn{1}{c}{--} \\
\bottomrule
\end{tabular*}
\begin{tablenotes}
\footnotesize
\item RB is mean relative bias across retained replications, MSE is mean squared error, CP is empirical coverage of 95\% posterior credible intervals, ESS is effective sample size, and $\hat{R}$ is the Gelman--Rubin convergence diagnostic. Summaries use the retained-fit denominator reported in Table~\ref{tab:simulation_summary} in the main text. A dash denotes an unavailable quantity, undefined RB when the data-generating value is zero, or an unavailable sampling diagnostic for a structurally fixed parameter.
\end{tablenotes}
\end{threeparttable}
\end{table}
\clearpage
\end{landscape}
\begin{landscape}
\begin{table}[p]

\centering
\caption{Full parameter-level simulation results: S6 (CLOUD, TrendyOU, StationaryOU, and DiagOU).}
\label{tab:r8_s6_full_recovery}
\footnotesize
\sisetup{mode=math}
\setlength{\tabcolsep}{0.35pt}
\renewcommand{\arraystretch}{0.72}
\begin{threeparttable}
\begin{tabular*}{\linewidth}{@{\extracolsep{\fill}}l S[table-format=-1.3] *{4}{S[table-format=-4.3] S[table-format=3.3] S[table-format=3.1] S[table-format=5.1] S[table-format=1.2]}}
\toprule
Parameter & {True} & \multicolumn{5}{c}{CLOUD} & \multicolumn{5}{c}{TrendyOU} & \multicolumn{5}{c}{StationaryOU} & \multicolumn{5}{c}{DiagOU} \\
\cmidrule(lr){3-7}\cmidrule(lr){8-12}\cmidrule(lr){13-17}\cmidrule(l){18-22}
 & & {RB} & {MSE} & {CP} & {ESS} & {$\hat{R}$} & {RB} & {MSE} & {CP} & {ESS} & {$\hat{R}$} & {RB} & {MSE} & {CP} & {ESS} & {$\hat{R}$} & {RB} & {MSE} & {CP} & {ESS} & {$\hat{R}$} \\
\midrule
\multicolumn{22}{@{}l}{\textit{Drift $\boldsymbol\Gamma$}} \\
\midrule
$\Gamma_{1,1}$ & 0.262 & 0.092 & 0.433 & 93.9 & 367.4 & 1.05 & 2.058 & 0.427 & 64.4 & 379.7 & 1.06 & 1.584 & 0.274 & 66.0 & 579.2 & 1.02 & 22.736 & 38.457 & 20.6 & 1226.1 & 1.09 \\
$\Gamma_{1,2}$ & 1.349 & -0.025 & 0.251 & 99.0 & 403.9 & 1.04 & -0.523 & 0.906 & 79.5 & 297.4 & 1.05 & -0.569 & 0.691 & 42.6 & 467.5 & 1.02 & -1.000 & 1.819 & 0.0 & \multicolumn{1}{c}{--} & \multicolumn{1}{c}{--} \\
$\Gamma_{1,3}$ & -0.759 & 0.036 & 0.112 & 98.0 & 504.5 & 1.03 & 0.325 & 0.159 & 82.2 & 407.0 & 1.03 & -0.039 & 0.053 & 95.7 & 546.8 & 1.03 & -1.000 & 0.575 & 0.0 & \multicolumn{1}{c}{--} & \multicolumn{1}{c}{--} \\
$\Gamma_{1,4}$ & 0.095 & 0.059 & 0.135 & 99.0 & 429.2 & 1.05 & 4.575 & 0.584 & 72.6 & 276.2 & 1.08 & -0.768 & 0.138 & 95.7 & 396.1 & 1.03 & -1.000 & 0.009 & 0.0 & \multicolumn{1}{c}{--} & \multicolumn{1}{c}{--} \\
$\Gamma_{1,5}$ & -0.374 & -0.047 & 0.157 & 100.0 & 429.1 & 1.04 & -1.353 & 0.404 & 72.6 & 315.1 & 1.07 & 0.084 & 0.263 & 80.9 & 509.8 & 1.03 & -1.000 & 0.140 & 0.0 & \multicolumn{1}{c}{--} & \multicolumn{1}{c}{--} \\
$\Gamma_{1,6}$ & 0.464 & -0.067 & 0.255 & 96.9 & 435.4 & 1.03 & -1.005 & 0.555 & 83.6 & 281.5 & 1.07 & -2.016 & 1.076 & 17.0 & 496.0 & 1.03 & -1.000 & 0.215 & 0.0 & \multicolumn{1}{c}{--} & \multicolumn{1}{c}{--} \\
$\Gamma_{1,7}$ & 1.492 & -0.061 & 0.118 & 99.0 & 423.7 & 1.05 & -0.381 & 0.477 & 76.7 & 342.0 & 1.06 & -0.781 & 1.429 & 8.5 & 430.8 & 1.03 & -1.000 & 2.227 & 0.0 & \multicolumn{1}{c}{--} & \multicolumn{1}{c}{--} \\
$\Gamma_{1,8}$ & -0.626 & -0.032 & 0.130 & 96.9 & 479.9 & 1.05 & -0.271 & 0.146 & 94.5 & 413.5 & 1.04 & -0.837 & 0.323 & 42.6 & 525.0 & 1.02 & -1.000 & 0.392 & 0.0 & \multicolumn{1}{c}{--} & \multicolumn{1}{c}{--} \\
$\Gamma_{2,1}$ & -1.692 & -0.019 & 0.232 & 100.0 & 397.6 & 1.07 & -0.851 & 2.196 & 0.0 & 354.2 & 1.07 & -0.938 & 2.640 & 0.0 & 522.9 & 1.03 & -1.000 & 2.861 & 0.0 & \multicolumn{1}{c}{--} & \multicolumn{1}{c}{--} \\
$\Gamma_{2,2}$ & 1.015 & -0.007 & 0.082 & 99.0 & 493.6 & 1.02 & -0.211 & 0.177 & 84.9 & 376.2 & 1.05 & 0.034 & 0.166 & 95.7 & 492.2 & 1.03 & 5.836 & 36.865 & 2.1 & 2337.2 & 1.01 \\
$\Gamma_{2,3}$ & -0.236 & 0.028 & 0.118 & 98.0 & 527.1 & 1.03 & 0.562 & 0.075 & 87.7 & 544.9 & 1.03 & 0.386 & 0.058 & 95.7 & 652.2 & 1.02 & -1.000 & 0.056 & 0.0 & \multicolumn{1}{c}{--} & \multicolumn{1}{c}{--} \\
$\Gamma_{2,4}$ & -0.538 & -0.053 & 0.181 & 95.9 & 486.6 & 1.05 & -1.191 & 0.607 & 61.6 & 329.7 & 1.08 & -2.617 & 2.091 & 6.4 & 618.2 & 1.02 & -1.000 & 0.289 & 0.0 & \multicolumn{1}{c}{--} & \multicolumn{1}{c}{--} \\
$\Gamma_{2,5}$ & -0.785 & -0.025 & 0.111 & 98.0 & 475.2 & 1.04 & -0.726 & 0.443 & 61.6 & 293.5 & 1.10 & -1.495 & 1.468 & 6.4 & 536.2 & 1.02 & -1.000 & 0.616 & 0.0 & \multicolumn{1}{c}{--} & \multicolumn{1}{c}{--} \\
$\Gamma_{2,6}$ & -0.392 & -0.026 & 0.179 & 99.0 & 432.2 & 1.04 & 1.852 & 0.740 & 49.3 & 312.4 & 1.08 & 0.648 & 0.917 & 91.5 & 531.6 & 1.02 & -1.000 & 0.154 & 0.0 & \multicolumn{1}{c}{--} & \multicolumn{1}{c}{--} \\
$\Gamma_{2,7}$ & -0.636 & -0.037 & 0.184 & 96.9 & 444.6 & 1.05 & -0.754 & 0.364 & 65.8 & 354.1 & 1.06 & -1.573 & 1.141 & 6.4 & 497.2 & 1.02 & -1.000 & 0.405 & 0.0 & \multicolumn{1}{c}{--} & \multicolumn{1}{c}{--} \\
$\Gamma_{2,8}$ & -1.433 & -0.073 & 0.087 & 98.0 & 497.9 & 1.04 & -0.560 & 0.703 & 12.3 & 489.0 & 1.04 & -0.477 & 0.574 & 29.8 & 599.4 & 1.03 & -1.000 & 2.054 & 0.0 & \multicolumn{1}{c}{--} & \multicolumn{1}{c}{--} \\
$\Gamma_{3,1}$ & 0.006 & 0.053 & 0.147 & 99.0 & 482.9 & 1.04 & 183.728 & 1.321 & 21.9 & 350.6 & 1.05 & 220.364 & 1.858 & 19.1 & 569.2 & 1.03 & -1.000 & 0.000 & 0.0 & \multicolumn{1}{c}{--} & \multicolumn{1}{c}{--} \\
$\Gamma_{3,2}$ & 0.469 & 0.014 & 0.083 & 99.0 & 554.6 & 1.04 & -0.762 & 0.464 & 87.7 & 351.5 & 1.04 & -1.022 & 0.468 & 91.5 & 517.8 & 1.04 & -1.000 & 0.220 & 0.0 & \multicolumn{1}{c}{--} & \multicolumn{1}{c}{--} \\
$\Gamma_{3,3}$ & 1.136 & -0.048 & 0.053 & 99.0 & 756.3 & 1.02 & -0.171 & 0.105 & 89.0 & 577.8 & 1.02 & -0.119 & 0.071 & 95.7 & 788.2 & 1.01 & 4.571 & 29.787 & 0.0 & 2125.5 & 1.01 \\
$\Gamma_{3,4}$ & 0.615 & -0.078 & 0.062 & 99.0 & 690.3 & 1.05 & 0.378 & 0.237 & 93.2 & 434.6 & 1.06 & -0.218 & 0.258 & 100.0 & 510.1 & 1.03 & -1.000 & 0.378 & 0.0 & \multicolumn{1}{c}{--} & \multicolumn{1}{c}{--} \\
$\Gamma_{3,5}$ & -0.152 & 0.055 & 0.086 & 98.0 & 583.5 & 1.02 & 0.526 & 0.118 & 100.0 & 310.1 & 1.08 & 2.566 & 0.376 & 74.5 & 452.5 & 1.02 & -1.000 & 0.023 & 0.0 & \multicolumn{1}{c}{--} & \multicolumn{1}{c}{--} \\
$\Gamma_{3,6}$ & 1.650 & -0.055 & 0.071 & 99.0 & 561.9 & 1.04 & -0.276 & 0.431 & 90.4 & 322.8 & 1.08 & -0.539 & 1.021 & 44.7 & 511.1 & 1.04 & -1.000 & 2.724 & 0.0 & \multicolumn{1}{c}{--} & \multicolumn{1}{c}{--} \\
$\Gamma_{3,7}$ & -1.162 & -0.071 & 0.114 & 100.0 & 524.1 & 1.05 & -0.253 & 0.311 & 79.5 & 366.2 & 1.06 & 0.064 & 0.225 & 87.2 & 587.8 & 1.03 & -1.000 & 1.350 & 0.0 & \multicolumn{1}{c}{--} & \multicolumn{1}{c}{--} \\
$\Gamma_{3,8}$ & -0.626 & -0.012 & 0.057 & 98.0 & 723.1 & 1.03 & -0.765 & 0.334 & 71.2 & 569.8 & 1.03 & -0.645 & 0.244 & 70.2 & 745.0 & 1.02 & -1.000 & 0.392 & 0.0 & \multicolumn{1}{c}{--} & \multicolumn{1}{c}{--} \\
$\Gamma_{4,1}$ & -0.613 & -0.049 & 0.266 & 99.0 & 417.4 & 1.04 & -0.066 & 0.280 & 87.7 & 279.0 & 1.08 & -2.896 & 3.346 & 0.0 & 466.6 & 1.05 & -1.000 & 0.376 & 0.0 & \multicolumn{1}{c}{--} & \multicolumn{1}{c}{--} \\
$\Gamma_{4,2}$ & 0.091 & -0.032 & 0.140 & 95.9 & 559.6 & 1.04 & 2.937 & 0.415 & 86.3 & 343.0 & 1.08 & 1.805 & 0.477 & 87.2 & 343.8 & 1.04 & -1.000 & 0.008 & 0.0 & \multicolumn{1}{c}{--} & \multicolumn{1}{c}{--} \\
$\Gamma_{4,3}$ & -0.545 & -0.072 & 0.065 & 98.0 & 754.3 & 1.02 & -0.750 & 0.224 & 69.9 & 646.1 & 1.03 & -1.118 & 0.493 & 53.2 & 522.7 & 1.02 & -1.000 & 0.298 & 0.0 & \multicolumn{1}{c}{--} & \multicolumn{1}{c}{--} \\
$\Gamma_{4,4}$ & -0.008 & -0.048 & 0.158 & 95.9 & 602.4 & 1.03 & -57.461 & 0.376 & 83.6 & 386.3 & 1.04 & -183.164 & 2.553 & 8.5 & 359.2 & 1.05 & -1008.739 & 66.525 & 0.0 & 3211.8 & 1.00 \\
$\Gamma_{4,5}$ & -1.806 & -0.089 & 0.204 & 92.9 & 500.0 & 1.05 & -0.416 & 0.768 & 50.7 & 233.7 & 1.09 & -1.077 & 4.040 & 0.0 & 343.0 & 1.05 & -1.000 & 3.262 & 0.0 & \multicolumn{1}{c}{--} & \multicolumn{1}{c}{--} \\
$\Gamma_{4,6}$ & -0.988 & -0.056 & 0.181 & 95.9 & 468.4 & 1.03 & -1.022 & 1.379 & 64.4 & 238.0 & 1.08 & -1.953 & 4.088 & 8.5 & 416.2 & 1.03 & -1.000 & 0.977 & 0.0 & \multicolumn{1}{c}{--} & \multicolumn{1}{c}{--} \\
$\Gamma_{4,7}$ & -1.266 & -0.046 & 0.216 & 95.9 & 482.9 & 1.05 & -0.479 & 0.633 & 78.1 & 304.2 & 1.07 & -1.620 & 4.312 & 0.0 & 502.6 & 1.05 & -1.000 & 1.602 & 0.0 & \multicolumn{1}{c}{--} & \multicolumn{1}{c}{--} \\
$\Gamma_{4,8}$ & 1.641 & -0.083 & 0.077 & 98.0 & 562.9 & 1.02 & -0.313 & 0.324 & 58.9 & 462.7 & 1.06 & -0.682 & 1.378 & 8.5 & 488.4 & 1.02 & -1.000 & 2.692 & 0.0 & \multicolumn{1}{c}{--} & \multicolumn{1}{c}{--} \\
$\Gamma_{5,1}$ & 1.127 & -0.077 & 0.211 & 99.0 & 416.8 & 1.07 & -0.962 & 1.247 & 9.6 & 415.8 & 1.06 & -0.923 & 1.206 & 4.3 & 509.2 & 1.03 & -1.000 & 1.270 & 0.0 & \multicolumn{1}{c}{--} & \multicolumn{1}{c}{--} \\
\bottomrule
\end{tabular*}
\begin{tablenotes}
\footnotesize
\item RB is mean relative bias across retained replications, MSE is mean squared error, CP is empirical coverage of 95\% posterior credible intervals, ESS is effective sample size, and $\hat{R}$ is the Gelman--Rubin convergence diagnostic. Summaries use the retained-fit denominator reported in Table~\ref{tab:simulation_summary} in the main text. A dash denotes an unavailable quantity, undefined RB when the data-generating value is zero, or an unavailable sampling diagnostic for a structurally fixed parameter.
\end{tablenotes}
\end{threeparttable}
\end{table}
\clearpage
\end{landscape}

\begin{landscape}
\begin{table}[p]
\ContinuedFloat

\centering
\caption{Full parameter-level simulation results: S6 (continued).}
\label{tab:r8_s6_full_recovery_2}
\footnotesize
\sisetup{mode=math}
\setlength{\tabcolsep}{0.35pt}
\renewcommand{\arraystretch}{0.72}
\begin{threeparttable}
\begin{tabular*}{\linewidth}{@{\extracolsep{\fill}}l S[table-format=-1.3] *{4}{S[table-format=-4.3] S[table-format=3.3] S[table-format=3.1] S[table-format=5.1] S[table-format=1.2]}}
\toprule
Parameter & {True} & \multicolumn{5}{c}{CLOUD} & \multicolumn{5}{c}{TrendyOU} & \multicolumn{5}{c}{StationaryOU} & \multicolumn{5}{c}{DiagOU} \\
\cmidrule(lr){3-7}\cmidrule(lr){8-12}\cmidrule(lr){13-17}\cmidrule(l){18-22}
 & & {RB} & {MSE} & {CP} & {ESS} & {$\hat{R}$} & {RB} & {MSE} & {CP} & {ESS} & {$\hat{R}$} & {RB} & {MSE} & {CP} & {ESS} & {$\hat{R}$} & {RB} & {MSE} & {CP} & {ESS} & {$\hat{R}$} \\
\midrule
\multicolumn{22}{@{}l}{\textit{Drift $\boldsymbol\Gamma$ (continued)}} \\
\midrule
$\Gamma_{5,2}$ & -0.074 & -0.065 & 0.091 & 100.0 & 473.2 & 1.05 & -4.681 & 0.334 & 91.8 & 280.9 & 1.06 & -2.181 & 0.146 & 89.4 & 327.3 & 1.03 & -1.000 & 0.006 & 0.0 & \multicolumn{1}{c}{--} & \multicolumn{1}{c}{--} \\
$\Gamma_{5,3}$ & 0.197 & -0.021 & 0.077 & 98.0 & 585.9 & 1.05 & 0.916 & 0.095 & 91.8 & 504.3 & 1.04 & 0.367 & 0.048 & 93.6 & 603.3 & 1.04 & -1.000 & 0.039 & 0.0 & \multicolumn{1}{c}{--} & \multicolumn{1}{c}{--} \\
$\Gamma_{5,4}$ & 1.068 & 0.029 & 0.095 & 99.0 & 479.1 & 1.06 & 0.090 & 0.251 & 91.8 & 266.0 & 1.08 & -0.609 & 0.500 & 59.6 & 392.0 & 1.03 & -1.000 & 1.141 & 0.0 & \multicolumn{1}{c}{--} & \multicolumn{1}{c}{--} \\
$\Gamma_{5,5}$ & 0.356 & 0.063 & 0.062 & 100.0 & 546.3 & 1.03 & 0.136 & 0.104 & 94.5 & 372.5 & 1.08 & -1.746 & 0.481 & 19.1 & 612.4 & 1.04 & 15.143 & 31.542 & 27.8 & 1165.6 & 1.07 \\
$\Gamma_{5,6}$ & -1.204 & -0.033 & 0.179 & 94.9 & 516.9 & 1.04 & -0.866 & 1.325 & 38.4 & 378.4 & 1.04 & -0.637 & 0.682 & 19.1 & 393.9 & 1.03 & -1.000 & 1.450 & 0.0 & \multicolumn{1}{c}{--} & \multicolumn{1}{c}{--} \\
$\Gamma_{5,7}$ & -0.492 & -0.085 & 0.099 & 100.0 & 474.6 & 1.07 & -0.138 & 0.133 & 97.3 & 400.3 & 1.06 & -0.095 & 0.050 & 97.9 & 318.5 & 1.04 & -1.000 & 0.242 & 0.0 & \multicolumn{1}{c}{--} & \multicolumn{1}{c}{--} \\
$\Gamma_{5,8}$ & -0.653 & -0.024 & 0.079 & 96.9 & 628.5 & 1.03 & 0.715 & 0.292 & 65.8 & 442.1 & 1.06 & 0.227 & 0.069 & 85.1 & 374.4 & 1.06 & -1.000 & 0.427 & 0.0 & \multicolumn{1}{c}{--} & \multicolumn{1}{c}{--} \\
$\Gamma_{6,1}$ & -0.071 & 0.033 & 0.170 & 100.0 & 452.0 & 1.03 & 1.098 & 0.123 & 87.7 & 343.9 & 1.09 & -1.903 & 0.140 & 85.1 & 462.8 & 1.04 & -1.000 & 0.005 & 0.0 & \multicolumn{1}{c}{--} & \multicolumn{1}{c}{--} \\
$\Gamma_{6,2}$ & 0.126 & 0.068 & 0.122 & 95.9 & 496.1 & 1.02 & 1.107 & 0.151 & 95.9 & 362.4 & 1.07 & -0.066 & 0.133 & 95.7 & 390.2 & 1.02 & -1.000 & 0.016 & 0.0 & \multicolumn{1}{c}{--} & \multicolumn{1}{c}{--} \\
$\Gamma_{6,3}$ & -1.664 & -0.053 & 0.120 & 95.9 & 697.4 & 1.03 & -0.516 & 0.769 & 2.7 & 566.7 & 1.04 & -0.530 & 0.802 & 2.1 & 735.6 & 1.01 & -1.000 & 2.769 & 0.0 & \multicolumn{1}{c}{--} & \multicolumn{1}{c}{--} \\
$\Gamma_{6,4}$ & 0.559 & -0.028 & 0.128 & 93.9 & 604.4 & 1.02 & -0.796 & 0.371 & 74.0 & 316.6 & 1.08 & -2.096 & 1.547 & 23.4 & 447.9 & 1.03 & -1.000 & 0.312 & 0.0 & \multicolumn{1}{c}{--} & \multicolumn{1}{c}{--} \\
$\Gamma_{6,5}$ & 1.455 & -0.065 & 0.082 & 99.0 & 515.7 & 1.04 & -0.516 & 0.708 & 45.2 & 304.2 & 1.08 & -0.634 & 0.904 & 10.6 & 407.9 & 1.04 & -1.000 & 2.118 & 0.0 & \multicolumn{1}{c}{--} & \multicolumn{1}{c}{--} \\
$\Gamma_{6,6}$ & 0.538 & 0.064 & 0.114 & 99.0 & 501.3 & 1.02 & 0.201 & 0.280 & 86.3 & 341.3 & 1.10 & -0.406 & 0.107 & 89.4 & 541.1 & 1.03 & 11.713 & 41.939 & 0.0 & 2254.2 & 1.01 \\
$\Gamma_{6,7}$ & 0.456 & -0.068 & 0.179 & 96.9 & 513.1 & 1.02 & -0.262 & 0.118 & 97.3 & 347.1 & 1.06 & -0.336 & 0.132 & 89.4 & 432.5 & 1.05 & -1.000 & 0.208 & 0.0 & \multicolumn{1}{c}{--} & \multicolumn{1}{c}{--} \\
$\Gamma_{6,8}$ & 0.679 & -0.078 & 0.082 & 98.0 & 645.1 & 1.03 & -0.414 & 0.148 & 80.8 & 491.3 & 1.06 & -0.220 & 0.060 & 93.6 & 690.0 & 1.03 & -1.000 & 0.461 & 0.0 & \multicolumn{1}{c}{--} & \multicolumn{1}{c}{--} \\
$\Gamma_{7,1}$ & -1.702 & -0.078 & 0.459 & 93.9 & 364.1 & 1.07 & -0.809 & 2.053 & 2.7 & 378.0 & 1.05 & -1.101 & 3.674 & 0.0 & 601.4 & 1.01 & -1.000 & 2.896 & 0.0 & \multicolumn{1}{c}{--} & \multicolumn{1}{c}{--} \\
$\Gamma_{7,2}$ & 0.849 & -0.001 & 0.061 & 100.0 & 510.4 & 1.05 & 0.338 & 0.496 & 86.3 & 419.4 & 1.06 & 0.193 & 0.183 & 95.7 & 611.2 & 1.02 & -1.000 & 0.721 & 0.0 & \multicolumn{1}{c}{--} & \multicolumn{1}{c}{--} \\
$\Gamma_{7,3}$ & 1.421 & -0.013 & 0.118 & 94.9 & 560.5 & 1.04 & -0.203 & 0.227 & 79.5 & 538.9 & 1.04 & -0.058 & 0.098 & 95.7 & 693.5 & 1.02 & -1.000 & 2.019 & 0.0 & \multicolumn{1}{c}{--} & \multicolumn{1}{c}{--} \\
$\Gamma_{7,4}$ & 0.605 & 0.038 & 0.108 & 99.0 & 515.1 & 1.03 & 0.509 & 0.492 & 83.6 & 374.2 & 1.06 & 0.153 & 0.359 & 93.6 & 531.4 & 1.01 & -1.000 & 0.366 & 0.0 & \multicolumn{1}{c}{--} & \multicolumn{1}{c}{--} \\
$\Gamma_{7,5}$ & 0.191 & 0.066 & 0.132 & 98.0 & 429.2 & 1.06 & 0.156 & 0.274 & 90.4 & 336.3 & 1.08 & -1.287 & 0.304 & 87.2 & 514.8 & 1.02 & -1.000 & 0.036 & 0.0 & \multicolumn{1}{c}{--} & \multicolumn{1}{c}{--} \\
$\Gamma_{7,6}$ & 0.914 & -0.076 & 0.183 & 96.9 & 440.6 & 1.05 & -0.997 & 1.277 & 65.8 & 366.2 & 1.06 & -1.244 & 1.510 & 29.8 & 577.6 & 1.02 & -1.000 & 0.835 & 0.0 & \multicolumn{1}{c}{--} & \multicolumn{1}{c}{--} \\
$\Gamma_{7,7}$ & 0.394 & 0.047 & 0.368 & 94.9 & 390.6 & 1.04 & 1.693 & 0.598 & 69.9 & 391.8 & 1.06 & 2.311 & 1.021 & 42.6 & 538.0 & 1.02 & 18.262 & 53.273 & 0.0 & 2749.7 & 1.00 \\
$\Gamma_{7,8}$ & 0.546 & 0.025 & 0.098 & 98.0 & 530.9 & 1.04 & 1.292 & 0.608 & 43.8 & 601.0 & 1.04 & 1.312 & 0.600 & 42.6 & 685.8 & 1.03 & -1.000 & 0.298 & 0.0 & \multicolumn{1}{c}{--} & \multicolumn{1}{c}{--} \\
$\Gamma_{8,1}$ & -0.127 & -0.077 & 0.123 & 98.0 & 457.2 & 1.07 & -2.809 & 0.201 & 82.2 & 455.3 & 1.04 & -1.833 & 0.149 & 91.5 & 643.5 & 1.02 & -1.000 & 0.016 & 0.0 & \multicolumn{1}{c}{--} & \multicolumn{1}{c}{--} \\
$\Gamma_{8,2}$ & 1.406 & -0.098 & 0.084 & 98.0 & 512.7 & 1.04 & -0.107 & 0.123 & 100.0 & 579.2 & 1.04 & -0.199 & 0.184 & 97.9 & 591.2 & 1.02 & -1.000 & 1.975 & 0.0 & \multicolumn{1}{c}{--} & \multicolumn{1}{c}{--} \\
$\Gamma_{8,3}$ & 0.173 & -0.086 & 0.053 & 99.0 & 723.8 & 1.03 & -1.314 & 0.081 & 90.4 & 764.9 & 1.03 & -1.207 & 0.079 & 87.2 & 852.8 & 1.01 & -1.000 & 0.030 & 0.0 & \multicolumn{1}{c}{--} & \multicolumn{1}{c}{--} \\
$\Gamma_{8,4}$ & -1.023 & 0.034 & 0.067 & 98.0 & 636.9 & 1.02 & -0.106 & 0.085 & 98.6 & 489.7 & 1.04 & -0.339 & 0.308 & 83.0 & 537.1 & 1.02 & -1.000 & 1.046 & 0.0 & \multicolumn{1}{c}{--} & \multicolumn{1}{c}{--} \\
$\Gamma_{8,5}$ & 0.194 & 0.015 & 0.069 & 100.0 & 519.5 & 1.03 & 0.777 & 0.105 & 91.8 & 390.9 & 1.07 & 2.373 & 0.269 & 61.7 & 555.8 & 1.02 & -1.000 & 0.037 & 0.0 & \multicolumn{1}{c}{--} & \multicolumn{1}{c}{--} \\
$\Gamma_{8,6}$ & -0.612 & 0.041 & 0.084 & 96.9 & 537.5 & 1.02 & 0.514 & 0.240 & 90.4 & 420.9 & 1.06 & 0.075 & 0.058 & 97.9 & 603.2 & 1.02 & -1.000 & 0.374 & 0.0 & \multicolumn{1}{c}{--} & \multicolumn{1}{c}{--} \\
$\Gamma_{8,7}$ & -0.811 & -0.091 & 0.083 & 98.0 & 534.2 & 1.03 & -0.164 & 0.107 & 97.3 & 491.3 & 1.04 & -0.286 & 0.143 & 89.4 & 584.9 & 1.02 & -1.000 & 0.658 & 0.0 & \multicolumn{1}{c}{--} & \multicolumn{1}{c}{--} \\
$\Gamma_{8,8}$ & 0.508 & 0.092 & 0.049 & 100.0 & 608.2 & 1.02 & 0.312 & 0.069 & 95.9 & 635.5 & 1.02 & 0.421 & 0.086 & 91.5 & 803.5 & 1.01 & 12.693 & 43.997 & 0.0 & 2417.5 & 1.01 \\
\addlinespace[0.25em]
\multicolumn{22}{@{}l}{\textit{Loadings $\boldsymbol\Lambda$}} \\
\midrule
$\lambda_{1}$ & 0.800 & 0.002 & 0.006 & 100.0 & 1428.8 & 1.01 & 0.044 & 0.016 & 95.9 & 937.1 & 1.04 & 0.789 & 0.432 & 2.1 & 1686.5 & 1.01 & 0.030 & 0.008 & 99.0 & 1977.3 & 1.01 \\
$\lambda_{2}$ & 1.200 & -0.113 & 0.023 & 92.9 & 1309.5 & 1.01 & -0.071 & 0.029 & 91.8 & 699.3 & 1.03 & 0.500 & 0.405 & 10.6 & 1312.0 & 1.01 & -0.101 & 0.019 & 96.9 & 1824.0 & 1.00 \\
\bottomrule
\end{tabular*}
\begin{tablenotes}
\footnotesize
\item RB is mean relative bias across retained replications, MSE is mean squared error, CP is empirical coverage of 95\% posterior credible intervals, ESS is effective sample size, and $\hat{R}$ is the Gelman--Rubin convergence diagnostic. Summaries use the retained-fit denominator reported in Table~\ref{tab:simulation_summary} in the main text. A dash denotes an unavailable quantity, undefined RB when the data-generating value is zero, or an unavailable sampling diagnostic for a structurally fixed parameter.
\end{tablenotes}
\end{threeparttable}
\end{table}
\clearpage
\end{landscape}

\begin{landscape}
\begin{table}[p]
\ContinuedFloat

\centering
\caption{Full parameter-level simulation results: S6 (continued).}
\label{tab:r8_s6_full_recovery_3}
\footnotesize
\sisetup{mode=math}
\setlength{\tabcolsep}{0.35pt}
\renewcommand{\arraystretch}{0.72}
\begin{threeparttable}
\begin{tabular*}{\linewidth}{@{\extracolsep{\fill}}l S[table-format=-1.3] *{4}{S[table-format=-4.3] S[table-format=3.3] S[table-format=3.1] S[table-format=5.1] S[table-format=1.2]}}
\toprule
Parameter & {True} & \multicolumn{5}{c}{CLOUD} & \multicolumn{5}{c}{TrendyOU} & \multicolumn{5}{c}{StationaryOU} & \multicolumn{5}{c}{DiagOU} \\
\cmidrule(lr){3-7}\cmidrule(lr){8-12}\cmidrule(lr){13-17}\cmidrule(l){18-22}
 & & {RB} & {MSE} & {CP} & {ESS} & {$\hat{R}$} & {RB} & {MSE} & {CP} & {ESS} & {$\hat{R}$} & {RB} & {MSE} & {CP} & {ESS} & {$\hat{R}$} & {RB} & {MSE} & {CP} & {ESS} & {$\hat{R}$} \\
\midrule
\multicolumn{22}{@{}l}{\textit{Loadings $\boldsymbol\Lambda$ (continued)}} \\
\midrule
$\lambda_{3}$ & 1.100 & -0.122 & 0.022 & 91.8 & 1415.7 & 1.01 & -0.051 & 0.020 & 94.5 & 831.8 & 1.03 & 0.611 & 0.484 & 2.1 & 1628.8 & 1.01 & -0.105 & 0.018 & 95.9 & 1970.2 & 1.02 \\
$\lambda_{4}$ & 0.900 & -0.033 & 0.009 & 95.9 & 1307.2 & 1.01 & 0.509 & 0.230 & 5.5 & 1250.6 & 1.01 & 0.819 & 0.581 & 0.0 & 1520.4 & 1.01 & -0.011 & 0.008 & 96.9 & 1844.5 & 1.01 \\
$\lambda_{5}$ & 1.400 & -0.176 & 0.067 & 49.0 & 1427.6 & 1.01 & 0.257 & 0.161 & 53.4 & 1089.9 & 1.01 & 0.499 & 0.540 & 8.5 & 1426.5 & 1.01 & -0.166 & 0.061 & 60.8 & 1895.8 & 1.01 \\
$\lambda_{6}$ & 1.000 & -0.082 & 0.010 & 91.8 & 1003.9 & 1.01 & 0.492 & 0.259 & 2.7 & 836.6 & 1.01 & 0.812 & 0.680 & 0.0 & 1038.3 & 1.01 & -0.059 & 0.008 & 94.8 & 1366.0 & 1.01 \\
$\lambda_{7}$ & 0.700 & 0.031 & 0.006 & 94.9 & 1783.9 & 1.01 & 0.053 & 0.008 & 97.3 & 1824.9 & 1.00 & 0.248 & 0.039 & 61.7 & 2159.0 & 1.00 & 0.082 & 0.009 & 93.8 & 2343.4 & 1.00 \\
$\lambda_{8}$ & 1.100 & -0.017 & 0.006 & 99.0 & 1087.7 & 1.01 & 0.039 & 0.012 & 94.5 & 905.0 & 1.01 & 0.249 & 0.091 & 38.3 & 1081.3 & 1.01 & 0.017 & 0.007 & 97.9 & 1358.7 & 1.01 \\
$\lambda_{9}$ & 0.900 & 0.019 & 0.005 & 95.9 & 1258.8 & 1.01 & 0.082 & 0.013 & 91.8 & 1129.6 & 1.01 & 0.277 & 0.073 & 31.9 & 1369.4 & 1.01 & 0.063 & 0.009 & 90.7 & 1597.1 & 1.00 \\
$\lambda_{10}$ & 1.200 & -0.036 & 0.006 & 95.9 & 1049.8 & 1.01 & 0.266 & 0.113 & 16.4 & 949.0 & 1.01 & 0.635 & 0.602 & 0.0 & 1127.2 & 1.00 & -0.003 & 0.005 & 97.9 & 1433.5 & 1.01 \\
$\lambda_{11}$ & 0.800 & -0.004 & 0.003 & 98.0 & 1542.4 & 1.02 & 0.290 & 0.059 & 11.0 & 1530.1 & 1.01 & 0.666 & 0.300 & 0.0 & 1882.6 & 1.00 & 0.040 & 0.005 & 94.8 & 2241.8 & 1.00 \\
$\lambda_{12}$ & 1.000 & -0.012 & 0.004 & 99.0 & 1214.6 & 1.01 & 0.273 & 0.081 & 9.6 & 1189.5 & 1.01 & 0.607 & 0.381 & 0.0 & 1468.1 & 1.01 & 0.033 & 0.005 & 97.9 & 1659.1 & 1.01 \\
$\lambda_{13}$ & 0.800 & -0.015 & 0.006 & 100.0 & 1654.8 & 1.01 & 0.054 & 0.013 & 95.9 & 1210.8 & 1.02 & 0.703 & 0.340 & 2.1 & 2194.4 & 1.00 & -0.002 & 0.006 & 99.0 & 2070.0 & 1.00 \\
$\lambda_{14}$ & 1.200 & -0.099 & 0.019 & 92.9 & 1520.5 & 1.01 & -0.007 & 0.019 & 97.3 & 849.5 & 1.03 & 0.530 & 0.439 & 4.3 & 1614.7 & 1.01 & -0.104 & 0.021 & 94.8 & 1844.8 & 1.01 \\
$\lambda_{15}$ & 1.100 & -0.098 & 0.017 & 93.9 & 1642.1 & 1.01 & 0.005 & 0.016 & 98.6 & 986.5 & 1.02 & 0.661 & 0.566 & 4.3 & 1997.6 & 1.01 & -0.098 & 0.017 & 95.9 & 1975.6 & 1.01 \\
$\lambda_{16}$ & 0.900 & -0.035 & 0.006 & 96.9 & 1367.8 & 1.01 & 0.500 & 0.219 & 6.8 & 1192.2 & 1.02 & 0.796 & 0.537 & 0.0 & 1642.3 & 1.00 & -0.016 & 0.007 & 99.0 & 1864.8 & 1.01 \\
$\lambda_{17}$ & 1.400 & -0.171 & 0.063 & 52.0 & 1460.9 & 1.01 & 0.268 & 0.173 & 53.4 & 1035.2 & 1.02 & 0.510 & 0.550 & 6.4 & 1439.1 & 1.01 & -0.167 & 0.062 & 53.6 & 1912.5 & 1.00 \\
$\lambda_{18}$ & 1.000 & -0.093 & 0.013 & 87.8 & 1045.4 & 1.01 & 0.484 & 0.253 & 4.1 & 809.6 & 1.02 & 0.785 & 0.641 & 0.0 & 1102.3 & 1.01 & -0.079 & 0.011 & 92.8 & 1397.7 & 1.01 \\
$\lambda_{19}$ & 0.700 & 0.039 & 0.005 & 99.0 & 1650.6 & 1.01 & 0.056 & 0.007 & 97.3 & 1749.0 & 1.01 & 0.279 & 0.046 & 55.3 & 2092.9 & 1.00 & 0.080 & 0.008 & 93.8 & 2355.0 & 1.00 \\
$\lambda_{20}$ & 1.100 & -0.025 & 0.005 & 99.0 & 1037.7 & 1.01 & 0.022 & 0.008 & 98.6 & 799.9 & 1.03 & 0.207 & 0.063 & 53.2 & 994.3 & 1.01 & 0.005 & 0.006 & 100.0 & 1351.8 & 1.01 \\
$\lambda_{21}$ & 0.900 & 0.021 & 0.005 & 98.0 & 1184.5 & 1.01 & 0.060 & 0.010 & 90.4 & 998.5 & 1.01 & 0.307 & 0.088 & 19.1 & 1244.1 & 1.01 & 0.057 & 0.009 & 91.8 & 1635.7 & 1.00 \\
$\lambda_{22}$ & 1.200 & -0.041 & 0.008 & 94.9 & 1033.4 & 1.01 & 0.256 & 0.105 & 16.4 & 954.1 & 1.01 & 0.656 & 0.635 & 0.0 & 1292.8 & 1.00 & 0.001 & 0.006 & 96.9 & 1435.5 & 1.01 \\
$\lambda_{23}$ & 0.800 & -0.010 & 0.003 & 98.0 & 1499.4 & 1.01 & 0.277 & 0.056 & 20.5 & 1490.7 & 1.01 & 0.681 & 0.311 & 0.0 & 1973.6 & 1.01 & 0.044 & 0.005 & 94.8 & 2211.8 & 1.00 \\
$\lambda_{24}$ & 1.000 & -0.004 & 0.005 & 96.9 & 1185.1 & 1.01 & 0.284 & 0.090 & 16.4 & 1133.7 & 1.01 & 0.668 & 0.462 & 0.0 & 1535.0 & 1.01 & 0.046 & 0.008 & 91.8 & 1638.8 & 1.01 \\
\addlinespace[0.25em]
\multicolumn{22}{@{}l}{\textit{Covariate slopes $\boldsymbol\Phi$}} \\
\midrule
$\Phi_{1,1}$ & 0.400 & 0.091 & 0.006 & 95.9 & 2137.8 & 1.01 & -1.000 & 0.160 & 0.0 & \multicolumn{1}{c}{--} & \multicolumn{1}{c}{--} & -1.000 & 0.160 & 0.0 & \multicolumn{1}{c}{--} & \multicolumn{1}{c}{--} & 0.059 & 0.006 & 96.9 & 3089.3 & 1.02 \\
$\Phi_{1,2}$ & -0.200 & 0.061 & 0.002 & 93.9 & 2094.7 & 1.00 & -1.000 & 0.040 & 0.0 & \multicolumn{1}{c}{--} & \multicolumn{1}{c}{--} & -1.000 & 0.040 & 0.0 & \multicolumn{1}{c}{--} & \multicolumn{1}{c}{--} & 0.034 & 0.001 & 95.9 & 3066.7 & 1.01 \\
$\Phi_{2,1}$ & -0.300 & 0.051 & 0.003 & 95.9 & 1944.4 & 1.00 & -1.000 & 0.090 & 0.0 & \multicolumn{1}{c}{--} & \multicolumn{1}{c}{--} & -1.000 & 0.090 & 0.0 & \multicolumn{1}{c}{--} & \multicolumn{1}{c}{--} & 0.030 & 0.002 & 94.8 & 3260.4 & 1.00 \\
$\Phi_{2,2}$ & 0.500 & 0.044 & 0.002 & 98.0 & 1293.5 & 1.01 & -1.000 & 0.250 & 0.0 & \multicolumn{1}{c}{--} & \multicolumn{1}{c}{--} & -1.000 & 0.250 & 0.0 & \multicolumn{1}{c}{--} & \multicolumn{1}{c}{--} & 0.013 & 0.001 & 99.0 & 1850.3 & 1.01 \\
$\Phi_{3,1}$ & 0.200 & -0.019 & 0.001 & 98.0 & 2166.2 & 1.00 & -1.000 & 0.040 & 0.0 & \multicolumn{1}{c}{--} & \multicolumn{1}{c}{--} & -1.000 & 0.040 & 0.0 & \multicolumn{1}{c}{--} & \multicolumn{1}{c}{--} & -0.058 & 0.001 & 96.9 & 3826.4 & 1.00 \\
$\Phi_{3,2}$ & 0.100 & -0.012 & 0.000 & 94.9 & 2352.9 & 1.00 & -1.000 & 0.010 & 0.0 & \multicolumn{1}{c}{--} & \multicolumn{1}{c}{--} & -1.000 & 0.010 & 0.0 & \multicolumn{1}{c}{--} & \multicolumn{1}{c}{--} & -0.050 & 0.000 & 95.9 & 4182.9 & 1.00 \\
$\Phi_{4,1}$ & -0.100 & 0.025 & 0.001 & 92.9 & 1989.3 & 1.00 & -1.000 & 0.010 & 0.0 & \multicolumn{1}{c}{--} & \multicolumn{1}{c}{--} & -1.000 & 0.010 & 0.0 & \multicolumn{1}{c}{--} & \multicolumn{1}{c}{--} & -0.012 & 0.001 & 96.9 & 3662.7 & 1.00 \\
$\Phi_{4,2}$ & -0.300 & 0.011 & 0.000 & 94.9 & 1599.9 & 1.01 & -1.000 & 0.090 & 0.0 & \multicolumn{1}{c}{--} & \multicolumn{1}{c}{--} & -1.000 & 0.090 & 0.0 & \multicolumn{1}{c}{--} & \multicolumn{1}{c}{--} & -0.030 & 0.000 & 91.8 & 2588.6 & 1.01 \\
$\Phi_{5,1}$ & 0.400 & 0.082 & 0.006 & 95.9 & 2351.3 & 1.01 & -1.000 & 0.160 & 0.0 & \multicolumn{1}{c}{--} & \multicolumn{1}{c}{--} & -1.000 & 0.160 & 0.0 & \multicolumn{1}{c}{--} & \multicolumn{1}{c}{--} & 0.061 & 0.005 & 97.9 & 3200.1 & 1.01 \\
$\Phi_{5,2}$ & -0.200 & 0.035 & 0.001 & 96.9 & 2297.7 & 1.00 & -1.000 & 0.040 & 0.0 & \multicolumn{1}{c}{--} & \multicolumn{1}{c}{--} & -1.000 & 0.040 & 0.0 & \multicolumn{1}{c}{--} & \multicolumn{1}{c}{--} & 0.018 & 0.001 & 96.9 & 3145.2 & 1.01 \\
$\Phi_{6,1}$ & -0.300 & 0.083 & 0.004 & 89.8 & 1921.8 & 1.01 & -1.000 & 0.090 & 0.0 & \multicolumn{1}{c}{--} & \multicolumn{1}{c}{--} & -1.000 & 0.090 & 0.0 & \multicolumn{1}{c}{--} & \multicolumn{1}{c}{--} & 0.058 & 0.004 & 92.8 & 3244.3 & 1.00 \\
\bottomrule
\end{tabular*}
\begin{tablenotes}
\footnotesize
\item RB is mean relative bias across retained replications, MSE is mean squared error, CP is empirical coverage of 95\% posterior credible intervals, ESS is effective sample size, and $\hat{R}$ is the Gelman--Rubin convergence diagnostic. Summaries use the retained-fit denominator reported in Table~\ref{tab:simulation_summary} in the main text. A dash denotes an unavailable quantity, undefined RB when the data-generating value is zero, or an unavailable sampling diagnostic for a structurally fixed parameter.
\end{tablenotes}
\end{threeparttable}
\end{table}
\clearpage
\end{landscape}

\begin{landscape}
\begin{table}[p]
\ContinuedFloat

\centering
\caption{Full parameter-level simulation results: S6 (continued).}
\label{tab:r8_s6_full_recovery_4}
\footnotesize
\sisetup{mode=math}
\setlength{\tabcolsep}{0.35pt}
\renewcommand{\arraystretch}{0.72}
\begin{threeparttable}
\begin{tabular*}{\linewidth}{@{\extracolsep{\fill}}l S[table-format=-1.3] *{4}{S[table-format=-4.3] S[table-format=3.3] S[table-format=3.1] S[table-format=5.1] S[table-format=1.2]}}
\toprule
Parameter & {True} & \multicolumn{5}{c}{CLOUD} & \multicolumn{5}{c}{TrendyOU} & \multicolumn{5}{c}{StationaryOU} & \multicolumn{5}{c}{DiagOU} \\
\cmidrule(lr){3-7}\cmidrule(lr){8-12}\cmidrule(lr){13-17}\cmidrule(l){18-22}
 & & {RB} & {MSE} & {CP} & {ESS} & {$\hat{R}$} & {RB} & {MSE} & {CP} & {ESS} & {$\hat{R}$} & {RB} & {MSE} & {CP} & {ESS} & {$\hat{R}$} & {RB} & {MSE} & {CP} & {ESS} & {$\hat{R}$} \\
\midrule
\multicolumn{22}{@{}l}{\textit{Covariate slopes $\boldsymbol\Phi$ (continued)}} \\
\midrule
$\Phi_{6,2}$ & 0.500 & 0.061 & 0.002 & 94.9 & 1322.5 & 1.01 & -1.000 & 0.250 & 0.0 & \multicolumn{1}{c}{--} & \multicolumn{1}{c}{--} & -1.000 & 0.250 & 0.0 & \multicolumn{1}{c}{--} & \multicolumn{1}{c}{--} & 0.041 & 0.002 & 99.0 & 1866.4 & 1.00 \\
$\Phi_{7,1}$ & 0.200 & -0.003 & 0.001 & 95.9 & 2116.6 & 1.00 & -1.000 & 0.040 & 0.0 & \multicolumn{1}{c}{--} & \multicolumn{1}{c}{--} & -1.000 & 0.040 & 0.0 & \multicolumn{1}{c}{--} & \multicolumn{1}{c}{--} & -0.042 & 0.001 & 96.9 & 3792.4 & 1.00 \\
$\Phi_{7,2}$ & 0.100 & 0.004 & 0.000 & 94.9 & 2307.6 & 1.00 & -1.000 & 0.010 & 0.0 & \multicolumn{1}{c}{--} & \multicolumn{1}{c}{--} & -1.000 & 0.010 & 0.0 & \multicolumn{1}{c}{--} & \multicolumn{1}{c}{--} & -0.015 & 0.000 & 95.9 & 4069.8 & 1.00 \\
$\Phi_{8,1}$ & -0.100 & 0.018 & 0.001 & 95.9 & 1963.3 & 1.00 & -1.000 & 0.010 & 0.0 & \multicolumn{1}{c}{--} & \multicolumn{1}{c}{--} & -1.000 & 0.010 & 0.0 & \multicolumn{1}{c}{--} & \multicolumn{1}{c}{--} & -0.029 & 0.001 & 95.9 & 3671.9 & 1.00 \\
$\Phi_{8,2}$ & -0.300 & 0.012 & 0.000 & 96.9 & 1553.3 & 1.01 & -1.000 & 0.090 & 0.0 & \multicolumn{1}{c}{--} & \multicolumn{1}{c}{--} & -1.000 & 0.090 & 0.0 & \multicolumn{1}{c}{--} & \multicolumn{1}{c}{--} & -0.035 & 0.001 & 92.8 & 2573.1 & 1.00 \\
\addlinespace[0.25em]
\multicolumn{22}{@{}l}{\textit{Population slopes $\boldsymbol\alpha$}} \\
\midrule
$\alpha_{1}$ & 0.500 & 0.073 & 0.004 & 94.9 & 1617.9 & 1.01 & 0.300 & 0.029 & 34.2 & 743.9 & 1.04 & -1.000 & 0.250 & 0.0 & \multicolumn{1}{c}{--} & \multicolumn{1}{c}{--} & 0.040 & 0.003 & 95.9 & 2189.2 & 1.04 \\
$\alpha_{2}$ & -0.300 & 0.028 & 0.001 & 95.9 & 1997.5 & 1.00 & -0.117 & 0.002 & 60.3 & 1490.8 & 1.01 & -1.000 & 0.090 & 0.0 & \multicolumn{1}{c}{--} & \multicolumn{1}{c}{--} & -0.006 & 0.001 & 99.0 & 3054.2 & 1.00 \\
$\alpha_{3}$ & 0.200 & 0.011 & 0.001 & 96.9 & 2297.5 & 1.00 & 0.392 & 0.007 & 2.7 & 1702.1 & 1.01 & -1.000 & 0.040 & 0.0 & \multicolumn{1}{c}{--} & \multicolumn{1}{c}{--} & -0.034 & 0.001 & 96.9 & 3777.2 & 1.01 \\
$\alpha_{4}$ & -0.400 & -0.007 & 0.001 & 96.9 & 1773.1 & 1.01 & -0.153 & 0.004 & 23.3 & 1573.7 & 1.01 & -1.000 & 0.160 & 0.0 & \multicolumn{1}{c}{--} & \multicolumn{1}{c}{--} & -0.046 & 0.001 & 94.8 & 2753.8 & 1.00 \\
$\alpha_{5}$ & 0.500 & 0.048 & 0.003 & 96.9 & 1916.4 & 1.01 & 0.247 & 0.019 & 39.7 & 919.7 & 1.02 & -1.000 & 0.250 & 0.0 & \multicolumn{1}{c}{--} & \multicolumn{1}{c}{--} & 0.028 & 0.002 & 100.0 & 2270.4 & 1.02 \\
$\alpha_{6}$ & -0.300 & 0.017 & 0.002 & 93.9 & 2022.4 & 1.01 & -0.109 & 0.002 & 65.8 & 1416.5 & 1.01 & -1.000 & 0.090 & 0.0 & \multicolumn{1}{c}{--} & \multicolumn{1}{c}{--} & 0.001 & 0.002 & 95.9 & 3114.8 & 1.00 \\
$\alpha_{7}$ & 0.200 & 0.024 & 0.001 & 95.9 & 2177.3 & 1.00 & 0.432 & 0.008 & 2.7 & 1471.4 & 1.01 & -1.000 & 0.040 & 0.0 & \multicolumn{1}{c}{--} & \multicolumn{1}{c}{--} & -0.013 & 0.001 & 94.8 & 3754.6 & 1.00 \\
$\alpha_{8}$ & -0.400 & -0.003 & 0.001 & 96.9 & 1699.8 & 1.01 & -0.153 & 0.004 & 28.8 & 1478.7 & 1.01 & -1.000 & 0.160 & 0.0 & \multicolumn{1}{c}{--} & \multicolumn{1}{c}{--} & -0.050 & 0.001 & 91.8 & 2757.3 & 1.00 \\
\addlinespace[0.25em]
\multicolumn{22}{@{}l}{\textit{Measurement effects $\boldsymbol\beta$}} \\
\midrule
$\beta_{1,1}$ & 0.300 & 0.021 & 0.045 & 92.9 & 5043.3 & 1.00 & -0.063 & 0.039 & 93.2 & 5525.3 & 1.00 & 0.002 & 0.032 & 95.7 & 8043.5 & 1.00 & 0.032 & 0.043 & 92.8 & 8084.1 & 1.00 \\
$\beta_{1,2}$ & 0.500 & 0.026 & 0.008 & 98.0 & 4670.9 & 1.00 & -0.012 & 0.010 & 97.3 & 5196.4 & 1.00 & 0.012 & 0.007 & 97.9 & 6506.5 & 1.00 & 0.013 & 0.008 & 97.9 & 6851.0 & 1.00 \\
$\beta_{2,1}$ & 0.100 & -0.388 & 0.035 & 94.9 & 4551.3 & 1.00 & -0.611 & 0.031 & 98.6 & 5161.9 & 1.00 & -0.298 & 0.029 & 100.0 & 7517.5 & 1.00 & -0.338 & 0.032 & 95.9 & 7635.5 & 1.00 \\
$\beta_{2,2}$ & 0.200 & -0.158 & 0.012 & 93.9 & 5077.4 & 1.00 & -0.193 & 0.012 & 94.5 & 5626.0 & 1.00 & -0.081 & 0.011 & 91.5 & 7374.4 & 1.00 & -0.163 & 0.012 & 92.8 & 7740.8 & 1.00 \\
$\beta_{3,1}$ & -0.100 & -0.055 & 0.040 & 96.9 & 4685.1 & 1.00 & -0.136 & 0.036 & 94.5 & 5169.4 & 1.00 & -0.418 & 0.037 & 95.7 & 7303.5 & 1.00 & -0.072 & 0.038 & 96.9 & 7657.1 & 1.00 \\
$\beta_{3,2}$ & 0.200 & -0.099 & 0.009 & 96.9 & 5164.9 & 1.00 & -0.188 & 0.010 & 95.9 & 5606.2 & 1.00 & -0.059 & 0.007 & 97.9 & 7522.0 & 1.00 & -0.117 & 0.010 & 95.9 & 7801.6 & 1.00 \\
$\beta_{4,1}$ & 0.200 & 0.156 & 0.027 & 94.9 & 4829.7 & 1.00 & 0.090 & 0.026 & 93.2 & 5478.3 & 1.00 & 0.131 & 0.029 & 95.7 & 6962.7 & 1.00 & 0.144 & 0.025 & 95.9 & 7991.6 & 1.00 \\
$\beta_{4,2}$ & -0.100 & 0.027 & 0.007 & 94.9 & 5465.8 & 1.00 & 0.008 & 0.007 & 91.8 & 6056.4 & 1.00 & 0.231 & 0.006 & 97.9 & 7875.2 & 1.00 & 0.004 & 0.007 & 94.8 & 8344.1 & 1.00 \\
$\beta_{5,1}$ & 0.100 & -0.174 & 0.030 & 99.0 & 4621.9 & 1.00 & -0.096 & 0.029 & 97.3 & 5165.8 & 1.00 & -0.405 & 0.030 & 97.9 & 6753.1 & 1.00 & -0.183 & 0.029 & 97.9 & 7525.0 & 1.00 \\
$\beta_{5,2}$ & 0.300 & -0.098 & 0.010 & 94.9 & 4725.3 & 1.00 & -0.175 & 0.012 & 89.0 & 5163.4 & 1.00 & -0.238 & 0.015 & 85.1 & 6633.5 & 1.00 & -0.115 & 0.011 & 91.8 & 7136.3 & 1.00 \\
$\beta_{6,1}$ & -0.100 & -0.152 & 0.021 & 92.9 & 4821.8 & 1.00 & -0.334 & 0.023 & 91.8 & 5473.8 & 1.00 & -0.031 & 0.024 & 93.6 & 7449.0 & 1.00 & -0.150 & 0.022 & 92.8 & 7841.7 & 1.00 \\
$\beta_{6,2}$ & -0.200 & -0.015 & 0.006 & 91.8 & 4959.8 & 1.00 & 0.016 & 0.005 & 90.4 & 5439.7 & 1.00 & 0.011 & 0.005 & 95.7 & 7157.3 & 1.00 & -0.024 & 0.006 & 91.8 & 7901.8 & 1.00 \\
$\beta_{7,1}$ & -0.200 & -0.000 & 0.015 & 96.9 & 5292.1 & 1.00 & 0.029 & 0.016 & 95.9 & 5738.4 & 1.00 & -0.063 & 0.016 & 100.0 & 7791.4 & 1.00 & 0.011 & 0.015 & 96.9 & 8148.9 & 1.00 \\
$\beta_{7,2}$ & -0.100 & 0.011 & 0.004 & 93.9 & 5493.4 & 1.00 & 0.021 & 0.004 & 94.5 & 5889.2 & 1.00 & 0.067 & 0.004 & 89.4 & 7984.7 & 1.00 & -0.010 & 0.004 & 92.8 & 8469.1 & 1.00 \\
$\beta_{8,1}$ & 0.400 & -0.021 & 0.012 & 93.9 & 4848.8 & 1.00 & -0.007 & 0.013 & 90.4 & 5176.1 & 1.00 & 0.042 & 0.015 & 87.2 & 6610.1 & 1.00 & -0.038 & 0.011 & 94.8 & 7564.3 & 1.00 \\
$\beta_{8,2}$ & 0.100 & 0.029 & 0.003 & 94.9 & 5144.6 & 1.00 & 0.018 & 0.003 & 93.2 & 5463.3 & 1.00 & 0.154 & 0.003 & 97.9 & 7114.4 & 1.00 & 0.037 & 0.003 & 96.9 & 7667.3 & 1.00 \\
$\beta_{9,1}$ & 0.200 & 0.014 & 0.013 & 92.9 & 4831.6 & 1.00 & -0.009 & 0.015 & 90.4 & 5188.5 & 1.00 & 0.171 & 0.017 & 85.1 & 7221.7 & 1.00 & 0.009 & 0.014 & 93.8 & 7770.0 & 1.00 \\
$\beta_{9,2}$ & -0.400 & 0.007 & 0.003 & 95.9 & 4502.0 & 1.00 & -0.000 & 0.002 & 98.6 & 4734.5 & 1.00 & -0.001 & 0.003 & 95.7 & 6082.3 & 1.00 & 0.002 & 0.003 & 96.9 & 6778.0 & 1.00 \\
$\beta_{10,1}$ & -0.300 & -0.038 & 0.012 & 96.9 & 4412.2 & 1.00 & -0.048 & 0.012 & 97.3 & 4790.5 & 1.00 & -0.022 & 0.011 & 95.7 & 6139.3 & 1.00 & -0.041 & 0.012 & 95.9 & 7044.2 & 1.00 \\
$\beta_{10,2}$ & 0.200 & -0.068 & 0.003 & 93.9 & 4457.6 & 1.00 & -0.089 & 0.004 & 90.4 & 4806.2 & 1.00 & -0.095 & 0.003 & 89.4 & 6162.3 & 1.00 & -0.053 & 0.004 & 93.8 & 6997.9 & 1.00 \\
\bottomrule
\end{tabular*}
\begin{tablenotes}
\footnotesize
\item RB is mean relative bias across retained replications, MSE is mean squared error, CP is empirical coverage of 95\% posterior credible intervals, ESS is effective sample size, and $\hat{R}$ is the Gelman--Rubin convergence diagnostic. Summaries use the retained-fit denominator reported in Table~\ref{tab:simulation_summary} in the main text. A dash denotes an unavailable quantity, undefined RB when the data-generating value is zero, or an unavailable sampling diagnostic for a structurally fixed parameter.
\end{tablenotes}
\end{threeparttable}
\end{table}
\clearpage
\end{landscape}

\begin{landscape}
\begin{table}[p]
\ContinuedFloat

\centering
\caption{Full parameter-level simulation results: S6 (continued).}
\label{tab:r8_s6_full_recovery_5}
\footnotesize
\sisetup{mode=math}
\setlength{\tabcolsep}{0.35pt}
\renewcommand{\arraystretch}{0.72}
\begin{threeparttable}
\begin{tabular*}{\linewidth}{@{\extracolsep{\fill}}l S[table-format=-1.3] *{4}{S[table-format=-4.3] S[table-format=3.3] S[table-format=3.1] S[table-format=5.1] S[table-format=1.2]}}
\toprule
Parameter & {True} & \multicolumn{5}{c}{CLOUD} & \multicolumn{5}{c}{TrendyOU} & \multicolumn{5}{c}{StationaryOU} & \multicolumn{5}{c}{DiagOU} \\
\cmidrule(lr){3-7}\cmidrule(lr){8-12}\cmidrule(lr){13-17}\cmidrule(l){18-22}
 & & {RB} & {MSE} & {CP} & {ESS} & {$\hat{R}$} & {RB} & {MSE} & {CP} & {ESS} & {$\hat{R}$} & {RB} & {MSE} & {CP} & {ESS} & {$\hat{R}$} & {RB} & {MSE} & {CP} & {ESS} & {$\hat{R}$} \\
\midrule
\multicolumn{22}{@{}l}{\textit{Measurement effects $\boldsymbol\beta$ (continued)}} \\
\midrule
$\beta_{11,1}$ & 0.100 & -0.074 & 0.014 & 91.8 & 5097.2 & 1.00 & -0.009 & 0.011 & 94.5 & 5554.0 & 1.00 & -0.231 & 0.015 & 91.5 & 7452.9 & 1.00 & -0.084 & 0.014 & 92.8 & 8065.6 & 1.00 \\
$\beta_{11,2}$ & -0.100 & -0.088 & 0.003 & 93.9 & 5205.6 & 1.00 & -0.053 & 0.003 & 93.2 & 5637.9 & 1.00 & -0.053 & 0.003 & 97.9 & 7590.2 & 1.00 & -0.114 & 0.003 & 93.8 & 8048.9 & 1.00 \\
$\beta_{12,1}$ & 0.200 & -0.072 & 0.012 & 93.9 & 4724.0 & 1.00 & -0.052 & 0.012 & 98.6 & 5105.3 & 1.00 & -0.109 & 0.016 & 91.5 & 6757.4 & 1.00 & -0.076 & 0.012 & 93.8 & 7551.3 & 1.00 \\
$\beta_{12,2}$ & 0.300 & -0.020 & 0.003 & 95.9 & 4448.8 & 1.00 & -0.039 & 0.003 & 94.5 & 4867.9 & 1.00 & -0.073 & 0.004 & 93.6 & 6272.9 & 1.00 & -0.006 & 0.003 & 94.8 & 6857.0 & 1.00 \\
$\beta_{13,1}$ & 0.300 & -0.042 & 0.028 & 99.0 & 5102.2 & 1.00 & -0.125 & 0.028 & 100.0 & 5700.9 & 1.00 & -0.166 & 0.025 & 100.0 & 8193.4 & 1.00 & -0.048 & 0.028 & 100.0 & 8205.9 & 1.00 \\
$\beta_{13,2}$ & 0.500 & -0.061 & 0.011 & 93.9 & 4917.9 & 1.00 & -0.090 & 0.012 & 94.5 & 5304.9 & 1.00 & -0.110 & 0.014 & 95.7 & 7190.6 & 1.00 & -0.073 & 0.012 & 95.9 & 6991.6 & 1.00 \\
$\beta_{14,1}$ & 0.100 & -0.081 & 0.029 & 98.0 & 4607.0 & 1.00 & 0.164 & 0.025 & 98.6 & 5089.8 & 1.00 & -0.502 & 0.027 & 95.7 & 7564.6 & 1.00 & -0.109 & 0.028 & 97.9 & 7676.5 & 1.00 \\
$\beta_{14,2}$ & 0.200 & -0.180 & 0.010 & 93.9 & 5121.9 & 1.00 & -0.205 & 0.011 & 90.4 & 5493.3 & 1.00 & -0.166 & 0.010 & 95.7 & 7679.7 & 1.00 & -0.199 & 0.010 & 93.8 & 7780.3 & 1.00 \\
$\beta_{15,1}$ & -0.100 & -0.072 & 0.047 & 93.9 & 4612.4 & 1.00 & -0.255 & 0.048 & 93.2 & 5190.2 & 1.00 & 0.088 & 0.057 & 89.4 & 7491.8 & 1.00 & -0.031 & 0.047 & 94.8 & 7781.9 & 1.00 \\
$\beta_{15,2}$ & 0.200 & -0.222 & 0.013 & 90.8 & 5191.7 & 1.00 & -0.254 & 0.013 & 91.8 & 5569.2 & 1.00 & -0.182 & 0.012 & 93.6 & 7932.4 & 1.00 & -0.202 & 0.012 & 92.8 & 7782.2 & 1.00 \\
$\beta_{16,1}$ & 0.200 & -0.065 & 0.025 & 94.9 & 4783.6 & 1.00 & -0.138 & 0.023 & 93.2 & 5500.3 & 1.00 & 0.035 & 0.026 & 91.5 & 7315.2 & 1.00 & -0.064 & 0.027 & 94.8 & 8144.8 & 1.00 \\
$\beta_{16,2}$ & -0.100 & 0.296 & 0.009 & 93.9 & 5592.2 & 1.00 & 0.204 & 0.008 & 93.2 & 6147.5 & 1.00 & 0.419 & 0.007 & 93.6 & 8161.1 & 1.00 & 0.249 & 0.008 & 93.8 & 8415.9 & 1.00 \\
$\beta_{17,1}$ & 0.100 & 0.208 & 0.032 & 94.9 & 4572.9 & 1.00 & 0.087 & 0.029 & 95.9 & 5301.5 & 1.00 & -0.056 & 0.032 & 93.6 & 7088.5 & 1.00 & 0.227 & 0.034 & 93.8 & 7432.7 & 1.00 \\
$\beta_{17,2}$ & 0.300 & -0.185 & 0.013 & 90.8 & 4771.1 & 1.00 & -0.241 & 0.015 & 87.7 & 5349.5 & 1.00 & -0.270 & 0.019 & 85.1 & 6945.0 & 1.00 & -0.199 & 0.013 & 88.7 & 7268.9 & 1.00 \\
$\beta_{18,1}$ & -0.100 & -0.103 & 0.016 & 98.0 & 4860.4 & 1.00 & -0.010 & 0.014 & 97.3 & 5508.8 & 1.00 & -0.193 & 0.016 & 97.9 & 7451.8 & 1.00 & -0.211 & 0.017 & 93.8 & 7910.3 & 1.00 \\
$\beta_{18,2}$ & -0.200 & -0.016 & 0.004 & 98.0 & 4939.3 & 1.00 & -0.006 & 0.004 & 98.6 & 5536.5 & 1.00 & -0.018 & 0.004 & 95.7 & 7531.1 & 1.00 & -0.035 & 0.004 & 97.9 & 7721.9 & 1.00 \\
$\beta_{19,1}$ & -0.200 & -0.065 & 0.017 & 92.9 & 5241.1 & 1.00 & -0.067 & 0.015 & 94.5 & 5677.7 & 1.00 & 0.055 & 0.018 & 93.6 & 7893.4 & 1.00 & -0.066 & 0.016 & 93.8 & 8189.8 & 1.00 \\
$\beta_{19,2}$ & -0.100 & 0.105 & 0.005 & 94.9 & 5477.4 & 1.00 & 0.130 & 0.005 & 93.2 & 6000.4 & 1.00 & -0.085 & 0.006 & 93.6 & 8199.5 & 1.00 & 0.118 & 0.005 & 93.8 & 8360.0 & 1.00 \\
$\beta_{20,1}$ & 0.400 & 0.027 & 0.012 & 96.9 & 4781.7 & 1.00 & 0.039 & 0.012 & 98.6 & 5027.7 & 1.00 & -0.000 & 0.009 & 100.0 & 6883.8 & 1.00 & 0.024 & 0.011 & 97.9 & 7413.4 & 1.00 \\
$\beta_{20,2}$ & 0.100 & -0.032 & 0.003 & 96.9 & 5131.9 & 1.00 & -0.026 & 0.003 & 95.9 & 5486.7 & 1.00 & 0.073 & 0.003 & 95.7 & 7331.8 & 1.00 & -0.027 & 0.003 & 96.9 & 7773.4 & 1.00 \\
$\beta_{21,1}$ & 0.200 & -0.021 & 0.013 & 91.8 & 4806.9 & 1.00 & -0.043 & 0.013 & 93.2 & 5243.7 & 1.00 & 0.054 & 0.014 & 93.6 & 7207.4 & 1.00 & -0.031 & 0.014 & 90.7 & 7768.0 & 1.00 \\
$\beta_{21,2}$ & -0.400 & -0.012 & 0.003 & 93.9 & 4565.7 & 1.00 & -0.010 & 0.003 & 94.5 & 4717.3 & 1.00 & -0.037 & 0.003 & 95.7 & 6117.2 & 1.00 & -0.011 & 0.003 & 95.9 & 6941.1 & 1.00 \\
$\beta_{22,1}$ & -0.300 & 0.006 & 0.012 & 96.9 & 4416.3 & 1.00 & 0.007 & 0.013 & 95.9 & 4905.8 & 1.00 & 0.092 & 0.016 & 97.9 & 6439.3 & 1.00 & -0.009 & 0.013 & 97.9 & 7071.4 & 1.00 \\
$\beta_{22,2}$ & 0.200 & -0.003 & 0.003 & 95.9 & 4463.4 & 1.00 & 0.003 & 0.003 & 98.6 & 4939.9 & 1.00 & -0.065 & 0.003 & 97.9 & 6393.9 & 1.00 & -0.002 & 0.003 & 96.9 & 6855.2 & 1.00 \\
$\beta_{23,1}$ & 0.100 & -0.031 & 0.013 & 92.9 & 5169.9 & 1.00 & -0.139 & 0.015 & 91.8 & 5609.9 & 1.00 & -0.093 & 0.016 & 89.4 & 7625.7 & 1.00 & -0.006 & 0.014 & 93.8 & 8051.4 & 1.00 \\
$\beta_{23,2}$ & -0.100 & -0.046 & 0.003 & 95.9 & 5282.1 & 1.00 & -0.018 & 0.003 & 97.3 & 5814.6 & 1.00 & 0.021 & 0.002 & 100.0 & 7730.6 & 1.00 & -0.057 & 0.003 & 96.9 & 8101.5 & 1.00 \\
$\beta_{24,1}$ & 0.200 & 0.034 & 0.015 & 92.9 & 4737.1 & 1.00 & 0.012 & 0.016 & 90.4 & 5201.9 & 1.00 & 0.014 & 0.015 & 93.6 & 6916.1 & 1.00 & 0.038 & 0.015 & 91.8 & 7402.7 & 1.00 \\
$\beta_{24,2}$ & 0.300 & -0.005 & 0.003 & 94.9 & 4448.7 & 1.00 & -0.025 & 0.003 & 94.5 & 4882.9 & 1.00 & -0.049 & 0.003 & 95.7 & 6389.5 & 1.00 & -0.000 & 0.003 & 92.8 & 6823.3 & 1.00 \\
\addlinespace[0.25em]
\multicolumn{22}{@{}l}{\textit{Random-effect scales $\boldsymbol\sigma_b$}} \\
\midrule
$\sigma_{b,1}$ & 3.700 & -0.051 & 0.098 & 93.9 & 1655.6 & 1.00 & -0.084 & 0.157 & 82.2 & 1628.5 & 1.01 & -0.087 & 0.168 & 78.7 & 1810.0 & 1.01 & -0.063 & 0.116 & 90.7 & 1720.3 & 1.04 \\
$\sigma_{b,2}$ & 3.400 & -0.095 & 0.157 & 82.7 & 1506.1 & 1.01 & -0.133 & 0.258 & 60.3 & 1315.8 & 1.02 & -0.130 & 0.265 & 61.7 & 1535.1 & 1.01 & -0.119 & 0.219 & 75.3 & 1441.8 & 1.06 \\
$\sigma_{b,3}$ & 4.800 & -0.127 & 0.435 & 60.2 & 1765.0 & 1.01 & -0.155 & 0.617 & 45.2 & 1620.0 & 1.01 & -0.152 & 0.598 & 38.3 & 1873.0 & 1.00 & -0.143 & 0.531 & 48.5 & 1761.3 & 1.03 \\
$\sigma_{b,4}$ & 3.100 & -0.039 & 0.052 & 93.9 & 1661.7 & 1.00 & -0.062 & 0.076 & 87.7 & 1600.9 & 1.01 & -0.046 & 0.060 & 95.7 & 1773.6 & 1.00 & -0.048 & 0.058 & 94.8 & 1945.6 & 1.00 \\
$\sigma_{b,5}$ & 4.000 & -0.143 & 0.389 & 46.9 & 1601.9 & 1.01 & -0.178 & 0.571 & 30.1 & 1493.0 & 1.00 & -0.178 & 0.549 & 31.9 & 1696.0 & 1.00 & -0.162 & 0.478 & 39.2 & 1934.0 & 1.00 \\
\bottomrule
\end{tabular*}
\begin{tablenotes}
\footnotesize
\item RB is mean relative bias across retained replications, MSE is mean squared error, CP is empirical coverage of 95\% posterior credible intervals, ESS is effective sample size, and $\hat{R}$ is the Gelman--Rubin convergence diagnostic. Summaries use the retained-fit denominator reported in Table~\ref{tab:simulation_summary} in the main text. A dash denotes an unavailable quantity, undefined RB when the data-generating value is zero, or an unavailable sampling diagnostic for a structurally fixed parameter.
\end{tablenotes}
\end{threeparttable}
\end{table}
\clearpage
\end{landscape}

\begin{landscape}
\begin{table}[p]
\ContinuedFloat

\centering
\caption{Full parameter-level simulation results: S6 (continued).}
\label{tab:r8_s6_full_recovery_6}
\footnotesize
\sisetup{mode=math}
\setlength{\tabcolsep}{0.35pt}
\renewcommand{\arraystretch}{0.72}
\begin{threeparttable}
\begin{tabular*}{\linewidth}{@{\extracolsep{\fill}}l S[table-format=-1.3] *{4}{S[table-format=-4.3] S[table-format=3.3] S[table-format=3.1] S[table-format=5.1] S[table-format=1.2]}}
\toprule
Parameter & {True} & \multicolumn{5}{c}{CLOUD} & \multicolumn{5}{c}{TrendyOU} & \multicolumn{5}{c}{StationaryOU} & \multicolumn{5}{c}{DiagOU} \\
\cmidrule(lr){3-7}\cmidrule(lr){8-12}\cmidrule(lr){13-17}\cmidrule(l){18-22}
 & & {RB} & {MSE} & {CP} & {ESS} & {$\hat{R}$} & {RB} & {MSE} & {CP} & {ESS} & {$\hat{R}$} & {RB} & {MSE} & {CP} & {ESS} & {$\hat{R}$} & {RB} & {MSE} & {CP} & {ESS} & {$\hat{R}$} \\
\midrule
\multicolumn{22}{@{}l}{\textit{Random-effect scales $\boldsymbol\sigma_b$ (continued)}} \\
\midrule
$\sigma_{b,6}$ & 6.100 & -0.099 & 0.425 & 43.9 & 1157.8 & 1.01 & -0.117 & 0.577 & 32.9 & 1235.3 & 1.01 & -0.114 & 0.551 & 36.2 & 1531.9 & 1.01 & -0.106 & 0.474 & 40.2 & 1716.1 & 1.01 \\
$\sigma_{b,7}$ & 5.100 & -0.049 & 0.114 & 80.6 & 1197.7 & 1.01 & -0.056 & 0.128 & 78.1 & 1287.7 & 1.01 & -0.057 & 0.142 & 78.7 & 1580.7 & 1.01 & -0.050 & 0.118 & 81.4 & 1801.4 & 1.00 \\
$\sigma_{b,8}$ & 1.700 & -0.011 & 0.010 & 95.9 & 1363.9 & 1.01 & -0.011 & 0.013 & 93.2 & 1226.6 & 1.01 & -0.011 & 0.012 & 100.0 & 1359.5 & 1.01 & -0.042 & 0.014 & 92.8 & 1649.2 & 1.01 \\
$\sigma_{b,9}$ & 2.500 & -0.007 & 0.013 & 95.9 & 1313.1 & 1.01 & -0.009 & 0.012 & 100.0 & 1328.8 & 1.01 & -0.009 & 0.016 & 95.7 & 1521.9 & 1.01 & -0.016 & 0.014 & 95.9 & 1753.6 & 1.00 \\
$\sigma_{b,10}$ & 3.300 & -0.044 & 0.044 & 83.7 & 1130.7 & 1.01 & -0.039 & 0.042 & 86.3 & 1194.9 & 1.01 & -0.062 & 0.067 & 80.9 & 1368.7 & 1.00 & -0.058 & 0.060 & 73.2 & 1635.1 & 1.01 \\
$\sigma_{b,11}$ & 4.200 & -0.032 & 0.048 & 91.8 & 1213.2 & 1.01 & -0.030 & 0.051 & 90.4 & 1370.5 & 1.01 & -0.044 & 0.074 & 87.2 & 1672.6 & 1.00 & -0.035 & 0.052 & 88.7 & 1800.5 & 1.01 \\
$\sigma_{b,12}$ & 2.800 & -0.029 & 0.026 & 88.8 & 1267.3 & 1.01 & -0.036 & 0.028 & 89.0 & 1350.6 & 1.01 & -0.055 & 0.046 & 85.1 & 1605.9 & 1.00 & -0.037 & 0.030 & 85.6 & 1732.8 & 1.01 \\
$\sigma_{b,13}$ & 3.700 & -0.061 & 0.106 & 93.9 & 1665.4 & 1.00 & -0.089 & 0.172 & 78.1 & 1712.3 & 1.01 & -0.094 & 0.186 & 74.5 & 1927.4 & 1.01 & -0.074 & 0.129 & 90.7 & 1823.1 & 1.02 \\
$\sigma_{b,14}$ & 3.400 & -0.099 & 0.153 & 78.6 & 1645.9 & 1.01 & -0.122 & 0.222 & 68.5 & 1313.4 & 1.01 & -0.163 & 0.354 & 48.9 & 1602.6 & 1.01 & -0.131 & 0.242 & 61.9 & 1431.7 & 1.02 \\
$\sigma_{b,15}$ & 4.800 & -0.124 & 0.407 & 62.2 & 1862.1 & 1.00 & -0.143 & 0.538 & 50.7 & 1596.1 & 1.01 & -0.155 & 0.618 & 38.3 & 2004.0 & 1.01 & -0.142 & 0.514 & 49.5 & 1795.0 & 1.03 \\
$\sigma_{b,16}$ & 3.100 & -0.029 & 0.047 & 94.9 & 1670.1 & 1.00 & -0.062 & 0.080 & 84.9 & 1667.0 & 1.01 & -0.063 & 0.074 & 87.2 & 1974.4 & 1.01 & -0.039 & 0.053 & 92.8 & 1965.3 & 1.01 \\
$\sigma_{b,17}$ & 4.000 & -0.139 & 0.365 & 55.1 & 1575.9 & 1.01 & -0.170 & 0.534 & 42.5 & 1472.3 & 1.01 & -0.165 & 0.517 & 42.6 & 1798.2 & 1.00 & -0.158 & 0.455 & 40.2 & 1927.4 & 1.00 \\
$\sigma_{b,18}$ & 6.100 & -0.101 & 0.447 & 40.8 & 1173.0 & 1.01 & -0.123 & 0.626 & 28.8 & 1212.4 & 1.01 & -0.129 & 0.687 & 25.5 & 1635.4 & 1.01 & -0.109 & 0.508 & 34.0 & 1760.0 & 1.01 \\
$\sigma_{b,19}$ & 5.100 & -0.052 & 0.118 & 82.7 & 1233.3 & 1.01 & -0.057 & 0.134 & 79.5 & 1313.0 & 1.00 & -0.065 & 0.147 & 74.5 & 1603.9 & 1.01 & -0.054 & 0.120 & 83.5 & 1790.0 & 1.01 \\
$\sigma_{b,20}$ & 1.700 & -0.023 & 0.010 & 96.9 & 1328.6 & 1.01 & -0.017 & 0.010 & 91.8 & 1203.7 & 1.02 & -0.033 & 0.013 & 97.9 & 1327.9 & 1.01 & -0.042 & 0.014 & 92.8 & 1652.6 & 1.01 \\
$\sigma_{b,21}$ & 2.500 & -0.012 & 0.013 & 96.9 & 1339.9 & 1.01 & -0.020 & 0.012 & 100.0 & 1331.6 & 1.01 & -0.020 & 0.016 & 97.9 & 1512.5 & 1.01 & -0.019 & 0.014 & 96.9 & 1775.6 & 1.01 \\
$\sigma_{b,22}$ & 3.300 & -0.031 & 0.034 & 91.8 & 1149.5 & 1.01 & -0.031 & 0.042 & 91.8 & 1178.1 & 1.01 & -0.038 & 0.045 & 91.5 & 1492.2 & 1.01 & -0.042 & 0.043 & 90.7 & 1621.3 & 1.01 \\
$\sigma_{b,23}$ & 4.200 & -0.026 & 0.040 & 93.9 & 1210.8 & 1.00 & -0.029 & 0.046 & 86.3 & 1334.0 & 1.01 & -0.033 & 0.051 & 89.4 & 1706.8 & 1.00 & -0.026 & 0.044 & 91.8 & 1790.3 & 1.01 \\
$\sigma_{b,24}$ & 2.800 & -0.010 & 0.021 & 90.8 & 1279.4 & 1.01 & -0.016 & 0.024 & 93.2 & 1350.0 & 1.01 & -0.033 & 0.034 & 87.2 & 1645.7 & 1.01 & -0.019 & 0.023 & 88.7 & 1752.4 & 1.00 \\
\addlinespace[0.25em]
\multicolumn{22}{@{}l}{\textit{Thresholds $\boldsymbol\theta$}} \\
\midrule
$\theta_{1}$ & 2.300 & -0.074 & 0.071 & 92.9 & 1935.7 & 1.00 & -0.077 & 0.075 & 90.4 & 2112.5 & 1.01 & 0.377 & 0.818 & 4.3 & 2031.0 & 1.00 & -0.089 & 0.082 & 91.8 & 3029.7 & 1.02 \\
$\theta_{6,1}$ & -4.000 & -0.045 & 0.121 & 89.8 & 568.1 & 1.02 & -0.065 & 0.159 & 84.9 & 849.3 & 1.01 & -0.253 & 1.099 & 6.4 & 1088.1 & 1.01 & -0.051 & 0.127 & 88.7 & 1127.6 & 1.02 \\
$\theta_{6,2}$ & -1.000 & -0.019 & 0.083 & 94.9 & 482.5 & 1.03 & -0.036 & 0.081 & 95.9 & 804.4 & 1.01 & -0.797 & 0.727 & 17.0 & 1078.8 & 1.01 & -0.026 & 0.079 & 95.9 & 997.6 & 1.02 \\
$\theta_{6,3}$ & 2.700 & -0.053 & 0.108 & 90.8 & 524.8 & 1.02 & -0.062 & 0.109 & 89.0 & 857.0 & 1.01 & 0.218 & 0.459 & 55.3 & 1160.6 & 1.01 & -0.058 & 0.104 & 90.7 & 1078.8 & 1.01 \\
$\theta_{2}$ & 1.500 & -0.067 & 0.063 & 90.8 & 2098.7 & 1.01 & -0.066 & 0.057 & 89.0 & 2236.7 & 1.01 & 0.843 & 1.688 & 0.0 & 1671.7 & 1.01 & -0.085 & 0.067 & 91.8 & 3297.3 & 1.01 \\
$\theta_{7,1}$ & -5.500 & -0.024 & 0.095 & 93.9 & 717.5 & 1.01 & -0.028 & 0.111 & 89.0 & 838.7 & 1.02 & 0.052 & 0.151 & 85.1 & 1178.2 & 1.01 & -0.024 & 0.095 & 91.8 & 1401.3 & 1.01 \\
$\theta_{7,2}$ & -2.500 & -0.032 & 0.071 & 91.8 & 566.4 & 1.02 & -0.039 & 0.086 & 86.3 & 665.4 & 1.02 & 0.142 & 0.184 & 74.5 & 932.4 & 1.01 & -0.036 & 0.073 & 89.7 & 1138.6 & 1.01 \\
$\theta_{7,3}$ & 2.600 & 0.003 & 0.063 & 93.9 & 621.2 & 1.01 & -0.008 & 0.061 & 95.9 & 728.0 & 1.01 & -0.173 & 0.264 & 51.1 & 943.8 & 1.01 & 0.005 & 0.062 & 94.8 & 1244.9 & 1.01 \\
$\theta_{3}$ & 0.800 & -0.077 & 0.065 & 92.9 & 1646.7 & 1.01 & -0.031 & 0.057 & 93.2 & 1857.8 & 1.01 & 1.532 & 1.563 & 0.0 & 1880.6 & 1.00 & -0.087 & 0.069 & 92.8 & 3059.0 & 1.00 \\
$\theta_{8,1}$ & -4.500 & 0.001 & 0.043 & 95.9 & 1873.4 & 1.01 & 0.004 & 0.050 & 94.5 & 1824.1 & 1.00 & 0.156 & 0.555 & 12.8 & 1662.0 & 1.01 & -0.004 & 0.043 & 94.8 & 3006.0 & 1.01 \\
$\theta_{8,2}$ & -1.500 & 0.000 & 0.017 & 90.8 & 2821.7 & 1.00 & 0.000 & 0.017 & 91.8 & 2921.5 & 1.00 & 0.449 & 0.476 & 0.0 & 1933.8 & 1.01 & -0.005 & 0.017 & 92.8 & 4781.6 & 1.00 \\
$\theta_{8,3}$ & 2.000 & -0.013 & 0.015 & 94.9 & 2280.0 & 1.00 & -0.013 & 0.018 & 95.9 & 2098.4 & 1.01 & -0.341 & 0.483 & 0.0 & 2226.1 & 1.00 & -0.019 & 0.017 & 96.9 & 3250.5 & 1.00 \\
$\theta_{4}$ & 1.200 & -0.005 & 0.031 & 95.9 & 1877.3 & 1.00 & -0.032 & 0.035 & 95.9 & 2254.4 & 1.01 & -0.627 & 0.601 & 0.0 & 2187.1 & 1.00 & -0.011 & 0.030 & 95.9 & 3545.2 & 1.00 \\
$\theta_{9,1}$ & -3.000 & -0.006 & 0.031 & 91.8 & 1480.8 & 1.01 & -0.001 & 0.029 & 95.9 & 1600.9 & 1.01 & 0.177 & 0.324 & 14.9 & 1789.5 & 1.01 & -0.010 & 0.031 & 92.8 & 2743.4 & 1.01 \\
\bottomrule
\end{tabular*}
\begin{tablenotes}
\footnotesize
\item RB is mean relative bias across retained replications, MSE is mean squared error, CP is empirical coverage of 95\% posterior credible intervals, ESS is effective sample size, and $\hat{R}$ is the Gelman--Rubin convergence diagnostic. Summaries use the retained-fit denominator reported in Table~\ref{tab:simulation_summary} in the main text. A dash denotes an unavailable quantity, undefined RB when the data-generating value is zero, or an unavailable sampling diagnostic for a structurally fixed parameter.
\end{tablenotes}
\end{threeparttable}
\end{table}
\clearpage
\end{landscape}

\begin{landscape}
\begin{table}[p]
\ContinuedFloat

\centering
\caption{Full parameter-level simulation results: S6 (continued).}
\label{tab:r8_s6_full_recovery_7}
\footnotesize
\sisetup{mode=math}
\setlength{\tabcolsep}{0.35pt}
\renewcommand{\arraystretch}{0.72}
\begin{threeparttable}
\begin{tabular*}{\linewidth}{@{\extracolsep{\fill}}l S[table-format=-1.3] *{4}{S[table-format=-4.3] S[table-format=3.3] S[table-format=3.1] S[table-format=5.1] S[table-format=1.2]}}
\toprule
Parameter & {True} & \multicolumn{5}{c}{CLOUD} & \multicolumn{5}{c}{TrendyOU} & \multicolumn{5}{c}{StationaryOU} & \multicolumn{5}{c}{DiagOU} \\
\cmidrule(lr){3-7}\cmidrule(lr){8-12}\cmidrule(lr){13-17}\cmidrule(l){18-22}
 & & {RB} & {MSE} & {CP} & {ESS} & {$\hat{R}$} & {RB} & {MSE} & {CP} & {ESS} & {$\hat{R}$} & {RB} & {MSE} & {CP} & {ESS} & {$\hat{R}$} & {RB} & {MSE} & {CP} & {ESS} & {$\hat{R}$} \\
\midrule
\multicolumn{22}{@{}l}{\textit{Thresholds $\boldsymbol\theta$ (continued)}} \\
\midrule
$\theta_{9,2}$ & 0.000 & \multicolumn{1}{c}{--} & 0.020 & 94.9 & 1541.7 & 1.01 & \multicolumn{1}{c}{--} & 0.020 & 97.3 & 1808.5 & 1.01 & \multicolumn{1}{c}{--} & 0.335 & 2.1 & 2074.8 & 1.01 & \multicolumn{1}{c}{--} & 0.021 & 95.9 & 3061.3 & 1.01 \\
$\theta_{9,3}$ & 3.000 & 0.011 & 0.030 & 92.9 & 1594.1 & 1.01 & 0.002 & 0.034 & 93.2 & 1721.6 & 1.01 & -0.178 & 0.311 & 10.6 & 2104.5 & 1.01 & 0.011 & 0.030 & 94.8 & 2769.9 & 1.01 \\
$\theta_{5}$ & 2.100 & -0.106 & 0.088 & 85.7 & 1820.4 & 1.00 & -0.146 & 0.143 & 78.1 & 2090.5 & 1.00 & -0.571 & 1.486 & 0.0 & 2021.9 & 1.00 & -0.121 & 0.099 & 81.4 & 3310.6 & 1.00 \\
$\theta_{10,1}$ & -6.000 & -0.028 & 0.090 & 90.8 & 1156.5 & 1.01 & -0.025 & 0.091 & 93.2 & 1260.5 & 1.01 & -0.197 & 1.462 & 2.1 & 1460.7 & 1.00 & -0.036 & 0.117 & 89.7 & 2054.6 & 1.01 \\
$\theta_{10,2}$ & -2.000 & -0.047 & 0.036 & 93.9 & 1132.0 & 1.01 & -0.043 & 0.038 & 93.2 & 1486.0 & 1.01 & -0.535 & 1.174 & 0.0 & 2068.7 & 1.00 & -0.059 & 0.043 & 90.7 & 2330.4 & 1.01 \\
$\theta_{10,3}$ & 1.500 & -0.007 & 0.029 & 92.9 & 1228.7 & 1.01 & 0.012 & 0.031 & 94.5 & 1735.7 & 1.01 & 0.640 & 0.960 & 0.0 & 1975.7 & 1.00 & -0.010 & 0.030 & 93.8 & 2637.5 & 1.00 \\
$\theta_{13}$ & 2.300 & -0.057 & 0.066 & 94.9 & 1952.8 & 1.00 & -0.056 & 0.069 & 95.9 & 2142.5 & 1.00 & 0.286 & 0.518 & 27.7 & 2271.0 & 1.00 & -0.069 & 0.076 & 91.8 & 3168.7 & 1.01 \\
$\theta_{11,1}$ & -5.000 & -0.006 & 0.063 & 95.9 & 893.1 & 1.01 & -0.011 & 0.071 & 90.4 & 1078.7 & 1.01 & -0.149 & 0.611 & 17.0 & 1450.3 & 1.01 & -0.008 & 0.062 & 95.9 & 1734.1 & 1.01 \\
$\theta_{11,2}$ & -1.000 & 0.022 & 0.049 & 91.8 & 671.4 & 1.01 & 0.014 & 0.052 & 87.7 & 862.7 & 1.01 & -0.669 & 0.485 & 4.3 & 1250.2 & 1.01 & 0.017 & 0.048 & 92.8 & 1372.0 & 1.01 \\
$\theta_{11,3}$ & 2.500 & -0.022 & 0.052 & 90.8 & 733.2 & 1.02 & -0.018 & 0.048 & 90.4 & 951.7 & 1.01 & 0.246 & 0.427 & 21.3 & 1449.1 & 1.01 & -0.024 & 0.054 & 91.8 & 1494.9 & 1.01 \\
$\theta_{14}$ & 1.500 & -0.037 & 0.051 & 91.8 & 2155.2 & 1.00 & -0.051 & 0.046 & 93.2 & 2225.6 & 1.01 & 0.658 & 1.048 & 4.3 & 1818.3 & 1.01 & -0.062 & 0.055 & 90.7 & 3398.6 & 1.00 \\
$\theta_{12,1}$ & -4.000 & -0.003 & 0.035 & 96.9 & 1355.9 & 1.01 & -0.011 & 0.039 & 95.9 & 1586.7 & 1.01 & -0.230 & 0.878 & 0.0 & 1915.1 & 1.00 & -0.006 & 0.040 & 94.8 & 2439.9 & 1.00 \\
$\theta_{12,2}$ & -0.500 & 0.036 & 0.022 & 94.9 & 1373.4 & 1.01 & 0.015 & 0.020 & 97.3 & 1856.0 & 1.00 & -1.632 & 0.686 & 0.0 & 2382.0 & 1.00 & 0.030 & 0.022 & 93.8 & 2818.6 & 1.00 \\
$\theta_{12,3}$ & 2.000 & -0.024 & 0.028 & 91.8 & 1477.6 & 1.01 & -0.018 & 0.027 & 91.8 & 1981.2 & 1.00 & 0.373 & 0.583 & 0.0 & 2285.1 & 1.00 & -0.022 & 0.027 & 91.8 & 2988.3 & 1.00 \\
$\theta_{15}$ & 0.800 & -0.108 & 0.073 & 90.8 & 1648.2 & 1.01 & -0.106 & 0.075 & 89.0 & 1912.8 & 1.00 & 1.144 & 0.936 & 2.1 & 2078.6 & 1.00 & -0.124 & 0.074 & 90.7 & 3100.2 & 1.00 \\
$\theta_{18,1}$ & -4.000 & -0.043 & 0.125 & 89.8 & 564.6 & 1.01 & -0.064 & 0.173 & 84.9 & 846.1 & 1.01 & -0.286 & 1.413 & 6.4 & 1128.8 & 1.01 & -0.051 & 0.136 & 87.6 & 1121.0 & 1.02 \\
$\theta_{18,2}$ & -1.000 & -0.042 & 0.082 & 93.9 & 475.6 & 1.02 & -0.056 & 0.080 & 94.5 & 820.9 & 1.01 & -0.941 & 0.991 & 14.9 & 1153.8 & 1.01 & -0.047 & 0.085 & 91.8 & 969.2 & 1.02 \\
$\theta_{18,3}$ & 2.700 & -0.051 & 0.099 & 92.9 & 514.2 & 1.02 & -0.069 & 0.107 & 87.7 & 863.1 & 1.01 & 0.254 & 0.590 & 34.0 & 1289.3 & 1.00 & -0.059 & 0.105 & 92.8 & 1042.5 & 1.01 \\
$\theta_{16}$ & 1.200 & 0.016 & 0.034 & 92.9 & 1867.9 & 1.00 & -0.026 & 0.037 & 91.8 & 2272.7 & 1.01 & -0.723 & 0.796 & 2.1 & 2491.2 & 1.00 & 0.014 & 0.033 & 93.8 & 3535.6 & 1.00 \\
$\theta_{19,1}$ & -5.500 & -0.019 & 0.076 & 95.9 & 712.5 & 1.02 & -0.026 & 0.088 & 97.3 & 848.6 & 1.01 & 0.052 & 0.155 & 85.1 & 1177.4 & 1.01 & -0.021 & 0.073 & 97.9 & 1430.4 & 1.01 \\
$\theta_{19,2}$ & -2.500 & -0.029 & 0.053 & 93.9 & 560.4 & 1.02 & -0.031 & 0.049 & 94.5 & 670.1 & 1.01 & 0.148 & 0.186 & 78.7 & 939.1 & 1.01 & -0.033 & 0.052 & 93.8 & 1160.7 & 1.01 \\
$\theta_{19,3}$ & 2.600 & -0.011 & 0.050 & 94.9 & 617.0 & 1.02 & -0.024 & 0.050 & 94.5 & 731.1 & 1.01 & -0.209 & 0.341 & 42.6 & 955.0 & 1.01 & -0.008 & 0.049 & 95.9 & 1258.6 & 1.01 \\
$\theta_{17}$ & 2.100 & -0.113 & 0.102 & 85.7 & 1804.6 & 1.00 & -0.152 & 0.146 & 72.6 & 2064.5 & 1.01 & -0.635 & 1.854 & 0.0 & 2216.6 & 1.01 & -0.129 & 0.119 & 78.4 & 3299.3 & 1.00 \\
$\theta_{20,1}$ & -4.500 & -0.012 & 0.033 & 96.9 & 1878.4 & 1.01 & -0.011 & 0.028 & 97.3 & 1670.4 & 1.01 & 0.129 & 0.392 & 23.4 & 1605.5 & 1.01 & -0.014 & 0.037 & 97.9 & 3076.7 & 1.01 \\
$\theta_{20,2}$ & -1.500 & -0.016 & 0.011 & 96.9 & 2875.4 & 1.00 & -0.012 & 0.009 & 97.3 & 2854.8 & 1.01 & 0.419 & 0.413 & 0.0 & 1958.4 & 1.01 & -0.018 & 0.012 & 97.9 & 4835.9 & 1.00 \\
$\theta_{20,3}$ & 2.000 & -0.009 & 0.013 & 98.0 & 2200.7 & 1.00 & -0.014 & 0.015 & 98.6 & 1959.0 & 1.01 & -0.352 & 0.507 & 0.0 & 2578.4 & 1.01 & -0.011 & 0.015 & 97.9 & 3232.8 & 1.00 \\
$\theta_{21,1}$ & -3.000 & -0.006 & 0.028 & 91.8 & 1478.4 & 1.01 & -0.010 & 0.029 & 94.5 & 1566.0 & 1.01 & 0.173 & 0.304 & 14.9 & 1759.2 & 1.00 & -0.008 & 0.029 & 91.8 & 2820.0 & 1.01 \\
$\theta_{21,2}$ & 0.000 & \multicolumn{1}{c}{--} & 0.022 & 90.8 & 1587.4 & 1.01 & \multicolumn{1}{c}{--} & 0.021 & 94.5 & 1814.8 & 1.00 & \multicolumn{1}{c}{--} & 0.311 & 10.6 & 2288.9 & 1.00 & \multicolumn{1}{c}{--} & 0.022 & 91.8 & 3124.4 & 1.01 \\
$\theta_{21,3}$ & 3.000 & -0.000 & 0.034 & 88.8 & 1627.8 & 1.01 & -0.006 & 0.035 & 90.4 & 1705.0 & 1.01 & -0.186 & 0.348 & 10.6 & 2217.4 & 1.00 & -0.000 & 0.035 & 90.7 & 2809.5 & 1.01 \\
$\theta_{22,1}$ & -6.000 & -0.027 & 0.080 & 94.9 & 1205.4 & 1.01 & -0.029 & 0.090 & 94.5 & 1294.2 & 1.01 & -0.184 & 1.279 & 4.3 & 1671.8 & 1.00 & -0.030 & 0.093 & 93.8 & 2059.3 & 1.00 \\
$\theta_{22,2}$ & -2.000 & -0.034 & 0.033 & 92.9 & 1155.4 & 1.01 & -0.029 & 0.035 & 93.2 & 1528.5 & 1.01 & -0.513 & 1.085 & 0.0 & 2242.6 & 1.00 & -0.039 & 0.034 & 94.8 & 2297.0 & 1.00 \\
$\theta_{22,3}$ & 1.500 & -0.009 & 0.023 & 96.9 & 1202.9 & 1.01 & -0.011 & 0.026 & 95.9 & 1753.7 & 1.01 & 0.669 & 1.045 & 0.0 & 2118.7 & 1.00 & -0.011 & 0.024 & 96.9 & 2561.3 & 1.00 \\
$\theta_{23,1}$ & -5.000 & -0.013 & 0.058 & 94.9 & 861.1 & 1.01 & -0.018 & 0.068 & 93.2 & 1067.2 & 1.01 & -0.159 & 0.673 & 6.4 & 1475.6 & 1.01 & -0.013 & 0.058 & 95.9 & 1706.3 & 1.01 \\
\bottomrule
\end{tabular*}
\begin{tablenotes}
\footnotesize
\item RB is mean relative bias across retained replications, MSE is mean squared error, CP is empirical coverage of 95\% posterior credible intervals, ESS is effective sample size, and $\hat{R}$ is the Gelman--Rubin convergence diagnostic. Summaries use the retained-fit denominator reported in Table~\ref{tab:simulation_summary} in the main text. A dash denotes an unavailable quantity, undefined RB when the data-generating value is zero, or an unavailable sampling diagnostic for a structurally fixed parameter.
\end{tablenotes}
\end{threeparttable}
\end{table}
\clearpage
\end{landscape}

\begin{landscape}
\begin{table}[p]
\ContinuedFloat

\centering
\caption{Full parameter-level simulation results: S6 (continued).}
\label{tab:r8_s6_full_recovery_8}
\footnotesize
\sisetup{mode=math}
\setlength{\tabcolsep}{0.35pt}
\renewcommand{\arraystretch}{0.72}
\begin{threeparttable}
\begin{tabular*}{\linewidth}{@{\extracolsep{\fill}}l S[table-format=-1.3] *{4}{S[table-format=-4.3] S[table-format=3.3] S[table-format=3.1] S[table-format=5.1] S[table-format=1.2]}}
\toprule
Parameter & {True} & \multicolumn{5}{c}{CLOUD} & \multicolumn{5}{c}{TrendyOU} & \multicolumn{5}{c}{StationaryOU} & \multicolumn{5}{c}{DiagOU} \\
\cmidrule(lr){3-7}\cmidrule(lr){8-12}\cmidrule(lr){13-17}\cmidrule(l){18-22}
 & & {RB} & {MSE} & {CP} & {ESS} & {$\hat{R}$} & {RB} & {MSE} & {CP} & {ESS} & {$\hat{R}$} & {RB} & {MSE} & {CP} & {ESS} & {$\hat{R}$} & {RB} & {MSE} & {CP} & {ESS} & {$\hat{R}$} \\
\midrule
\multicolumn{22}{@{}l}{\textit{Thresholds $\boldsymbol\theta$ (continued)}} \\
\midrule
$\theta_{23,2}$ & -1.000 & -0.041 & 0.044 & 94.9 & 646.4 & 1.01 & -0.045 & 0.046 & 95.9 & 851.9 & 1.01 & -0.740 & 0.586 & 4.3 & 1270.5 & 1.01 & -0.047 & 0.044 & 93.8 & 1356.6 & 1.01 \\
$\theta_{23,3}$ & 2.500 & 0.002 & 0.051 & 92.9 & 724.0 & 1.01 & 0.000 & 0.048 & 94.5 & 933.6 & 1.01 & 0.277 & 0.538 & 8.5 & 1473.4 & 1.01 & 0.005 & 0.052 & 92.8 & 1493.0 & 1.01 \\
$\theta_{24,1}$ & -4.000 & -0.005 & 0.041 & 95.9 & 1316.5 & 1.01 & -0.011 & 0.045 & 93.2 & 1585.5 & 1.00 & -0.220 & 0.811 & 0.0 & 2116.0 & 1.00 & -0.006 & 0.043 & 92.8 & 2447.6 & 1.01 \\
$\theta_{24,2}$ & -0.500 & -0.003 & 0.029 & 90.8 & 1330.4 & 1.01 & -0.017 & 0.032 & 91.8 & 1871.0 & 1.00 & -1.695 & 0.755 & 0.0 & 2499.4 & 1.00 & -0.006 & 0.028 & 91.8 & 2814.6 & 1.01 \\
$\theta_{24,3}$ & 2.000 & -0.009 & 0.032 & 92.9 & 1443.8 & 1.01 & -0.009 & 0.031 & 94.5 & 1988.6 & 1.00 & 0.409 & 0.718 & 0.0 & 2337.3 & 1.00 & -0.009 & 0.031 & 94.8 & 2945.4 & 1.01 \\
\addlinespace[0.25em]
\multicolumn{22}{@{}l}{\textit{Correlations $\boldsymbol\rho$}} \\
\midrule
$\rho_{1}$ & 0.288 & -0.307 & 0.014 & 91.8 & 529.7 & 1.03 & -1.805 & 0.278 & 0.0 & 392.0 & 1.03 & -2.685 & 0.602 & 0.0 & 579.4 & 1.02 & -1.000 & 0.083 & 0.0 & \multicolumn{1}{c}{--} & \multicolumn{1}{c}{--} \\
$\rho_{2}$ & -0.072 & -0.505 & 0.007 & 96.9 & 479.4 & 1.03 & 0.288 & 0.006 & 98.6 & 371.5 & 1.02 & -5.936 & 0.187 & 0.0 & 608.4 & 1.01 & -1.000 & 0.005 & 0.0 & \multicolumn{1}{c}{--} & \multicolumn{1}{c}{--} \\
$\rho_{3}$ & -0.152 & -0.267 & 0.008 & 94.9 & 446.0 & 1.05 & -2.009 & 0.101 & 6.8 & 352.4 & 1.03 & 1.685 & 0.069 & 2.1 & 499.7 & 1.04 & -1.000 & 0.023 & 0.0 & \multicolumn{1}{c}{--} & \multicolumn{1}{c}{--} \\
$\rho_{4}$ & -0.215 & -0.563 & 0.021 & 89.8 & 697.7 & 1.02 & -1.571 & 0.121 & 6.8 & 564.3 & 1.02 & -4.116 & 0.789 & 0.0 & 774.8 & 1.02 & -1.000 & 0.046 & 0.0 & \multicolumn{1}{c}{--} & \multicolumn{1}{c}{--} \\
$\rho_{5}$ & 0.258 & -0.305 & 0.012 & 92.9 & 559.2 & 1.03 & -1.872 & 0.241 & 0.0 & 396.9 & 1.04 & -2.829 & 0.537 & 0.0 & 594.9 & 1.02 & -1.000 & 0.067 & 0.0 & \multicolumn{1}{c}{--} & \multicolumn{1}{c}{--} \\
$\rho_{6}$ & -0.241 & -0.278 & 0.010 & 92.9 & 394.5 & 1.03 & -0.097 & 0.007 & 95.9 & 326.1 & 1.03 & -2.287 & 0.308 & 0.0 & 565.6 & 1.03 & -1.000 & 0.058 & 0.0 & \multicolumn{1}{c}{--} & \multicolumn{1}{c}{--} \\
$\rho_{7}$ & -0.019 & -1.056 & 0.007 & 93.9 & 392.6 & 1.03 & -12.768 & 0.068 & 17.8 & 328.2 & 1.04 & 18.674 & 0.135 & 0.0 & 501.1 & 1.02 & -1.000 & 0.000 & 0.0 & \multicolumn{1}{c}{--} & \multicolumn{1}{c}{--} \\
$\rho_{8}$ & -0.105 & -0.340 & 0.006 & 93.9 & 678.2 & 1.02 & -2.024 & 0.049 & 8.2 & 876.8 & 1.01 & 0.566 & 0.006 & 91.5 & 1238.4 & 1.01 & -1.000 & 0.011 & 0.0 & \multicolumn{1}{c}{--} & \multicolumn{1}{c}{--} \\
$\rho_{9}$ & -0.102 & -0.171 & 0.004 & 96.9 & 671.7 & 1.02 & 3.774 & 0.152 & 0.0 & 697.3 & 1.02 & -0.514 & 0.006 & 80.9 & 1103.4 & 1.02 & -1.000 & 0.010 & 0.0 & \multicolumn{1}{c}{--} & \multicolumn{1}{c}{--} \\
$\rho_{10}$ & 0.128 & -0.210 & 0.008 & 96.9 & 980.0 & 1.01 & -2.795 & 0.134 & 1.4 & 1086.7 & 1.02 & -4.817 & 0.382 & 0.0 & 1352.0 & 1.01 & -1.000 & 0.016 & 0.0 & \multicolumn{1}{c}{--} & \multicolumn{1}{c}{--} \\
$\rho_{11}$ & -0.089 & -0.517 & 0.008 & 94.9 & 844.2 & 1.02 & -7.314 & 0.431 & 0.0 & 757.2 & 1.02 & -8.551 & 0.587 & 0.0 & 825.0 & 1.01 & -1.000 & 0.008 & 0.0 & \multicolumn{1}{c}{--} & \multicolumn{1}{c}{--} \\
$\rho_{12}$ & -0.086 & -0.103 & 0.004 & 99.0 & 592.7 & 1.04 & -2.203 & 0.039 & 17.8 & 842.2 & 1.02 & 0.958 & 0.010 & 76.6 & 1115.9 & 1.01 & -1.000 & 0.007 & 0.0 & \multicolumn{1}{c}{--} & \multicolumn{1}{c}{--} \\
$\rho_{13}$ & 0.087 & -0.141 & 0.004 & 96.9 & 625.8 & 1.01 & -5.670 & 0.245 & 0.0 & 715.8 & 1.01 & -0.947 & 0.010 & 68.1 & 1111.1 & 1.01 & -1.000 & 0.008 & 0.0 & \multicolumn{1}{c}{--} & \multicolumn{1}{c}{--} \\
$\rho_{14}$ & 0.002 & -1.891 & 0.002 & 99.0 & 1041.7 & 1.01 & -82.146 & 0.029 & 19.2 & 1025.2 & 1.02 & -216.977 & 0.184 & 0.0 & 1109.1 & 1.01 & -1.000 & 0.000 & 0.0 & \multicolumn{1}{c}{--} & \multicolumn{1}{c}{--} \\
$\rho_{15}$ & 0.128 & -0.170 & 0.006 & 99.0 & 1337.8 & 1.01 & -0.653 & 0.012 & 91.8 & 1104.9 & 1.02 & 2.331 & 0.092 & 0.0 & 1524.0 & 1.00 & -1.000 & 0.016 & 0.0 & \multicolumn{1}{c}{--} & \multicolumn{1}{c}{--} \\
$\rho_{16}$ & -0.279 & -0.133 & 0.004 & 99.0 & 1171.9 & 1.01 & -1.013 & 0.083 & 1.4 & 1007.1 & 1.01 & -0.194 & 0.006 & 83.0 & 1511.8 & 1.00 & -1.000 & 0.078 & 0.0 & \multicolumn{1}{c}{--} & \multicolumn{1}{c}{--} \\
$\rho_{17}$ & 0.054 & 0.016 & 0.005 & 91.8 & 949.4 & 1.01 & 1.604 & 0.013 & 65.8 & 935.8 & 1.01 & 6.238 & 0.116 & 0.0 & 1238.7 & 1.01 & -1.000 & 0.003 & 0.0 & \multicolumn{1}{c}{--} & \multicolumn{1}{c}{--} \\
$\rho_{18}$ & 0.051 & -0.348 & 0.003 & 96.9 & 953.3 & 1.01 & -3.717 & 0.040 & 8.2 & 960.1 & 1.01 & -8.913 & 0.212 & 0.0 & 1172.8 & 1.01 & -1.000 & 0.003 & 0.0 & \multicolumn{1}{c}{--} & \multicolumn{1}{c}{--} \\
$\rho_{19}$ & -0.175 & -0.184 & 0.006 & 96.9 & 1482.5 & 1.03 & -1.666 & 0.092 & 4.1 & 1168.0 & 1.02 & 1.263 & 0.051 & 2.1 & 1547.7 & 1.00 & -1.000 & 0.031 & 0.0 & \multicolumn{1}{c}{--} & \multicolumn{1}{c}{--} \\
$\rho_{20}$ & -0.082 & -0.115 & 0.004 & 94.9 & 1333.5 & 1.01 & 4.804 & 0.157 & 0.0 & 1006.0 & 1.02 & -0.302 & 0.005 & 87.2 & 1562.0 & 1.01 & -1.000 & 0.007 & 0.0 & \multicolumn{1}{c}{--} & \multicolumn{1}{c}{--} \\
$\rho_{21}$ & -0.00005 & 43.255 & 0.003 & 94.9 & 1028.9 & 1.01 & 3313.815 & 0.031 & 23.3 & 1140.5 & 1.01 & 9087.684 & 0.209 & 0.0 & 1333.0 & 1.01 & -1.000 & 0.000 & 0.0 & \multicolumn{1}{c}{--} & \multicolumn{1}{c}{--} \\
$\rho_{22}$ & -0.007 & 1.130 & 0.003 & 95.9 & 1105.2 & 1.01 & -51.573 & 0.142 & 0.0 & 1015.8 & 1.01 & -83.738 & 0.371 & 0.0 & 1220.6 & 1.01 & -1.000 & 0.000 & 0.0 & \multicolumn{1}{c}{--} & \multicolumn{1}{c}{--} \\
$\rho_{23}$ & -0.311 & -0.338 & 0.016 & 88.8 & 665.3 & 1.02 & 0.386 & 0.019 & 75.3 & 609.0 & 1.02 & 0.980 & 0.095 & 0.0 & 773.0 & 1.02 & -1.000 & 0.097 & 0.0 & \multicolumn{1}{c}{--} & \multicolumn{1}{c}{--} \\
$\rho_{24}$ & -0.242 & -0.363 & 0.013 & 85.7 & 459.1 & 1.02 & -0.250 & 0.009 & 90.4 & 421.6 & 1.03 & -2.259 & 0.303 & 0.0 & 469.8 & 1.03 & -1.000 & 0.059 & 0.0 & \multicolumn{1}{c}{--} & \multicolumn{1}{c}{--} \\
$\rho_{25}$ & 0.077 & -0.674 & 0.008 & 92.9 & 465.5 & 1.02 & 1.996 & 0.030 & 56.2 & 492.5 & 1.03 & -5.374 & 0.175 & 0.0 & 569.6 & 1.02 & -1.000 & 0.006 & 0.0 & \multicolumn{1}{c}{--} & \multicolumn{1}{c}{--} \\
$\rho_{26}$ & -0.082 & -0.098 & 0.004 & 99.0 & 657.0 & 1.01 & -2.348 & 0.041 & 17.8 & 838.5 & 1.02 & 0.978 & 0.009 & 72.3 & 1129.3 & 1.01 & -1.000 & 0.007 & 0.0 & \multicolumn{1}{c}{--} & \multicolumn{1}{c}{--} \\
$\rho_{27}$ & 0.157 & -0.127 & 0.004 & 99.0 & 661.2 & 1.02 & -3.388 & 0.286 & 0.0 & 788.7 & 1.02 & -0.907 & 0.023 & 25.5 & 1162.6 & 1.01 & -1.000 & 0.025 & 0.0 & \multicolumn{1}{c}{--} & \multicolumn{1}{c}{--} \\
$\rho_{28}$ & 0.008 & 0.258 & 0.003 & 96.9 & 1099.5 & 1.02 & -20.801 & 0.031 & 27.4 & 1049.1 & 1.01 & -57.660 & 0.207 & 0.0 & 1272.8 & 1.01 & -1.000 & 0.000 & 0.0 & \multicolumn{1}{c}{--} & \multicolumn{1}{c}{--} \\
\bottomrule
\end{tabular*}
\begin{tablenotes}
\footnotesize
\item RB is mean relative bias across retained replications, MSE is mean squared error, CP is empirical coverage of 95\% posterior credible intervals, ESS is effective sample size, and $\hat{R}$ is the Gelman--Rubin convergence diagnostic. Summaries use the retained-fit denominator reported in Table~\ref{tab:simulation_summary} in the main text. A dash denotes an unavailable quantity, undefined RB when the data-generating value is zero, or an unavailable sampling diagnostic for a structurally fixed parameter.
\end{tablenotes}
\end{threeparttable}
\end{table}
\clearpage
\end{landscape}

\begin{table}[p]
\centering
\caption{Additional grouped recovery diagnostics for S5 ($R=8$, real-eigenvalue drift).}
\label{tab:r8_s5_summary}
\CloudStatTableStyle
\setlength{\tabcolsep}{2.3pt}
\begin{threeparttable}
\begin{tabular*}{\linewidth}{@{\extracolsep{\fill}}llrrrrr@{}}
\toprule
Parameter group & Model & \shortstack{Med.\\$|\mathrm{RB}|$} & \shortstack{Mean\\MSE} & \shortstack{Mean CP\\(\%)} & \shortstack{Median\\ESS} & \shortstack{Max.\\$\widehat R$}\\
\midrule
Drift $\boldsymbol\Gamma$ & CLOUD & \multicolumn{1}{c}{--} & 0.134 & \multicolumn{1}{c}{--} & 362.9 & 1.071 \\
 & TrendyOU & \multicolumn{1}{c}{--} & 0.159 & \multicolumn{1}{c}{--} & 312.0 & 1.098 \\
 & StationaryOU & \multicolumn{1}{c}{--} & 0.208 & \multicolumn{1}{c}{--} & 280.0 & 1.082 \\
 & DiagOU & \multicolumn{1}{c}{--} & 0.122 & \multicolumn{1}{c}{--} & 381.7 & 1.072 \\
\addlinespace[1.5pt]
Correlations $\boldsymbol\rho$ & CLOUD & 0.284 & 0.014 & 97.1 & 383.7 & 1.066 \\
 & TrendyOU & 1.966 & 0.111 & 35.7 & 555.5 & 1.066 \\
 & StationaryOU & 2.343 & 0.204 & 22.5 & 591.5 & 1.070 \\
 & DiagOU & 1.000 & 0.024 & 0.0 & -- & -- \\
\addlinespace[1.5pt]
Loadings $\boldsymbol\Lambda$ & CLOUD & \multicolumn{1}{c}{--} & 0.021 & \multicolumn{1}{c}{--} & 816.5 & 1.036 \\
 & TrendyOU & \multicolumn{1}{c}{--} & 0.080 & \multicolumn{1}{c}{--} & 864.5 & 1.075 \\
 & StationaryOU & \multicolumn{1}{c}{--} & 0.401 & \multicolumn{1}{c}{--} & 1129.5 & 1.022 \\
 & DiagOU & \multicolumn{1}{c}{--} & 0.013 & \multicolumn{1}{c}{--} & 1239.0 & 1.025 \\
\addlinespace[1.5pt]
Covariate slopes $\boldsymbol\Phi$ & CLOUD & 0.081 & 0.004 & 93.0 & 1573.2 & 1.038 \\
 & TrendyOU & 1.000 & 0.086 & 0.0 & -- & -- \\
 & StationaryOU & 1.000 & 0.086 & 0.0 & -- & -- \\
 & DiagOU & 0.027 & 0.002 & 95.8 & 2001.6 & 1.035 \\
\addlinespace[1.5pt]
Population slopes $\boldsymbol\alpha$ & CLOUD & 0.052 & 0.003 & 95.2 & 1061.2 & 1.026 \\
 & TrendyOU & 0.240 & 0.016 & 39.7 & 991.1 & 1.059 \\
 & StationaryOU & 1.000 & 0.135 & 0.0 & -- & -- \\
 & DiagOU & 0.018 & 0.002 & 96.6 & 1384.8 & 1.025 \\
\addlinespace[1.5pt]
Measurement effects $\boldsymbol\beta$ & CLOUD & 0.043 & 0.014 & 95.1 & 5410.2 & 1.003 \\
 & TrendyOU & 0.053 & 0.014 & 94.5 & 7438.9 & 1.003 \\
 & StationaryOU & 0.053 & 0.014 & 94.3 & 7081.2 & 1.003 \\
 & DiagOU & 0.045 & 0.013 & 95.1 & 9316.7 & 1.004 \\
\addlinespace[1.5pt]
Thresholds $\boldsymbol\theta$ & CLOUD & 0.018 & 0.051 & 94.3 & 1639.3 & 1.021 \\
 & TrendyOU & 0.019 & 0.055 & 93.7 & 2369.1 & 1.014 \\
 & StationaryOU & 0.319 & 0.755 & 12.7 & 1512.0 & 1.024 \\
 & DiagOU & 0.016 & 0.052 & 94.0 & 2995.1 & 1.012 \\
\addlinespace[1.5pt]
Random-effect scales $\boldsymbol\sigma_b$ & CLOUD & 0.048 & 0.140 & 81.8 & 1312.0 & 1.013 \\
 & TrendyOU & 0.052 & 0.183 & 75.5 & 1619.2 & 1.015 \\
 & StationaryOU & 0.057 & 0.182 & 76.5 & 1581.2 & 1.020 \\
 & DiagOU & 0.051 & 0.173 & 78.0 & 1581.3 & 1.028 \\
\bottomrule
\end{tabular*}
\begin{tablenotes}
\footnotesize
\item RB is computed only for nonzero data-generating values; the displayed value is the median absolute per-parameter RB within the group. Mean MSE and mean CP summarize all reported parameters in the group. For the 28 correlation parameters, only unique off-diagonal elements are included.
\item Median absolute RB and mean CP for drift and loading parameters are intentionally omitted because they already appear in Table~\ref{tab:simulation_summary} in the main text; their MSE, ESS, and convergence diagnostics remain here. Full parameter-level results appear in separate tables.
\item Structural zero restrictions are included at their implied values. Hence, TrendyOU and StationaryOU have $\boldsymbol{\Phi}=0$, StationaryOU has $\boldsymbol{\alpha}=0$, and DiagOU has zero off-diagonal dynamic dependence. Dashes denote unavailable sampling diagnostics for fixed parameters.
\end{tablenotes}
\end{threeparttable}
\end{table}
\begin{table}[p]
\centering
\caption{Additional grouped recovery diagnostics for S6 ($R=8$, oscillatory drift).}
\label{tab:r8_s6_summary}
\CloudStatTableStyle
\setlength{\tabcolsep}{2.3pt}
\begin{threeparttable}
\begin{tabular*}{\linewidth}{@{\extracolsep{\fill}}llrrrrr@{}}
\toprule
Parameter group & Model & \shortstack{Med.\\$|\mathrm{RB}|$} & \shortstack{Mean\\MSE} & \shortstack{Mean CP\\(\%)} & \shortstack{Median\\ESS} & \shortstack{Max.\\$\widehat R$}\\
\midrule
Drift $\boldsymbol\Gamma$ & CLOUD & \multicolumn{1}{c}{--} & 0.137 & \multicolumn{1}{c}{--} & 507.4 & 1.075 \\
 & TrendyOU & \multicolumn{1}{c}{--} & 0.471 & \multicolumn{1}{c}{--} & 375.2 & 1.100 \\
 & StationaryOU & \multicolumn{1}{c}{--} & 0.852 & \multicolumn{1}{c}{--} & 524.0 & 1.060 \\
 & DiagOU & \multicolumn{1}{c}{--} & 6.098 & \multicolumn{1}{c}{--} & 2295.7 & 1.094 \\
\addlinespace[1.5pt]
Correlations $\boldsymbol\rho$ & CLOUD & 0.291 & 0.007 & 95.1 & 675.0 & 1.051 \\
 & TrendyOU & 2.114 & 0.103 & 26.2 & 813.6 & 1.036 \\
 & StationaryOU & 2.508 & 0.203 & 21.0 & 1113.5 & 1.040 \\
 & DiagOU & 1.000 & 0.024 & 0.0 & -- & -- \\
\addlinespace[1.5pt]
Loadings $\boldsymbol\Lambda$ & CLOUD & \multicolumn{1}{c}{--} & 0.013 & \multicolumn{1}{c}{--} & 1338.7 & 1.019 \\
 & TrendyOU & \multicolumn{1}{c}{--} & 0.082 & \multicolumn{1}{c}{--} & 1016.9 & 1.037 \\
 & StationaryOU & \multicolumn{1}{c}{--} & 0.387 & \multicolumn{1}{c}{--} & 1494.2 & 1.012 \\
 & DiagOU & \multicolumn{1}{c}{--} & 0.014 & \multicolumn{1}{c}{--} & 1844.7 & 1.018 \\
\addlinespace[1.5pt]
Covariate slopes $\boldsymbol\Phi$ & CLOUD & 0.030 & 0.002 & 95.3 & 2042.0 & 1.010 \\
 & TrendyOU & 1.000 & 0.086 & 0.0 & -- & -- \\
 & StationaryOU & 1.000 & 0.086 & 0.0 & -- & -- \\
 & DiagOU & 0.035 & 0.002 & 96.0 & 3222.2 & 1.020 \\
\addlinespace[1.5pt]
Population slopes $\boldsymbol\alpha$ & CLOUD & 0.021 & 0.002 & 96.0 & 1957.0 & 1.015 \\
 & TrendyOU & 0.200 & 0.009 & 32.2 & 1475.1 & 1.037 \\
 & StationaryOU & 1.000 & 0.135 & 0.0 & -- & -- \\
 & DiagOU & 0.031 & 0.002 & 96.1 & 2905.8 & 1.038 \\
\addlinespace[1.5pt]
Measurement effects $\boldsymbol\beta$ & CLOUD & 0.051 & 0.014 & 94.9 & 4830.6 & 1.003 \\
 & TrendyOU & 0.052 & 0.014 & 94.4 & 5394.6 & 1.003 \\
 & StationaryOU & 0.077 & 0.015 & 94.5 & 7309.4 & 1.003 \\
 & DiagOU & 0.051 & 0.014 & 94.8 & 7754.4 & 1.005 \\
\addlinespace[1.5pt]
Thresholds $\boldsymbol\theta$ & CLOUD & 0.020 & 0.053 & 93.1 & 1343.2 & 1.026 \\
 & TrendyOU & 0.025 & 0.058 & 92.4 & 1635.7 & 1.020 \\
 & StationaryOU & 0.314 & 0.719 & 14.6 & 1849.4 & 1.014 \\
 & DiagOU & 0.021 & 0.056 & 92.9 & 2690.4 & 1.021 \\
\addlinespace[1.5pt]
Random-effect scales $\boldsymbol\sigma_b$ & CLOUD & 0.046 & 0.149 & 80.7 & 1334.2 & 1.015 \\
 & TrendyOU & 0.059 & 0.207 & 73.8 & 1332.8 & 1.023 \\
 & StationaryOU & 0.062 & 0.220 & 71.9 & 1620.7 & 1.008 \\
 & DiagOU & 0.052 & 0.180 & 75.8 & 1760.7 & 1.063 \\
\bottomrule
\end{tabular*}
\begin{tablenotes}
\footnotesize
\item RB is computed only for nonzero data-generating values; the displayed value is the median absolute per-parameter RB within the group. Mean MSE and mean CP summarize all reported parameters in the group. For the 28 correlation parameters, only unique off-diagonal elements are included.
\item Median absolute RB and mean CP for drift and loading parameters are intentionally omitted because they already appear in Table~\ref{tab:simulation_summary} in the main text; their MSE, ESS, and convergence diagnostics remain here. Full parameter-level results appear in separate tables.
\item Structural zero restrictions are included at their implied values. Hence, TrendyOU and StationaryOU have $\boldsymbol{\Phi}=0$, StationaryOU has $\boldsymbol{\alpha}=0$, and DiagOU has zero off-diagonal dynamic dependence. Dashes denote unavailable sampling diagnostics for fixed parameters.
\end{tablenotes}
\end{threeparttable}
\end{table}

\begin{table}[!htbp]
\centering
\caption{Additional computational diagnostics for the matched $R=2$ benchmark, complementing the target counts and median times in Table~\ref{tab:sim_computational_summary} in the main text. The 90th percentile of time to target is conditional on reaching the strict computational target. ``Sampling start'' is the median time at which all chains finished warm-up. The ratio is the paired LOU/CLOUD ratio of times to target. D/T/I counts denote nonattainment due to a diagnostic failure, the time limit, or infrastructure censoring, respectively.}
\label{tab:outcomes}
\CloudStatTableStyle
\begin{threeparttable}
\begin{tabular*}{\linewidth}{@{\extracolsep{\fill}}llccccc@{}}
\toprule
Scenario & Method & \shortstack{Target success\\Wilson 95\% CI} & \shortstack{Time to target\\p90 (h)} & \shortstack{Sampling start\\(h)} & \shortstack{Paired time-to-\\target ratio} & D/T/I \\
\midrule
Oscillatory & CLOUD & 96.3--100.0\% & 5.56 & 2.46 & -- & 0/0/0 \\
Oscillatory & LOU & 27.3--45.8\% & 29.32 & 5.53 & 4.11 [3.65, 5.15] & 13/51/0 \\
Nonoscillatory & CLOUD & 96.3--100.0\% & 4.54 & 1.92 & -- & 0/0/0 \\
Nonoscillatory & LOU & 45.2--64.4\% & 25.96 & 4.76 & 3.93 [3.29, 4.87] & 5/38/2 \\
\bottomrule
\end{tabular*}
\end{threeparttable}
\end{table}

\begin{table}[!htbp]
\centering
\caption{Matched-draw computational efficiency and HMC geometry at the common checkpoint after 1,000 post-warm-up draws ($n=100$ paired datasets per scenario). Continuous entries are medians across datasets; Supplementary Figure~\ref{fig:geometry} shows paired distributions for the principal efficiency and geometry ratios. ESS/s is the minimum bulk ESS divided by elapsed wall-clock seconds. Count rows report the number of fits with the indicated warning or event.}
\label{tab:geometry}
\CloudStatTableStyle
\begin{threeparttable}
\begin{tabular*}{\linewidth}{@{\extracolsep{\fill}}lcccc@{}}
\toprule
& \multicolumn{2}{c}{Oscillatory} & \multicolumn{2}{c}{Nonoscillatory}\\
\cmidrule(lr){2-3}\cmidrule(lr){4-5}
Metric & CLOUD & LOU & CLOUD & LOU\\
\midrule
\multicolumn{5}{l}{\emph{Efficiency and recovery}}\\
Elapsed time, h & 3.13 & 7.56 & 2.40 & 6.71\\
Minimum bulk ESS & 263 & 83 & 301 & 113\\
Minimum tail ESS & 495 & 218 & 549 & 267\\
$\Gamma$ RMSE & 0.0420 & 0.0422 & 0.0431 & 0.0427\\
\addlinespace
\multicolumn{5}{l}{\emph{HMC geometry diagnostics}}\\
Median leapfrog count & 127 & 511 & 127 & 511\\
Median step size & 0.0341 & 0.0289 & 0.0336 & 0.0294\\
Minimum E-BFMI & 0.500 & 0.111 & 0.503 & 0.151\\
Minimum E-BFMI $<0.3$, $n$/100 & 0 & 100 & 0 & 100\\
At least one divergence, $n$/100 & 0 & 30 & 0 & 15\\
At least one maximum-tree-depth hit, $n$/100 & 0 & 40 & 0 & 38\\
\bottomrule
\end{tabular*}
\end{threeparttable}
\end{table}

\begin{table}[!htbp]
\centering
\caption{Parameter recovery among the 91 paired datasets for which both methods succeeded, pooled across the two $R=2$ scenarios. RMSE entries are the median [interquartile range] of per-dataset RMSE. CP is the pooled empirical 95\% credible-interval coverage over the corresponding parameter--replicate combinations. For $\Omega$, only the unique off-diagonal correlation is used, avoiding duplicate symmetric entries and fixed unit diagonals.}
\label{tab:recovery}
\CloudStatTableStyle
\begin{threeparttable}
\begin{tabular*}{\linewidth}{@{\extracolsep{\fill}}lcccc@{}}
\toprule
Parameter group & CLOUD RMSE & LOU RMSE & CLOUD CP & LOU CP\\
\midrule
$\Gamma$ & 0.037 [0.024, 0.053] & 0.036 [0.024, 0.052] & 95.6\% & 95.3\%\\
$\Omega$ & 0.031 [0.011, 0.060] & 0.032 [0.011, 0.058] & 95.6\% & 96.7\%\\
Factor loadings & 0.409 [0.333, 0.581] & 0.408 [0.332, 0.580] & 95.0\% & 94.7\%\\
Measurement-covariate effects & 0.142 [0.115, 0.162] & 0.142 [0.115, 0.163] & 95.4\% & 95.2\%\\
Thresholds & 0.333 [0.262, 0.399] & 0.333 [0.264, 0.397] & 95.0\% & 95.1\%\\
Random-effect SDs & 0.408 [0.306, 0.504] & 0.407 [0.305, 0.504] & 94.2\% & 94.5\%\\
\bottomrule
\end{tabular*}
\end{threeparttable}
\end{table}

\begin{table}[!p]
\centering
\caption{Mapping of MDS-UPDRS Part III item scores to prespecified domains.}
\label{tab:ppmi_item_domain_mapping}
\CloudStatTableStyle
\begin{threeparttable}
\begin{tabular*}{\linewidth}{@{\extracolsep{\fill}}p{0.20\linewidth}rp{0.65\linewidth}@{}}
\toprule
Domain & Scores & Original PPMI data-dictionary item codes \\
\midrule
Axial/midline & 8 & NP3SPCH, NP3FACXP, NP3RISNG, NP3GAIT, NP3FRZGT, NP3PSTBL, NP3POSTR, NP3BRADY \\
Left upper limb & 3 & NP3FTAPL, NP3HMOVL, NP3PRSPL \\
Lower limb & 4 & NP3TTAPR, NP3TTAPL, NP3LGAGR, NP3LGAGL \\
Postural/kinetic tremor & 4 & NP3PTRMR, NP3PTRML, NP3KTRMR, NP3KTRML \\
Rest tremor & 6 & NP3RTCON, NP3RTALJ, NP3RTARU, NP3RTALU, NP3RTARL, NP3RTALL \\
Right upper limb & 3 & NP3FTAPR, NP3HMOVR, NP3PRSPR \\
Rigidity & 5 & NP3RIGN, NP3RIGRU, NP3RIGLU, NP3RIGRL, NP3RIGLL \\
\bottomrule
\end{tabular*}
\begin{tablenotes}
\footnotesize
\item \emph{Notes.} The partition contains 33 ordinal item scores, each on the 0--4 MDS-UPDRS scale. These domains follow a clinically motivated seven-factor interpretation.
\end{tablenotes}
\end{threeparttable}
\end{table}
\clearpage
\begin{table}[!p]
\centering
\caption{Measurement alignment for the four OU specifications.}
\label{tab:ppmi_measurement_alignment}
\CloudStatTableStyle
\begin{threeparttable}
\begin{tabular*}{\linewidth}{@{\extracolsep{\fill}}lrrr@{}}
\toprule
Model & Loading rank correlation $\rho_S$ & Median loading ratio & Median cutpoint difference \\
\midrule
CLOUD & 1.000 & 1.000 & 0.000 \\
TrendyOU & 0.997 & 1.006 & 0.024 \\
StationaryOU & 0.996 & 1.026 & 0.452 \\
DiagOU & 0.970 & 1.013 & 0.203 \\
\bottomrule
\end{tabular*}
\begin{tablenotes}
\footnotesize
\item \emph{Notes.} Loading rank correlation $\rho_S$ is the Spearman rank correlation between each comparator's item loadings and the CLOUD loadings; the median loading ratio and median cutpoint difference are relative to CLOUD. These diagnostics summarize measurement comparability across temporal specifications within the prespecified factor structure.
\end{tablenotes}
\end{threeparttable}
\end{table}
\clearpage
\begin{table}[!p]
\centering
\caption{Annual within-participant changes in observed motor-domain scores in the training set.}
\label{tab:ppmi_annual_change}
\CloudStatTableStyle
\begin{threeparttable}
\begin{tabular*}{\linewidth}{@{\extracolsep{\fill}}lrlrr@{}}
\toprule
Domain & Points/year & 95\% CI & BH q & Participants \\
\midrule
Axial/midline & +0.377 & [+0.299, +0.455] & 8.92e-20 & 1201 \\
Rigidity & +0.188 & [+0.120, +0.256] & 1.86e-07 & 1202 \\
Lower limb & +0.240 & [+0.175, +0.305] & 2.46e-12 & 1203 \\
Left upper limb & +0.145 & [+0.091, +0.199] & 3.39e-07 & 1204 \\
Right upper limb & +0.054 & [+0.001, +0.107] & 0.0517 & 1204 \\
Rest tremor & -0.018 & [-0.086, +0.050] & 0.609 & 1204 \\
Postural/kinetic tremor & -0.066 & [-0.106, -0.026] & 0.00174 & 1204 \\
\bottomrule
\end{tabular*}
\begin{tablenotes}
\footnotesize
\item \emph{Notes.} Participant fixed-effects regressions were fitted, and participant-clustered confidence intervals were computed using a first-order finite-sample-corrected cluster-robust (CR1) variance estimator. BH $q$ is a Benjamini--Hochberg adjusted $q$ value. Native domain-score sums require complete observed domain assessments. These are time-variation-weighted linear associations, not dropout-adjusted population trends.
\end{tablenotes}
\end{threeparttable}
\end{table}
\clearpage
\begin{table}[!p]
\centering
\caption{Observed-score age-by-time and RBDSQ-by-time progression associations in the training set (distinct from the CLOUD latent contrasts in Figure~\ref{fig:ppmi_progression}B in the main text).}
\label{tab:ppmi_empirical_covariates}
\CloudStatTableStyle
\begin{threeparttable}
\begin{tabular*}{\linewidth}{@{\extracolsep{\fill}}llrlr@{}}
\toprule
Domain & Covariate & Observed effect & Observed 95\% CI & BH $q$ \\
\midrule
Axial/midline & Age & +0.005 & [-0.009, +0.019] & 0.515 \\
Axial/midline & RBDSQ & +0.019 & [+0.006, +0.033] & 0.050278 \\
Rigidity & Age & +0.010 & [-0.007, +0.026] & 0.409 \\
Rigidity & RBDSQ & +0.016 & [+0.00035, +0.032] & 0.224 \\
Lower limb & Age & -0.011 & [-0.030, +0.007] & 0.409 \\
Lower limb & RBDSQ & +0.018 & [+0.00015, +0.036] & 0.224 \\
Left upper limb & Age & -0.011 & [-0.031, +0.009] & 0.409 \\
Left upper limb & RBDSQ & +0.011 & [-0.008, +0.030] & 0.409 \\
Right upper limb & Age & -0.003 & [-0.022, +0.017] & 0.776 \\
Right upper limb & RBDSQ & +0.011 & [-0.009, +0.030] & 0.409 \\
Rest tremor & Age & +0.005 & [-0.006, +0.017] & 0.464 \\
Rest tremor & RBDSQ & +0.007 & [-0.004, +0.017] & 0.409 \\
Postural/kinetic tremor & Age & -0.008 & [-0.019, +0.003] & 0.409 \\
Postural/kinetic tremor & RBDSQ & -0.004 & [-0.014, +0.007] & 0.515 \\
\bottomrule
\end{tabular*}
\begin{tablenotes}
\footnotesize
\item \emph{Notes.} RBDSQ denotes the Rapid Eye Movement Sleep Behavior Disorder Screening Questionnaire. BH denotes Benjamini--Hochberg adjustment. Observed coefficients are conditional changes in mean item severity within each domain (0--4) per year per standardized covariate unit. All 14 observed BH-adjusted $q$ values exceed 0.05; the closest is 0.050278 (axial/midline RBDSQ), which rounds to 0.050 to three decimal places. Coefficients are the fitted point estimates, with 95\% participant-clustered first-order cluster-robust (CR1) $t$ confidence intervals. Interval endpoints are reported with sufficient precision to distinguish values close to zero.
\end{tablenotes}
\end{threeparttable}
\end{table}
\clearpage
\begin{table}[!p]
\centering
\caption{Inferential details for the 22 selected one-year associations in Figure~\ref{fig:ppmi_interactions} in the main text; the point estimates shown there are not repeated.}
\label{tab:ppmi_selected_associations}
\CloudStatTableStyle
\begin{threeparttable}
\begin{tabular*}{\linewidth}{@{\extracolsep{\fill}}llrc@{}}
\toprule
Source $\to$ target & CLOUD 99\% CrI & Empirical BH $q$ & Sign agreement \\
\midrule
Ax $\to$ LL & [+0.185, +0.346] & 2.09e-10 & Y \\
LUL $\to$ LL & [+0.164, +0.343] & 1.28e-12 & Y \\
RUL $\to$ LUL & [-0.197, -0.126] & 3.69e-10 & Y \\
LUL $\to$ RUL & [-0.184, -0.076] & 5.35e-11 & Y \\
Ax $\to$ Rig & [+0.058, +0.207] & 2.6e-07 & Y \\
LUL $\to$ PKT & [+0.048, +0.193] & 0.106 & Y \\
LUL $\to$ Rig & [+0.062, +0.191] & 0.000173 & Y \\
Ax $\to$ RUL & [+0.061, +0.177] & 2.68e-10 & Y \\
RT $\to$ Ax & [-0.146, -0.084] & 2.57e-06 & Y \\
Ax $\to$ LUL & [+0.049, +0.163] & 2.01e-10 & Y \\
LL $\to$ RUL & [-0.170, -0.025] & 0.644 & Y \\
RT $\to$ Rig & [-0.142, -0.059] & 0.0059 & Y \\
RT $\to$ LL & [-0.140, -0.042] & 0.0234 & Y \\
LL $\to$ Rig & [-0.176, -0.009] & 0.142 & Y \\
Ax $\to$ RT & [-0.133, -0.029] & 0.000698 & Y \\
RT $\to$ RUL & [-0.110, -0.039] & 0.007 & Y \\
LL $\to$ LUL & [-0.132, -0.003] & 0.106 & N \\
RT $\to$ LUL & [-0.107, -0.033] & 5.32e-05 & Y \\
PKT $\to$ RT & [-0.107, -0.031] & 0.0271 & N \\
LL $\to$ Ax & [-0.128, -0.000] & 0.297 & N \\
PKT $\to$ Rig & [+0.002, +0.109] & 0.155 & Y \\
RT $\to$ PKT & [+0.009, +0.104] & 1.31e-15 & Y \\
\bottomrule
\end{tabular*}
\begin{tablenotes}
\footnotesize
\item \emph{Notes.} Figure~\ref{fig:ppmi_interactions} in the main text displays the latent and empirical point estimates; this table provides their complementary uncertainty intervals and multiplicity-adjusted evidence. CrI denotes a posterior credible interval; BH $q$ denotes the Benjamini--Hochberg multiplicity-adjusted significance value. The latent transitions and empirical annualized changes remain on different scales. Nineteen of 22 associations agree in sign; 15 of 22 also meet the empirical BH false-discovery-rate (FDR) criterion at 0.05. Abbreviations: Ax, axial/midline; Rig, rigidity; LL, lower limb; LUL/RUL, left/right upper limb; RT, rest tremor; PKT, postural/kinetic tremor. Both summaries use the same development cohort and describe conditional associations.
\end{tablenotes}
\end{threeparttable}
\end{table}
\clearpage
\begin{table}[!p]
\centering
\caption{Longitudinal mean and dependence structures in the eight PPMI forecast models.}
\label{tab:ppmi_v174_models}
\CloudStatTableStyle
\begin{threeparttable}
\begin{tabular*}{\linewidth}{@{\extracolsep{\fill}}p{0.16\linewidth}p{0.36\linewidth}p{0.37\linewidth}@{}}
\toprule
Model & Longitudinal mean or latent trajectory & Dependence and comparison role \\
\midrule
CLOUD & Covariate-dependent latent mean $t(\boldsymbol\alpha+\boldsymbol\Phi\mathbf x_i^{(2)})$ & Coupled OU drift and stationary correlation; joint latent dynamics \\
TrendyOU & Population trend $t\boldsymbol\alpha$ & Coupled OU drift and correlation; no covariate-dependent slope \\
StationaryOU & Zero moving latent mean & Coupled OU drift and correlation; zero latent mean slope \\
DiagOU & Population trend $t\boldsymbol\alpha$ & Independent centered OU residuals; also omits covariate effects on the mean in PPMI \\
GP-IRT & Seven-domain Mat\'ern-3/2 latent Gaussian processes & Smooth temporal latent process without parametric OU drift \\
MLTLMM & Seven anchored latent random intercepts and slopes & Correlated subject-level latent coefficients \\
GLMM & Ordinal item-specific fixed effects and correlated subject intercepts and slopes & Cumulative-logit mixed-effects regression \\
GAMM & GLMM-style subject effects plus penalized nonlinear time functions & Cumulative-logit additive mixed-effects regression \\
\bottomrule
\end{tabular*}
\begin{tablenotes}
\footnotesize\item \emph{Notes.} In PPMI, DiagOU simultaneously removes cross-domain residual dependence and covariate-dependent mean slopes. All models produce five-category predictive distributions for the same held-out item outcomes.
\end{tablenotes}
\end{threeparttable}
\end{table}
\clearpage
\begin{table}[!p]
\centering
\caption{Participant-paired ranked probability score differences for third-occasion, approximately one-year-ahead forecasts (CLOUD minus comparator), supplementing Table~\ref{tab:ppmi_forecast} in the main text.}
\label{tab:ppmi_rps_differences}
\CloudStatTableStyle
\begin{threeparttable}
\begin{tabular*}{\linewidth}{@{\extracolsep{\fill}}lrl@{}}
\toprule
Comparator & $\Delta$RPS & 95\% bootstrap CI \\
\midrule
TrendyOU & +0.00004 & [-0.00020, +0.00029] \\
StationaryOU & +0.00023 & [-0.00020, +0.00068] \\
DiagOU & -0.00017 & [-0.00076, +0.00041] \\
MLTLMM & -0.00881 & [-0.01042, -0.00723] \\
GP-IRT & -0.00400 & [-0.00517, -0.00280] \\
GLMM & -0.00003 & [-0.00097, +0.00088] \\
GAMM & -0.00000 & [-0.00098, +0.00100] \\
\bottomrule
\end{tabular*}
\begin{tablenotes}
\footnotesize
\item \emph{Notes.} Negative RPS differences favor CLOUD; all OU-variant intervals include zero. The paired NLPD differences and intervals are reported exclusively in Table~\ref{tab:ppmi_forecast} in the main text. Forecast comparisons are exploratory because test outcomes had been inspected previously.
\end{tablenotes}
\end{threeparttable}
\end{table}
\clearpage
\begin{table}[!p]
\centering
\caption{Differences in pooled ordinal quadratic weighted kappa (CLOUD minus comparator).}
\label{tab:ppmi_kappa_differences}
\CloudStatTableStyle
\begin{threeparttable}
\begin{tabular*}{\linewidth}{@{\extracolsep{\fill}}lrl@{}}
\toprule
Comparator & Paired $\Delta\kappa_w$ & 95\% bootstrap CI \\
\midrule
TrendyOU & -0.0002 & [-0.0050, +0.0041] \\
StationaryOU & +0.0121 & [+0.0063, +0.0182] \\
DiagOU & +0.0117 & [+0.0028, +0.0214] \\
MLTLMM & +0.0888 & [+0.0686, +0.1104] \\
GP-IRT & +0.0609 & [+0.0461, +0.0754] \\
GLMM & -0.0170 & [-0.0291, -0.0041] \\
GAMM & -0.0198 & [-0.0321, -0.0065] \\
\bottomrule
\end{tabular*}
\begin{tablenotes}
\footnotesize
\item \emph{Notes.} Kappa is calculated from pooled modal category decisions; probabilistic scores use participant-balanced aggregation of item-level values. Gains over StationaryOU and DiagOU are statistically distinguishable but small; GLMM and GAMM have higher kappa values.
\end{tablenotes}
\end{threeparttable}
\end{table}
\clearpage

\begin{figure}[htbp]
\centering
\includegraphics[width=0.95\textwidth]{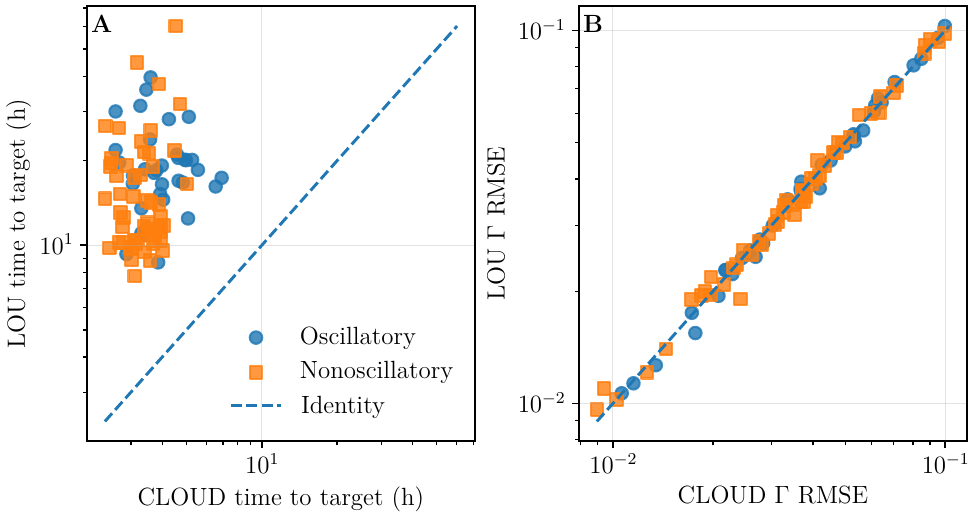}
\caption{Paired comparison restricted to the 91 datasets for which both methods reached the sampling-precision criterion (36 oscillatory and 55 nonoscillatory). Panel A plots LOU versus CLOUD time to target; every point lies above the identity line. Panel B plots per-dataset drift-matrix root mean squared error (RMSE); the points are tightly concentrated around the identity line, showing comparable drift recovery conditional on successful computation. Both axes are logarithmic.}
\label{fig:common}
\end{figure}

\begin{figure}[htbp]
\centering
\includegraphics[width=0.92\textwidth]{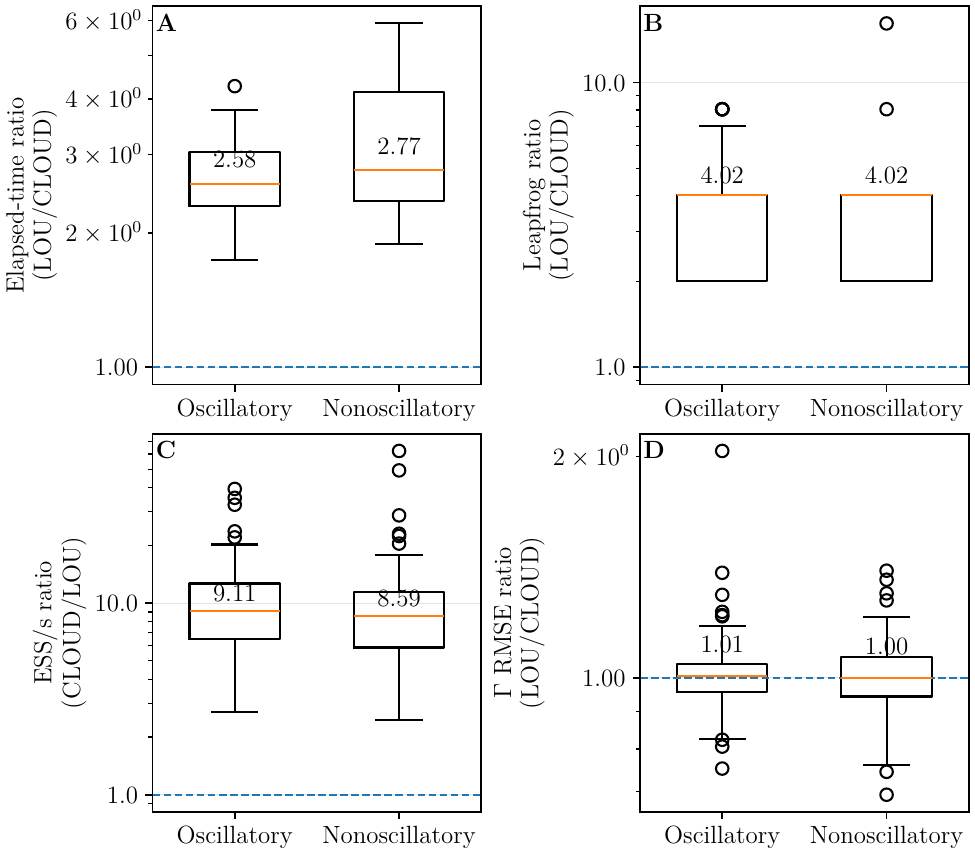}
\caption{Paired computational comparison at the common checkpoint after 1,000 post-warm-up draws. Panels (A)--(D) show the LOU/CLOUD elapsed-time ratio, LOU/CLOUD ratio of median leapfrog counts, CLOUD/LOU ratio of minimum bulk effective sample size (ESS) per second, and LOU/CLOUD drift-matrix root mean squared error (RMSE) ratio, respectively. Values above one favor CLOUD in panels A--C; a drift-RMSE ratio near one in D indicates comparable estimation accuracy. Reference lines mark one, and the printed values are medians; axes use logarithmic scales. Additional diagnostics, including bulk ESS and energy information, are reported in Supplementary Table~\ref{tab:geometry}.}
\label{fig:geometry}
\end{figure}

\begin{figure}[p]
\centering
\includegraphics[width=0.96\textwidth,height=0.79\textheight,keepaspectratio]{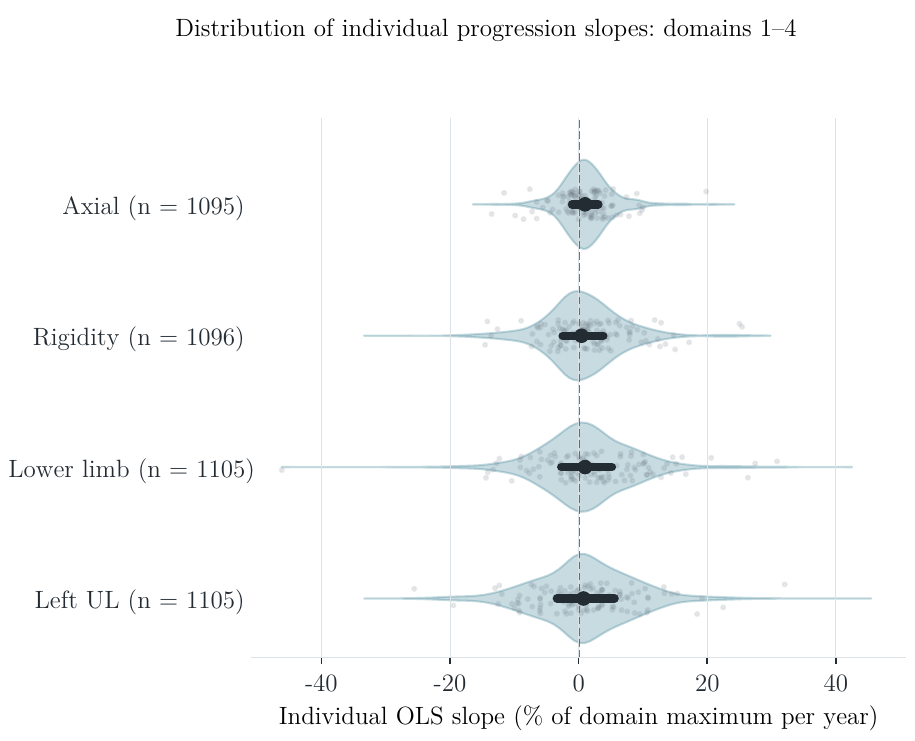}
\caption{Distributions of participant-specific annualized slopes in observed domain scores for axial/midline, rigidity, lower limb, and left upper limb in the development set. Slopes were estimated by ordinary least squares (OLS) and are expressed as percentages of the corresponding domain maximum per year.}
\label{fig:ppmi_web_slope_dist1}
\end{figure}

\begin{figure}[p]
\centering
\includegraphics[width=0.96\textwidth,height=0.79\textheight,keepaspectratio]{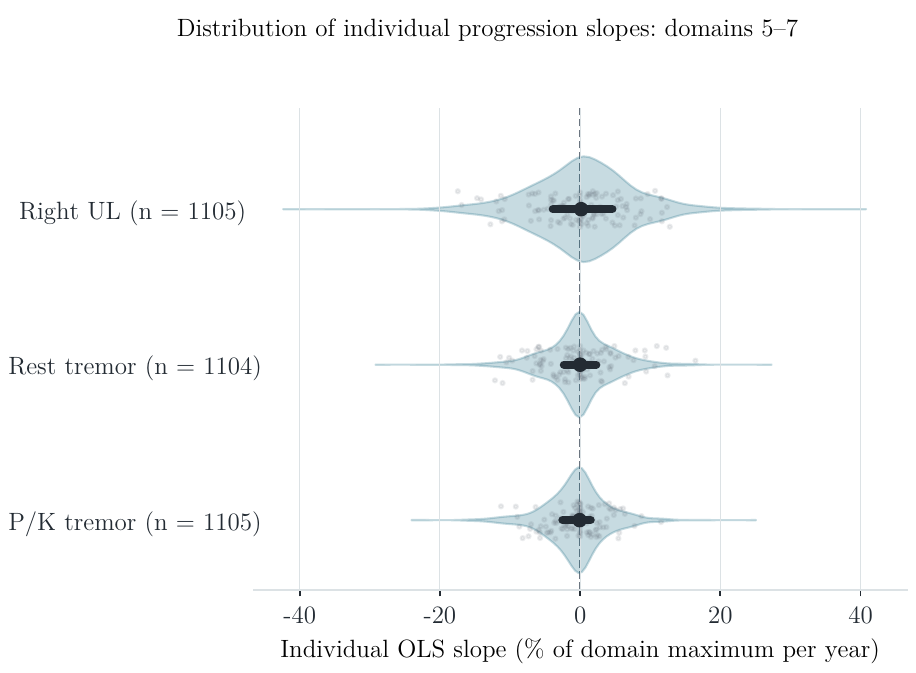}
\caption{Distributions of participant-specific annualized slopes in observed domain scores for right upper limb, rest tremor, and postural/kinetic tremor in the development set. Slopes were estimated by ordinary least squares (OLS) and are expressed as percentages of the corresponding domain maximum per year.}
\label{fig:ppmi_web_slope_dist2}
\end{figure}

\begin{figure}[p]
\centering
\includegraphics[width=0.96\textwidth,height=0.79\textheight,keepaspectratio]{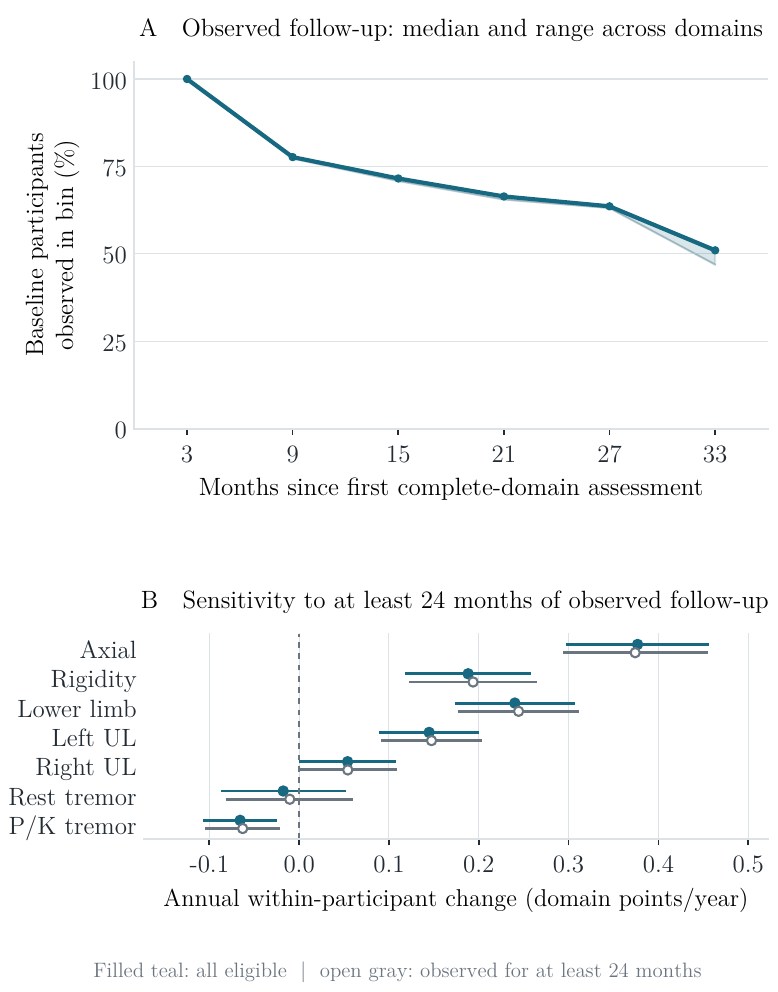}
\caption{Follow-up composition over time and sensitivity to an observed follow-up span of 24 months or longer.}
\label{fig:ppmi_web_followup_composition}
\end{figure}

\begin{figure}[p]
\centering
\includegraphics[width=0.96\textwidth,height=0.79\textheight,keepaspectratio]{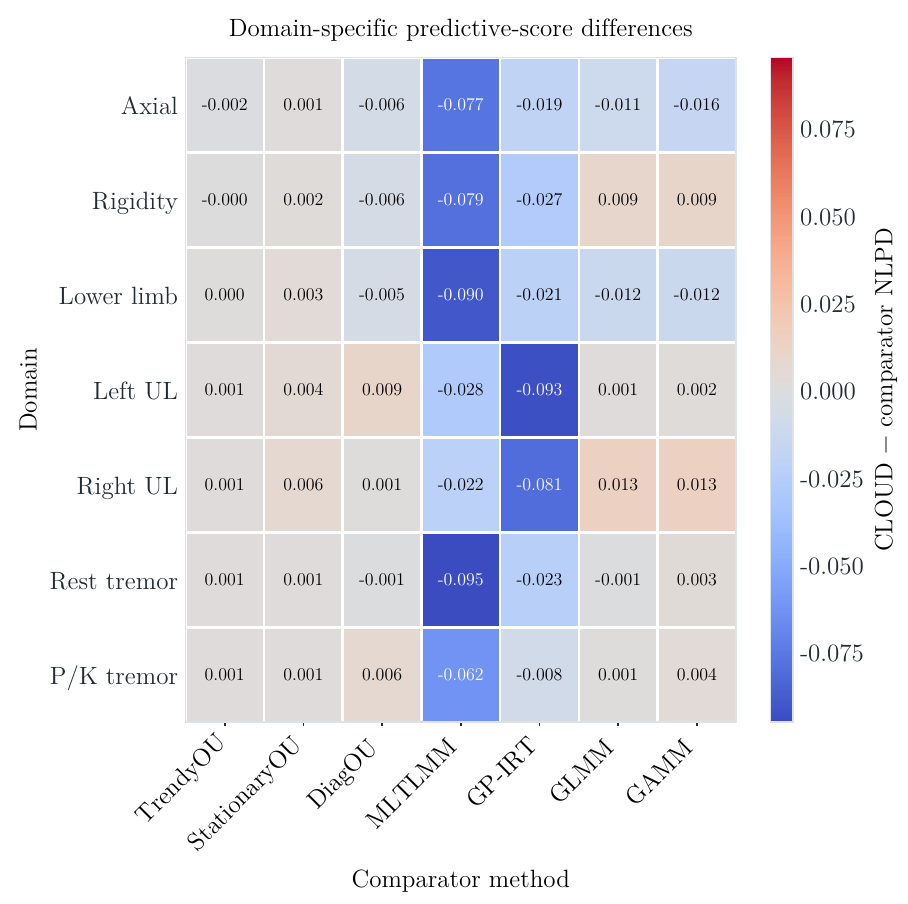}
\caption{Domain-specific held-out negative log predictive density (NLPD) differences (CLOUD minus comparator) across all seven motor domains and seven benchmark models. Negative values favor CLOUD; positive values favor the comparator. TrendyOU has a population latent mean slope; StationaryOU has zero latent mean slope; DiagOU has independent centered OU residuals. UL denotes upper limb; P/K denotes postural/kinetic.}
\label{fig:ppmi_web_domain_prediction1}
\end{figure}

\end{document}